\documentclass[twoside,12pt]{amsart}

\usepackage{amsmath,amssymb,amsrefs}
\usepackage{graphicx}
\usepackage{bm}
\usepackage{hyperref} 
\usepackage[mathscr]{euscript}
\usepackage{arydshln}
\usepackage{enumerate}

\numberwithin{equation}{section}

\newtheorem{theorem}{Theorem}[section]
\newtheorem{corollary}[theorem]{Corollary}
\newtheorem{proposition}[theorem]{Proposition}
\newtheorem{lemma}[theorem]{Lemma}
\theoremstyle{definition}
\newtheorem{definition}[theorem]{Definition}
\theoremstyle{remark}
\newtheorem{remark}[theorem]{Remark}

{\unskip\nobreak\hskip 1em plus 1fil\nobreak$\Box$
\parfillskip=0pt%
\endtrivlist}

\newcommand{\R}{\mathbb{R}}
\newcommand{\C}{\mathbb{C}}
\newcommand{\Z}{\mathbb{Z}}

\newcommand{\B}{\mathcal{B}}
\newcommand{\brill}{{\mathcal{B}}}

\newcommand{\be}{{\bf e}}
\newcommand{\bk}{{\bf k}}
\newcommand{\bfm}{{\bf m}}
\newcommand{\bn}{{\bf n}}

\newcommand{\tK}{{\tilde{\bf K}}}
\newcommand{\bK}{{\bf K}}
\newcommand{\bKp}{{\bf K'}}
\newcommand{\bv}{{\bf v}}
\newcommand{\bw}{{\bf w}}
\newcommand{\bx}{{\bf x}}
\newcommand{\bz}{{\bf z}}

\newcommand{\by}{{\bf y}}
\newcommand{\bA}{{\bf A}}

\newcommand{\bkappa}{{\bf \kappa}}

\newcommand{\Honeycomb}{{\bf H}}

\newcommand{\vtilde}{{\bm{\mathfrak{v}}}}
\newcommand{\ktilde}{{\bm{\mathfrak{K}}}}

\newcommand{\kpar}{{k_{\parallel}}}

\newcommand{\inner}[1]{\left\langle#1\right\rangle}
\newcommand{\D}{\partial}

\newcommand{\eps}{\varepsilon}
\newcommand{\nit}{\noindent}
\newcommand{\nn}{\nonumber}

\newcommand{\lamsharp}{{v_{_F}}}
\newcommand{\lamsharplam}{{v^\lambda_{_F}}}

\newcommand{\supp}{\text{supp}}

\begin{document}

\title{Honeycomb Schr\"odinger operators\\ 
 in the strong binding regime}

\author{C. L. Fefferman}
\address{Department of Mathematics, Princeton University, Princeton, NJ, USA}
\email{cf@math.princeton.edu}

\author{J. P. Lee-Thorp}
\address{Courant Institute of Mathematical Sciences, New York University, New York, NY, USA}
\email{leethorp@cims.nyu.edu}

\author{M. I. Weinstein}
\address{Department of Applied Physics and Applied Mathematics and Department of Mathematics, Columbia University, New York, NY, USA}
\email{miw2103@columbia.edu}

\date{\today}

\keywords{Schr\"odinger equation,  Dirac point, Floquet-Bloch spectrum, Topological insulator,  Spectral gap, Edge state, Honeycomb lattice, strong binding regime, tight binding limit}

\begin{abstract} In this article, we study the Schr\"odinger operator for a large class of periodic potentials with the symmetry of a hexagonal tiling of the plane. The potentials we consider are  superpositions of localized potential wells, centered on the vertices of a regular  honeycomb structure corresponding to the single electron model of graphene and its artificial analogues.
We consider this Schr\"odinger operator in the regime of strong binding, where the depth of the potential wells is large. Our main result is that for sufficiently deep potentials, the lowest two Floquet-Bloch dispersion surfaces, when appropriately rescaled, converge uniformly to those of the two-band tight-binding model (Wallace, 1947 \cites{Wallace:47}). 
Furthermore, we establish as corollaries,  in the regime of strong binding, results on (a) the existence of spectral gaps for honeycomb potentials that break $\mathcal{P}\mathcal{T}$ symmetry and (b) the existence of topologically protected edge states -- states which propagate parallel to and are localized transverse to a line-defect or ``edge'' - for a large class of rational edges, and which are robust to a class of  large transverse-localized perturbations of the edge.  
We believe that the ideas of this article may be applicable in other settings for which a tight-binding model emerges in an extreme parameter limit.
\end{abstract}

\maketitle   

\pagestyle{myheadings}
\thispagestyle{plain}
\markboth{Honeycomb Schr\"odinger operators in the strong binding regime }{C.L. Fefferman, J.P. Lee-Thorp, M.I. Weinstein}

\section{Introduction}

 In this article, we study the Schr\"odinger operator, $-\Delta + V$, for a   large class of periodic potentials with the symmetry of a hexagonal tiling of the plane. The potentials we consider are  superpositions of atomic localized potential wells, $V_0$, supported in discs centered 
on the vertices of a regular honeycomb structure corresponding to the single electron model of graphene
 and to its artificial analogues.

We consider this Schr\"odinger operator in the regime of strong binding, where the depth of the potential is large. Our main result is that for sufficiently deep potentials,  the lowest two Floquet-Bloch dispersion surfaces, when appropriately rescaled, converge uniformly to those of the two-band tight-binding model, introduced by P. R. Wallace in 1947 in his pioneering study of graphite \cites{Wallace:47}. 
Furthermore, our main results, together with previous results in \cites{FW:12} and \cites{FLW-2d_edge:16},  yield:\\
 (a) results on the existence of spectral gaps for Schr\"odinger operators with honeycomb potentials, perturbed in such a way as to   break $\mathcal{P}\mathcal{T}$ symmetry (the composition of parity-inversion and time-reversal symmetries), and\\
  (b) results on the existence of {\it topologically protected edge states} for Schr\"odinger operators with honeycomb potentials perturbed by a class of line-defects or {\it edges}, assumed to be parallel to vectors in the underlying period lattice.
  
 Spectral gaps play a central role in energy transport properties of crystalline media. 
 Edge states are time-harmonic solutions which are plane-wave-like (propagating) parallel to the edge and localized transverse to the edge.
Topologically protected edge states, due to their immunity against strong perturbations, have potential as a highly robust  means of energy transport. 

We comment briefly on terminology. An edge is frequently understood to mean an abrupt termination of bulk structure.
The terms ``edge'' for a line-defect across which there is a change in a key characteristic of the structure, and ``edge state'' are also used in the physics literature; see, for example, \cites{HR:07, RH:08,Shvets-PTI:13}. 
The edge states we discuss are of the latter type.
In particular, our edge states in structures with a domain wall defect are localized transverse to a line in the direction of a period lattice vector, a rational ``edge''.  
In this paper, {\it topological protection}  refers to the stability of bifurcations of edge states from Dirac points (a bifurcation from the intersection of continuous spectral bands) against a class of transverse-localized (even large) perturbations of the Hamiltonian. Although there is evidence from tight-binding models and numerical simulations of continuum PDE models of 
stability against fully localized perturbations, a precise mathematical theory is an open problem. 

Finally, we believe that the ideas of this article may be applicable in
other settings for which a tight-binding model emerges in an extreme parameter limit.

\subsection{Graphene and its artificial analogues - physical motivation}

Graphene is a two-dimensional material consisting of a single atomic layer of carbon atoms arranged in a regular honeycomb structure. 
It has been a  subject of intense interest  and exploration by  the fundamental and applied scientific, and engineering communities 
since its experimental fabrication and study in the middle of  the last decade \cites{geim2007rise,Kim-etal:05}. Many of graphene{'}s  novel electronic properties are related to conical intersections of its dispersion surfaces (Dirac points) and the corresponding effective Dirac (massless Fermionic) dynamics of wave-packets. These properties  can be understood by considering 
the band structure near the Fermi level for a Hamiltonian which only incorporates the $\pi-$ electrons \cites{geim2007rise,RMP-Graphene:09,FW:14}.
In this approximate model, the band structure is that of the two-dimensional Schr\"odinger operator with a  honeycomb lattice potential. 

Since many of  graphene{'}s properties are related to quantum mechanical problems governed by a class of energy-conserving wave equations in a medium with special symmetries, wave systems of this general type, in other physical settings, {\it e.g.} electronic, optical, acoustic, have received a great deal of recent attention by theorists and experimentalists. These have been dubbed {\it artificial-graphene} and have been explored, for example, in electronic physics \cites{artificial-graphene:11},  photonics \cites{Rechtsman-etal:13,plotnik2013observation,BKMM_prl:13} and acoustics \cites{khanikaev-acoustic-graphene:15}.

One such property, observed in electronic and photonic systems with honeycomb symmetry, is the existence of 
topologically protected {\it edge states}.   Edge states are modes which are
(i) pseudo-periodic (plane-wave-like  or propagating) parallel to a line-defect, and (ii)  localized transverse to the line-defect; see the schematic in Figure \ref{fig:mode_schematic}.
{\it Topological protection}, refers to the persistence of these modes and their properties, 
even when the line-defect is subjected to strong local perturbations. 
 In applications, edge states are of great interest due to their potential as robust vehicles for channeling energy. 

The extensive physics literature on topologically robust edge states goes back to investigations of the quantum Hall effect; see, for example, \cites{H:82, TKNN:82, Hatsugai:93, wen1995topological} and the rigorous mathematical articles \cites{Macris-Martin-Pule:99,EG:02, EGS:05,Taarabt:14}.  
In \cites{HR:07, RH:08} a proposal for realizing {\it photonic edge states}  in periodic electromagnetic structures which exhibit the magneto-optic effect was made. In this case, the edge is realized via a domain wall across which 
the Faraday axis is reversed. 
Since the magneto-optic effect breaks time-reversal symmetry, as does the magnetic field in the Hall effect, the resulting edge states are unidirectional.

Other realizations of edges in photonic and electromagnetic systems, {\it e.g.} between periodic dielectric and conducting structures, between periodic structures and free-space, have been explored through experiment and numerical simulation; see, for example \cites{Soljacic-etal:08,Fan-etal:08,Rechtsman-etal:13a,Shvets-PTI:13,Shvets:14}.

 The prevalent approaches to the theoretical study of these systems are: the tight-binding (discrete) approximation
  (see, for example, \cites{RMP-Graphene:09,MPFG:09}), the nearly free-electron (or free-photon) approximation (see, for example, \cites{HR:07,RH:08}) or direct numerical simulation (see, for example, \cites{bahat2008symmetry}). 
   In the tight-binding approximation, wave functions (Floquet-Bloch modes) are approximated by superpositions of local ground states of deep (high-contrast) potential wells, each of whose amplitudes interacts weakly with its nearest neighbors. In the nearly free-electron approximation, the potential in the Schr\"odinger operator  is treated as a small (low-contrast) perturbation of the Laplacian; see, for example,  \cites{Ashcroft-Mermin:76}. 
 
 Analytical results on the behavior of dispersion surfaces of Schr\"odinger operators with {\it generic} honeycomb  lattice potentials, which are not limited to these approximations in that there are no assumptions on the size of the potential, were obtained in \cites{FW:12,FLW-MAMS:17} using bifurcation theory  from the nearly free-electron limit, combined with methods of complex analysis to extend the analysis globally in the contrast (coupling) parameter. The results of the present article concern Schr\"odinger equations in the {\it strong binding regime} (deep or high-contrast potentials / strong coupling) and its relation to the {\it tight-binding (discrete) limit}. Before outlining our results we discuss the celebrated  two-band  tight-binding model of  Wallace (1947) \cites{Wallace:47}. 
  
\subsection{Wallace{'}s two-band tight-binding model of graphite}\label{wallace-47}
 The (regular) honeycomb structure, $\Honeycomb$, is the union of two interpenetrating equilateral triangular lattices, $\Lambda_A$ and $\Lambda_B$, where $\Lambda_A=\bv_A+\Lambda_h$, $\Lambda_B=\bv_B+\Lambda_h$ and 
 $\Lambda_h=\Z\bv_1\oplus\Z\bv_2$; see Figure \ref{lattice-and-itsdual}.
 \begin{figure}
\centering 
\includegraphics[width=0.75\textwidth]{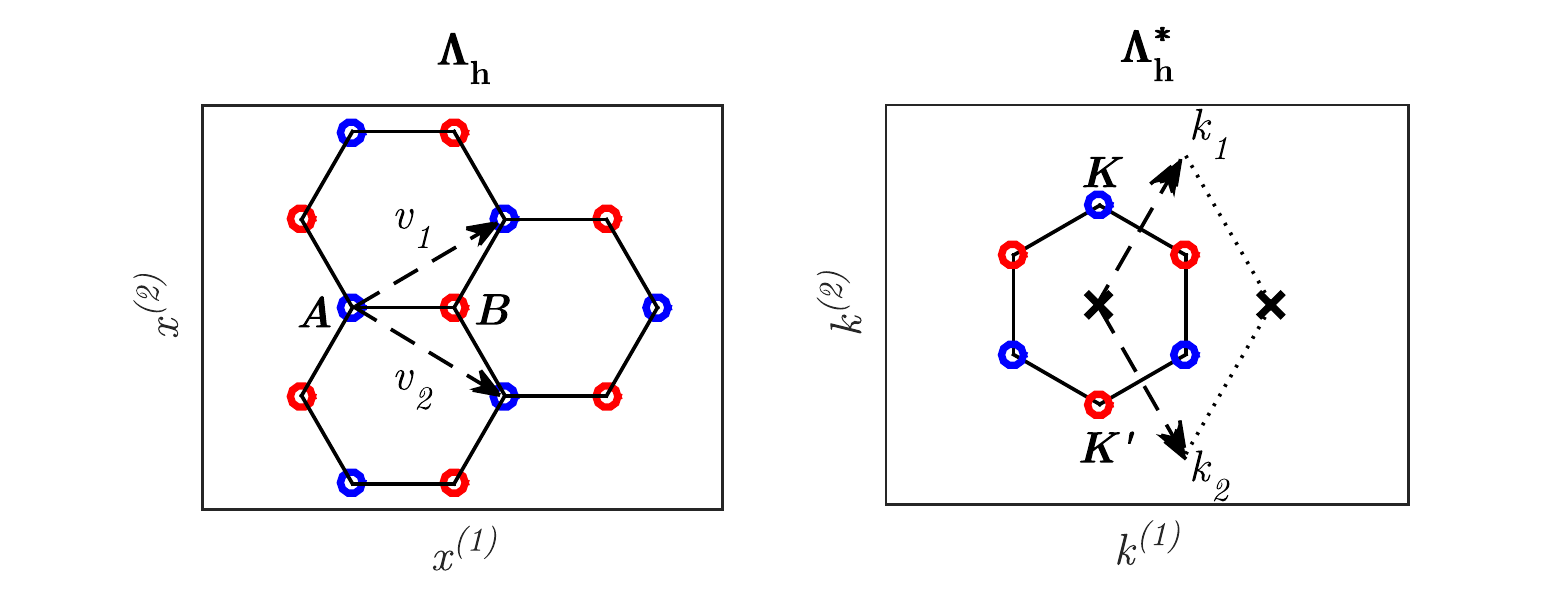}
\caption{\small 
{\bf Left panel}: ${\bf A}=(0,0)$, ${\bf B}=(\frac{1}{\sqrt3},0)$. 
Honeycomb structure, ${\bf H}$,  is the union of two sub-lattices  $\Lambda_{\bf A}={\bf A}+\Lambda_h$ (blue) 
and $\Lambda_{\bf B}={\bf B}+\Lambda_h$ (red); several hexagons shown. The lattice vectors  $\{\bv_1,\bv_2\}$ generate $\Lambda_h$.
{\bf Right panel}:
Brillouin zone, $\brill_h$, and dual basis $\{\bk_1,\bk_2\}$. $\bK$ and $\bK'$ are labeled.}
\label{lattice-and-itsdual}
\end{figure}
 \begin{figure}
\centering 
\includegraphics[width=0.75\textwidth]{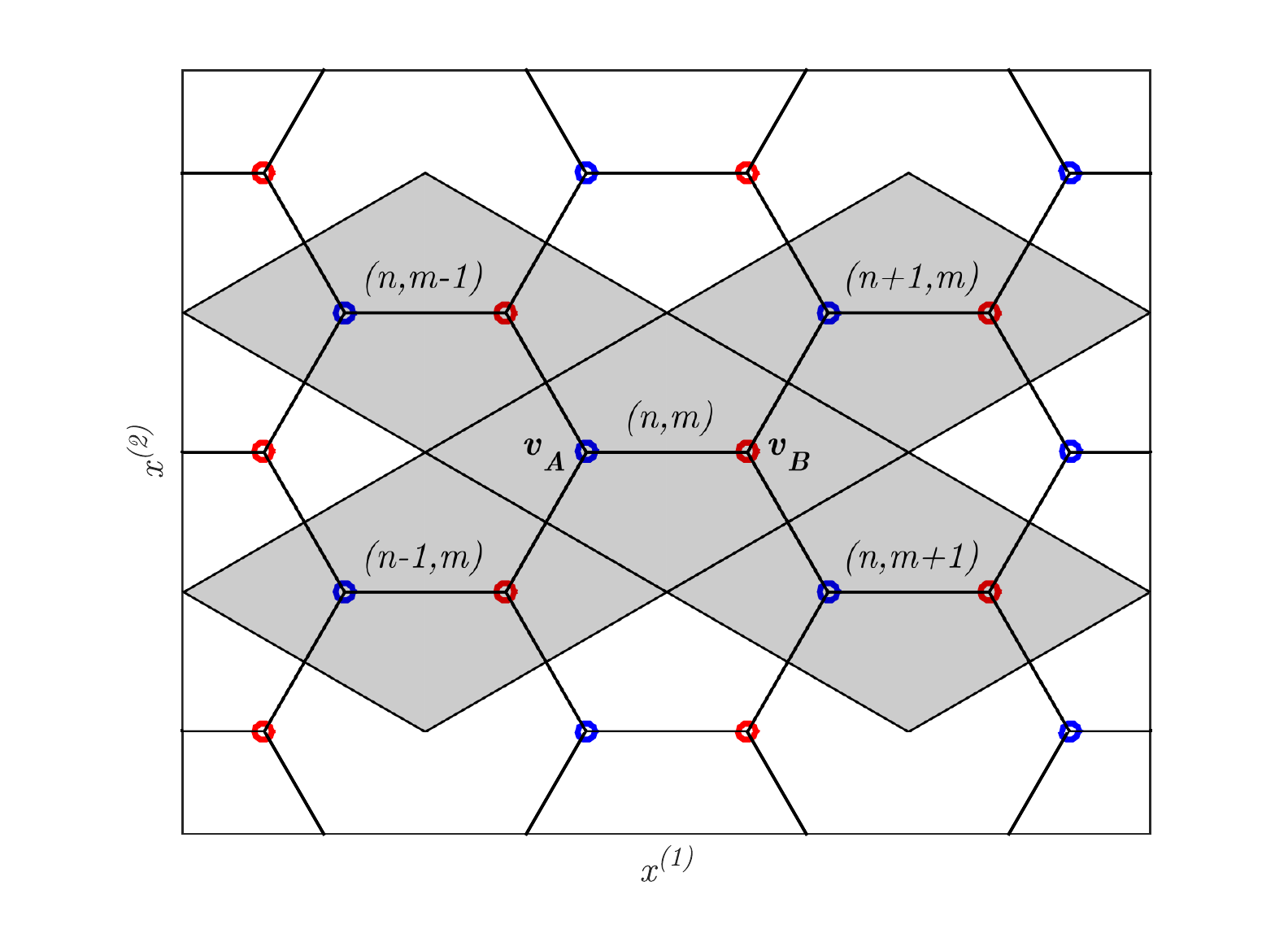}
\caption{\small Four (shaded) tiles of a tiling of $\R^2$. These contain the lattice points, whose amplitudes couple 
to  $\psi_A^{(n,m)}$ and $\psi_B^{(n,m)}$ according to the {\it tight-binding model} \eqref{tb-eqn}.}
\label{fig:diamond_tiling}
\end{figure}
Figure \ref{fig:diamond_tiling} displays the honeycomb structure and several shaded cells of a tiling of $\R^2$ by diamond-shaped fundamental period cells. In each fundamental period cell, there are two atomic sites ($A-$ type and $B-$ type). Let $\left(\vec\psi_A(T),\vec\psi_B(T)\right)^{\rm t}=\left(\psi_A^{(n,m)}(T),\psi_B^{(n,m)}(T)\right)^{\rm t}_{(n,m)\in\Z^2}\in l^2(\Z^2)\times l^2(\Z^2)$  denote the time-dependent amplitudes of the ground states centered at the $A-$ and $B-$ sites of the cell with label $(n,m)$, {\it i.e}
  the period cell containing $\bv_A+ n\bv_1+m\bv_2$ and $\bv_B+ n\bv_1+m\bv_2$. 
Recall that modes of the full honeycomb structure are assumed to be superpositions of interacting lattice-translates of ground states, concentrated on the support of deep potential wells. Each $A$-site amplitude interacts with its three nearest neighbor $B$-site amplitudes, and analogously for each $B$-site amplitude. The discrete equations are: 

\begin{align}
i\partial_T\ 
\begin{bmatrix}
\psi_A^{n,m}\\ \psi_B^{n,m}\
\end{bmatrix} &=\ t\ 
 \begin{bmatrix}  \psi_B^{n,m}\ +\ \psi_B^{n,m-1}\ +\  \psi_B^{n-1,m}\\
 \psi_A^{n,m}\ +\ \psi_A^{n+1,m}\ +\ \psi_A^{n,m+1}
 \end{bmatrix}\nn\\
 & =\ \frac{1}{|\brill_h|}\ \int_{\brill_h}\ e^{i\bk\cdot(n\bv_1+m\bv_2)}\ t\ H_{_{\rm TB}}(\bk)
 \begin{bmatrix}\widehat{\psi}_A(\bk) \\ \widehat{\psi}_B(\bk)\end{bmatrix}\ d\bk
\label{tb-eqn} \end{align}
 where $t$ denotes a non-zero coupling (``hopping'') coefficient,
 \begin{align}
 H_{_{\rm TB}}(\bk)\ &\equiv\ \begin{pmatrix} 0 & -\overline{\gamma(\bk)}
  \\ -\gamma(\bk)& 0\end{pmatrix},
  \label{HTBa}\\
  \gamma(\bk)\ &\equiv\ \sum_{\nu=1,2,3}e^{ i\bk\cdot \be_{B,\nu}}
 \ =\ e^{ i\bk\cdot \be_{B,1}}
\ \left(\ 1\ +\ e^{i\bk\cdot\bv_1}\ +\ e^{i\bk\cdot\bv_2}\ \right)\quad \textrm{(see \eqref{gamma-def0})}\ ,
\label{HTBb} \end{align}
 and  $\left(\widehat\psi_A(\bk),\widehat\psi_B(\bk)\right)^{\rm t}=\left(e^{-i\bk\cdot(\be_{B,1}/2)}\widetilde\psi_A(\bk),
  -e^{i\bk\cdot(\be_{B,1}/2)}\widetilde\psi_B(\bk)\right)^{\rm t}$, 
  where 
$\left(\widetilde\psi_A(\bk),\widetilde\psi_B(\bk)\right)^{\rm t}$ denotes the discrete Fourier transform of 
  $\left(\vec\psi_A(T),\vec\psi_B(T)\right)^{\rm t}$. The vectors $\be_{B,\nu}, \nu=1,2,3$ are the three vectors directed
  from any point in $\Lambda_B$ to its three nearest neighbors in $\Lambda_A$, and analogously for 
  $\be_{A,\nu}, \nu=1,2,3$; see Figure \ref{fig:fundamental-cell} below.
  
  Large but finite-time validity of such discrete approximations to time-dependent continuum Schr\"odinger equations, 
for certain initial data data, was studied in  \cites{ACZ:12,PSM:08}.

 The system \eqref{tb-eqn} has two dispersion surfaces.  To derive these explicitly, let
 $(\psi_A^{n,m}, \psi_B^{n,m})^{\rm t} = (  \alpha_A,\alpha_B)^{\rm t}  e^{-i\mathcal{E} T}\ e^{i(n\bv_1+m\bv_2)\cdot\bk}$, 
  where $\alpha_A$ and $\alpha_B$ are constants and   $\bk$ varies over the Brillouin zone, $\brill_h$; 
   see  Figure \ref{lattice-and-itsdual} and Section \ref{preliminaries}. 
 Substitution into \eqref{tb-eqn} yields the dispersion relation for the two spectral bands of the tight-binding model:
 \begin{align}
 \mathcal{E}_\pm(\bk)\ &=\  \pm\ |t|\ \mathscr{W}_{_{TB}}(\bk)\ ,\ \ \bk\in\brill_h\ ,\ \ {\rm where}\nn\\
 \mathscr{W}_{_{TB}}(\bk) &\ \equiv\ |\gamma(\bk)|\ =\  \left|\ 1+e^{i\bk\cdot\bv_1}+e^{i\bk\cdot\bv_2}\ \right| ; \label{tb-dispersion}\end{align}
$\mathcal{E}_\pm(\bk)$ are the two eigenvalue branches of $H_{_{\rm TB}}(\bk)$.
A plot of the two dispersion surfaces $\bk\in\R^2\mapsto\
\mathcal{E}_\pm(\bk),\ \ \bk\in\brill_h$ is shown in Figure \ref{fig:tight-binding-dsurfaces}.
\begin{figure}
\centering 
\includegraphics[width=0.65\textwidth]{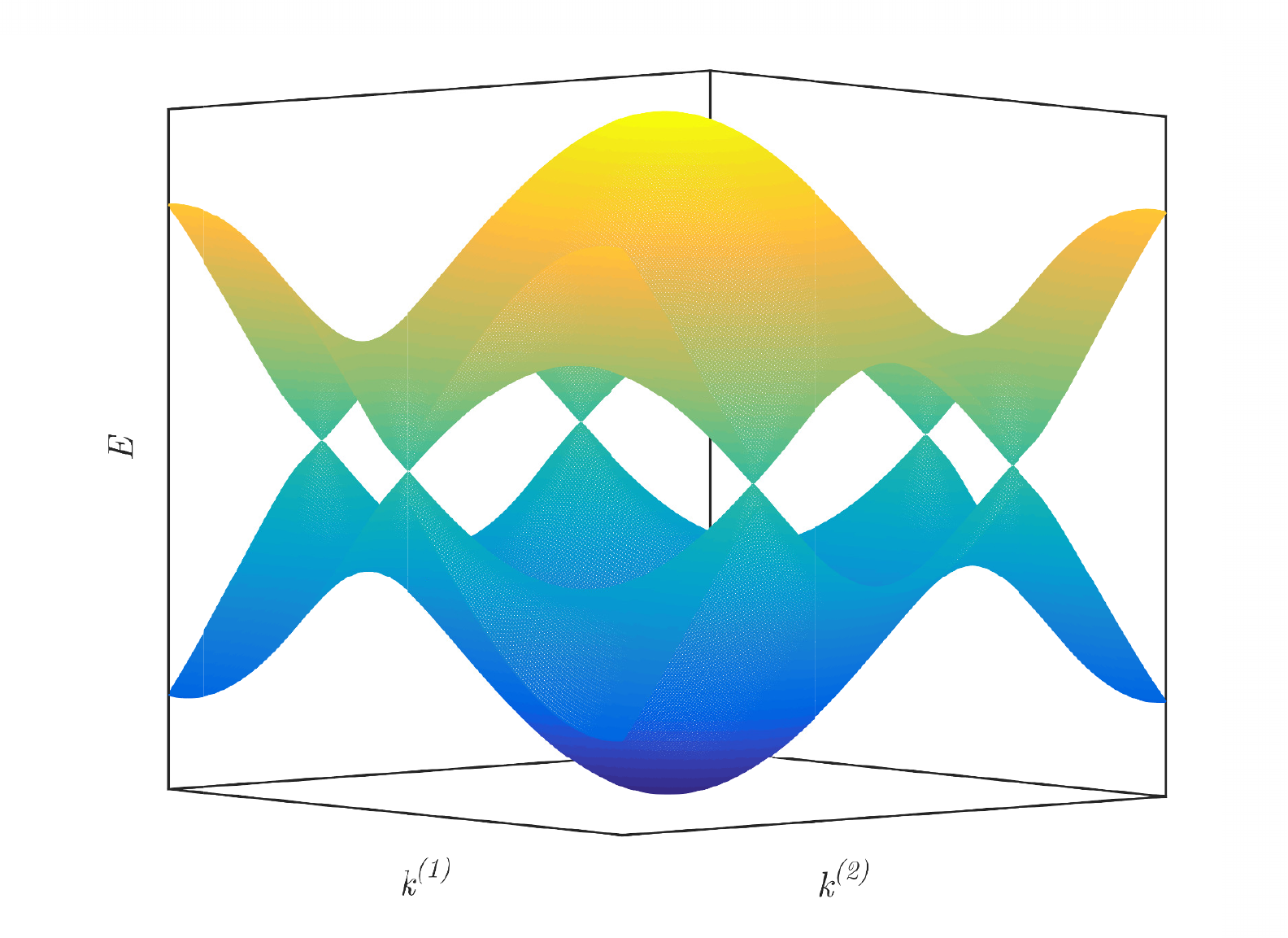}
\caption{\small Dispersion surfaces of Wallace{'}s  2-band tight-binding model. Dispersion relation displayed in  \eqref{tb-dispersion}.}
\label{fig:tight-binding-dsurfaces}
\end{figure}
 We note that the two dispersion surfaces are $\Lambda_h^*-$ periodic with respect to $\bk$, and touch conically 
  ($\mathscr{W}_{_{TB}}(\bk)=0$) at the six vertices of $\brill_h$,
  and their translates by the dual lattice, $\Lambda_h^*$; see Lemma \ref{gamma0}. The energy / quasi-momentum pairs
   at these conical intersection points are so-called Dirac points; see Section \ref{dirac-points}.

\subsection{Summary of results}\label{summaryofresults}
%

We study the  continuous Schr\"odinger operator, $-\Delta+\lambda^2 V(\bx)$, with honeycomb lattice potential, $V(\bx)$,  defined on $\R^2$ and $\lambda>\lambda_\star$ sufficiently large. Our particular model is one where $V$ is a superposition of ``atomic'' potential wells, $V_0(\bx)$, supported within the union of discs, centered on points of the honeycomb structure, $\Honeycomb$; see Figure \ref{fig:discs}. The detailed assumptions on $V_0(\bx)$ (Section \ref{atomic-well}) ensure that $V(\bx)$ is a honeycomb lattice potential in the sense of \cites{FW:12}, {\it i.e.} real-valued, periodic with respect equilateral triangular lattice, inversion symmetric and rotationally invariant by $120^\circ$.

 \begin{figure}
\centering 
\includegraphics[width=.75\textwidth]{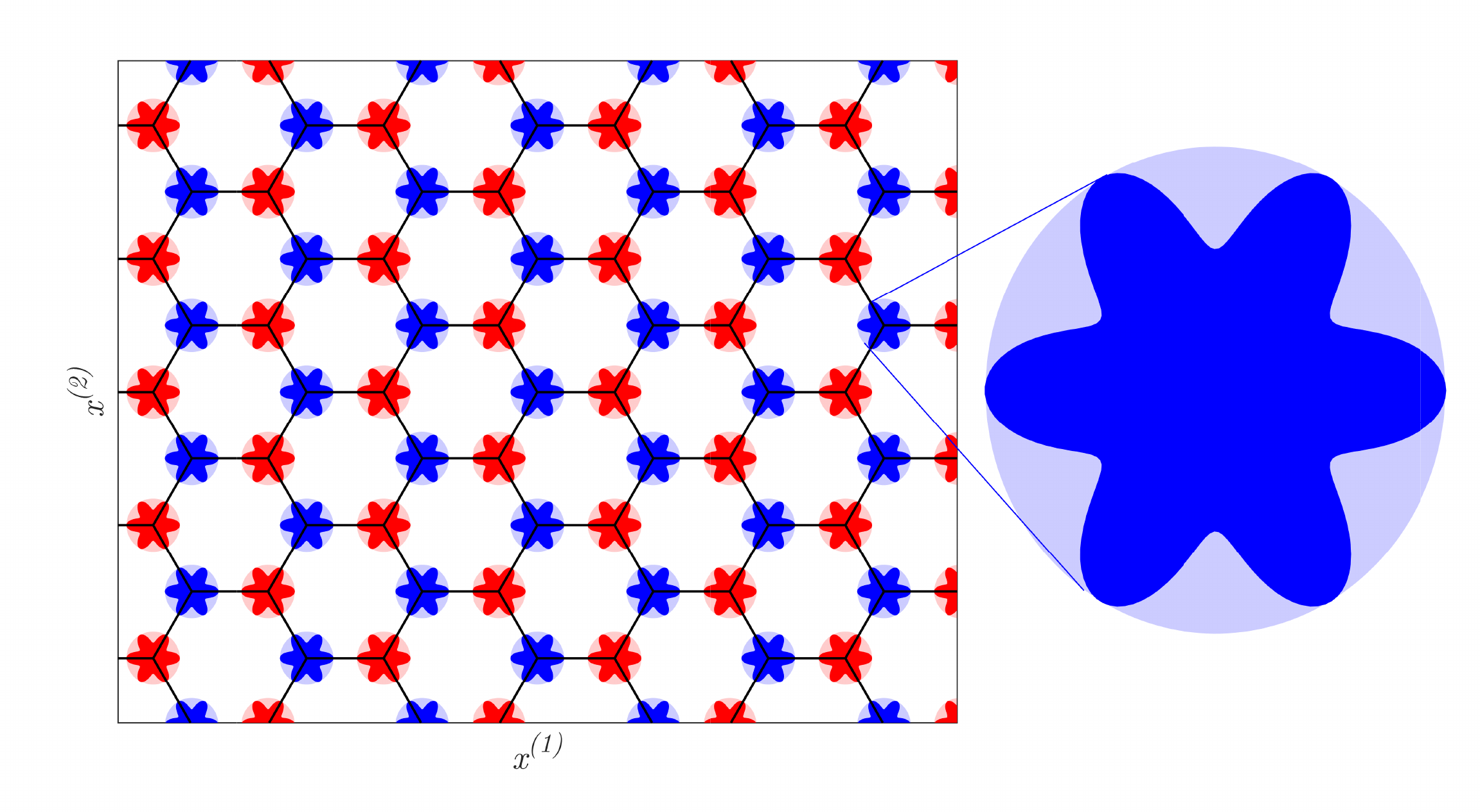}
\caption{\small Lightly shaded discs of radius $r_0<\frac{1}{2}|\be_{A,1}|$, centered at each $\bv\in\Honeycomb=\Lambda_A\cup\Lambda_B$. A copy of the atomic potential $V_0$, satisfying hypotheses $(PW_1)-(PW_4)$, {\bf (GS)} and {\bf (EG)}  in Section  \ref{atomic-well}, is supported within each disc. The $60^\circ$ degree rotationally invariant support of $V_0$ is darkly shaded. }
\label{fig:discs}
\end{figure}
 
For $\bk$ varying over the Brillouin zone, $\brill_h$, let 
$E_1^\lambda(\bk)\le E^\lambda_2(\bk)\le\dots\le E^\lambda_b(\bk)\le\dots$ (listed with multiplicity) denote the Floquet-Bloch  spectrum of $ -(\nabla+i\bk)^2 +  \lambda^2 V(\bx)$, considered with $\Lambda_h-$ periodic boundary conditions.  The graphs of the mappings: $\bk\mapsto E_b^\lambda(\bk)$ are the dispersion surfaces. We study the following
 
\noindent {\bf Problem:} Precisely describe the behavior of
the dispersion surfaces of $-\Delta +  \lambda^2 V(\bx)$, obtained from the low-lying (two lowest) eigenvalues of $ -(\nabla+i\bk)^2 + \lambda^2 V(\bx)$:
 \begin{align}
&  \bk\mapsto E_1^\lambda(\bk)\ =\ E_-^\lambda(\bk)\qquad {\rm and}\qquad  
\bk\mapsto E_2^\lambda(\bk)\ =\ E_+^\lambda(\bk) ,\nn\end{align}
for all  $\lambda>\lambda_\star$ sufficiently large.
 This is called the regime of {\it strong binding}.
We refer to the rescaled, $\lambda\to\infty$, limiting behavior as the {\it tight-binding limit}.
%
%
%

\nit {\bf The Main Theorem, Theorem \ref{main-theorem}}: For $\lambda>\lambda_\star$ sufficiently large,  the rescaled low-lying dispersion  surfaces, $\bk\mapsto E^\lambda_\pm(\bk)$, converge uniformly to  Wallace{'}s (1947) two-band  tight-binding model defined on a honeycomb structure. Specifically, for a suitable energy, $E^\lambda_D$, and $\rho_\lambda>0$, we have
\begin{align}
&\left(\ E^\lambda_-(\bk)-E^\lambda_D\ \right) / \rho_\lambda \to -\mathscr{W}_{_{TB}}(\bk)\ \ {\rm and}\ \ 
\left(\ E^\lambda_+(\bk)-E^\lambda_D\ \right) / \rho_\lambda \to +\mathscr{W}_{_{TB}}(\bk)\ ,
\label{conv_to_wallace-47}\\
&\textrm{as $\lambda\to\infty$, uniformly in $\bk\in\brill_h$, the Brillouin zone.}\nn
\end{align}
We also prove estimates and convergence of the derivatives of $E^\lambda_\pm(\bk)$ on appropriate domains. 
Furthermore, in {\bf Theorem 6.2} we establish the scaled norm-convergence of the resolvent of $-\Delta+\lambda^2V$ to that of the tight-binding Hamiltonian, $H_{_{\rm TB}}$. 
 
The first and second dispersion surfaces intersect precisely at the quasi-momenta located at the six vertices of $\B_h$
 at the energy-level, $E_D^\lambda$.
 The parameter $\rho_\lambda$, displayed in \eqref{rho-lam-def},
 is given
by an exponentially small overlap integral involving the atomic potential $V_0$, the ground state of $-\Delta+\lambda^2 V_0$
and the ground state translated to a nearest neighbor site of $\Honeycomb$; see Proposition \ref{prop:rho-lam-bounds}. 

 Theorem \ref{main-theorem} implies that, in the strong binding regime (all $\lambda$ sufficiently large),
  the only intersections of the lowest two dispersion surfaces occur at Dirac points, situated at the vertices, $\bK_\star$, of $\brill_h$. 
  Moreover, \eqref{conv_to_wallace-47} and the Taylor expansion of $\mathscr{W}_{_{TB}}(\bk)$ (Lemma \ref{gamma0}) near vertices gives:
  \begin{equation}
   E_\pm^\lambda(\bk)\ =\ E_D^\lambda\ \pm\ |\lamsharplam|\ |\bk-\bK_\star|
   + \mathcal{O}(\rho_\lambda|\bk-\bK_\star|^2) , 
   \label{Epm-expand}
   \end{equation}
   where 
   \begin{equation}
   |\lamsharplam|\ =\ \left[\ \frac{\sqrt3}{2}+\mathcal{O}(e^{-c\lambda})\ \right]\rho_\lambda
   \label{vF-expand}
   \end{equation}
   is the Fermi velocity (see Definition \ref{dirac-pt-defn} and \cites{FW:12}), the velocity of  ``quasi-particles'' (wave-packets)
   which are spectrally concentrated near Dirac points; see, for example, \cites{RMP-Graphene:09,FW:14}.
 \medskip
 
 \begin{remark}[Exchange of Dirac Points]\label{exchange}
Consider the Schroedinger operator $-\Delta + \lambda^2 V$. As in Theorem \ref{main-theorem},
 we take $V(\bx)$ to be a superposition 
of atomic wells $V_0(\bx)$ centered at honeycomb lattice sites. As shown in Appendix A of 
 \cites{FLW-2d_edge:16} it is possible to choose $V_0(\bx)$ so that $V_{_{1,1}}<0$, where  $V_{_{1,1}}$ denotes the $(1,1)-$ Fourier coefficient of $V(\bx)$. It follows from \cites{FW:12} that for $\lambda$ sufficiently small and positive $-\Delta + \lambda^2 V$ has Dirac points situated at the intersection of the $2^{nd}$ and $3^{rd}$ spectral bands.
Furthermore, by Theorem \ref{main-theorem} for $\lambda > \lambda_\star$ sufficiently large 
$-\Delta + \lambda^2 V$ has Dirac points situated
at the intersection of the $1^{st}$ and $2^{nd}$ spectral bands.
It follows that  for such honeycomb Schroedinger operators there is an {\it exchange of Dirac points} from the $2^{nd}$ and $3^{rd}$ bands
to the $1^{st}$ and $2^{nd}$ bands as $\lambda$ is increased across a finite value of $\lambda=\lambda_{cr}$, where $0< \lambda_{cr}<\lambda_\star$. A further result in \cites{FW:12} (see also Appendix D of \cites{FLW-MAMS:17}) is that $-\Delta + \lambda^2 V$ has Dirac points for all $\lambda>0$ outside a discrete set $\widetilde{C}$. Such examples prove
that the exceptional set $\widetilde{C}$ can be non-empty.
 \end{remark}

Theorem \ref{main-theorem}, together with results in \cites{FW:12} and \cites{FLW-2d_edge:16},
imply the following corollaries:
\begin{enumerate}
\item[(A)] {\it Corollary \ref{spectral-gaps}: Spectral gaps when breaking $\mathcal{P}\mathcal{T}$ symmetry}. Honeycomb lattice potentials, $V$, have the property that the associated Schr\"odinger Hamiltonian, $-\Delta+\lambda^2 V$,  commutes with the composition of inversion with respect to an appropriate center, and complex conjugation, $\mathcal{C}\circ\mathcal{I}$. This is also known as {\it $\mathcal{P}\mathcal{T}$ symmetry}.  For $\lambda$ large, a  consequence of Theorem \ref{main-theorem} and the results in \cites{FW:12} is the existence of spectral gaps about Dirac points (see Section \ref{dirac-points}) of $-\Delta+\lambda^2 V$ when the Hamiltonian is perturbed in such a way as to break $\mathcal{P}\mathcal{T}$ symmetry.
%

A review of other mechanisms for construction spectral gaps, also in the high-contrast regime, appears in 
\cites{Hempel-Post:03}; see also \cites{Figotin-Kuchment:96a,Figotin-Kuchment:96b}.

\item[(B)]  {\it Corollary \ref{rational-edge-states}: Protected edge states in honeycomb structures with line-defects}. Edge states are time-harmonic solutions of the  Schr\"odinger equation, which are propagating parallel to a
 line-defect (edge) and are  localized transverse to it; see Figure \ref{fig:mode_schematic}.
  In \cites{FLW-2d_edge:16} (see also \cites{FLW-2d_materials:15}), we develop a theory of protected edge states for  honeycomb structures, perturbed by a class of line-defects (domain walls) in the direction of an element of $\Lambda_h$  ({\it rational edges}). The key hypothesis is a {\it spectral no-fold condition}.  Our main result, Theorem \ref{main-theorem},  implies the validity of the spectral no-fold condition, and hence the existence of edge states for a large class of rational edges, in the strong binding regime. 
  
  Previous analytical work on topologically protected edge states in periodic structures with line-defects has focused on approximate tight-binding models; see, for example,  
   \cites{RMP-Graphene:09,delplace2011zak,Graf-Porta:13}. 
   
  Finally, we remark that the effect of interacting electrons in graphene, in the tight-binding limit,
   have been studied in \cites{Giuliani-Mastropietro:10,Giuliani-Mastropietro-Porta:12}.
  \end{enumerate}

\subsection{Outline}\label{outline}

 In Section \ref{fb_basics} we review basic Floquet-Bloch theory of Schr\"odinger operators with periodic potentials.
 Section \ref{HS} introduces the honeycomb structure, $\Honeycomb$, which is the union of the two sublattices $\Lambda_A$ and $\Lambda_B$.
 Section \ref{atomic-well} discusses hypotheses on the atomic potential well, $V_0$, and Section \ref{periodization} treats its  $\Lambda_h-$ periodization, $V$, obtained via summation over translates by vectors $\bv\in\Honeycomb$.  The atomic potential, $V_0$, is assumed to have compact support, $|\bx|<r_0$, where  $r_0<r_{critical}$ and $r_{critical}$ (which is less than than half the distance nearest neighbor lattice points in ${\bf H}$) is determined  by a Geometric Lemma presented in Section \ref{ME-bounds}. 
 The assertions of this lemma are easily seen to hold for $r_0$ positive and sufficiently small. A non-trivial lower bound for $r_{critical}$ is of interest in applications and for this we require the Geometric Lemma. 
 
 In Section \ref{main-results} we state our main result, Theorem \ref{main-theorem}, on the low-lying dispersion surfaces of 
  $-\Delta+\lambda^2V$, for $\lambda>\lambda_\star$ sufficiently large. We also state and prove consequences for Schr\"odinger operators on $\R^2$ with perturbed honeycomb structures in the regime of strong binding:  Corollary \ref{spectral-gaps} on spectral gaps and Corollary \ref{rational-edge-states} on protected edge states. 
  
Section \ref{dirac-points} reviews the notion of Dirac points and results  on the existence of Dirac points
  for generic honeycomb lattice potentials \cites{FW:12,FLW-MAMS:17}; see also 
  \cites{Colin-de-Verdiere:91,Grushin:09,berkolaiko-comech:15,Lee:16}. Dirac points are energy / quasi-momentum pairs, 
  which occur at quasi-momenta located at the vertices of the Brillouin zone, $\brill_h$, and at which neighboring dispersion surfaces touch conically.
  For an extensive discussion of Dirac points and edge states for nanotube structures in the context of  quantum graphs, see \cites{Kuchment-Post:07} and \cites{Do-Kuchment:13}.

  The proof of the main theorem, Theorem \ref{main-theorem}, on the large $\lambda$ behavior of low-lying dispersion surfaces is carried out in Sections  \ref{low-lying-fb} through \ref{ME-bounds}.
   In Section \ref{low-lying-fb} we construct approximate Floquet-Bloch modes, $p^\lambda_{\bk,A}(\bx)$ and $p^\lambda_{\bk,B}(\bx)$ (associated with the sublattices $\Lambda_A$ and $\Lambda_B$) for the two lowest spectral bands of the
  Floquet-Bloch Hamiltonian $H^\lambda(\bk)=-(\nabla_\bx+i\bk)^2+\lambda^2 V(\bx)-E_0^\lambda$\ ( $\bk\in\R^2$),  in terms of the ground state eigenpair, $\left(E_0^\lambda,p_0^\lambda(\bx)\right)$, of the atomic Hamiltonian, $-\Delta+\lambda^2 V_0$. 
  
  In Section \ref{en-estimates} we first derive energy estimates for the family of Floquet-Bloch Hamiltonians $H^\lambda(\bk)=-(\nabla_\bx+i\bk)^2+\lambda^2 V(\bx)-E_0^\lambda$, $\bk\in\R^2$, restricted to the $L^2(\R^2/\Lambda_h)-$ orthogonal complement of these approximate  Floquet-Bloch modes.  We then use these estimates to prove resolvent bounds on this subspace.
  
  In Sections \ref{dirac-points-sb} through \ref{ME-bounds} we apply resolvent bounds on $H^\lambda(\bk)$,  in a Lyapunov-Schmidt reduction scheme, for $\lambda$ large, to the 2D subspace ${\rm span}\{p^\lambda_{\bk,A}(\bx), p^\lambda_{\bk,B}(\bx)\}$. The main steps are  (i) a proof of key properties of Dirac points at the vertices of the Brillouin zone, $\brill_h$ (Theorem \ref{solve-L2tau-evp}) and (ii) a study of uniform convergence of the (rescaled) low-lying dispersion maps $\bk\mapsto E^\lambda_\pm(\bk)$ to the dispersion surfaces of Wallace{'}s tight-binding model (Propositions \ref{Proposition-RDSADP} and \ref{Proposition-nearDPs}).

  In Section \ref{LS-reduction} we characterize  the low-lying dispersion surfaces as the locus of points, $(\Omega,\bk)\in\R\times\R^2$ satisfying $\det\mathcal{M}^{\lambda}(\Omega,\bk)=0$. For each $\lambda$ sufficiently large $(\Omega,\bk)\mapsto \mathcal{M}^{\lambda}(\Omega,\bk) $ is an analytic map from a subset $U\subset\C\times\C^2$ into the space of $2\times2$ matrices, where $\mathcal{M}^{\lambda}(\Omega,\bk)$ is Hermitian for real $\Omega$ and $\bk$.  
   In Section \ref{sec:M-expanded} we expand $\mathcal{M}^{\lambda}(\Omega,\bk)$ for large $\lambda$ and in Sections \ref{dispersion-surfaces} and \ref{mu_pm-rescaled}  we introduce and analyze a rescaling of $\det\mathcal{M}^{\lambda}(\Omega,\bk)$ to complete the proof of our main result, Theorem \ref{main-theorem}.  
  
  Section \ref{ME-bounds} contains estimates that facilitate our control of the large $\lambda-$ perturbation theory in terms of an intrinsic (exponentially small) parameter, $\rho_\lambda$. This parameter has the form of an integral of the product of: the atomic potential well $V_0(\bx)$, the atomic ground state, $p_0^\lambda(\bx)$, and the translate of $p_0^\lambda(\bx)$ to a nearest neighbor lattice site in $\Honeycomb$. An important tool is a lemma in Euclidean geometry, used to bound the ground state $-\Delta+\lambda^2 V_0(\bx)$ and to quantify the maximum allowable size of the support of the atomic potential well in our proofs.
  
In the remainder of this section we discuss some definitions, notation and  conventions used
 throughout the paper.

\subsection{Notations and conventions}\label{preliminaries}
  
We denote by 
$ \Lambda_h=\Z{\bf v}_1 \oplus \Z{\bf v}_2,$
 the equilateral triangular lattice generated by the basis vectors:
\begin{align}
 \bv_1 &=\ \left( \begin{array}{c} \frac{\sqrt{3}}{2} \\ {}\\  \frac{1}{2}\end{array} \right),\ \ 
\bv_2 =\ \left(\begin{array}{c} \frac{\sqrt{3}}{2} \\ {}\\ -\frac{1}{2} \end{array}\right)\ .\label{v12-def}
\end{align}
The dual lattice $\Lambda_h^* =\ \Z {\bf k}_1\oplus \Z{\bf k}_2$ is spanned by the dual basis vectors:
\begin{align}
&  \bk_1=\ 2\pi\left(\begin{array}{c} \frac{ \sqrt{3} }{3}\\ {}\\ 1\end{array}\right),\ \ \
 \bk_2 = 2\pi\left(\begin{array}{c} \frac{ \sqrt{3} }{3}\\ {}\\ -1\end{array}\right).\ 
 \label{k1k2-def}
\end{align}
Note that $\bk_{\ell}\cdot \bv_{{\ell'}}=2\pi\delta_{\ell{\ell'}}$.
The Brillouin zone, $\brill_h$, is the hexagon in $\R^2_\bk$ consisting of all points which are closer to the origin than to any other point in $\Lambda_h^*$; see Figure \ref{lattice-and-itsdual}.

 Denote by $\bK$ and $\bKp$ the vertices of  $\brill_h$ given by:
\begin{equation}
\bK\equiv\frac{1}{3}\left(\bk_1-\bk_2\right)= \begin{pmatrix}0\\ \frac{4\pi}{3}\end{pmatrix},\ \ \bKp\equiv-\bK\ .
\label{KKprime}
\end{equation}
All six  vertices of $\brill_h$ can be generated from $\bK$ and $\bK^\prime$ by application of the rotation matrix, $R$,
 which rotates a vector in $\mathbb{R}^2$ clockwise by $2\pi/3$ about the origin. The matrix $R$ is given by
\begin{equation}
R\ =\ \left(\begin{array}{cc} -\frac{1}{2} & \frac{\sqrt{3}}{2}\\
                                                 {} & {}\\
                                           -\frac{\sqrt{3}}{2} & -\frac{1}{2}
                                           \end{array}\right)
  \label{Rdef}\end{equation}
Note the relations:
\begin{equation*}
R^*\bv_1\ =\ -\bv_2,\qquad R^*\bv_2\ =\ \bv_1-\bv_2
\end{equation*}
%

In Section \ref{HS} we make the choice of diamond-shaped fundamental cell for the honeycomb structure, $D$,
 shown in Figure \ref{fig:fundamental-cell}. 
$D$ contains two base-points ,
 \begin{equation}
\bv_A\ =\ \begin{pmatrix} 0\\ 0\end{pmatrix},\qquad \bv_B\ =\ \begin{pmatrix} \frac{1}{\sqrt3}\\ 0\end{pmatrix} ,
\label{vAB-def}\end{equation}
of the sublattices, $\Lambda_A$ and $\Lambda_B$, which comprise the honeycomb structure. The location 
\begin{equation}
\bx_c\ =\ \frac12\ \begin{pmatrix} \frac{1}{\sqrt3} \\ -1\end{pmatrix}\ .
\label{x_c-def}
\end{equation}
marks the center of a hexagon and is a vertex of $D$.
%

Additional frequently used notations and conventions are:
\begin{enumerate}[(1)]
 \item For $\bfm=(m_1,m_2)\in\Z^2$, $\bfm\vec\bk=m_1\bk_1+m_2\bk_2$ and $\bfm\vec\bv=m_1\bv_1+m_2\bv_2$.
 \item  $\bk=(k^{(1)},k^{(2)})$.
 \item $\tK$ will be used to denote a generic quasi-momentum. 
\item $\bK_\star$ will be used to denote for a generic element of  $\Lambda_\bK^*\cup \Lambda_{\bK'}^*=\left(\bK + \Lambda_h^*\right)\cup\left(\bK' + \Lambda_h^*\right)$. These are the vertices of the Brillouin zone, $\mathcal{B}_h$, and their translates by the dual lattice.
 \item $E_D^\lambda$ or $E_D$ denotes the energy of a Dirac point. 
  \item $x\lesssim y$ if and only if there exists $C>0$ such that $x \leq Cy$.
  And  $x \approx y$ if and only if $x \lesssim y$ and $y \lesssim x$. We shall discuss below 
    the dependencies of constants $C$.
  \item $\left\langle f,g\right\rangle$ is an inner product, which is linear in $g$ for fixed $f$, and conjugate linear in $f$, for fixed $g$.
  \item Let $\Lambda$ denote an arbitrary lattice in the plane, $\R^2$. For $s\in\R$, the space $H^s(\R^2/\Lambda) $ consists of complex-valued and $\Lambda-$ periodic functions $f$ on $\R^2$, whose Fourier coefficients, $\{\widehat{f}(\bfm)\}_{_{\bfm\in\Z^2}}$, satisfy
 \[ \|f\|_{_{H^s(\R^2/\Lambda)}}^2\ \equiv\ \sum_{\bfm\in\Z^2} (1+|\bfm|^2)^s\ |\widehat{f}(\bfm)|^2\ <\ \infty\ .\]
 \item For ${\bf F}=(F_1,F_2,\dots,F_m)$, with each $F_j\in \mathcal{Y}$, a normed linear space, we write
  $\|{\bf F}\|_{_{\mathcal{Y}}}=\sum_{j=1}^m\|F_j\|_{_{\mathcal{Y}}}$. 
  \end{enumerate}

  We study  $-\Delta+\lambda^2 V_0(\bx)$ and its $\Lambda_h-$ periodic variants. Here, $\lambda>0$ is a coupling constant, assumed to satisfy $\lambda>\lambda_\star$ for a large enough $\lambda_\star$; and  $V_0(\bx)$ is a given potential defined on $\R^2$.
  We write $c, C, C'$ {\it etc.} to denote constants which depend on $V_0$. A discussion of the  precise dependencies of constants is given in Section \ref{sec:constants}.
    These symbols may denote different constants in different occurrences. As a result of the above conventions, it is correct to assert, for example, $\lambda^{10}e^{-c\lambda}\le e^{-c\lambda}$.

Finally, for relations involving norms and inner products in which we do not explicitly indicate the relevant function space, it is to be understood that these are taken in $L^2(\R^2/\Lambda_h)$.

\subsection{Acknowledgements} The authors thank G. Berkolaiko for his interest and feedback concerning this work. We thank the referee for a careful reading and incisive comments which, in particular, inspired Theorem \ref{res-conv} and Section \ref{resolvent}. This research was supported in part by National Science  Foundation grants DMS-1265524 (CLF) and  DMS-1412560 (MIW \& JPL-T) and Simons Foundation Math + X Investigator Award \#376319 (MIW).

\section{Floquet-Bloch Theory and Honeycomb Lattice Potentials}\label{fb_basics} 

We begin with a review of Floquet-Bloch theory. For the theory, see \cites{Eastham:74,Kuchment:12,Kuchment:16,RS4} and \cites{avron-simon:78,Skriganov:85,Gerard:90,Karpeshina:97,Gerard-Nier:98}.

\subsection{Fourier analysis on $L^2(\R/\Lambda)$}\label{fourier-analysis}

Let $\{\vtilde_1,\vtilde_2\}$ be a linearly independent set in $\R^2$, and introduce the
\begin{align}
&\text{\bf Lattice: } \Lambda = \Z\vtilde_1\oplus\Z\vtilde_2 = \{m_1\vtilde_1 + m_2\vtilde_2 \ : \ m_1,m_2 \in \Z \} ;   \nn \\
&\text{\bf Dual lattice: }
\Lambda^\ast = \Z\ktilde_1\oplus\Z\ktilde_2 = \{ \bfm\vec\ktilde=m_1\ktilde_1 + m_2\ktilde_2  : m_1,m_2 \in \Z \} , \nn\\ 
&\qquad\qquad\qquad\qquad\qquad \ktilde_i\cdot\vtilde_j = 2\pi\delta_{ij},\ 1\leq i,j \leq 2;   \nn\\
& \text{\bf Fundamental period cell, }\ \Omega\subset\R^2; \nn\\
&\text{\bf Brillouin zone: }  \B,\ \textrm{a choice of fundamental dual cell}. \nn 
 \end{align}

 \begin{definition}[The spaces $L^2(\R^2/\Lambda)$ and $L^2_\bk$] \ \ 
 \label{L2-spaces}
 \begin{enumerate}
  \item[(a)] $L^2(\R^2/\Lambda)$ denotes the space of $L^2_{loc}$ functions which are  $\Lambda-$ periodic: 
  $ f\in L^2(\R^2/\Lambda)$ if and only if $f(\bx+\vtilde)=f(\bx)$ for almost all $ \bx\in\R^2,\ \vtilde\in\Lambda$; and 
$f\in L^2(\Omega)$.
  \item[(b)] $L^2_\bk$ denotes the space of $L^2_{loc}$ functions which satisfy a pseudo-periodic boundary condition: $ f(\bx+\vtilde)=e^{i\bk\cdot\vtilde}f(\bx)$ for all $\vtilde\in\Lambda$ and almost all $\bx\in\R^2$; {\it i.e.} $e^{-i\bk\cdot\bx}f(\bx)\in L^2(\R^2/\Lambda)$.
  
\nit  For $f$ and $g$ in $L^2_\bk$, $\overline{f}g$ is in $L^2(\R^2/\Lambda)$ and we define the inner product by
  \begin{equation*}
   \inner{f,g}_{L^2_\bk} = \int_\Omega \overline{f(\bx)} g(\bx) d\bx.
  \end{equation*}
 \end{enumerate}
 \end{definition}

\subsection{Floquet-Bloch theory}\label{flo-bl-theory}

Let $Q(\bx)$ denote a real-valued potential which is periodic with respect to $\Lambda$. We shall assume throughout this paper that $Q\in C^\infty(\R^2/\Lambda)$, although we expect that this condition can be relaxed significantly without much extra work; see Remark \ref{energy-gap}.
Introduce the  Schr\"odinger Hamiltonian 
$H \equiv -\Delta + Q(\bx)$. 
For each $\bk\in\R^2$, we study the {\it Floquet-Bloch eigenvalue problem} on $L^2_\bk$:
\begin{align}
 \label{fl-bl-evp}
&H \Phi(\bx;\bk) = E(\bk) \Phi(\bx;\bk), \ \ \bx\in\R^2, \\
&\Phi(\bx+\vtilde;\bk)=e^{i\bk\cdot\vtilde}\Phi(\bx;\bk), \ \ \vtilde\in\Lambda \nn .
\end{align}
An $L^2_\bk-$ solution of \eqref{fl-bl-evp} is called a {\it Floquet-Bloch} state.

Since the $\bk-$ pseudo-periodic boundary condition in \eqref{fl-bl-evp} is invariant under translations in the dual period lattice, $\Lambda^\ast$, it suffices to restrict our attention to $\bk\in\B$, where  $\B$, the {\it Brillouin Zone},  is a fundamental cell in $\bk-$ space.

An equivalent formulation to \eqref{fl-bl-evp} is obtained by setting $\Phi(\bx;\bk)=e^{i\bk\cdot\bx}p(\bx;\bk)$. Then, 
\begin{equation}
 \label{fl-bl-evp-per}
 H(\bk) p(\bx;\bk) = E(\bk) p(\bx;\bk), \ \bx\in\R^2, \quad p(\bx+\vtilde)=p(\bx;\bk),\ \  \vtilde\in\Lambda,
\end{equation}
where
$ H(\bk) \equiv -(\nabla+i\bk)^2 + Q(\bx)$ 
is a self-adjoint operator on  $L^2(\R^2/\Lambda)$.
 The eigenvalue problem  \eqref{fl-bl-evp-per}, has a discrete set of eigenvalues
$E_1(\bk)\leq E_2(\bk)\leq \cdots \leq E_b(\bk)\leq \cdots$,
with $L^2(\R^2/\Lambda)-$ eigenfunctions $p_b(\bx;\bk),\  b=1,2,3,\ldots$. 
 The maps $\bk\in\B\mapsto E_j(\bk)$ are, in general, Lipschitz continuous functions; for example,
  see \cites{avron-simon:78,Kuchment:12,Kuchment:16} and Appendix A of \cites{FW:14}. For each $\bk\in\B$, the set $\{p_j(\bx;\bk)\}_{j\geq1}$ can be taken to be a  complete orthonormal basis for $L^2(\R^2/\Lambda)$.
%
%
%
%

As $\bk$ varies over $\B$, $E_b(\bk)$ sweeps out a closed interval in $\R$. The union over $b\ge1$ of these closed intervals is exactly the $L^2(\R^2)-$ spectrum of $-\Delta+Q(\bx)$: 
$ \text{spec} \left(H\right) = \bigcup_{\bk\in\B} \text{spec} \left(H(\bk)\right).$
Furthermore, the set $\{\Phi_b(\bx;\bk)\}_{b\geq1,\bk\in\B}$ is complete in $L^2(\R^2)$. For a suitable normalization of
 $\Phi_b(\bx;\bk)$, we have 
\begin{equation*}
 f(\bx) = \sum_{b\geq1} \int_\B \inner{\Phi_b(\cdot;\bk),f(\cdot)}_{L^2(\R^2)} \Phi_b(\bx;\bk) d\bk
 \equiv \sum_{b\geq1} \int_\B \widetilde{f}_b(\bk) \Phi_b(\bx;\bk) d\bk ,
\end{equation*}
where the sum converges in the $L^2$ norm.

\section{Honeycomb Structure}\label{HS}

Denote by $\Lambda_h\subset \R^2$, the equilateral triangular lattice specified in Section \ref{preliminaries}. Recall the base points $\bv_A$ and $\bv_B$, defined in \eqref{vAB-def}.


\noindent{\bf Sublattices: $\Lambda_A$ and $\Lambda_B$ and the honeycomb structure $\Honeycomb$}: 
Generate the $A-$ sublattice, $\Lambda_A$, 
and the $B-$ sublattice, $\Lambda_B$:
$\Lambda_I = \bv_I+\Lambda_h\subset\R^2,\ \ I=A, B$ .
%
The honeycomb structure is defined to be:
\begin{equation*}
\Honeycomb = \Lambda_A\ \cup\ \Lambda_B\subset\R^2.
\end{equation*}


\noindent{\bf $D$, fundamental domain for $\R^2/\Lambda_h$:}\ Let $D\subset \R^2$ denote the diamond-shaped fundamental domain for the torus, $\R^2/\Lambda_h$,
shown in Figure \ref{fig:fundamental-cell}. Choose $D$ so that  $\bv_A, \bv_B\in D$. $\bx_c$ is the center of a hexagon (not a point in $\Honeycomb$) and a vertex of the parallelogram $D$.
For any $F\in L^1(\R^2)$ we have
\begin{equation*}
\int_{\bx\in\R^2}\ F(\bx)\ d\bx\ =\ \int_{\bx\in D}\ \sum_{\bv\in\Lambda_h} F(\bx-\bv)\ d\bx .
\end{equation*}

 \begin{figure}
\centering 
\includegraphics[width=0.75\textwidth]{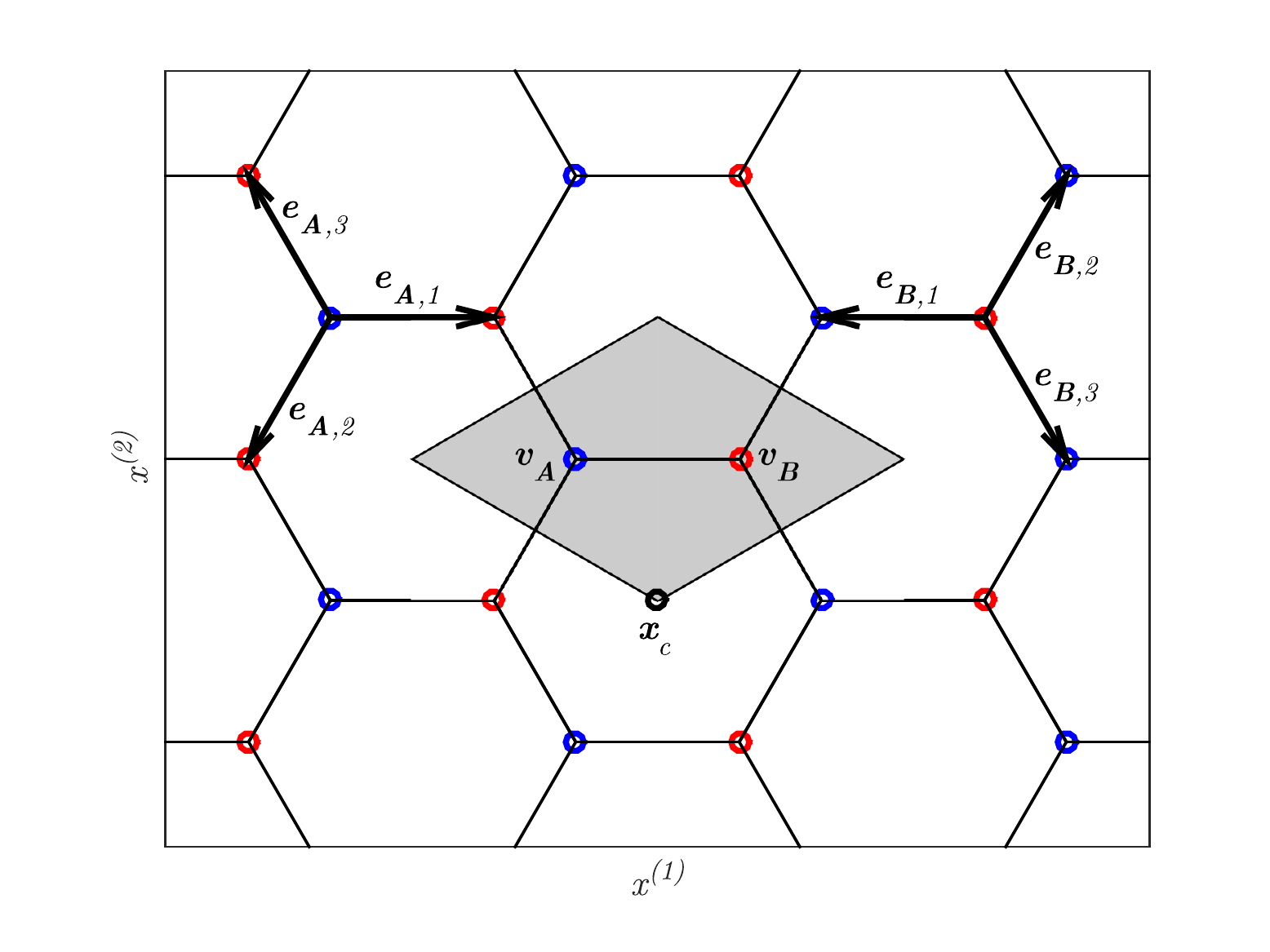}
\caption{\small Diamond-shaped (shaded) fundamental domain, $D$, containing two base points of the honeycomb, $\Honeycomb$: $\bv_A=(0,0)$ and $\bv_B=(1/\sqrt3,0)$. Hexagon center, $\bx_c$,  is a point relative to which the honeycomb potential $V$ is $120^\circ$ rotationally invariant and inversion symmetric. Indicated are: vectors $\be_{A,\nu},\ \nu=1,2,3$ from a typical site in $\Lambda_A$ pointing to its three nearest neighbors in $\Lambda_B$, and $\be_{B,\nu}=-\be_{A,\nu},\ \nu=1,2,3$ from a site typical $\Lambda_B$ pointing to its three nearest neighbors in $\Lambda_A$. }
\label{fig:fundamental-cell}
\end{figure}


\noindent{\bf Nearest neighbors in $\Honeycomb$:} For any fixed $\bv\in\Lambda_A$, the points in $\Honeycomb$ which are nearest to $\bv$ are the three points in the lattice $\Lambda_B$ given by:
\begin{equation}
\bv + \be_{A,1},\ \bv + \be_{A,2},\ {\rm and}\ \bv + \be_{A,3} .
\label{vA-nneighbors}
\end{equation}
Here $\be_{A,\nu},\ \nu=1,2,3$ are shown in Figure \ref{fig:fundamental-cell}. 
Thus, 
\begin{equation}
 \be_{A,\nu}=R^{\nu-1}\be_{A,1}=R^{\nu-1}\begin{pmatrix} \frac{1}{\sqrt3}\\ 0\end{pmatrix},\ \ \nu=1,2,3 ,
 \label{eAB-def}
 \end{equation}
where $R$ is the $120^\circ$ clockwise rotation matrix; see \eqref{Rdef}.

Similarly, for any $\bw\in\Lambda_B$, the points in $\Honeycomb$ which are nearest to $\bw$ are the three points in $\Lambda_A$:
\begin{equation}
\bw + \be_{B,1},\ \bw + \be_{B,2},\ {\rm and}\ \bw + \be_{B,3},
\label{vB-nneighbors}
\end{equation}
where $\be_{B,\nu},\ \nu=1,2,3$ are shown in Figure \ref{fig:fundamental-cell}. Note that $\be_{A,\nu}=-\be_{B,\nu}$.

\section{Atomic potential well, $V_0(\bx)$, and ground state: $(p_0^\lambda(\bx), E_0^\lambda)$}
\label{atomic-well}

Fix a smooth potential well $V_0(\bx)$ on $\R^2$ with the following properties.
\begin{itemize}
\item[($PW_1$)] $-1\le V_0(\bx)\le0,\ \bx\in\R^2$.
\item[($PW_2$)] support $V_0\ \subset\ \{\bx\in\R^2: |\bx|<r_0\}$, where $r_0<r_{critical}$. Here, 
\[ 0.33 |\be_{A,1}|\le r_{critical}< 0.5|\be_{A,1}|,\ \] as determined in Geometric Lemma \ref{euclid1}, and 
$|\be_{A,1}|= 1/\sqrt3$ is the distance between nearest neighbor vertices.
\item[($PW_3$)]
 $V_0(\bx)$ is invariant under a $2\pi/3$ ($120^\circ$) rotation about the origin, $\bx=0$.
\item[($PW_4$)] $V_0(\bx)$ is inversion-symmetric with respect to the origin; $V_0(-\bx)=V_0(\bx)$.
\end{itemize}

Consider the ``atomic'' Schr\"odinger operator $-\Delta + \lambda^2 V_0(\bx)$ in $L^2(\R^2)$. Let $p_0^\lambda(\bx), E_0^\lambda$, respectively, be the ground state eigenfunction and strictly negative ground state eigenvalue of $-\Delta + \lambda^2 V_0(\bx)$. This eigenpair is simple and, by the symmetries of $V_0$,  $p_0^\lambda(\bx)$ is invariant under a  $60^\circ$ rotation about the origin. 
 
 We normalize $p_0^\lambda(\bx)$ so that $p_0^\lambda(\bx)>0$ for all $\bx\in\R^2$, and 
\begin{equation*}
\int_{\R^2} |\ p_0^\lambda(\bx)\ |^2\ d\bx\ =\ 1.
\end{equation*}
Note that since $V_0\in L^\infty(\R^2)$, $p_0^\lambda\in H^2(\R^2)$. The ground state, $p_0^\lambda$, satisfies the
following pointwise bound 
 \begin{align}
 p_0^\lambda(\bx) &\le
  \begin{cases} 
C_1\ e^{-c_1\lambda|\bx|}, & |\bx|\ge r_0+c_0\\
C_2\  \lambda,& |\bx|< r_0+c_0 .
  \end{cases}
  \label{expo-decay1}
  \end{align}
  where $\supp(V_0)\subset B(0,r_0)$ and $c_0>0$, and $C_1, C_2$ and $c_1$ are constants that depend on 
   $V_0$, $r_0$ and $c_0$; see Corollary \ref{cor:expo-decay}.

In addition to hypotheses $(PW_1)-(PW_4)$ on $V_0(\bx)$, we assume the following two properties of the Hamiltonian
 $-\Delta + \lambda^2 V_0(\bx)$:

 
\nit{\bf (GS) Ground state energy upper bound:} For $\lambda$ large, $E_0^\lambda$, the ground state energy of  $-\Delta + \lambda^2 V_0(\bx)$, 
satisfies the upper bound 
\begin{equation} E_0^\lambda\  \le\ -C \lambda^2 .
\label{GS}\end{equation}
Here, $C$ is a strictly positive constant depending on $V_0$. A simple consequence of the variational characterization of $E_0^\lambda$ is the lower bound $E_0^\lambda \ge -\|V_0\|_{_{L^\infty}} \lambda^2=-\lambda^2$. However, the upper bound \eqref{GS} requires further restrictions on $V_0$.
%
%

\nit{\bf (EG) Energy gap property:} There exists $c_{gap}>0$, such that if $\psi\in H^2(\R^2)$ and $\left\langle p_0^\lambda,\psi\right\rangle_{_{L^2(\R^2)}} =\ 0$, then
\begin{equation}
\left\langle\ \left(-\Delta + \lambda^2 V_0\right)\psi,\psi\ \right\rangle_{_{L^2(\R^2)}}\ \ge\ (E_0^\lambda+c_{gap})\ \|\psi\|_{_{L^2(\R^2)}}^2.\label{EG}
\end{equation}
\subsection{Examples of the energy gap property, (EG)}\label{energy-gap}

\begin{enumerate}
\item Let $V_0(\bx)$ be a smooth potential well. For simplicity, assume that $V_0$ has a single non-degenerate minimum at $\bx={\bf 0}$: $\min_{_{\bx\in\R^2}} V_0(\bx)=V_0({\bf 0})=-1$,  $D^2V({\bf 0})=I$, $-1\le V_0(\bx)\le0$ and $V_0(\bx)\to0$ sufficiently  rapidly as $|\bx|\to\infty$.  Then, a simple argument based on the scaling: $\by=\lambda^{1\over2}\bx$,  indicates that for fixed $N\ge1$ and $\lambda>\lambda_N$ sufficiently large, the first $N-$ eigenvalues of $-\Delta_\bx+\lambda^2 V_0(\bx)$
satisfy, to leading order:
\begin{equation}
 E_j^\lambda\ =\ -\lambda^2  + \lambda\ h_j,\ \ 1\le j\le N,
 \end{equation}
 where $h_j$ is the $j^{th}$ eigenvalue of the harmonic oscillator  Hamiltonian $-\Delta_\by +\frac12 |\by|^2$.
Rigorous results are presented in \cites{Simon:83,CFKS:87}.
In this case we have that $c_{gap}$ is of order $\lambda$. 
 %
 %
 
\item Consider piecewise constant cylindrical well, defined by the potential 
\begin{equation}
V_0(\bx)= 
\begin{cases} 
-1 & \textrm{ for $|\bx|<R$}\\
0 &\textrm{ for $|\bx|\ge R$}
\end{cases}
\end{equation}
(Strictly speaking this choice of $V_0$ does not satisfy the above smoothness hypotheses, but it is not
difficult to extend our conclusions concerning $c_{gap}$, to the case where the discontinuity of $V_0$ is smoothed out.)
Solutions which are regular at $|\bx|=0$ and are decaying as $|\bx|\to\infty$ are of the form $e^{im\theta}u(r)$, where
\begin{equation*}
u(r) = \begin{cases} 
\alpha_0\ J_m(\sqrt{\lambda^2-|E|}\ |\bx| ) &\textrm{ for $|\bx|\le R$}\\
\alpha\  I_m(|E|^{\frac12}|\bx|) & \textrm{ for $|\bx|>R$} .
\end{cases}
\end{equation*}
Here, $J_m(r)$ and $I_m(r)$ are respectively the  Bessel and modified Bessel functions of order $m\ge0$;
 $J_m(r)$ is regular at $r=0$ and $I_m(r)$ decays exponentially as $r\to\infty$. 
Eigenvalues, $E$,  are given by solutions of
\begin{equation}
J_m(\sqrt{\lambda^2-|E|}\ R) =  \frac{\sqrt{\lambda^2-|E|}}{|E|^{1\over2}}\ \frac{I_m^\prime(|E|^{1\over2})}{I_m(|E|^{1\over2})}\ J_m^\prime(\sqrt{\lambda^2-|E|}\ R),
\end{equation}
Consider say the first two eigenvalues of $-\Delta + \lambda^2 V_0(\bx)$. For $\lambda$ sufficiently large, these
eigenvalues are given, to leading order by: $E= -\lambda^2+ \left(\rho/R\right)^2 $, where $\rho$ 
varies over the  roots of $J_m(\rho)=0$.  Therefore, in this case $c_{gap}$ is of order $1$.

\end{enumerate}

\subsection{Bounds on the derived (intrinsic) small parameter, $\rho_\lambda$}\label{rho-lambda}
 Let
\begin{equation}
\rho_\lambda \ \equiv\ \int \lambda^2 |V_0(\by)|\ p_0^\lambda(\by)\ p_0^\lambda(\by+\be_{A,1})\ d\by\ . 
\label{rho-lam-def}
\end{equation}

\begin{proposition}\label{prop:rho-lam-bounds}
There exist positive constants $\lambda_\star$, $c_1$ and $c_2$, all depending on $V_0$,  such that for all
 $\lambda>\lambda_\star$,
 \begin{equation}
e^{-c_1\lambda}\ \lesssim\ \rho_\lambda\ \lesssim\ \ e^{-c_2\lambda} 
\label{rho-lambda-bounds}
\end{equation}
\end{proposition}

\nit The proof is given in Section \ref{ME-bounds}.

\begin{remark}\label{rho-indep}
By hypotheses $(PW_1)-(PW_4)$ on $V_0$, the expression $\int \lambda^2 |V_0(\by)|\ p_0^\lambda(\by)\ p_0^\lambda(\by+\be_{I,\nu})\ d\by$ 
is independent of  $I\in\{A,B\}$ and $\nu=1,2,3$, and is equal to $\rho_\lambda$.
\end{remark}

\section{ $\Honeycomb-$ periodization of  $V_0$ and the Bloch- Hamiltonian $H^\lambda(\bk)$}\label{periodization}

We construct a honeycomb potential by summing translates of $V_0(\bx)$ over the honeycomb structure:
\begin{equation}
V(\bx)\ =\ \sum_{\bv\in\Honeycomb}V_0(\bx-\bv),\ \ \Honeycomb=\Lambda_A\cup\Lambda_B.
\label{Vdef}
\end{equation}
Since  ${\rm supp}(V_0)\subset B({\bf 0},r_0)$ (hypothesis $(PW_2)$)  we have that
\begin{equation}
\textrm{if\quad dist$(\bz,\Lambda_A\cup\Lambda_B)> r_0$,\quad then\ V(\bz)=0\ .}
\label{Veq0}
\end{equation}

 We next remark on the symmetries of $V(\bx)$, defined by \eqref{Vdef}. Let $\bx_c$ denote the center  
  of the fundamental hexagon, depicted in Figure \ref{fig:fundamental-cell}. 
  
  \noindent{\it $120^{\circ}$ rotation with respect to $\bx_c$:}\ Given a point $\bx\in\R^2$, its $2\pi/3$ counterclockwise rotation about $\bx_c$, denoted $\widehat{\bx}_{\mathcal R}$, satisfies:
 \begin{equation*}
  \widehat{\bx}_{\mathcal R}-\bx_c \equiv R^*(\bx-\bx_c).
 \end{equation*}
  
   \noindent{\it Inversion with respect to $\bx_c$:}\ Given a point $\bx\in\R^2$, its inversion with respect to $\bx_c$, denoted $\widehat{\bx}_{\mathcal{I}}$, satisfies:
     \begin{equation*}
  \widehat{\bx}_{\mathcal{I}}-\bx_c \equiv -(\bx-\bx_c).
  \end{equation*}
The following proposition on the symmetries of $V(\bx)$, defined in \eqref{Vdef}, can be easily verified. 
\begin{proposition}\label{VisHLP}
$V(\bx)$, defined in \eqref{Vdef},  is a honeycomb lattice potential in the sense of \cites{FW:12}. That is, $V$ is real-valued, smooth, $\Lambda_h-$ periodic and, with respect to $\bx_c$,
 $V$ is rotationally invariant by $120^{\circ}$ and inversion symmetric. That is, for all $\bx\in\R^2$,
 \begin{align}
V(\widehat{\bx}_{\mathcal R}) &\equiv V(\bx_c+R^*(\bx-\bx_c)) = V(\bx), \nn \\ 
V(\widehat{\bx}_{\mathcal I}) &\equiv  V(2\bx_c-\bx) = V(\bx) . \label{V-inv}
 \end{align}
 \end{proposition}

Let 
\begin{equation}
V^\lambda(\bx)\ =\ \lambda^2 \sum_{\bv\in\Honeycomb}V_0(\bx-\bv)\ -\ E_0^\lambda,\ \ \textrm{for}\ \bx\in\R^2 .
\label{Vlam-def}
\end{equation}
 $V^\lambda(\bx)$ is a $\Lambda_h-$ periodic function on $\R^2$ and consequently it may be regarded as a function on $\R^2/\Lambda_h$. By Proposition \ref{VisHLP} (see also  \cites{FW:12}), $V^\lambda(\bx)$ is a honeycomb lattice potential. 
 
 We next study the family of Floquet-Bloch eigenvalue problems:
  \begin{align}
&  H^\lambda\psi\ =\ E\psi,\ \ \psi(\bx+\bv)=e^{i\bk\cdot\bv}\psi(\bx),\ \bv\in\Lambda_h\ \textrm{where}\ \label{evp-Hlam}\\
&\quad  H^\lambda=-\Delta+V^\lambda(\bx) =-\Delta + \lambda^2 V(\bx) - E_0^\lambda,\label{Hlam}
 \end{align}
where  $\bk$ varies over the Brillouin zone, $\brill_h$. Equivalently, we may study, for $\bk\in\brill_h$,
 \begin{align*}
&  H^\lambda(\bk)\ p \ = \ E\ p,\ \ p(\bx+\bv)\ =\ p(\bx),\ \ \bv\in\Lambda_h,\ \  \textrm{where}\\ 
&H^\lambda(\bk) \equiv -(\nabla_\bx+i\bk)^2+ \lambda^2 V(\bx) - E_0^\lambda\ . 
\end{align*}

 An important role in the spectral properties of $H^\lambda$, for large $\lambda$,  is played by the function $\gamma(\bk)$,   defined for $\bk\in\C^2$ by
 \begin{align}
\gamma(\bk) &\equiv \sum_{\nu=1,2,3}e^{ i\bk\cdot \be_{B,\nu}}
 \ =\ e^{ i\bk\cdot \be_{B,1}}
\ \left(\ 1\ +\ e^{i\bk\cdot\bv_1}\ +\ e^{i\bk\cdot\bv_2}\ \right)\ .
\label{gamma-def0}\end{align}
Here, $\be_{B,\nu},\ \nu=1,2,3$ are defined in Section \ref{HS} and indicated in Figure \ref{fig:fundamental-cell}. 

\begin{lemma}\label{gamma0}
\begin{enumerate}
\item[(i)] For $\bk\in\R^2$, $\gamma(\bk)=0$ if and only if $\bk\in \bK + \Lambda_h^*$ or $\bk\in -\bK + \Lambda_h^*$.
\item[(ii)] $120^{\circ}$ rotational invariance:\  $\gamma(R\bk)=\gamma(\bk)$, where $R$ is the clockwise $120^{\circ}$ rotation matrix about $\bk=0$; see \eqref{Rdef}.
\item[(iii)]  Recall that $\mathscr{W}_{_{\rm TB}}(\bk)=|\gamma(\bk)|=|1\ +\ e^{i\bk\cdot\bv_1}\ +\ e^{i\bk\cdot\bv_2}|$, 
 for $\bk\in\R^2$. \\ 
For $\bK_\star\in \left(\bK + \Lambda_h^*\right)\cup\left(-\bK + \Lambda_h^*\right)$, we have the expansion %
 \begin{equation*}
 |\mathscr{W}_{_{\rm TB}}(\bK_\star+\bkappa)|^2\  =\ \frac34\ |\bkappa|^2 + \sum_{|\bfm|=3}\bkappa^\bfm F_{0,\bfm}(\bkappa),\ 
 \end{equation*}
 for $\bkappa\in\R^2$ and $|\bkappa|<c$,  
 where $F_{0,\bfm}$ are smooth functions. 
\end{enumerate}
\end{lemma}
\begin{proof}[Proof of Lemma \ref{gamma0}] (i)  For $\bk\in\R^2$, the three points $1,  e^{i\bk\cdot\bv_1}$ and $ e^{i\bk\cdot\bv_2}$
 lie on the unit circle, and $1+ e^{i\bk\cdot\bv_1} + e^{i\bk\cdot\bv_2}=0$ if and only if their center of mass is zero. Hence, either
(a) $e^{i\bk\cdot\bv_1}=\tau$ and $e^{i\bk\cdot\bv_2}=\overline\tau$ or 
(b) $e^{i\bk\cdot\bv_1}=\overline\tau$ and $e^{i\bk\cdot\bv_2}=\tau$, where $\tau=e^{2\pi i/3}=-1/2+i\sqrt3/2, 
\overline\tau=e^{-2\pi i/3}$ are the non-trivial cube roots of unity. Consider case (a); case (b) is handled similarly.
 For case (a): $\bk=(k^{(1)}, k^{(2)})$ satisfies: $\bk\cdot\bv_1=2\pi/3\ (mod \ 2\pi)$ and 
 $\bk\cdot\bv_2=-2\pi/3\ (mod \ 2\pi)$. That is, 
 \begin{align*}
 k^{(1)}\frac{\sqrt{3}}2+k^{(2)}\frac12\ =\ \frac{2\pi}3+2m_1\pi\ ,\ \qquad  k^{(1)}\frac{\sqrt{3}}2-k^{(2)}\frac12\ =\ -\frac{2\pi}3+2m_2\pi\ ,
 \end{align*}
 where $m_1$ and $m_2$ are integers. Therefore, $k^{(1)}=2\pi (m_1+m_2)/\sqrt3$ and
  $k^{(2)}=4\pi/3+2\pi (m_1-m_2)$ or equivalently $\bk=\bK+m_1\bk_1+m_2\bk_2$,
   where  $\bk_1$, $\bk_2$ and $\bK$ are displayed in \eqref{k1k2-def} and \eqref{KKprime}.
   
   (ii): Let $R$ denote clockwise rotation by $120^\circ$. Then, since the action of $R^*$ on the $\be_{B,\nu},\ \nu=1,2,3$ merely permutes their ordering, we have $\gamma(R\bk)=\gamma(\bk)$.
   
   (iii): Taylor expanding  $e^{-i\bk\cdot\be_{B,1}}\ \gamma(\bk)=1+e^{i\bk\cdot\bv_1}+e^{i\bk\cdot\bv_2}$ 
    (see \eqref{gamma-def0}) about $\bK$, we find at leading order
   \[ e^{-i(\bK+\bkappa)\cdot\be_{B,1}} \gamma(\bK+\bkappa) = i[\tau(\bkappa\cdot\bv_1)+\overline\tau(\bkappa\cdot\bv_2)]+ \mathcal{O}(|\bkappa|^2)
   =-\frac{\sqrt{3}}{2} \left( \bkappa_1-i\bkappa_2\right)+ \mathcal{O}(|\bkappa|^2).\]
   Therefore,  for $\bkappa\in\R^2$ with $|\bkappa|$ small, $ |\mathscr{W}_{_{\rm TB}}(\bK+\bkappa)|^2\equiv|\gamma(\bK+\bkappa)|^2=\frac34|\bkappa|^2+\ 
   \mathcal{O}(|\bkappa|^3)$.
  This completes the proof of the Lemma \ref{gamma0}.
  \end{proof}
 
 \section{Main results}\label{main-results}
 
 In this section we state our main theorem on $-\Delta + \lambda^2 V(\bx)$ in the regime of strong binding $(\lambda\gg1)$. We also state two corollaries on spectral gaps and protected edge states for perturbed honeycomb structures.
  Throughout, we assume hypotheses $(PW_1)-(PW_4)$ on $V_0(\bx)$, and hypotheses {\bf (GS)} and {\bf (EG)} on the ground state energy and spectral gap for $-\Delta+\lambda^2 V_0(\bx)$. These were delineated in Section \ref{atomic-well}.
 \begin{theorem}[Low-lying dispersion surfaces in the strong binding regime]\label{main-theorem}
 Let $E_0^\lambda$ denote the ground state eigenvalue of the atomic Hamiltonian, $-\Delta +\lambda^2 V_0$ (Section \ref{atomic-well}).  Let $V(\bx)$ denote the $\Lambda_h-$ periodic potential obtained by summing $V_0(\bx)$ over the honeycomb structure, $\Honeycomb= (\bv_A+\Lambda_h)\cup (\bv_B+\Lambda_h)$; see  \eqref{Vdef}. We consider the family of Floquet-Bloch eigenvalue problems for the periodic Schr\"odinger operator $-\Delta +\lambda^2 V(\bx)$, depending on the quasi-momentum $\bk\in\R^2$;
  see  \eqref{fl-bl-evp} and \eqref{fl-bl-evp-per}. 
 
 Fix $\beta_{max}$, a non-negative integer.  There exist positive constants $\lambda_\star>0$ (sufficiently large),  $\widehat{C}>0$ and $c, c_{**} >0$,  depending only on $V_0(\bx)$ and  $\beta_{max}$,
  such that for all
  $\lambda>\lambda_\star$ the following hold:
  \begin{enumerate}
  \item For each $\bk\in\R^2$, there are precisely two eigenvalues, $E$, of the operator $-(\nabla+i\bk)^2+\lambda^2 V(\bx)$ with periodic boundary conditions, satisfying:
   \[ E_0^\lambda\ -\ \widehat{C}\ \rho_\lambda\ \le\ E\ \le\ 
 E_0^\lambda\ +\ \widehat{C}\ \rho_\lambda\ ,\]
 where $(E_0^\lambda , p_0^\lambda)$ is the ground state eigenpair of  $-\Delta+\lambda^2V_0$ and 
\begin{equation*}
 \rho_\lambda=  \lambda^2\int\ |V_0(\by)|\ p_0^\lambda(\by)\ p_0^\lambda(\by+\be_{A,1})\ d\by 
 \end{equation*}
 satisfies the bounds: $e^{-c_1\lambda}\lesssim \rho_\lambda\lesssim e^{-c_2\lambda}$,
  provided in Proposition \ref{prop:rho-lam-bounds}.

  \item For each $\bk\in\R^2$, we denote the two eigenvalues in part (1) by $E^\lambda_-(\bk)\le E^\lambda_+(\bk)$.
  These are equal to $E_1(\bk)$ and $E_2(\bk)$, the first two  band dispersion functions of $-\Delta+\lambda^2 V$.
\item For $\bk\in\B_h$, the graphs of $\bk\mapsto E^\lambda_\pm(\bk)$ intersect at the six vertices of the regular hexagon,  $\B_h$, at the shared energy-level, $E_D^\lambda$. The pairs $(\bK_\star,E_D^\lambda)$, where $\bK_\star$ varies over the vertices of  $\B_h$, are called Dirac points of $-\Delta+\lambda^2 V$.

In particular, for each vertex $\bK_\star$ of $\B_h$, the operator $-(\nabla+i\bK_\star)^2+\lambda^2V(\bx)$, with periodic boundary conditions, has a double eigenvalue:
   \begin{equation}
   E^\lambda_D=E_0^\lambda+\rho_\lambda h_0^\lambda,\qquad |h_0^\lambda|\lesssim e^{-c\lambda} .
   \nn\end{equation}

 \item Convergence of dispersion surfaces: Let  $\mathscr{W}_{_{\rm TB}}(\bk)\ \equiv\ |1+e^{i\bk\cdot\bv_1}+e^{i\bk\cdot\bv_2}|$; see \eqref{tb-dispersion} and Lemma \ref{gamma0}. 
 
\noindent (a) \underline{ Low-lying dispersion surfaces away from Dirac points:}\\  
For all $\bk\in\R^2$ such that  $ \mathscr{W}_{_{\rm TB}}(\bk)\ \ge\lambda^{-\frac14}$:

 \begin{equation}
  \left| \partial_\bk^\beta\left\{ \left( E^\lambda_\pm(\bk) - E_D^\lambda \right)/\rho_\lambda 
   - \left[  \pm \mathscr{W}_{_{\rm TB}}(\bk) \right] \right\} \right|
  \le  e^{-c\lambda} ,\ \ |\beta|\le\beta_{max} .
   \label{near-conv}
  \end{equation}
%
%

\noindent (b) \underline{Low-lying dispersion surfaces near Dirac points:} \\
    For any vertex, $\bK_\star$, of $\brill_h$ and all $\bk$ satisfying
   $0<|\bk-\bK_\star|<c_{\star\star}$:
\begin{align*}
 \left| \partial_\bk^\beta\left\{ \left( E^\lambda_\pm(\bk) - E_D^\lambda \right)/\rho_\lambda 
   - \left[  \pm \mathscr{W}_{_{\rm TB}}(\bk) \right] \right\} \right|
&\le e^{-c\lambda} |\bk-\bK_\star|^{1-|\beta|} ,\ \ |\beta|\le\beta_{max} .
\end{align*}
  \end{enumerate}
  \end{theorem}
  
\noindent Theorem \ref{main-theorem} is a consequence of Proposition \ref{Proposition-RDSADP} and Proposition \ref{Proposition-nearDPs}.

We furthermore establish convergence of the scaled resolvent of $-\Delta+\lambda^2 V(\bx)$ to that of the tight-binding
 Hamiltonian.
 \begin{theorem}[Scaled convergence of the resolvent]\label{res-conv}
 Let $H^\lambda=-\Delta+\lambda^2V(\bx)-E_D^\lambda$ and introduce the scaled operator
 $\widetilde{H}^\lambda=H^\lambda/\rho^\lambda$.
  Further, let  $\widetilde{H}_\bk^\lambda\ =\  \widetilde{H}^\lambda\Big|_{_{L^2_\bk}}:H^2_\bk(\R^2/\Lambda_h)\to L^2_\bk(\R^2/\Lambda_h)$.
 Define  the mapping $\mathscr{J}_\bk^\lambda:\C^2\to L^2_\bk$:
 \begin{equation}
 \begin{pmatrix}\alpha\\ \beta\end{pmatrix}\mapsto \mathscr{J}_\bk^\lambda\begin{bmatrix}\alpha_A\\ \alpha_B\end{bmatrix} =
 \alpha_A P_{A,\bk}^\lambda(\bx)+\alpha_B P_{B,\bk}^\lambda(\bx),
  \label{scrJ}
 \end{equation}
 where $P_{I,\bk}^\lambda$, $I=A,B$,  defined in \eqref{P-bk-I}, denote approximate Floquet-Bloch modes, defined by a weighted sum of translates in $\Lambda_A$ (respectively 
 $\Lambda_B$) of the ground state, $\varphi_0^\lambda$, of $V_0$.

  Then, for any fixed $z\in\C\setminus\R$,
 \begin{equation}
 \|\ \left(\ \widetilde{H}_\bk^\lambda-zI\ \right)^{-1}\
  -\ \mathscr{J}_\bk^\lambda\ \left(\ H_{_{\rm TB}}(\bk)-zI\ \right)^{-1}\ \left(\mathscr{J}_\bk^\lambda\right)^*\ \|_{_{L^2_\bk\to L^2_\bk}}
  \lesssim e^{-c\lambda}\ ,
  \label{res-diff}
  \end{equation}
 uniformly in $\bk\in\brill_h$, for $\lambda>\lambda_\star$ .
  \end{theorem}
 \nit Theorem \ref{res-conv} is proved in Section \ref{resolvent}.
  In the following two subsections we discuss consequences of Theorem \ref{main-theorem}.
 
 \subsection{Spectral gaps for  $\mathcal{P}\mathcal{T}-$ breaking perturbations}\label{Gaps}

 \begin{corollary}\label{spectral-gaps} 
Let $V(\bx)$ be in the class of honeycomb potentials studied in Theorem \ref{main-theorem}.
Consider the perturbed  honeycomb Schr\"odinger
\begin{equation} H^{\lambda,\eta}=
-\Delta + \lambda^2 V(\bx) + \eta W(\bx),\label{Hlameta-def}
\end{equation}
where $\lambda^2>0$ and  $\eta$ is a real parameter and $W$ is such that:
\begin{enumerate}
\item  $W(\bx)$ is real-valued and $\Lambda_h$ periodic.
\item  $W(\bx)$ breaks inversion symmetry. Specifically, we take $\bx_c=0$ and assume 
 \begin{equation}
 W(-\bx)=-W(\bx);\label{Wodd}
 \end{equation}
 compare with \eqref{V-inv}. 
 \item 
 \begin{equation}
  \vartheta^\lambda_\sharp\ \equiv\ \left\langle \Phi^\lambda_1,W\Phi^\lambda_1\right\rangle\ne0,
  \label{vartheta}
  \end{equation}
  where $\Phi^\lambda_1$ is defined in Definition \ref{dirac-pt-defn} in Section \ref{dirac-points}.
\end{enumerate}
 Then, there exists a large constant $\lambda_\star>0$, such that for all $\lambda>\lambda_\star$ the following holds:
there is a constant $\eta_\star>0$, where $\eta_\star$ is sufficiently small and depends on $\lambda, V$ and $W$  such that 
if  $0<|\eta|<\eta_\star$, then the spectrum of $H^{\lambda,\eta}$ has a gap
   with energy $E_D^\lambda$ in its interior.  
 \end{corollary}

 \begin{proof}[Proof of Corollary \ref{spectral-gaps}]
   As shown in Corollary \ref{tb-dirac-pts}, there exist  Dirac points, 
   $(\bK_\star,E^\lambda_D)$, at the vertices, $\bK_\star$,  of $\brill_h$ for all $\lambda$ sufficiently large. For $\eta$ small, let $\bk\mapsto E^{(\lambda,\eta)}_\pm(\bk)$, denote the dispersion maps which are small perturbations of the maps  $\bk\mapsto E^{(\lambda,0)}_\pm(\bk)\equiv E^{\lambda}_\pm(\bk)$.

 The  proof of Corollary \ref{spectral-gaps} is based on the following two steps.
   \begin{enumerate}
   \item {\bf Claim:} There exists a constant $\lambda_\star$, such that for all $\lambda>\lambda_\star$ the following holds:
    There exist small constants $c_1>0$ and  $\eta_1>0$ such that 
     such that for all $0<|\eta|<\eta_1$  and all $\bk\in\brill_h$ satisfying $|\bk-\bK_\star|<c_1\lambda^{-1}$, where $\bK_\star$ varies over the vertices of $\brill_h$,
     \begin{align}
     \label{near-dispersion}
     E_\pm^{(\lambda,\eta)}(\bk)\ &=\ E_D^\lambda\ \pm\
      \sqrt{|\lamsharplam|^2\ |\bk-\bK_\star|^2\ +\ (\vartheta^\lambda_\sharp)^2\ \eta^2\ +\ \mathcal{O}(|\eta|^3)} \\
     & \qquad\qquad\qquad \times\ \left(\ 1+\mathcal{O}(|\bk-\bK_\star|\ +\ |\eta|\ )\ \right) \nn
     \end{align}
     
     The proof of this claim is via a Lyapunov-Schmidt reduction strategy 
    very similar to that in Appendix F of \cites{FLW-MAMS:17}. The essential difference is the need to separately treat  quasi-momenta
    within and outside a small $\lambda-$ dependent neighborhood of vertices $\bK_\star$ of $\brill_h$.
Expand solutions of the Floquet-Bloch eigenvalue problem for $H^{\lambda,\eta}$ in the form $\psi=\alpha_1\phi_1+\alpha_2\phi_2+\widetilde\psi$, where $\phi_j(\bx)=e^{-i\bK\cdot\bx}\Phi^\lambda_j,\ j=1,2$ are $\Lambda_h-$ periodic and $\widetilde\psi\perp\rm{span}\{\phi_1,\phi_2\}$. A coupled system for $(\alpha_1,\alpha_2,\widetilde\psi)$ is obtained by projecting the eigenvalue problem onto $\rm{span}\{\phi_1,\phi_2\}$  and its orthogonal complement.
 The projection onto  $\rm{span}\{\phi_1,\phi_2\}$ gives a system of two equations for $\alpha_1$ and $\alpha_2$, which depends on $\widetilde{\psi}$. To construct the mapping $(\alpha_1,\alpha_2)\mapsto \widetilde{\psi}[\alpha_1,\alpha_2]$ requires smallness of:
  \[ |\bk-\bK_\star| \ \left\|\nabla_\bx\left(-(\nabla+i\bK_\star)^2+\lambda^2V-E^\lambda_D\right)^{-1}P_\perp\right\|_{_{L^2(\R^2/\Lambda_h)\to L^2(\R^2/\Lambda_h)}},  \]
  where $P_\perp$ is the orthogonal projection onto $\rm{span}\{\phi_1,\phi_2\}^\perp$. The energy estimates of Section \ref{en-estimates} and in particular \eqref{deriv-bounds} in Proposition \ref{prop2-resolvent} below imply that this quantity is bounded by $C|\bk-\bK_\star|\lambda$, for some $C>0$. It follows that there exist small positive constants, $c_1$ and $\eta_1$, and a large constant, $\lambda_\star$, such that for $\lambda>\lambda_\star$, $|\bk-\bK_\star|<c_1\lambda^{-1}$ and $|\eta|<\eta_1$ we obtain a reduction to a two-dimensional problem, which yields \eqref{near-dispersion}.

     By \eqref{near-dispersion}, we can choose $0<c_2<c_1$ and $0<\eta_2<\eta_1$ such that for all $\bK_\star$ varying over the set of vertices of $\brill_h$, and all $0<|\eta|<\eta_2$ and all 
     $\bk\in\brill_h$, such that $|\bk-\bK_\star|<c_2\lambda^{-1}$:  the energies indicated by the maps:
     $\bk\mapsto E_\pm^{(\lambda,\eta)}(\bk)$ lie outside the interval about $E_D^\lambda$:       $(E_D^\lambda-\frac12\vartheta^\lambda_\sharp\eta,E_D^\lambda+\frac12\vartheta^\lambda_\sharp\eta)$. 
 
 %
   \item  Consider now quasimomenta, $\bk$ in the compact set:
   \begin{equation*}
  \mathscr{C}(c_2,\lambda) \equiv \left\{ \bk\in\brill_h : |\bk-\bK_\star|\ge c_2\lambda^{-1},\ \textrm{where $\bK_\star$ varies over the vertices of $\brill_h$} \right\} .
   \end{equation*}
First let $\eta=0$.  Theorem \ref{main-theorem} implies that for such quasi-momenta, the rescaled dispersion maps: 
   $\bk\mapsto\mu^\lambda_\pm(\bk) \equiv \left(\ E^{(\lambda,0)}_\pm(\bk)\ -\ E_D^\lambda\ \right)/\rho_\lambda$
are uniformly close to the Wallace dispersion relation, $\pm\mathscr{W}_{_{\rm TB}}(\bk)$ for $\lambda>\lambda_\star$ sufficiently large; see \eqref{near-conv}.
In particular, for $\bk\in\mathscr{C}(c_2,\lambda)$
\[
\left|\ E_\pm^{(\lambda,0)}(\bk)\ -\ \left(\ E_D^\lambda\pm\rho_\lambda\mathscr{W}_{_{\rm TB}}(\bk)\ \right)\ \right|\ \le\ \rho_\lambda e^{-c\lambda}
\]
and therefore, for $\lambda>\lambda_\star$ (if necessary, take $\lambda_\star$ to be larger than our earlier choices), 
\begin{equation} \left|\ E^{(\lambda,0)}_\pm(\bk)\ -\ E_D^\lambda\ \right|\ \ge\ \frac12\ \rho_\lambda \min_{\bk\in\mathscr{C}(c_2)}\  \mathscr{W}_{_{\rm TB}}(\bk)\ge C\ \rho_\lambda>0,
\label{le-c2-est}
\end{equation}
since within $\brill_h$, the Wallace dispersion relation  yields  energy $0$ only at the vertices of $\brill_h$ (Lemma \ref{gamma0}).

 Finally, we turn  to the perturbed dispersion maps $\bk\mapsto E^{(\lambda,\eta)}_\pm(\bk)$ on the compact quasi-momentum set 
$ \mathscr{C}(c_2,\lambda)$. By perturbation theory, about the eigenvalues $E^{(\lambda,0)}_\pm(\bk)$ 
($\bk\in\mathscr{C}(c_2,\lambda)$), there is a small and positive constant, $g_0$, such that for $\eta_\star(\lambda)\equiv g_0\rho_\lambda>0$, with $\lambda>\lambda_\star$ large enough:
\begin{equation}
 \left|\ E^{(\lambda,\eta)}_\pm(\bk)\ -\ E_D^\lambda\ \right|\ \ge\ \frac12 C\ \rho_\lambda>0
\label{ge-c2-est}
\end{equation}
for all $0<|\eta|<\eta_\star(\lambda)$ and all $\bk\in\mathscr{C}(c_2,\lambda)$. 
\end{enumerate}

By \eqref{le-c2-est} and \eqref{ge-c2-est} we have for all  $\bk\in\brill_h$, all $\lambda>\lambda_\star$
and all $0<\eta<\eta_\star(\lambda)$:
\begin{equation}
 \left|\ E^{(\lambda,\eta)}_\pm(\bk)\ -\ E_D^\lambda\ \right|\ \ge\ 
 c_\star(\lambda,\eta)\ \equiv\ \min\left\{\ \frac12 C\ \rho_\lambda\ ,\ \frac12\vartheta_\sharp\eta\ \right\}>0,
\label{ge-c2-est2}
\end{equation}
Under the above conditions on $\lambda$ and $\eta$, the energies indicated by the maps:
     $\bkappa\mapsto E_\pm^{(\lambda,\eta)}(\bk)$, where $\bk$ varies over the full Brillouin zone, $\brill_h$, 
     lie outside the open interval about $E_D^\lambda$:       $(E_D^\lambda-c_\star,E_D^\lambda+c_\star)$.

The proof of Corollary \ref{spectral-gaps} is now complete. 
\end{proof}

 \subsection{Protected edge states rational edges}\label{Protected-States}

Edge states are time-harmonic solutions of the  Schr\"odinger equation, which are propagating parallel to a
 line-defect (edge) and are  localized transverse to it; see the schematic in Figure \ref{fig:mode_schematic}.
  In \cites{FLW-2d_edge:16} (see also \cites{FLW-2d_materials:15}), we develop a theory of protected edge states for  honeycomb structures, perturbed by a class of line-defects (domain walls) in the direction of an element of $\Lambda_h$. We first present a terse summary. 
  
 Recall the spanning vectors of the equilateral triangular lattice, $\bv_1$ and $\bv_2$; see \eqref{v12-def}.  Given  a pair of  non-negative integers $a_1, b_1$, which are relatively prime, let $\vtilde_1=a_1\bv_1+b_1\bv_2$. We call the line $\R\vtilde_1$ the $\vtilde_1-$ edge. Since $a_1, b_1$ are relatively prime,  there exists a second pair of  relatively prime integers: $a_2,b_2$ such that $a_1b_2-a_2b_1=1$.  Set $\vtilde_2 = a_2 \bv_1 + b_2 \bv_2$. 

It follows that $\Z\vtilde_1\oplus\Z\vtilde_2=\Z\bv_1\oplus\Z\bv_2=\Lambda_h$.
Since $a_1b_2-a_2b_1=1$, we have dual lattice vectors $\ktilde_1, \ktilde_2\in\Lambda_h^*$, given by
$ \ktilde_1=b_2\bk_1-a_2\bk_2$ and $ \ktilde_2=-b_1\bk_1+a_1\bk_2$, which satisfy
$\ktilde_\ell \cdot \vtilde_{\ell'} = 2\pi \delta_{\ell, \ell'}$, $1\leq \ell, \ell' \leq 2$.
Note that $\Z\ktilde_1\oplus\Z\ktilde_2=\Z\bk_1\oplus\Z\bk_2=\Lambda^*_h$.  
%
The choice $\vtilde_1=\bv_1$ (or equivalently $\bv_2$) generates a \emph{zigzag edge} and the choice $\vtilde_1=\bv_1+\bv_2$ generates the \emph{armchair edge};  see Figure \ref{fig:edges}.

\begin{figure}
\centering
\includegraphics[width=0.65\textwidth]{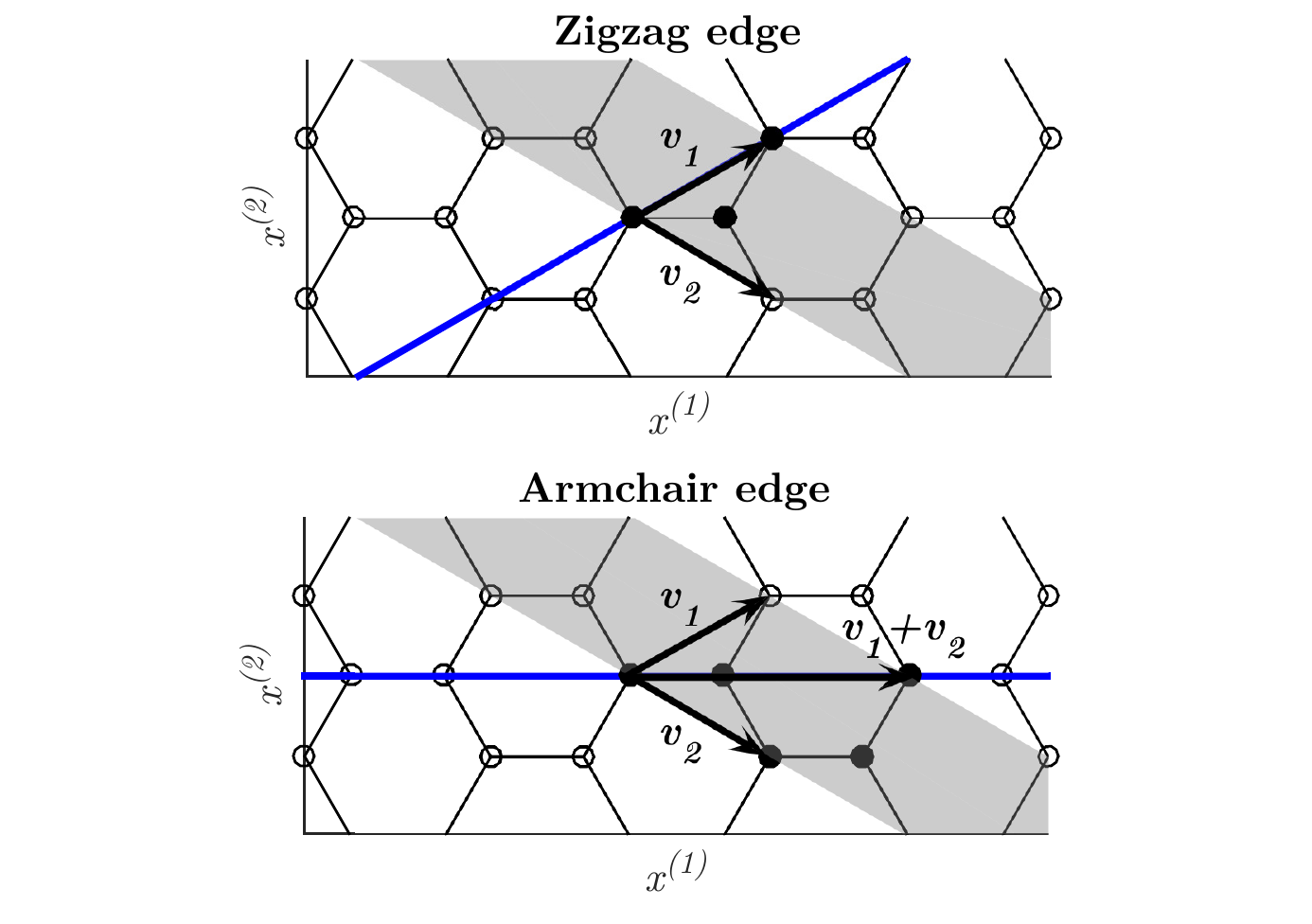}
\caption{\small
Bulk honeycomb structure, ${\bf H} =  ({\bf A } + \Lambda_h) \cup   ( {\bf B} + \Lambda_h)$.
{\bf Top panel}: Zigzag edge,  $\vtilde_1=\bv_1, \ktilde_2=\bk_2 $, $\R\bv_1 = \{\bx : \bk_2\cdot\bx=0\}$ (blue line). 
Shaded region is the fundamental domain of  cylinder, $\Sigma_{ZZ}$, corresponding to the zigzag edge.
{\bf Bottom panel}: Armchair edge, $\vtilde_1=\bv_1+\bv_2,  \ktilde_2=\bk_2-\bk_1$, $\R\left(\bv_1+\bv_2\right) = \{\bx : (\bk_1-\bk_2)\cdot\bx=0\}$ (blue line). 
Fundamental domain of cylinder, $\Sigma_{AC}$, corresponding to the armchair edge.
 }
\label{fig:edges}
\end{figure}

\begin{figure}
\centering
\includegraphics[width=0.6\textwidth]{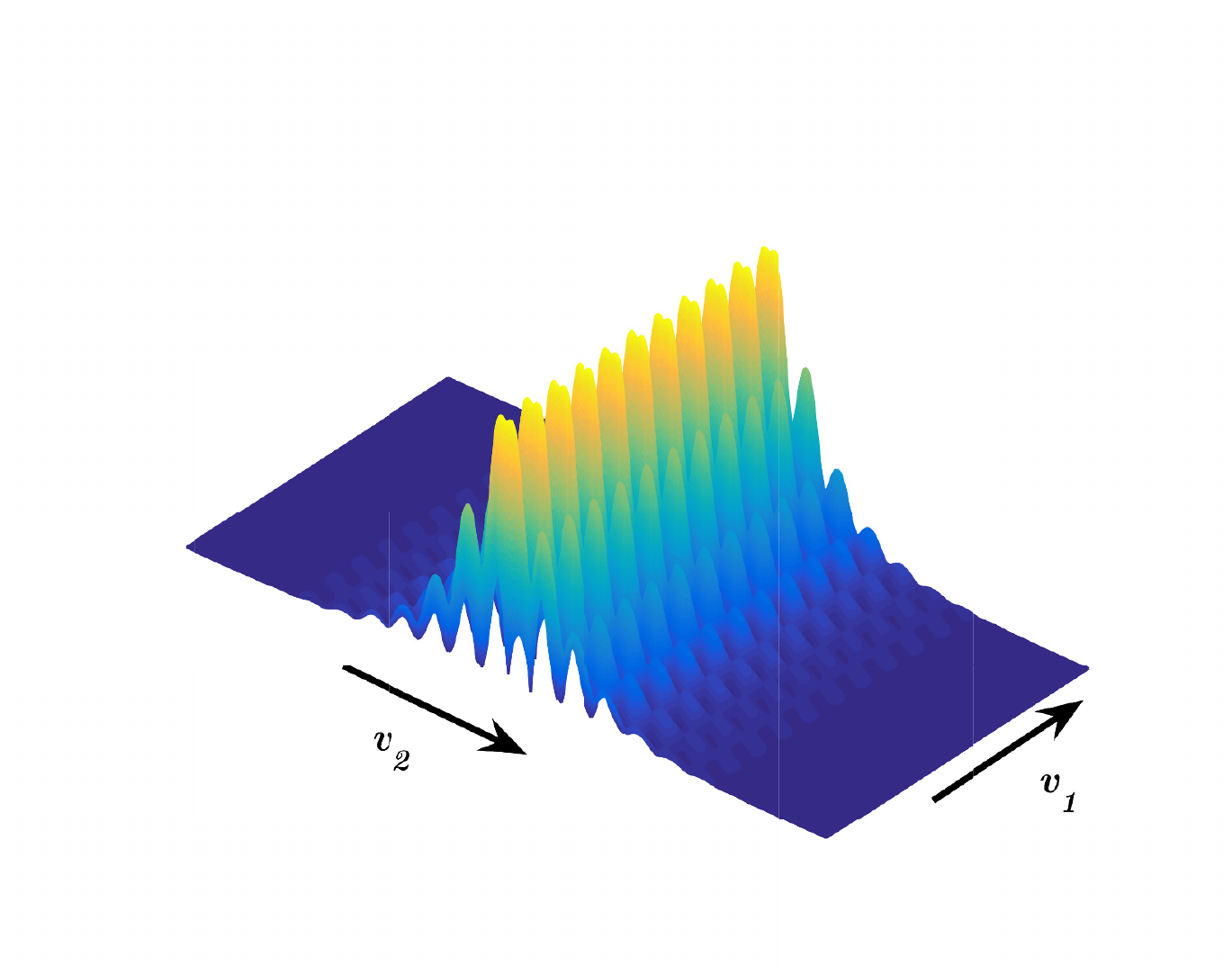}
\caption{\small
Edge state -- propagating (plane-wave like) parallel to a zigzag edge ($\R\bv_1$) and localized transverse to the edge}\label{fig:mode_schematic}
\end{figure}
%
%
 
Introduce the  family of Hamiltonians, depending on the real parameters $\lambda$ and $\delta$:
\begin{equation*}
H^{(\lambda,\delta)} \equiv -\Delta + \lambda^2 V(\bx) + \delta\kappa(\delta\ktilde_2\cdot \bx)W(\bx). 
\end{equation*}
$H^{(\lambda,0)}=-\Delta +\lambda^2 V(\bx)$ is the Hamiltonian for the unperturbed (bulk) honeycomb structure, studied in Theorem \ref{main-theorem}.
Here, $\delta$ will be taken to be sufficiently small, and $W(\bx)$ is $\Lambda_h-$ periodic
and odd with respect to the center, $\bx_c$, {\it i.e.} $W(2\bx_c-\bx)=-W(\bx)$. Thus, $W$ breaks inversion symmetry.  See Corollary \ref{spectral-gaps}.
The function $\kappa(\zeta)$ defines a \emph{domain wall}. We choose $\kappa$ to be sufficiently smooth and to satisfy $\kappa(0)=0$ and    $\kappa(\zeta)\to\pm\kappa_\infty\ne0$
as $\zeta\to\pm\infty$, \emph{e.g.} $\kappa(\zeta)=\tanh(\zeta)$. Without loss of generality, we assume $\kappa_\infty>0$. 

Note that $H^{(\lambda,\delta)}$ is invariant under translations  parallel to the  $\vtilde_1-$ edge, $\bx\mapsto\bx+\vtilde_1$, and hence there is a well-defined \emph{parallel quasi-momentum}, denoted $\kpar$.  Furthermore, $H^{(\lambda,\delta)}$ transitions adiabatically between the  asymptotic Hamiltonian $H_-^{(\lambda,\delta)}=H^{(\lambda,0)} - \delta\kappa_\infty W(\bx)$ as $\ktilde_2\cdot\bx\to-\infty$ to the asymptotic Hamiltonian $H_+^{(\lambda,\delta)}=H^{(\lambda,0)} + \delta\kappa_\infty W(\bx)$ as $\ktilde_2\cdot\bx\to\infty$. 
The  domain wall modulation of $W(\bx)$ realizes a phase-defect across the edge $\R\vtilde_1$.
 A variant of this construction was used in \cites{FLW-PNAS:14,FLW-MAMS:17} to interpolate between different asymptotic 1D dimer periodic potentials. 
 
We seek  \emph{$\vtilde_1-$ edge states} of $H^{(\lambda,\delta)}$, which are spectrally localized near the Dirac point, $(\bK_\star,E^\lambda_D)$, where $\bK_\star$ is a vertex of $\brill_h$.  
These are non-trivial solutions $\Psi$, with energies $E\approx E^\lambda_D$,  of the \emph{$\kpar-$ edge state eigenvalue problem (EVP):}
\begin{align}
H^{(\lambda,\delta)}\Psi &= E\Psi,\label{edge-evp}\\
\Psi(\bx+\vtilde_1) &= e^{i\kpar}\Psi(\bx) , \label{edge-bc1}\\
|\Psi(\bx)| \to \ 0\ &{\rm as} \ |\ktilde_2\cdot\bx|\to\infty , \label{edge-bc2}
\end{align}
for $\kpar\approx\bK_\star\cdot\vtilde_1$. The boundary conditions \eqref{edge-bc1} and \eqref{edge-bc2} imply, respectively, propagation parallel to, and localization transverse to, the edge $\R\vtilde_1$.

The edge state eigenvalue problem \eqref{edge-evp}-\eqref{edge-bc2} may be formulated in an appropriate Hilbert space. Introduce the cylinder $\Sigma\equiv \R^2/ \Z\vtilde_1$. 
Denote by  $H^s(\Sigma),\ s\ge0$, the  Sobolev spaces of functions defined on $\Sigma$. The pseudo-periodicity and decay conditions \eqref{edge-bc1}-\eqref{edge-bc2} are encoded by requiring $ \Psi \in H^s_\kpar(\Sigma)=H^s_\kpar$, for some $s\ge0$,  where
\[H^s_\kpar\equiv \ \left\{f : f(\bx)e^{-i\frac{\kpar}{2\pi}\ktilde_1\cdot\bx}\in H^s(\Sigma) \right\}.\]
We then formulate the EVP \eqref{edge-evp}-\eqref{edge-bc2} as:
\begin{equation}
H^{(\lambda,\delta)}\Psi = E\Psi, \quad \Psi\in H^2_{\kpar}(\Sigma).
\label{EVP2}
\end{equation}

Theorem 7.3 and Corollary 7.4 in \cites{FLW-2d_edge:16} (see also Theorem 4.1 of \cites{FLW-2d_materials:15}) formulate general hypotheses on the 
bulk honeycomb structure, $V(\bx)$, the domain-wall function, $\kappa(\zeta)$, and the asymptotic perturbation of the bulk structure, $W(\bx)$, which imply the existence of topologically protected $\vtilde_1-$ edge states, constructed as non-trivial eigenpairs $\delta\mapsto (\Psi^\delta,E^\delta)$ \eqref{EVP2} with $\kpar=\bK_\star\cdot\vtilde_1$, defined for all $|\delta|$ sufficiently small. This branch of solutions bifurcates, for $\delta\ne0$, from the intersection of spectral bands at $E_D$.
  
  The key hypothesis is a {\it spectral no-fold condition}, associated with $\vtilde_1-$ edge. In \cites{FLW-2d_edge:16}, this condition was verified for the zigzag edge, for 
 general honeycomb Schr\"odinger operators, $-\Delta +\eps V$, with low contrast; in particular, where
  \[\eps V_{1,1}\ \equiv\ \eps\ \int_{D}e^{-i(\bk_1+\bk_2)\cdot\bx} V(\bx)\ d\bx>0,\]
  with $|\eps|$ sufficiently small.

The main result of the present article, Theorem \ref{main-theorem}, immediately implies the validity of the spectral no-fold condition for all $\lambda>\lambda_\star$ sufficiently large, and hence 
  \begin{corollary}[Protected edge states for \underline{rational edges}]\label{rational-edge-states}
 In the strong binding regime, there exist  topologically protected edge states for the large class of  {\it rational edges}
  presented in Remark \ref{remarks_on_edgestates} below.  
   \end{corollary}

\nit The class of admissible edges in Corollary \ref{rational-edge-states}  is presented in Remark \ref{remarks_on_edgestates} below.
 
  We next present a brief explanation of the spectral no-fold condition. 
    Given an edge in the direction of  $\vtilde_1\in \Lambda_h$,
    there is an associated {\it dual slice} of the band structure, which passes through a chosen Dirac point, $(\bK_\star,E_D)$.  The dual slice is the set of curves $\xi\in[-1/2,1/2)\mapsto E_b(\bK_\star+\xi\ktilde_2),\ b\ge1$, where $\ktilde_2\cdot\vtilde_1=0$. The significance of the dual slice is that the edge-states constructed in \cites{FLW-2d_edge:16},
     which are propagating parallel to and localized transverse to $\vtilde_1$, 
      are superpositions of Floquet-Bloch modes with the quasi-momenta of the dual slice.
       In the current setting, we verify the spectral no-fold condition for the two lowest spectral bands
     of $-\Delta +\lambda^2V(\bx)$ for $\lambda>\lambda_\star$ sufficiently large. 
     The spectral no-fold condition states that the line $E=E_D$
     intersects the pair of dispersion curves 
     $\xi\in[-1/2,1/2)\mapsto E_1(\bK_\star+\xi\ktilde_2),\ E_2(\bK_\star+\xi\ktilde_2)$, only 
     at Dirac points.

     Figure \ref{fig:dispersion_curves_lambda} illustrates such pairs of dispersion curves, associated with several edges: zigzag, armchair and $(2,1)$ and two choices of $\lambda$: $\lambda=1$ and $5$. For $\lambda=1$, the no-fold condition holds only for the zigzag edge, but it holds for all three types of edges if $\lambda=5$.
     
\begin{remark}
\label{remarks_on_edgestates}
     \begin{enumerate}
     \item   At present, the results in \cites{FLW-2d_edge:16} are stated for edges, $\R\vtilde_1$, for which 
      $\xi\mapsto \bK_\star+\xi\ktilde_2$, $|\xi|\le1/2$, passes through only one independent Dirac point. 
      There are edges for which the dual slice passes through two independent Dirac points, {\it i.e.} where $\xi\mapsto \bK_\star+\xi\ktilde_2$, $|\xi|\le1/2$
       intersects both lattices $\bK_\star+\Lambda_h^*$ and $\bK_\star^\prime+\Lambda_h^*$. We are currently working on extending  the methods of \cites{FLW-2d_edge:16}   to this case. See  \cites{FLW-2d_materials:15} for a discussion and numerical simulations of the multi-branch bifurcation from the intersection of bands in this case. 
            %
     \item Our results imply that  given an edge in the direction $\vtilde^{(m,n)}_1=m\bv_1+n\bv_1$, where $m$ and $n$ are relatively prime integers, there is a threshold $\lambda_\star(m,n)$, such that for all $\lambda>\lambda_\star(m,n)$, the spectral no-fold
     condition holds and there exist edge states  \cites{FLW-2d_materials:15}. 
          \item As discussed above, in \cites{FLW-2d_edge:16} we prove the existence of zigzag edge states ($\vtilde^{(m,n)}_1=\vtilde^{(1,0)}_1$) for a class of line-defect perturbations of Schr\"odinger operators with weak honeycomb potentials: $-\Delta+\eps V_h(\bx),\ |\eps|\ll1$,  satisfying the additional condition:  $\eps V_{1,1}>0$.
      This analysis also suggests that if $\eps V_{1,1}<0$, then there are edge quasi-modes, whose energy slowly leaks into the bulk. By appropriate choice of atomic potential well 
     $V_0(\bx)$, we may  arrange for $V(\bx)$, such that $\eps V_{1,1}<0$; see Appendix A of \cites{FLW-2d_edge:16}.  For honeycomb potentials, $V$, arising from such atomic potentials, Corollary \ref{rational-edge-states}
      shows that there is necessarily a transition from a ``leaky'' resonance mode to a truly localized mode along the edge, for $\lambda$ above some finite $\lambda_\star$.
          \end{enumerate}
\end{remark}
\begin{figure}
\centering
\includegraphics[width=0.75\textwidth]{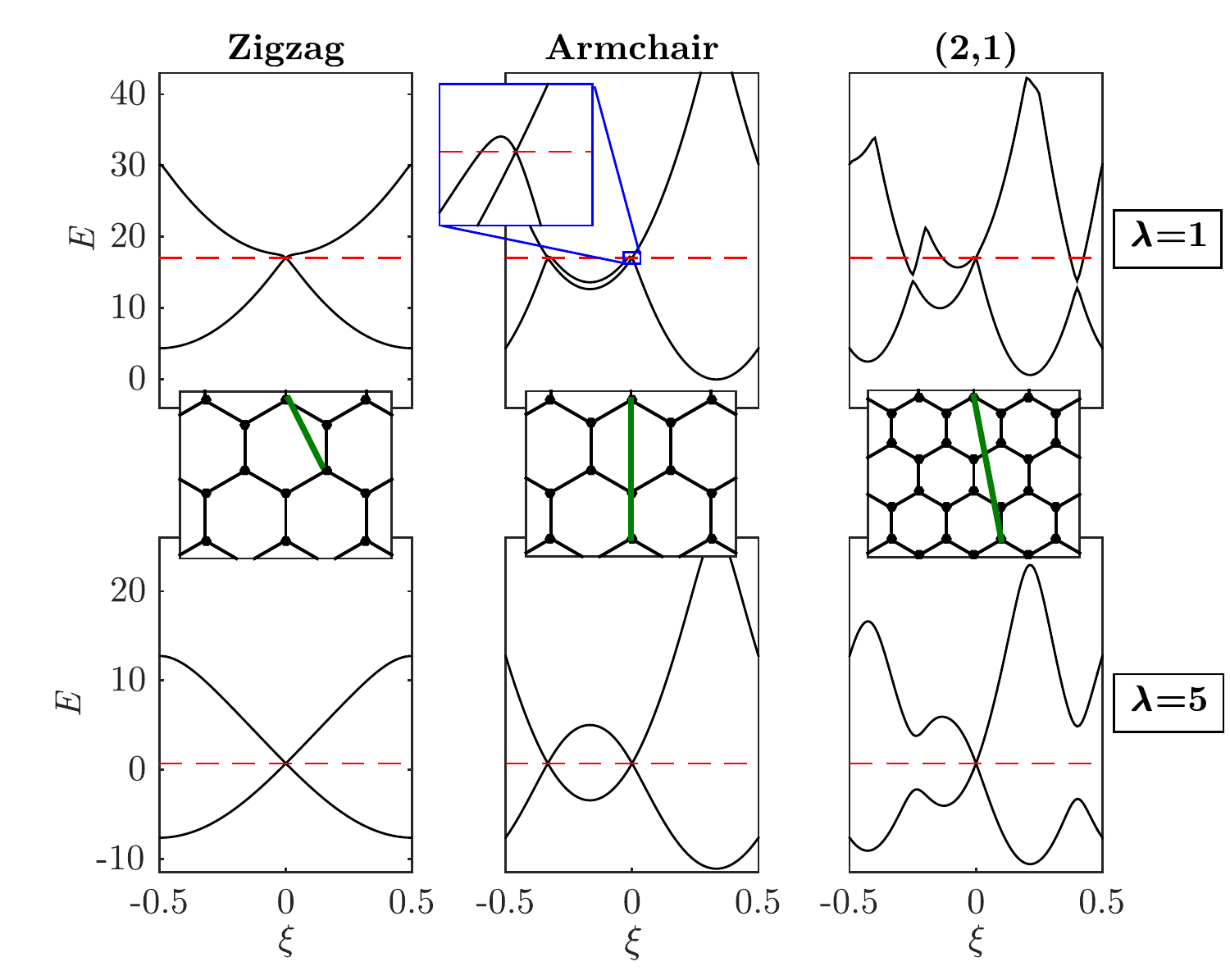}
\caption{\small
Band dispersion slices of $-\Delta+\lambda^2 V$ along the quasi-momentum segments:  $\bK+\xi\ktilde_2,\ |\xi|\le1/2$, for $\ktilde_2=\bk_2$ (zigzag), $\ktilde_2=-\bk_1+\bk_2$ (armchair) and $\ktilde_2=-\bk_1+2\bk_2$ ($(2,1)$), for $\lambda=1$ (top row) and $\lambda=5$ (bottom row). 
Numerical computations presented are for the simple trigonometric polynomial potential  $V(\bx)=\cos(\bk_1 \cdot\bx)+\cos(\bk_2 \cdot\bx)+\cos((\bk_1+\bk_2)\cdot\bx)$, satisfying $V_{1,1}>0$, which is not of the form of the potentials introduced in Section \ref{periodization}. Red dashed lines denote the Dirac point energy, $E^\lambda_D$.
{\bf Top row}: $\lambda=1$. Spectral no-fold condition holds for the zigzag edge but not the armchair and $(2,1)-$ edge.
{\bf Bottom row}: $\lambda=5$. Spectral no-fold condition holds for all three edges.
Insets between dispersion slice plots indicate zigzag, armchair and $(2,1)-$ (green) quasi-momentum segments (1D Brillouin zones) parametrized by $\xi$, for $0\leq\xi\leq1$.
}
\label{fig:dispersion_curves_lambda}
\end{figure}

\section{Dirac points}\label{dirac-points}

In this section we summarize results of  \cites{FW:12} on {\it Dirac points} of Schr\"odinger Hamiltonians, $H_V=-\Delta+V$, where $V$ is a honeycomb lattice potential. In the current context our Hamiltonian is $H^\lambda=-\Delta+\lambda^2 V(\bx)-E_0^\lambda$, with $V$  defined in \eqref{Vdef}.
 
Introduce the rotation and inversion operators, acting on functions:
\begin{equation} 
\mathcal{R}[f](\bx)=f(\bx_c+R^*(\bx-\bx_c)) \ {\rm and} \ \mathcal{I}[f](\bx)=f(\bx_c-(\bx-\bx_c))=f(2\bx_c-\bx).
\label{cRI-def}
\end{equation}

Let $\bK_\star$ denote any vertex of the Brillouin zone, $\B_h$, and let $f$ be a $\bK_\star-$ pseudo-periodic function. Since $R^*$ maps $\Lambda_h$ to itself and $R\bK_\star\in\bK_\star+\Lambda_h^*$, we have
\begin{align}
\mathcal{R}[f](\bx+\bv)&= f(\bx_c+R^*(\bx-\bx_c) + R^*\bv) = e^{i\bK_\star\cdot R^*\bv}f(\bx_c+R^*(\bx-\bx_c))\label{RfK}
\\ &=e^{iR\bK_\star\cdot \bv}f(\bx_c+R^*(x-\bx_c)) =e^{i\bK_\star\cdot \bv}f(\bx_c+R^*(\bx-\bx_c))\nn\\
&=
e^{i\bK_\star\cdot \bv}\ \mathcal{R}[f](\bx) .
\nn\end{align}
Therefore, in analogy with Proposition 2.2 of \cites{FW:12} we have
\begin{proposition}\label{commutator}
Let $\bK_\star$ be any of the six vertices of the Brillouin zone, $\brill_h$. 
Then, 
 $H^\lambda$ and $\mathcal{R}$ map a dense subspace of $L^2_{\bK_\star}$ to itself. Furthermore, restricted to this dense subspace, the commutator $[H^\lambda,\mathcal{R}]=H^\lambda\ \mathcal{R}-\mathcal{R}\ H^\lambda$ vanishes. In particular, if $\phi(\bx,\bK_\star)$ is a solution of the $\bK_\star-$ pseudo-periodic eigenvalue problem \eqref{evp-Hlam}, then $\mathcal{R}[\phi(\cdot,\bK_\star)](\bx)$ is also a solution of  \eqref{evp-Hlam}.
\end{proposition}

  Since $\mathcal{R}$ has eigenvalues $1,\tau$ and $\overline{\tau}$, it is natural to split $L^2_{\bK_\star}$, the space of $\bK_\star-$ pseudo-periodic functions, into the direct sum:
  \begin{equation}
   L^2_{\bK_\star}\ =\ L^2_{\bK_\star,1}\oplus L^2_{\bK_\star,\tau}\oplus L^2_{\bK_\star,\overline\tau}.
   \label{L2-directsum}
   \end{equation}
Here,  $L^2_{\bK_\star,\sigma}$, where $\sigma=1,\tau,\overline{\tau}$ and  $\tau=\exp(2\pi i/3)$, denote the invariant eigenspaces of $\mathcal{R}$:
   \begin{equation}
   L^2_{\bK_\star,\sigma}\ =\ \Big\{g\in  L^2_{\bK_\star}: \mathcal{R}g=\sigma g\Big\}\ .
 \label{L2Ksigma}   
\end{equation}
We also introduce  $H^s_{\bK_\star}$, $s\ge0$,  the subspace of 
functions $f\in L^2_{\bK_\star}$, such that $e^{-i\bK_\star\cdot\bx}f(\bx)\in H^2(\R^2/\Lambda_h)$.
 $H^s_{\bK_\star,\sigma}$,  for $s\ge0$ and $\sigma=1,\tau,\overline\tau$, a subspace of
 $L^2_{\bK_\star,\sigma}$ is defined analogously. 
  
Proposition \ref{commutator} and the decomposition \eqref{L2-directsum} imply that the $L^2_{\bK_\star}- $ Floquet-Bloch eigenvalue problem may be reduced to the three independent   $L^2_{\bK_\star,\sigma}-$ eigenvalue problems;
\begin{equation}
\left(-\Delta+\lambda^2V\right)\Psi\  =\ E\Psi,\qquad \Psi\in L^2_{\bK_\star,\sigma}\qquad \sigma=1,\tau,\overline\tau ;
\nn\end{equation}
In particular, see Definition \ref{dirac-pt-defn} of Dirac point below.

\begin{proposition}\label{tau-taubar}
Let $f\in L^2_{\bK_\star,\tau}$. Then, 
\[ \left(\mathcal{C}\circ\mathcal{I}\right)[f](\bx)\equiv \overline{f(2\bx_c-\bx)}\ \in\ L^2_{\bK_\star,\overline\tau},\ \]
where $\mathcal{C}$ denotes complex-conjugation and $\mathcal{I}$ is the inversion with respect to $\bx_c$, defined in \eqref{cRI-def}. 
\end{proposition}

\begin{proof}[Proof of Proposition \ref{tau-taubar}] Let $S=\mathcal{C}\circ\mathcal{I}$ and suppose $f\in L^2_{\bK_\star,\tau}$. Then, for any $\bv\in\Lambda_h$, 
\[ 
[Sf](\bx+\bv)=
\overline{f(2\bx_c-\bx-\bv)}=\overline{f(2\bx_c-\bx)e^{i\bK_\star\cdot(-\bv)}}=e^{i\bK_\star\cdot\bv} [Sf](\bx).\]
 Hence, $Sf\in L^2_{\bK_\star}$. Furthermore, 
 \begin{align*}
  \mathcal{R}[Sf](\bx) &=\  
  (Sf)\left(\bx_c+R^*(\bx-\bx_c)\right)\ =\ \overline{f\left(2\bx_c-[\bx_c+R^*(\bx-\bx_c)]\right)}
  \nn\\
  &=\ \overline{f\left(\bx_c+R^*[2\bx_c-\bx-\bx_c])]\right)} = \overline{\tau\ f(2\bx_c-\bx)}\ =\ \overline{\tau}\ S[f](\bx)\  
 \end{align*}
 which completes the proof.\end{proof}

    We next give  a precise definition of a Dirac point of a honeycomb Schroedinger operator; see \cites{FW:12}.
\begin{definition}[Dirac point]\label{dirac-pt-defn}
Let $V(\bx)$ be smooth,  real-valued and  $\Lambda_h-$ periodic on $\R^2$. Assume, in addition that $V$ is inversion symmetric with respect to $\bx_c$ and rotationally invariant by $120^\circ$; see \eqref{V-inv}. 
\footnote{ In \cites{FW:12} we implicitly assume that $\bx_c=0$ but, with obvious changes,
 the discussion of  \cites{FW:12} applies here with $\bx_c$ given by \eqref{x_c-def}.}
That is, $V$ is a honeycomb lattice potential in the sense of \cites{FW:12}. Consider the Schr\"odinger operator 
$H_V=-\Delta+V$.

Denote by $\brill_h$, the Brillouin zone. Let $\bK_\star\in\brill_h$ be one of the vertices of the Brillouin zone.
The energy / quasi-momentum pair $(\bK_\star,E_D)\in\brill_h\times\R$ is called a {\it Dirac point} if there exists $b_\star\ge1$  such that:
\begin{enumerate}
\item $E_D$ is an $L^2_{\bK_\star}-$ eigenvalue of  $H_V$ of multiplicity two.
\item The eigenspace for the eigenvalue $E_D$, $\textrm{Nullspace}\Big(H_V-E_D I\Big)$, is equal to
${\rm span}\Big\{ \Phi_1, \Phi_2\Big\}$, where 
 $\Phi_1\in L^2_{\bK_\star,\tau}$ is a solution of the 
$L^2_{\bK_\star,\tau}-$ Floquet-Bloch eigenvalue problem and 
$\Phi_2(\bx) = \left(\mathcal{C}\circ\mathcal{I}\right)[\Phi_1](\bx) = \overline{\Phi_1(2\bx_c-\bx)}\in L^2_{\bK_\star,\bar\tau}$ is a solution of the $L^2_{\bK_\star,\overline\tau}-$ Floquet-Bloch eigenvalue problem. We may take $\inner{\Phi_a, \Phi_b}_{L^2_{\bK_\star}} = \delta_{ab}$, $a,b=1,2$.
\item There exist constants $\lamsharp\in\C$, $\lamsharp\ne0$ and $\zeta_0>0$, Floquet-Bloch eigenpairs 
\[ \bk\mapsto (\Phi_{b_\star+1}(\bx;\bk),E_{b_\star+1}(\bk))\ \ {\rm and}\ \   \bk\mapsto (\Phi_{b_\star}(\bx;\bk),E_{b_\star}(\bk)),\]
  and Lipschitz continuous functions $e_j(\bk),\ j=b_\star, b_\star+1$, where $e_j(\bK_\star)=0$, defined for $|\bk-\bK_\star|<\zeta_0$ such   that
\begin{align*}
E_{b_\star+1}(\bk)-E_D\ &=\ + |\lamsharp|\ 
\left| \bk-\bK_\star \right|\ 
\left( 1\ +\ e_{b_\star+1}(\bk) \right),\nn\\
E_{b_\star}(\bk)-E_D\ &=\ - |\lamsharp|\ 
\left| \bk-\bK_\star \right|\ 
\left( 1\ +\ e_{b_\star}(\bk) \right) .\label{cones}
\end{align*}
In particular,  $|e_j(\bk)|\lesssim |\bk-\bK_\star|,\ j=b_\star, b_\star+1$.\\
\end{enumerate}
\end{definition}

\begin{remark}\label{onedp-implies-all}
{\ }
\begin{enumerate}
\item 
The quantity $|\lamsharp|$ is known as the Fermi velocity; see, for example, \cites{RMP-Graphene:09}. 
\item 
In \cites{FW:12} we prove that parts (1) and (2) of Definition \ref{dirac-pt-defn} imply part (3), although 
$\lamsharp$ may be zero. We then show 
 that for generic honeycomb lattice potentials, $\lamsharp\ne0$. No assumptions are made on 
the size (contrast/depth) of the potential. 
\item
 Note, from Proposition 4.1 of \cites{FW:12} that $\lamsharp$ is given in terms of the Floquet-Bloch modes $\Phi_1$
 and $\Phi_2$ by the expression:
 \begin{equation}
 \lamsharp\ =\ 2\left\langle\Phi_2,\frac{1}{i}\partial_{x_1}\Phi_1\right\rangle
 \label{vF-Phi}
 \end{equation}
\item Suppose $(\bK,E_D)$ is a Dirac point with corresponding eigenspace ${\rm span}\{ \Phi_1, \Phi_2\}$, where $\Phi_1\in L^2_{\bK,\tau}$ and $\Phi_2\in L^2_{\bK,\bar\tau}$. Then, since $\mathcal{R}$ commutes with $-\Delta+\lambda^2V$,  and since the quasi-momenta $R\bK$ and $R^2\bK$ yield equivalent pseudo-periodicity to $\bK$
 we have that $(\bK,E_D)$, $(R\bK,E_D)$  and $(R^2\bK,E_D)$ are all  Dirac points. Moreover, since $V$ is real-valued, complex-conjugation yields that $(\bK^\prime,E_D)$, $(R\bK^\prime,E_D)$  and $(R^2\bK^\prime,E_D)$ are all  Dirac points. So to establish that there are Dirac points located at the six vertices of $\brill_h$, it suffices to prove this for a single vertex of $\brill_h$. 
  \end{enumerate}
\end{remark}

\section{Approximation of low-lying Floquet-Bloch modes for large $\lambda$}\label{low-lying-fb}

In this section and in Section \ref{en-estimates} we assume that ${\rm supp}(V_0)\subset B({\bf 0},r_0)$,
where $0<r_0<\frac12|\be_{A,1}|\times(1-\delta_0)$, for some  small given $\delta_0>0$. 
The stricter constraint $(PW_2)$
on the size of ${\rm supp}(V_0)$ in the statement of Theorem \ref{main-theorem} is used in the analysis starting in 
 Section \ref{dirac-points-sb}.

Starting with the ground state eigenfunction, $p_0^\lambda(\bx)$, we define
\begin{align*}
p_\bk^\lambda(\bx)\ &=\ e^{-i\bk\cdot\bx}\ p_0^\lambda(\bx),\qquad \bk\in\C^2,\ \ \bx\in\R^2.
\end{align*}
For now, assume $|\Im\bk|<C_1$, where $C_1$ is a fixed positive number. (We shall later further constrain $\bk$ by $|\Im\bk|\le\lambda^{-1}$. ) Since $p_0^\lambda(\bx)$ satisfies the exponential decay bound (see Corollary \ref{cor:expo-decay}) $p_0^\lambda(\bx)\lesssim e^{-c\lambda |\bx|},\ |\bx|>R_0$, for $\lambda$ larger than some constant depending on $C_1$, $p_\bk^\lambda(\bx)$ is also exponentially decaying.

For fixed $\bk\in\R^2$ we have
\begin{equation*}
\left[\ -\left(\ \nabla_\bx+i\bk\right)^2\ +\ \lambda^2 V_0(\bx)\ -\ E_0^\lambda\ \right]p_\bk^\lambda(\bx)=0
\end{equation*}
and 
\begin{equation}
\left\langle\ \left(\ -\left(\ \nabla_\bx+i\bk\right)^2 + \lambda^2 V_0\ \right)\psi,\psi\ \right\rangle_{_{L^2(\R^2)}}\ \ge\ (E_0^\lambda+c_{gap})\ \|\psi\|_{_{L^2(\R^2)}}^2,
\label{H0k-bound}\end{equation}
for $\psi\in H^2(\R^2)$ such that $\left\langle p_\bk^\lambda,\psi\right\rangle_{_{L^2(\R^2)}}=0$, by property {\bf (EG)}
 in Section \ref{atomic-well}. 

 Introduce the  recentering of $p_\bk^\lambda$ at $\hat\bv\in\Honeycomb=\Lambda_A\cup \Lambda_B$:
\begin{equation*}
p_{\bk,\hat\bv}^\lambda(\bx)\ \equiv\ p_\bk^\lambda(\bx-\hat\bv)\ ,
\end{equation*}
and  define the $\Lambda_h-$ periodic  approximate Floquet-Bloch amplitudes:
\begin{equation}
p_{\bk,I}^\lambda(\bx)\ \equiv\ \sum_{\hat\bv\in\Lambda_I}\ p_{\bk,\hat\bv}^\lambda(\bx)
 \ =\ \sum_{\hat\bv\in\Lambda_I}\ e^{-i\bk\cdot(\bx-\hat\bv)}p_0^\lambda(\bx-\hat\bv)\ ,\ I = A , B.
\label{p-bk-I}
\end{equation}
and the  $\bk-$ pseudo-periodic approximate Floquet-Bloch modes:
\begin{equation}
P_{\bk,A}^\lambda(\bx)\equiv e^{i\bk\cdot\bx}\ p_{\bk,A}^\lambda(\bx),\qquad P_{\bk,B}^\lambda(\bx)\equiv e^{i\bk\cdot\bx}\ p_{\bk,B}^\lambda(\bx)\ .
\label{P-bk-I}
\end{equation}

\begin{remark}\label{remark-on-PkAB} In Theorem \ref{solve-L2tau-evp} we construct eigenstates of $-\Delta+\lambda^2 V$,\ $\lambda$ large: 
 $\Phi_1^\lambda(\bx)$ near $P_{\bK,A}^\lambda(\bx)\equiv e^{i\bK\cdot\bx}\ p_{\bK,A}^\lambda(\bx)\in L^2_{\bK,\tau}$ and 
 $\Phi_2^\lambda(\bx)$ near $P_{\bK,B}^\lambda(\bx)\equiv e^{i\bK\cdot\bx}\ p_{\bK,B}^\lambda(\bx)\in L^2_{\bK,\overline\tau}$.
\end{remark}

We find it useful to let the mapping $\bk\mapsto p_{\bk,I}^\lambda$ depend on a {\it complex} quasi-momentum, $\bk$,
 varying in an appropriate domain in $\C^2$. Corollary \ref{cor:expo-decay}, to be proved later, shows that $p_0^\lambda(\bx)$ satisfies the exponential decay bound $p_0^\lambda(\bx)\lesssim e^{-c\lambda |\bx|},\ |\bx|>r_0+c_0$. The function $p_{\bk,I}^\lambda(\bx)$ is $\Lambda_h-$ periodic on $\R^2$ so it may be regarded as a function on
  $\R^2/\Lambda_h$.
  
Furthermore, by the exponential decay of $p_{\bk}^\lambda(\bx)$, for all $\lambda$ larger than a constant which depends on $C_1$, 
 the series \eqref{p-bk-I} converges uniformly to an analytic function:
 \begin{align*}
&\textrm{$\bk\mapsto p_{\bk,I}^\lambda$  from $\{\bk\in\C^2:|\Im\bk|<C_1\}$ to $H^2(\R^2/\Lambda_h)$. }
\end{align*}

This property is used  in Section \ref{mu-near-dpts}, where we obtain derivative bounds on 
  the rescaled dispersion maps near Dirac points via Cauchy estimates for quasi-momenta, $\bk$, in a narrow strip,
    $\{ \bk\in\C^2\ :\ |\Im\bk|<\hat{c}\lambda^{-1}\}$, where $\hat{c}$ is a small constant.

We state further consequences of exponential decay of $p_0^\lambda$.
 For $|\Im \bk|<C_1$, 
\begin{align}
&\Big|\ \|p_{\bk,I}^\lambda\|_{_{L^2(\R^2/\Lambda_h)}}\ -\ 1\ \Big|\ \lesssim\ e^{-c\lambda},\ I= A, B.
\label{pkI-normalize}\\
&\Big|\ \left\langle p^\lambda_{\bk,J},p^\lambda_{\bk,I}\right\rangle_{_{L^2(\R^2/\Lambda_h)}}\ -\ \delta_{JI}\ \Big|\ \lesssim\ e^{-c\lambda},\ I= A, B, 
\label{pkI-normalize1} \quad \text{and} \\
&\|\ p_{\bk,I}^\lambda -p_{\bk,\bv}^\lambda\ \|_{_{L^2(\mathscr{Z}_\bv)}}\ \lesssim\ e^{-c\lambda},\ \bv\in\Lambda_I,\ \ I=A, B.
\label{pk-diff}
\end{align}
Here,  $\mathscr{Z}_\bv$, for $\bv\in\Lambda_I$, is the set of points in $\R^2$ which are at least as close to $\bv$ as to any other point in $\Lambda_I$.

For $|\Im \bk|<C_1$, we claim that:
 \begin{align}
& \|H^\lambda(\bk)p_{\bk,I}^\lambda\| \equiv \| \left[-(\nabla_\bx+i\bk)^2 + V^\lambda(\bx)\right]p_{\bk,I}^\lambda  \|_{_{L^2(\R^2/\Lambda_h)}} \lesssim e^{-c\lambda}, \ I=A, B. 
\label{Hlambda-pk}
 \end{align}
 Here, $V^\lambda(\bx)=\lambda^2 V(\bx)-E_D^\lambda=\lambda^2 \sum_{\bv\in{\bf H}}\ V_0(\bx-\bv)-E_D^\lambda$; see \eqref{Vlam-def}.
The bound \eqref{Hlambda-pk} follows from exponential decay of $p_\bk^\lambda$, a consequence of Corollary \ref{cor:expo-decay} (exponential decay of $p_0^\lambda(\bx)$), and  the observation
\begin{align}
H^\lambda(\bk) p_{\bk,\hat\bv}^\lambda(\bx) &\equiv \left[ -(\nabla_\bx+i\bk)^2 + V^\lambda(\bx) \right] p_{\bk,\hat\bv}^\lambda(\bx)\nn\\
& = \left[ -(\nabla_\bx+i\bk)^2 + \lambda^2 V_0(\bx-\hat\bv) - E_0^\lambda +
 \sum_{\bv\in\Honeycomb\setminus\{\hat\bv\}}\lambda^2 V_0(\bx-\bv) \right] p_{\bk}^\lambda(\bx-\hat\bv)\nn\\
 & =
 \sum_{\bv\in\Honeycomb\setminus\{\hat\bv\}} \lambda^2 V_0(\bx-\bv) p_{\bk}^\lambda(\bx-\hat\bv) . \label{Hpkhatv0}
\end{align}
%

Let $\tK\in\R^2$ denote a generic real quasi-momentum and assume  $|\bk-\tK|\lesssim\lambda^{-1}$. Then, using the exponential decay of $p_0^\lambda$, we obtain 
\begin{equation}
\left\|\ p_{\bk,I}^\lambda-p_{\tK,I}^\lambda\ \right\|\ \lesssim\ \lambda^{-1} \ .
\label{pktK-diff}
\end{equation}

The following lemma facilitates our working with a nearly orthogonal decomposition 
of $L^2(\R^2/\Lambda_h)$ in terms of ${\rm span}\left\{p_{\bk,I}^\lambda : I=A,B \right\}$ and 
${\rm span}\left\{p_{\tK,I}^\lambda : I=A,B \right\}^\perp$ provided the difference $|\bk-\tK|$ is sufficiently small. 

\begin{lemma}\label{orthog}
  Introduce the orthogonal projection
\begin{align*}
&\Pi_{_{AB}}:L^2(\R^2/\Lambda_h)\to\mathscr{H}_{_{AB}}\ \ {\rm onto}\nn\\
&\mathscr{H}_{_{AB}}\ =\ \left\{\tilde\psi\in L^2(\R^2/\Lambda_h)\ :\ \left\langle p_{\tK,I}^\lambda,\tilde\psi\right\rangle=0,\ {\rm for}\ I=A,B \right\},
\end{align*}
the orthogonal complement of ${\rm span}\left\{p_{\tK,I}^\lambda\ :\ I=A,B\ \right\}$ in $L^2(\R^2/\Lambda_h)$.\\
Assume $\tK\in\R^2$ and $|\bk-\tK|\lesssim\lambda^{-1}$.
\begin{enumerate}
\item\  
 Then, for $F\in L^2(\R^2/\Lambda_h)$, we have that
  \[ F=0\ \ \iff\ \ \Pi_{_{AB}}F=0\ \ {\rm and}\ \ \left\langle p_{_{\overline{k},I}}^\lambda, F\right\rangle=0, \ \ I=A, B.\]
\item\ Any $\psi\in H^2(\R^2/\Lambda_h)$ may be expressed in the form
\begin{equation}
\psi\ =\ \sum_{I=A,B}\ \alpha_I\ p_{\bk,I}^\lambda\ +\ \tilde\psi, \label{psi-decomp-AktK}\end{equation}
where  $\tilde{\psi}\in \mathscr{H}_{_{AB}}^2$
 and $\alpha_A, \alpha_B\in\C$.
\end{enumerate}
\end{lemma}

\begin{proof}[Proof of Lemma \ref{orthog}] To prove part (1), assume $\Pi_{_{AB}}F=0$ and $\left\langle p_{_{\overline{k},I}}^\lambda, F\right\rangle=0, \ \ I=A, B$. Then, $F=\sum_{I=A,B}\alpha_{_{I}} p_{_{\tK,I}}^\lambda$. 
Taking the inner product with $p_{_{\overline{k},A}}^\lambda$ and $p_{_{\overline{k},B}}^\lambda$ yields the equations
\[\sum_{I=A,B}\ \left\langle p_{_{\overline{k},J}}^\lambda, p_{_{\tK,I}}^\lambda\right\rangle\alpha_I\ =\  0, \quad \text{for\ \ } J=A, B.\]
The latter may be rewritten as
\begin{align}
\sum_{I=A,B}\ \Big[\ \left\langle p_{_{\tK,J}}^\lambda, p_{_{\tK,I}}^\lambda\right\rangle + 
\left\langle \left[p_{_{\overline{k},J}}^\lambda-p_{_{\tK,J}}^\lambda\right], p_{_{\tK,I}}^\lambda\right\rangle\ \Big]\ \alpha_I\ =\  0, \ J=A, B\ .
\end{align}
By \eqref{pkI-normalize1}, the first inner product within the square brackets is equal to $\delta_{JI}+ \mathcal{O}(e^{-c\lambda})$ and  by \eqref{pktK-diff}, the second inner product is $\mathcal{O}(\lambda^{-1})$. It follows that  $\alpha_I=0,\ I=A, B$. Hence, $F=0$.

To prove part (2), note that \eqref{psi-decomp-AktK} holds if and only if
\[ \left\langle p^\lambda_{\tK,J},\psi\right\rangle\ =\ 
\sum_{I=A,B} \left\langle p^\lambda_{\tK,J}, p^\lambda_{\bk,I}\right\rangle\ \alpha_I\ \ {\rm for}\ \  J=A,B.\]
By \eqref{pkI-normalize1}, $\left\langle p^\lambda_{\tK,J}, p^\lambda_{\bK,I}\right\rangle\ \alpha_I=\delta_{J,I}+\mathcal{O}(e^{-c\lambda})$, and we may solve for $\alpha_A$ and $\alpha_B$ for $\lambda$ sufficiently large. 
 This completes the proof of Lemma  \ref{orthog}.
\end{proof}

\section{Energy estimates}\label{en-estimates}

Throughout this section, we follow the convention (see Section \ref{preliminaries}) that in relations
 involving  norms and inner products for which the relevant function space is not explicitly indicated, it is to be understood that these are taken in $L^2(\R^2/\Lambda_h)$. 

The following result is the point of departure for our energy estimates and resolvent bounds.
\begin{lemma}\label{loc-en-lemma}
Fix $I= A$ or $ B$. 
Assume that ${\rm supp}(V_0)\subset B({\bf 0},r_0)$, 
where $0<r_0<\frac12|\be_{A,1}|\times(1-\delta_0)$ for $\delta_0>0$ and fixed. 

Suppose $\psi\in H^2(\R^2/\Lambda_h)$ and $\bk\in\R^2$.  Assume that $\textrm{supp}\ \psi\subset\{\bx\in\R^2/\Lambda_h: {\rm dist}(\bx,\Lambda_I)\le r_0\}$
  and that
\begin{equation*}
\left\langle p_{\bk,I}^\lambda, \psi\right\rangle_{L^2(\R^2/\Lambda_h)}\ =\ 0.
\end{equation*}
Then,
\begin{equation*}
\left\langle \left[-(\nabla_\bx+i\bk)^2+V^\lambda(\bx) \right]\psi,\psi \right\rangle_{L^2(\R^2/\Lambda_h)}\ge c\|\psi\|_{L^2(\R^2/\Lambda_h)}^2\ .
\end{equation*}
\end{lemma}

\begin{proof}[Proof of Lemma \ref{loc-en-lemma}]

 Fix $I= A$ or $ B$ and let $\psi^\sharp(\bx)=\psi(\bx)\ {\bf 1}_{\{|\bx-\bv_I|\le r_0\}}$ .
 Then, 
 \begin{align*}
 & \psi^\sharp\in H^2(\R^2),\ \ \textrm{supp}\ \psi^\sharp\subset\{\bx\in\R^2: |\bx-\bv_I|\le r_0\},\\
 & \psi(\bx)=\sum_{\bv\in\Lambda_h}\psi^\sharp(\bx-\bv),
 \end{align*}
 since the discs $\{\bx\in\R^2: |\bx-\bv|< r_0\}$, where $\bv\in\Lambda_I$, are disjoint subsets of $\R^2$;
  see Figure \ref{fig:discs}.
Using the $\Lambda_h-$ periodicity of $V^\lambda(\bx)$ we have:
 \begin{equation*}
\left\langle \left[-(\nabla_\bx+i\bk)^2+V^\lambda(\bx) \right]\psi,\psi \right\rangle_{L^2(\R^2/\Lambda_h)} 
= \left\langle \left[-(\nabla_\bx+i\bk)^2+V^\lambda(\bx) \right]\psi^\sharp,\psi^\sharp \right\rangle_{L^2(\R^2)}.
\end{equation*}

We shall make use of the following consequence of \eqref{H0k-bound}:\\
Let  $\eta\in H^2(\R^2)$ satisfy $\left\langle p_{\bk,\bv_I}^\lambda,\eta\right\rangle_{L^2(\R^2)}=0$. Then, 
 \begin{equation}
\left\langle\ \left(\ -\left(\ \nabla_\bx+i\bk\right)^2 + \lambda^2 V_0(\bx-\bv_I)\ - E_0^\lambda\ \right)\eta,\eta\ \right\rangle_{L^2(\R^2)}\ \ge\ c_{gap}\ \|\eta\|_{L^2(\R^2)}^2\ .
\label{H0k-bound1}\end{equation}
Recall that by hypothesis on $\psi$ and the $\Lambda_h-$ periodicity of $p_{\bk,I}^\lambda$:
\begin{equation*}
\left\langle p_{\bk,I}^\lambda,\psi^\sharp\right\rangle_{L^2(\R^2)}\ =\ 
\left\langle p_{\bk,I}^\lambda,\psi\right\rangle_{L^2(\R^2/\Lambda_h)}\ = 0.
\end{equation*}
So to make use of \eqref{H0k-bound1}, we should compare 
$\left\langle p_{\bk,\bv_I}^\lambda,\psi^\sharp\right\rangle_{L^2(\R^2)}$ with 
$\left\langle p_{\bk,I}^\lambda,\psi^\sharp\right\rangle_{L^2(\R^2)}$.

By \eqref{pk-diff} and the Cauchy-Schwarz inequality
\begin{equation*}
\Big|\ \left\langle p_{\bk,\bv_I}^\lambda,\psi^\sharp\right\rangle_{L^2(\R^2)}
\ -\ \left\langle p_{\bk,I}^\lambda,\psi^\sharp\right\rangle_{L^2(\R^2)}\ \Big|\ \lesssim\ e^{-c\lambda}\ \|\psi^\sharp\|_{L^2(\R^2)}.
\nn\end{equation*}
Hence, 
\begin{equation} \Big| \left\langle p_{\bk,\bv_I}^\lambda,\psi^\sharp\right\rangle_{L^2(\R^2)} \Big| \lesssim
e^{-c\lambda}\ \|\psi^\sharp\|_{L^2(\R^2)}.
\label{proj-vI}
\end{equation}
Recall that $\|p^\lambda_{\bk,\bv_I}\|_{L^2(\R^2)}=\|p_0^\lambda\|_{L^2(\R^2)}=1$. We may write
 \begin{equation}
 \psi^\sharp=  \alpha\ p^\lambda_{\bk,\bv_I}+\psi^{\sharp\sharp},\ \ \left\langle p^\lambda_{\bk,\bv_I},\psi^{\sharp\sharp}\right\rangle_{_{L^2(\R^2)}}=0\ ,
 \label{psi-sharp-decomp}
\end{equation} 
 where $\psi^{\sharp\sharp}\in H^2(\R^2)$ and $\alpha\in\C$ . From \eqref{proj-vI} we have $|\alpha|\le e^{-c\lambda}\ \|\psi^\sharp\|_{L^2(\R^2)}$  and therefore 
 \begin{equation}
 \|\psi^{\sharp\sharp}\|_{L^2(\R^2)}\ge (1-e^{-c\lambda})\|\psi^\sharp\|_{L^2(\R^2)}\ .
 \label{dblesharp}\end{equation}
 Since $\left\langle p^\lambda_{\bk,\bv_I},\psi^{\sharp\sharp}\right\rangle_{_{L^2(\R^2)}}=0$, \eqref{H0k-bound1} and \eqref{dblesharp} imply
 \begin{align*}
 &\left\langle\left[-(\nabla_\bx+i\bk)^2+\lambda^2V_0(\bx-\bv_I)-E_0^\lambda\right] \psi^{\sharp\sharp},\psi^{\sharp\sharp}\right\rangle_{L^2(\R^2)}\nn\\
 &\qquad\qquad \ge\ c_{gap}\| \psi^{\sharp\sharp}\|_{L^2(\R^2)}^2\ \ge \frac12 c_{gap}\ \| \psi^{\sharp}\|^2_{L^2(\R^2)}\ .
 \end{align*}
However using \eqref{psi-sharp-decomp} and the fact that $\left[-(\nabla_\bx+i\bk)^2+\lambda^2V_0(\bx-\bv_I)-E_0^\lambda\right] p_{\bk,\bv_I}^\lambda=0$ we have
 \begin{align*}
 &\left\langle\left[-(\nabla_\bx+i\bk)^2+\lambda^2V_0(\bx-\bv_I)-E_0^\lambda\right] \psi^{\sharp\sharp},\psi^{\sharp\sharp}\right\rangle_{L^2(\R^2)}\\
 &=\qquad \left\langle\left[-(\nabla_\bx+i\bk)^2+\lambda^2V_0(\bx-\bv_I)-E_0^\lambda\right] \psi^{\sharp},\psi^{\sharp}\right\rangle_{L^2(\R^2)}\ .
\end{align*}
Hence, 
 \begin{align*}
 &\left\langle\left[-(\nabla_\bx+i\bk)^2+\lambda^2V_0(\bx-\bv_I)-E_0^\lambda\right] \psi^{\sharp},\psi^{\sharp}\right\rangle_{L^2(\R^2)}\ \ge \frac12 c_{gap}\ \| \psi^{\sharp}\|^2_{L^2(\R^2)}.
 \end{align*}
Moreover, on $\textrm{supp}\ \psi^\sharp\subset B(\bv_I,r_0)$,  we have $\lambda^2V_0(\bx-\bv_I)-E^\lambda_0=V^\lambda(\bx)$ and therefore
 \begin{align}
 &\left\langle\left[-(\nabla_\bx+i\bk)^2+V^\lambda(\bx)\right] \psi^{\sharp},\psi^{\sharp}\right\rangle_{L^2(\R^2)}\ \ge \frac12 c_{gap}\ \| \psi^{\sharp}\|^2_{L^2(\R^2)}.
 \end{align}
 
 Finally, using that $\textrm{supp}\ \psi^\sharp\subset B(\bv_I,r_0)$ and that $\psi(\bx)=\sum_{\bv\in\Lambda_h}\psi^\sharp(\bx-\bv)$ for $\bx\in\R^2$, we conclude that 
  \begin{align}
 &\left\langle\left[-(\nabla_\bx+i\bk)^2+V^\lambda(\bx)\right] \psi,\psi\right\rangle_{L^2(\R^2/\Lambda_h)}\ \ge c'_{gap}\ \| \psi\|^2_{L^2(\R^2/\Lambda_h)}
 \end{align}
  for any $\psi\in H^2(\R^2/\Lambda_h)$ such that $\left\langle p_{\bk,I}^\lambda,\psi\right\rangle_{L^2(\R^2/\Lambda_h)}=0$
   and  $\textrm{supp}\ \psi\subset\{\bx\in\R^2/\Lambda_h: {\rm dist}(\bx,\Lambda_I)\le r_0\}$.
 This completes the proof of Lemma \ref{loc-en-lemma}.
 \end{proof}

\subsection{Localization and integration by parts}\label{ibp}

Let $\Theta\in C^\infty(\R^2/\Lambda_h)$ be real-valued and $\varphi\in H^2(\R^2/\Lambda_h)$. Then, for $\bk\in\R^2$,
\begin{align*}
&\left[-(\nabla_\bx+i\bk)^2+V^\lambda(\bx)\right]\left(\Theta\varphi\right)\nn\\
&=\quad \left[-\Delta_\bx-2i\bk\cdot\nabla_\bx+|\bk|^2+V^\lambda(\bx)\right]\left(\Theta\varphi\right)\nn\\
&=\quad \Theta\left[-(\nabla_\bx+i\bk)^2+V^\lambda(\bx)\right]\varphi-2\nabla_\bx\Theta\cdot\nabla_\bx\varphi-
(2i\bk\cdot\nabla_\bx\Theta)\varphi- (\Delta_\bx\Theta)\varphi\ .\nn
\end{align*}
Taking the $L^2(\R^2/\Lambda_h)-$
inner product with $\Theta\varphi$, we obtain, using self-adjointness:
\begin{align*}
&\left\langle \left[-(\nabla_\bx+i\bk)^2+V^\lambda(\bx)\right]\left(\Theta\varphi\right),\left(\Theta\varphi\right)\right\rangle_{L^2(\R^2/\Lambda_h)}\nn\\
&=\quad \Re\left\langle \Theta^2\left[-(\nabla_\bx+i\bk)^2+V^\lambda(\bx)\right]\varphi,\varphi \right\rangle_{L^2(\R^2/\Lambda_h)}
- 2\ \Re\left\langle \nabla_\bx\Theta\cdot\nabla_\bx\varphi,\Theta\varphi\right\rangle_{L^2(\R^2/\Lambda_h)}\nn\\
&\quad -\Re\left\langle
(2i\bk\cdot\nabla_\bx\Theta)\varphi, \Theta\varphi\right\rangle_{L^2(\R^2/\Lambda_h)} - \Re\left\langle (\Delta_\bx\Theta)\varphi,\Theta\varphi\right\rangle_{L^2(\R^2/\Lambda_h)}\ .
\nn\end{align*}
There are simplifications. First note that  
\[\Re\left\langle
(2i\bk\cdot\nabla_\bx\Theta)\varphi, \Theta\varphi\right\rangle_{L^2(\R^2/\Lambda_h)}=\Re\int\ -2i(\bk\cdot\nabla_\bx\Theta)\Theta|\varphi|^2d\bx=0.\]
Furthermore, 
\begin{align*}
-2\Re \left\langle \nabla_\bx\Theta\cdot\nabla_\bx\varphi,\Theta\varphi\right\rangle
&= -2\int_{\R^2/\Lambda_h}\left(\Theta\nabla_\bx\Theta\right)\cdot \Re\left(\varphi\nabla_\bx\overline{\varphi}\right)\ d\bx\nn\\
&=-2\int_{\R^2/\Lambda_h} \frac12\nabla_\bx(\Theta^2)\cdot \frac12\nabla_\bx|\varphi|^2 d\bx\nn\\
&=\frac12\int_{\R^2/\Lambda_h} \Delta_\bx(\Theta^2)\ |\varphi|^2 d\bx\ =
\left\langle \frac12\Delta_\bx(\Theta^2)\ \varphi,\varphi\right\rangle_{L^2(\R^2/\Lambda_h)}\ .
\end{align*}

In view of the above computations we have the following
\begin{lemma}[Integration by parts]\label{ibp-lemma} 
Let $\Theta\in C^\infty(\R^2/\Lambda_h)$ be real-valued, $\psi\in H^2(\R^2/\Lambda_h)$ and $\bk\in\R^2$. Then, 
\begin{align}
& \left\langle\left[-(\nabla_\bx+i\bk)^2+V^\lambda(\bx)\right]\left(\Theta\varphi\right),\left(\Theta\varphi\right)\right\rangle\nn\\
&=\qquad \Re\left\langle \Theta^2\left[-(\nabla_\bx+i\bk)^2+V^\lambda(\bx)\right]\varphi,\varphi \right\rangle\ +\ \left\langle\chi_{_\Theta}\varphi,\varphi\right\rangle 
\end{align}
where $\chi_{_\Theta}=\frac12\Delta_\bx(\Theta^2)-\Theta\Delta_\bx\Theta=|\nabla_\bx\Theta|^2$.
\end{lemma}

\subsection{Localized energy estimate}\label{localized-energy-estimate}

Assume that ${\rm supp}(V_0)\subset B({\bf 0},r_0)$, 
where $0<r_0<\frac12|\be_{A,1}|\times(1-\delta_0)$,\ $0<\delta_0<1$. 
Suppose $1<\delta^\prime<\delta^{\prime\prime}$ is such that 
\begin{equation}
 0<r_0<\delta^\prime r_0< \delta^{\prime\prime} r_0<\frac12|\be_{A,1}|\times(1-\delta_0)\ .
 \label{ddp} \end{equation} 
\begin{proposition}[Main localized energy estimate]\label{loc-en-est}
Fix $I\in\{A,B\}$, and  $\bk\in\R^2$. Assume $\psi\in H^2(\R^2/\Lambda_h)$ and $\left\langle p_{\bk,I}^\lambda,\psi\right\rangle_{L^2(\R^2/\Lambda_h)}=0$. Let $\Theta\in C_0^\infty(\R^2/\Lambda_h)$ be real-valued 
 and suppose that  
\begin{equation*}
\Theta(\bx) =
\begin{cases} 
1 & \textrm{if dist$(\bx,\Lambda_I)\le\delta^\prime r_0$} \\ 
0 & \textrm{if dist$(\bx,\Lambda_I)\ge \delta^{\prime\prime}r_0$}\ .
\end{cases}
\end{equation*}
Thus, 
\[
\textrm{supp}\ \Theta\subset \{\bx\in\R^2/\Lambda_h: {\rm dist}(\bx,\Lambda_I)\le \delta^{\prime\prime}r_0\}.
\]
Then, 
\begin{align*}
c\|\Theta\psi\|_{L^2(\R^2/\Lambda_h)}^2 & \le
 \left\langle\left[-(\nabla_\bx+i\bk)^2+V^\lambda(\bx)\right]\left(\Theta\psi\right),
 \left(\Theta\psi\right)\right\rangle_{L^2(\R^2/\Lambda_h)} \\
 &\qquad + e^{-c\lambda}\|\psi\|_{L^2(\R^2/\Lambda_h)}^2 .
\end{align*}
Here, the constants, $c$, are determined by $V_0$, $\delta_0$, $\delta^\prime$ and $\delta^{\prime\prime}$.
\end{proposition}

\begin{proof}[Proof of Proposition \ref{loc-en-est}]   
Let $\bk\in\R^2$. Suppose $\psi\in L^2(\R^2/\Lambda_h)$ is such that $\left\langle p^\lambda_{\bk,I},\psi\right\rangle=0$.  
We localize $\psi$ near $\Lambda_I$, while maintaining orthogonality, by defining $\varphi=\Theta(\psi-\alpha_I\ p_{\bk,I}^\lambda)$ with 
$\alpha_I\in\C$ chosen so that $\left\langle p^\lambda_{\bk,I},\varphi\right\rangle=0$. Hence, we require: 
\[\alpha_I \left\langle p_{\bk,I}^\lambda, \Theta p_{\bk,I}^\lambda\right\rangle = 
\left\langle p_{\bk,I}^\lambda, \Theta\psi\right\rangle=-\left\langle (1-\Theta)p_{\bk,I}^\lambda, \psi\right\rangle\ .\]
Using \eqref{pkI-normalize}, one sees that $\left|\ \left\langle p_{\bk,I}^\lambda, \Theta p_{\bk,I}^\lambda\right\rangle-1\ \right|$ 
 and $\|(1-\Theta)p_{\bk,I}^\lambda\|$ are $\lesssim e^{-c\lambda}$, and we conclude that 
 $ | \alpha_I |\ \lesssim\ e^{-c\lambda}\ \|\psi\|$.
Since 
\[ {\rm supp}\ \Theta(\psi-\alpha_I p_{\bk,I}^\lambda)\subset \{\bx\in\R^2/\Lambda_h: {\rm dist}(\bx,\Lambda_I)\le \delta^{\prime\prime}r_0\}\]
 and $\left\langle p^\lambda_{\bk,I},\Theta(\psi-\alpha_I\ p_{\bk,I}^\lambda)\right\rangle=0$, by Lemma \ref{loc-en-lemma} we have the lower bound:
\begin{align}
&\left\langle \left[-(\nabla_\bx+i\bk)^2+V^\lambda(\bx)\right]\Theta\cdot(\psi-\alpha_I p^\lambda_{\bk,I}),
\Theta\cdot(\psi-\alpha_I p^\lambda_{\bk,I}) \right\rangle\nn\\
&\qquad \ge c\ \|\Theta\cdot(\psi-\alpha_I p^\lambda_{\bk,I}) \|^2\nn\\
&\qquad \ge \frac{c}2\ \|\Theta \psi\|^2 - c \|\alpha_I\Theta p^\lambda_{\bk,I}\|^2
 \gtrsim \frac{c}2\ \|\Theta \psi\|^2 - e^{-c\lambda} \|\psi\|^2,
\label{lower-en}\end{align}
where the last inequality follows from the  bound  $| \alpha_I |\lesssim e^{-c\lambda}$.

On the other hand,  also using the above bound on $| \alpha_I |$, we see that
\begin{align}
&\left\langle \left[-(\nabla_\bx+i\bk)^2+V^\lambda(\bx)\right]\Theta\cdot(\psi-\alpha_I p^\lambda_{\bk,I}),
\Theta\cdot(\psi-\alpha_I p^\lambda_{\bk,I}) \right\rangle\nn\\
&\ =\ \left\langle \left[-(\nabla_\bx+i\bk)^2+V^\lambda(\bx)\right] (\Theta\psi),(\Theta\psi)\right\rangle\nn\\
&\qquad -\ 2\Re\ 
\left\langle \left[-(\nabla_\bx+i\bk)^2+V^\lambda(\bx)\right] (\alpha_I \Theta p_{\bk,I}^\lambda),(\Theta\psi)\right\rangle\nn\\
&\qquad +\ 
\left\langle \left[-(\nabla_\bx+i\bk)^2+V^\lambda(\bx)\right] (\alpha_I \Theta p_{\bk,I}^\lambda),(\alpha_I \Theta p_{\bk,I}^\lambda)\right\rangle\nn\\
&\lesssim\left\langle \left[-(\nabla_\bx+i\bk)^2+V^\lambda(\bx)\right] (\Theta\psi),(\Theta\psi)\right\rangle + e^{-c\lambda}\|\psi\|^2\ .
\label{upper-en}
\end{align}
Putting together  \eqref{lower-en} and \eqref{upper-en} completes the proof of Proposition \ref{loc-en-est}.
\end{proof}

\subsection{Global energy estimates} 

\begin{proposition}[Main global energy estimate]\label{global-en-est}
 Let $K_{max}>0$ be given. 
There exist constants $c$ and $\lambda_\star$, depending on $K_{max}$ such that for all $\lambda>\lambda_\star$
 the following holds: 
Let $\bk\in\R^2$ and  $|\bk|\le K_{max}$.  %
 Let $\psi\in H^2(\R^2/\Lambda_h)$ be such that 
\begin{equation*}
\left\langle p^\lambda_{\bk,A},\psi\right\rangle\ =\ 0\ \ {\rm and}\ \  \left\langle p^\lambda_{\bk,B},\psi\right\rangle\ =\ 0.
\end{equation*}
Then, 
\begin{align}
&c\lambda^{-2}\|(\nabla_\bx+i\bk)\psi\|^2 + c\|\psi\|^2 \le \left\langle \left[-(\nabla_\bx+i\bk)^2+V^\lambda(\bx)\right]\psi,\psi\right\rangle.
\label{main-en}
\end{align}
\end{proposition}
Before turning to the proof of Proposition \ref{global-en-est}, we first give three immediate corollaries.

 \begin{corollary}\label{atmost2evalues}
 Under the conditions of Proposition \ref{global-en-est}, for all $\lambda>\lambda_\star$
  the operator
  \[\textrm{ $-(\nabla_\bx+i\bk)^2+\lambda^2 V$ has at most 2 eigenvalues in the range $E<E_0^\lambda+\frac{1}{2}c$,}\]
where $c$ is the constant in \eqref{main-en}.
 \end{corollary}
Corollary \ref{atmost2evalues} follows from the variational characterization of eigenvalues of self-adjoint operators.

 \begin{corollary}\label{cor1-main-en}
Let $K_{max}>0$ be given. 
There exist constants $c$ and $\lambda_\star$, depending on $K_{max}$ such that for all $\lambda>\lambda_\star$
 the following holds: 
Let $\bk\in\R^2$ and  $|\bk|\le K_{max}$.  %
 Let $\psi\in H^2(\R^2/\Lambda_h)$ be such that 
\begin{equation*}
\left\langle p^\lambda_{\bk,A},\psi\right\rangle\ =\ 0\ \ {\rm and}\ \  \left\langle p^\lambda_{\bk,B},\psi\right\rangle\ =\ 0.
\end{equation*}
Then, 
\begin{align*}
&c\|\psi\|^2\ +\ c\lambda^{-2}\|(\nabla_\bx+i\bk)\psi\|^2\ \le\ \|\left[-(\nabla+i\bk)^2+V^\lambda(\bx)\right]\psi\|^2\ .
\end{align*}
 \end{corollary} 
To prove Corollary \ref{cor1-main-en}, note that for $\tilde{w}>0$ we have
\begin{align*}
&\left\langle \left[-(\nabla+i\bk)^2+V^\lambda(\bx)\right]\psi,\psi\right\rangle\ \le \frac{1}{4\tilde{w}^2}\ \|\left[-(\nabla+i\bk)^2+V^\lambda(\bx)\right]\psi\|^2
 + \tilde{w}^2\ \|\psi\|^2 .
 \end{align*}
 For small enough $\tilde{w}$, the term $\tilde{w}^2\ \|\psi\|_{L^2(\R^2/\Lambda_h)}^2$ may be absorbed back into the left-hand side of \eqref{main-en}. Corollary \ref{cor1-main-en} now follows.
  
 Next, since $\|\nabla_\bx\psi\|^2\ \le\  2\left(\|(\nabla_\bx+i\bk)\psi\|^2\ + \ |\bk|^2\ \|\psi\|^2\right)$, 
 Corollary \ref{cor1-main-en} implies
 
  \begin{corollary}\label{cor2-main-en}
 Let $K_{max}>0$ be given. Let $\bk\in\R^2$ with $|\bk|\le K_{max}$. 
There exist constants $c$ and $\lambda_\star$, depending on $K_{max}$ such that for all $\lambda>\lambda_\star$
 the following holds: Let $\psi\in H^2(\R^2/\Lambda_h)$ be such that 
\begin{equation*}
\left\langle p^\lambda_{\bk,A},\psi\right\rangle\ =\ 0\ \ {\rm and}\ \  \left\langle p^\lambda_{\bk,B},\psi\right\rangle\ =\ 0.
\end{equation*}
Then, 
\begin{align*}
& c\|\psi\|^2\ +\ c\lambda^{-2}\|\nabla_\bx\psi\|^2\ 
 \le\ \|\left[-(\nabla+i\bk)^2+V^\lambda(\bx)\right]\psi\|^2\ .
\end{align*}
 \end{corollary}  

\begin{proof}[Proof of Proposition \ref{global-en-est} (Main global energy estimate)]
Choose $\delta_0\in(0,1)$ and constants $\delta_1, \delta_2,\ \tilde\delta_1, \tilde\delta_2$ such that 
 $\tilde\delta_1<\tilde\delta_2<\delta_1<\delta_2$,
\begin{align}
& 0<r_0<\delta_1 r_0< \delta_2 r_0< \frac12|\be_{A,1}|(1-\delta_0) , \quad \text{and} \nn\\
& 0<r_0<\tilde\delta_1 r_0< \tilde\delta_2 r_0< \frac12|\be_{A,1}|(1-\delta_0) . \nn
\end{align} 

On $\R^2/\Lambda_h$, we introduce two partitions of unity:
\begin{equation}
1=\Theta_A^2+\Theta_B^2+\Theta_0^2,\qquad 1=\tilde\Theta_A^2+\tilde\Theta_B^2+\tilde\Theta_0^2,
\label{POU}\end{equation}
 where $\Theta_I$ and $\tilde\Theta_I$, $I=A,B,0$, are non-negative and $C^\infty$, and where 
  $\Theta_A$ and $\Theta_B$ have disjoint support and are localized, respectively, near $\Lambda_A$ and $\Lambda_B$. 
  Similarly, $\tilde\Theta_A$ and $\tilde\Theta_B$ have disjoint support and are localized, respectively, near $\Lambda_A$ and $\Lambda_B$.
%
In particular, for $I=A, B$:
  \begin{align*}
  \Theta_I\ \equiv\ \begin{cases} 
  1, & {\rm dist}(\bx,\Lambda_I)\le \delta_1 r_0\\
  0, & {\rm dist}(\bx,\Lambda_I)\ge \delta_2 r_0
  \end{cases}
 \end{align*}
 and $\Theta_0$ is defined via the first relation in \eqref{POU}.  Also,
  \begin{align*}
  \tilde\Theta_I\ \equiv\ \begin{cases} 
  1, & {\rm dist}(\bx,\Lambda_I)\le \tilde\delta_1 r_0\\
  0, & {\rm dist}(\bx,\Lambda_I)\ge \tilde\delta_2 r_0
  \end{cases}
 \end{align*}
 and $\tilde\Theta_0$ is defined via the second relation in \eqref{POU}. 
%
%
%
%
%
%

We assume  $\bk\in\R^2$. Note that the local energy estimate gives the following:\\
If $\psi\in H^2(\R^2/\Lambda_h)$ is such that $\left\langle p_{\bk,A}^\lambda,\psi\right\rangle=0$ and $ \left\langle p_{\bk,B}^\lambda,\psi\right\rangle=0$, then for $I= A, B$:
\begin{align}
&c\|\Theta_I\psi\|^2\ \le\ \left\langle H^\lambda(\bk)(\Theta_I \psi),(\Theta_I \psi)\right\rangle\ +\ e^{-c\lambda}\ \|\psi\|^2 , \label{TI-bd}\\
&c\|\tilde\Theta_I\psi\|^2\ \le\ \left\langle H^\lambda(\bk)(\tilde\Theta_I \psi),(\tilde\Theta_I \psi)\right\rangle\ +\ e^{-c\lambda}\ \|\psi\|^2 , \label{tTI-bd}
\end{align}
where $H^\lambda(\bk)=-(\nabla+i\bk)^2+V^\lambda(\bx)$.

Next consider $\Theta_0\psi$. For $\bx\in{\rm supp}\ \Theta_0$, we have 
${\rm dist}(\bx,\Lambda_A\cup\Lambda_B)\ge\delta_1 r_0>r_0$. 
On this set $V(\bx)=0$ (see \eqref{Veq0}) and  hence, by hypothesis {\bf (GS)}, \eqref{GS}, $V^\lambda(\bx)=-E^\lambda_0\ge c\lambda^2$.
It follows that, for all $\psi\in H^2(\R^2/\Lambda_h)$,
\begin{align}
c\lambda^2\|\Theta_0\psi\|^2 &\le  \left\langle V^\lambda(\bx)(\Theta_0\psi),(\Theta_0\psi)\right\rangle\ \le\  
\left\langle H^\lambda(\bk)(\Theta_0\psi),(\Theta_0\psi)\right\rangle\ .
\label{T0-bd}\end{align}
 Similarly, for all $\psi\in H^2(\R^2/\Lambda_h)$,
\begin{align}
c\lambda^2\|\tilde\Theta_0\psi\|^2 
&\le  \left\langle H^\lambda(\bk)(\tilde \Theta_0\psi),(\tilde\Theta_0\psi)\right\rangle\ .
\label{tT0-bd}\end{align}
Summing \eqref{TI-bd} over $I=A,B$ with \eqref{T0-bd}, and recalling \eqref{POU}, we obtain
\begin{align}
&c\|\psi\|^2\ +\ c\lambda^2\|\Theta_0\psi\|^2\ \le\ \sum_{I=A,B,0}\left\langle H^\lambda(\bk)(\Theta_I\psi),(\Theta_I\psi)\right\rangle\ +\ e^{-c\lambda}\|\psi\|^2\ .
\label{sumThetaAB0}
\end{align}
Furthermore, summing \eqref{tTI-bd} over $I=A,B$ with \eqref{tT0-bd} we obtain
\begin{align}
&c\|\psi\|^2\ +\ c\lambda^2\|\tilde\Theta_0\psi\|^2\ \le \sum_{I=A,B,0}\left\langle H^\lambda(\bk)(\tilde\Theta_I\psi),(\tilde\Theta_I\psi)\right\rangle\ + e^{-c\lambda}\|\psi\|^2\ .
\label{sumtildeThetaAB0}
\end{align}
Estimates \eqref{sumThetaAB0} and \eqref{sumtildeThetaAB0} hold for $\psi\in H^2(\R^2/\Lambda_h)$ such that 
$\left\langle p^\lambda_{\bk,J},\psi\right\rangle=0$ for $J\in\{A,B\}$. 

Next, we apply the integration-by-parts Lemma \ref{ibp-lemma} and again recall \eqref{POU} to conclude that
\begin{align}
& \sum_{I=A,B,0}\left\langle H^\lambda(\bk)(\Theta_I\psi),(\Theta_I\psi)\right\rangle 
= \left\langle H^\lambda(\bk)\psi,\psi\right\rangle\ +\ \left\langle \left[\sum_{I=A,B,0}\chi_{_{\Theta_I}}(\bx)\right]\psi,\psi\right\rangle,
\label{Theta-IMS}
\end{align}
where $\chi_{_{\Theta_I}}=\frac12\Delta_\bx(\Theta_I)^2-\Theta_I\Delta_\bx\Theta_I=|\nabla_\bx\Theta_I|^2$, for $I=A,B,0$.
An analogous formula to \eqref{Theta-IMS} holds for the  $\tilde\Theta_A^2+\tilde\Theta_B^2+\tilde\Theta_0^2=1$ partition of unity,
 where $\chi_{_{\Theta_I}}$ is replaced by $\chi_{_{\tilde\Theta_I}}=|\nabla_\bx\tilde\Theta_I|^2$, for $I=A,B,0$.
 
 Substituting \eqref{Theta-IMS} and its $\tilde\Theta-$ analogue into \eqref{sumThetaAB0} and \eqref{sumtildeThetaAB0}, respectively, yields
%
\begin{align}
&c\|\psi\|^2 + c\lambda^2\|\Theta_0\psi\|^2
\le \left\langle H^\lambda(\bk)\psi,\psi\right\rangle + 
\left\langle \sum_{I=A,B,0}\chi_{_{\Theta_I}}(\bx)\cdot \psi,\psi\right\rangle\ +\ e^{-c\lambda}\|\psi\|^2\ ,
\label{sumThetaAB0-1}
\end{align}
and similarly
\begin{align}
&c\|\psi\|^2 + c\lambda^2\|\tilde\Theta_0\psi\|^2
\le \left\langle H^\lambda(\bk)\psi,\psi\right\rangle\ +\ 
\left\langle \sum_{I=A,B,0}\chi_{_{\tilde\Theta_I}}(\bx)\cdot \psi,\psi\right\rangle\ +\ e^{-c\lambda}\|\psi\|^2\ .
\label{sumtildeThetaAB0-2}
\end{align}
From the definitions of 
$\Theta_I, \chi_{_{\Theta_I}}, \tilde\Theta_I, \chi_{_{\tilde\Theta_I}}$, we see that 
\begin{equation}
\left|\sum_{I=A,B,0}\chi_{_{\Theta_I}}(\bx) \right|\le C\ {\bf 1}_{\{\bx:{\rm dist}(\bx,\Honeycomb)\ge \delta_1r_0\}}\ ,\ \ 
 \Honeycomb=\Lambda_A\cup\Lambda_B.
\label{chi-bound1}
\end{equation}
Moreover, $\tilde\Theta_0=1$ for $\bx$ such that ${\rm dist}(\bx,\Honeycomb)\ge \tilde\delta_2r_0$ 
 and 
 \begin{equation}
 \left|\sum_{I=A,B,0}\chi_{_{\tilde\Theta_I}}(\bx) \right|\le C,\qquad \bx\in\R^2/\Lambda_h\ .
 \label{chi-bound2}
 \end{equation}

 By \eqref{sumThetaAB0-1} and \eqref{chi-bound1}, and since $\tilde\Theta_0=1$ on
  ${\rm supp}\left(\sum_{I=A,B,0}\chi_{_{\Theta_I}}\right)$  we have:
 \begin{align}
&c_1\|\psi\|^2  \le \left\langle H^\lambda(\bk)\psi,\psi\right\rangle\ +\ 
C_1\ \|\tilde\Theta_0\psi\|^2\ +\ C_2\ e^{-c\lambda}\|\psi\|^2\ .
\label{bound1}
\end{align}
By \eqref{sumtildeThetaAB0-2} and \eqref{chi-bound2},
\begin{align}
&c_1\|\psi\|^2 + c\lambda^2\|\tilde\Theta_0\psi\|^2
 \le \left\langle H^\lambda(\bk)\psi,\psi\right\rangle\ +\ C\ \|\psi\|^2\ .
\label{bound2}
\end{align}
Let $\hat{c}$ be a small enough constant and $\lambda$ be sufficiently large such that $\hat{c}+C_2e^{-c\lambda}<c_1/2$.
We consider two cases. If 
$C_1\|\tilde\Theta_0\psi\|^2\le\ \hat{c}\ \|\psi\|^2,$
then  \eqref{bound1} implies
 \begin{align*}
& \frac{c_1}2\|\psi\|^2  \le \left\langle H^\lambda(\bk)\psi,\psi\right\rangle .
\nn\end{align*}
On the other hand, if instead 
$C_1\|\tilde\Theta_0\psi\|^2\ge \hat{c} \|\psi\|^2,$
then \eqref{bound2} implies
\begin{align*}
& c\ \lambda^2\|\psi\|^2  \le \left\langle H^\lambda(\bk)\psi,\psi\right\rangle . \nn\end{align*}
Therefore, in either case  $H^\lambda(\bk)=-(\nabla_\bx+i\bk)^2+V^\lambda(\bx)$ satisfies
\begin{align}
& c\|\psi\|^2  \le \left\langle \left[-(\nabla_\bx+i\bk)^2+V^\lambda(\bx)\right]\psi,\psi\right\rangle .
\label{both-cases}\end{align}

To bound $\|(\nabla_\bx+i\bk)\psi\|^2$ we observe that 
\begin{align}
&\left\langle \left[-(\nabla_\bx+i\bk)^2\right]\psi,\psi\right\rangle\nn\\
&=\left\langle \left[-(\nabla_\bx+i\bk)^2+V^\lambda(\bx)\right]\psi,\psi\right\rangle
 - \left\langle V^\lambda(\bx)\psi,\psi\right\rangle\nn\\
 &\le \left\langle \left[-(\nabla_\bx+i\bk)^2+V^\lambda(\bx)\right]\psi,\psi\right\rangle
 + C\lambda^2\|\psi\|^2,
 \nn\end{align}
 since $ |V^\lambda(\bx)|\le C\lambda^2$ everywhere.  Therefore, by \eqref{both-cases}
 \begin{equation}
 \|(\nabla_\bx+i\bk)\psi\|^2\ \le\ C'\lambda^2\ \left\langle \left[-(\nabla_\bx+i\bk)^2+V^\lambda(\bx)\right]\psi,\psi\right\rangle\ .
 \label{grad-bound}\end{equation}
 Estimates \eqref{both-cases} and \eqref{grad-bound} imply the main global energy estimate and complete the proof of Proposition \ref{global-en-est}.
\end{proof}

\subsection{Global energy estimate on a fixed Hilbert space}\label{global-en-fixed}

We continue with the convention that  norms and inner products are taken in $L^2(\R^2/\Lambda_h)$, if not otherwise specified.

Proposition \ref{global-en-est} (see also Corollaries \ref{cor1-main-en} and \ref{cor2-main-en}) provides a lower bound on $H^\lambda(\bk)=-(\nabla_\bx+i\bk)^2+V^\lambda(\bx)$ subject to the $\bk-$ dependent orthogonality conditions:
 \begin{equation}\label{orthogs-2}
\left\langle p^\lambda_{\bk,A},\psi\right\rangle\ =\ 0\ \ {\rm and}\ \  \left\langle p^\lambda_{\bk,B},\psi\right\rangle\ =\ 0.
\end{equation}
For our Lyapunov-Schmidt reduction strategy of Section \ref{LS-reduction}, we require bounds on $H^\lambda(\bk)$
 and invertibility on a {\it fixed} subspace of the Hilbert space $L^2(\R^2/\Lambda_h)$, defined in terms of the conditions
 \eqref{orthogs-2} for {\it fixed} quasi-momentum, $\bk=\tK$.

 %
 \begin{corollary}\label{main-en-cor3}
 Fix $K_{max}>0$. There exist a small positive constants $\hat{c}$, which decreases with increasing $K_{max}$, 
  and $\lambda_\star$, depending on $V_0$ and $K_{max}$, 
 such that the following holds:
 Let $\tK\in\R^2$ with $|\tK|\le K_{max}$.  Let $\psi\in H^2(\R^2/\Lambda_h)$ with 
 \begin{equation}\label{orthogs-tK1}
\left\langle p^\lambda_{\tK,A},\psi\right\rangle\ =\ 0\ \ {\rm and}\ \  \left\langle p^\lambda_{\tK,B},\psi\right\rangle\ =\ 0.
\end{equation}
Then, for all $\bk\in\C^2$ such that $|\bk-\tK|<\hat{c} \lambda^{-1}$ we have
\begin{align*}
c\ \|\psi\|\ +\ c\lambda^{-1} \|\nabla_\bx\psi\|\
\le\ \left\|\left[-(\nabla_\bx+i\bk)^2+V^\lambda(\bx)\right]\psi\right\|\ .
\end{align*}
 \end{corollary}
 
\begin{proof}[Proof of Corollary \ref{main-en-cor3}] Fix $K_{max}>0$ and let 
\begin{equation}
 \tK\in\R^2,\quad |\tK|\le K_{max},\quad \bk\in\C^2\cap\{|\bk-\tK|<\hat{c}\lambda^{-1}\}\ ,
\label{chat-tK}
\end{equation}
 where $\hat{c}$ is to be chosen small enough below.  Let $\psi\in H^2(\R^2/\Lambda_h)$ be such that 
orthogonality conditions \eqref{orthogs-tK1} hold. 
By Corollary \ref{cor2-main-en}, with $\tK$ in place of $\bk$, we have 
\begin{align}
c\ \|\psi\|^2\ +\ c\lambda^{-2} \|\nabla_\bx\psi\|^2\
\le\ \left\|\left[-(\nabla_\bx+i\tK)^2+V^\lambda(\bx)\right]\psi\right\|^2\ .
\label{tK-Hlambda-bound-1}
\end{align}
 To conclude the proof, it suffices to bound the right hand side of estimate \eqref{tK-Hlambda-bound-1} by the same expression, but with $\tK$ replaced by $\bk$. 
Using \eqref{chat-tK}, we have
\begin{align*}
&\left\|\left[-(\nabla_\bx+i\tK)^2+V^\lambda(\bx)\right]\psi\right\|\  \\
&\quad \le   \left\|\left[-(\nabla_\bx+i\bk)^2+V^\lambda(\bx)\right]\psi \right\| + \left\|2i(\bk-\tK)\cdot(\nabla_\bx+i\bk)\psi\right\| \\
& \qquad \qquad +  \left|(\bk-\tK)\cdot(\bk-\tK)\right|\cdot \|\psi\| \\
 &\quad \le  \left\|\left[-(\nabla_\bx+i\bk)^2+V^\lambda(\bx)\right]\psi\right\| \\ 
 &\qquad \qquad + \hat{c}\lambda^{-1}\left(\ 2\left\|\nabla_\bx\psi\right\|+4K_{max}\ \|\psi\|\ +\ \hat{c}\lambda^{-1}\|\psi\|\right)\ .
 \end{align*}
 By \eqref{tK-Hlambda-bound-1}, the  latter three terms are controlled by
 $\hat{c}\ \left\|\left[-(\nabla_\bx+i\tK)^2+V^\lambda(\bx)\right]\psi\right\|$. Therefore, by choosing $\hat{c}$ sufficiently small, we find  
 \[\left\|\left[-(\nabla_\bx+i\tK)^2+V^\lambda(\bx)\right]\psi\right\|\lesssim \left\|\left[-(\nabla_\bx+i\bk)^2+V^\lambda(\bx)\right]\psi\right\|.\]
 Substituting this bound into  \eqref{tK-Hlambda-bound-1}, 
 completes  the proof of Corollary \ref{main-en-cor3}.
 \end{proof}

 \subsection{The resolvent}\label{resolvent-bounds}

 The following result is required to control the resolvent of $H^\lambda(\bk)$ on the subspace defined by the orthogonality conditions
 \eqref{orthogs-tK1}; see Lemma \ref{lem1-resolvent} and Proposition \ref{prop2-resolvent} below.
 
 \begin{corollary}\label{main-en-cor4} Fix $K_{max}>0$ and let $\tK\in\R^2$ with $|\tK|\le K_{max}$. Let $\psi\in H^2(\R^2/\Lambda_h)$ with $\left\langle p_{\tK,I}^\lambda,\psi\right\rangle=0$ for $I=A$ and $B$. Suppose that 
 $\psi$ satisfies 
 \begin{equation}
   H^\lambda(\tK)\ \psi\ =\ \varphi\ +\ \sum_{I=A,B}\ \mu_I\ p_{\tK,I}^\lambda , \label{HtkPsi}
   \end{equation}
 with $\mu_I\in\C$,\ and 
 $\left\langle p_{\tK,I}^\lambda,\varphi \right\rangle=0$ for $I=A$ and $B$.
 
 Then, 
 \begin{equation}
 c \|\psi\|\ +\ c\lambda^{-1}\ \|\nabla_\bx\psi\|\ \le\ \|\varphi\|\ .
\label{en-est-cor4} \end{equation}
 \end{corollary}
 
 \begin{proof}[Proof of Corollary \ref{main-en-cor4}] By Corollary \ref{cor2-main-en}, with $\bk=\tK$,
 \begin{align}
c \|\psi\| + c\lambda^{-1} \|\nabla_\bx\psi\|\
\le \left\|\left[-(\nabla_\bx+i\tK)^2+V^\lambda(\bx)\right]\psi\right\| \lesssim\
\|\varphi\| +  \sum_{I=A,B} |\mu_I| .
\label{eqn-cor4a}\end{align}
Taking the inner product of $p_{\tK,J}^\lambda$ with \eqref{HtkPsi}, using self-adjointness and the assumed orthogonality to $\varphi$, we obtain 
$\sum_I\left\langle p_{\tK,J}^\lambda,p_{\tK,I}^\lambda\right\rangle\mu_I=\left\langle H^\lambda(\tK)p_{\tK,J}^\lambda,\psi\right\rangle$. By the near-orthogonality 
relation \eqref{pkI-normalize1} and the Cauchy-Schwarz inequality, we have for $I=A, B$ that
  \[|\mu_I|\le\left(1+\mathcal{O}(e^{-c\lambda})\right)\ \sum_{J=A,B} \|H^\lambda(\tK)p_{\tK,J}^\lambda\|\ \|\psi\|. \]
Next, the bound \eqref{Hlambda-pk} implies $|\mu_I|\lesssim e^{-c\lambda}\|\psi\|$.  Therefore,  $\sum_I|\mu_I|$  can be absorbed into the left hand side of \eqref{eqn-cor4a}, for $\lambda$ large and the  estimate \eqref{en-est-cor4} follows. This completes the proof of Corollary \ref{main-en-cor4}.
\end{proof}

Fix $K_{max}>0$  and let $\tK\in\R^2$ with $|\tK|\le K_{max}$. We now introduce the Hilbert space, $\mathscr{H}_{_{AB}}$: 
\begin{equation}
\mathscr{H}_{_{AB}}\ \equiv\ \left[\ {\rm span}\left\{\ p^\lambda_{\tK,I}\ :\ I=A,B\ \right\}\ \right]^\perp\  \ {\rm in}\ \ L^2(\R^2/\Lambda_h)\ ,
\label{scr-H-def}
\end{equation}
and the associated orthogonal projection: $\Pi_{_{AB}}: L^2(\R^2/\Lambda_h)\to\mathscr{H}_{_{AB}}$.
 The space $\mathscr{H}_{_{AB}}$ depends on the choice of $\tK\in\R^2$. 
Also, introduce the subspace  $\mathscr{H}_{_{AB}}^2=\mathscr{H}_{_{AB}}\cap H^2(\R^2/\Lambda_h)$. The norms and inner products on $\mathscr{H}_{_{AB}}$ and $\mathscr{H}_{_{AB}}^2$ are those inherited from $L^2(\R^2/\Lambda_h)$ and 
$H^2(\R^2/\Lambda_h)$, respectively.
Recall that 
$H^\lambda(\tK)=-\left(\nabla_\bx+i\tK\right)^2+V^\lambda(\bx):H^2(\R^2/\Lambda_h)\to L^2(\R^2/\Lambda_h)$. 

For $\varphi\in\mathcal{H}_{_{AB}}$, we now study the solvability in $\mathcal{H}^2_{_{AB}}$ of $\Pi_{_{AB}}H^\lambda(\tK)\psi=\varphi$ (Lemma \ref{lem1-resolvent})  and then 
$\Pi_{_{AB}}\left(H^\lambda(\bk)-\Omega\right)\psi=\varphi$, for $\bk$ near $\tK$  
 (Proposition \ref{prop2-resolvent}). 
\begin{lemma}\label{lem1-resolvent}
For any $\varphi\in\mathscr{H}_{_{AB}}$, there exists one and only one 
$\psi\in\mathscr{H}_{_{AB}}^2$ such that 
\begin{equation}
\Pi_{{_{AB}}}\ H^\lambda(\tK)\ \psi\ =\ \varphi\ .
\label{PiHvphi}
\end{equation}
Moreover, that $\psi$ satisfies the bounds
\begin{equation}
c\ \left(\ \|\psi\|_{_{\mathscr{H}_{_{AB}}}}\ +\ \lambda^{-1}\ \|\nabla_\bx\psi\|_{_{L^2(\R^2/\Lambda_h)}}\ \right)\ \le\ 
\|\varphi\|_{_{\mathscr{H}_{_{AB}}}}.
\label{lem1-bound}\end{equation}
\end{lemma}

\begin{proof}[Proof of Lemma \ref{lem1-resolvent}]
We first prove that \eqref{PiHvphi} admits a solution $\psi\in\mathscr{H}_{_{AB}}^2$ for a dense subset of $\varphi\in\mathscr{H}_{_{AB}}$.  Indeed, if not, then there would exist a nontrivial $\varphi_0\in\mathscr{H}_{_{AB}}$  such that\begin{equation}
\left\langle \varphi_0,\Pi_{_{AB}}\ H^\lambda(\tK)\ \psi\right\rangle_{_{\mathscr{H}_{_{AB}}}}\ =\ 0,\ \ 
\textrm{for all}\ \psi\in\mathscr{H}_{_{AB}}^2\ .
\label{orthog-dense}
\end{equation}
Since $\varphi_0\in\mathscr{H}_{_{AB}}$, \eqref{orthog-dense} is equivalent to 
\begin{equation}
\left\langle \varphi_0, H^\lambda(\tK)\ \psi\right\rangle_{_{\mathscr{H}_{_{AB}}}}\ =\ 0,\ \ 
\textrm{for all}\ \psi\in\mathscr{H}_{_{AB}}^2\ .
\label{orthog-dense1}
\end{equation}
We shall show that $\varphi_0=0$ yielding a contradiction. To do so, we first show
that $\varphi_0\in\mathscr{H}^2_{_{AB}}$, so that we may write \eqref{orthog-dense1} as an orthogonality condition 
 on $H^\lambda(\tK)\varphi_0$.

Now, given any $\psi\in H^2(\R^2/\Lambda_h)$ we may write
\begin{equation}
\psi\ =\ \sum_{I=A,B}\ \alpha_I\ p^\lambda_{\tK,I}\ +\ \tilde\psi,\ \ \textrm{with}\ \ \tilde\psi\in \mathscr{H}_{_{AB}}^2\ .
\label{dense-psi-decomp}
\end{equation}
In  particular, 
\begin{equation*}
\left\langle p^\lambda_{\tK,J},\psi\right\rangle\ =\ \sum_{I=A,B}\ \alpha_I\ \left\langle p^\lambda_{\tK,J},p_{\tK,I}\right\rangle,\ \ 
{\rm for}\ \ J=A,B.
\nn\end{equation*}
By \eqref{pkI-normalize1},  $\left\langle p^\lambda_{\tK,J},p^\lambda_{\tK,J}\right\rangle$ differs from $\delta_{JI}$ by at most order $e^{-c\lambda}$. Therefore, 
$\alpha_I=\left\langle \sum_J\gamma_I^J p^\lambda_{\tK,J},\psi\right\rangle$, for a matrix $(\gamma^J_I)$, which is independent of $\psi$.  

Substituting \eqref{dense-psi-decomp} into \eqref{orthog-dense1}, we have
\begin{align}
\left\langle\varphi_0,H^\lambda(\tK)\psi\right\rangle\ &=\ 
\left\langle \varphi_0,H^\lambda(\tK)\ 
\left[
\sum_{I=A,B}\ \alpha_I\ p^\lambda_{\tK,I}\ +\ \tilde\psi 
\right]\ 
\right\rangle\nn\\ 
&=\ \sum_{I,J}\left\langle\gamma_I^J p^\lambda_{\tK,J},\psi\right\rangle\ \left\langle\varphi_0,H^\lambda(\tK)p^\lambda_{\tK,I} \right\rangle= \sum_{J=A,B}\left\langle\tilde{p}^\lambda_J,\psi\right\rangle\ ,
\label{tpJ}\end{align}
where $\tilde{p}^\lambda_J=\sum_{I=A,B}\gamma_I^J\overline{\left\langle\varphi_0,H^\lambda(\tK)p^\lambda_{\tK,I}\right\rangle }p^\lambda_{\tK,J}\ \in H^2(\R^2/\Lambda_h)$ is independent of $\psi$. 
 Rewriting \eqref{tpJ} we have
\begin{equation*}
\left\langle \varphi_0,-(\nabla_\bx+i\tK)^2\psi\right\rangle\ =\ 
\left\langle -V^\lambda\varphi_0+\ \sum_{J=A,B}\tilde{p}^\lambda_J\ ,\ \psi \right\rangle
\end{equation*}
for arbitrary $\psi\in H^2(\R^2/\Lambda_h)$.  Thus, 
 $-(\nabla_\bx+i\tK)^2\varphi_0=\ 
 -V^\lambda\varphi_0+\ \sum_{J=A,B}\tilde{p}_J^\lambda\in L^2(\R^2)$ in the sense of distributions, which implies that 
  $\varphi_0\in H^2(\R^2/\Lambda_h)$.  Furthermore, since  $\left\langle p_{\tK,I}^\lambda,\varphi_0\right\rangle=0$ for $I=A,B$,  we have $\varphi_0\in\mathscr{H}_{_{AB}}^2$ as claimed. Therefore, setting $\psi=\varphi_0$ in \eqref{orthog-dense1} gives  $\left\langle H^\lambda(\tK)\varphi_0,\varphi_0\right\rangle=0$. Applying Proposition \ref{global-en-est} we have
   $c\|\varphi_0\|^2 \le \left\langle H^\lambda(\tK)\varphi_0,\varphi_0\right\rangle=0$. Hence, $\varphi_0=0$.
  This proves that equation \eqref{PiHvphi}, 
$ \Pi_{_{AB}}\ H^\lambda(\tK)\ \psi\ =\ \varphi$,  
 has a solution $\psi\in\mathscr{H}_{_{AB}}^2$ for a dense subset of $\varphi\in\mathscr{H}_{_{AB}}$. 
 Moreover, the bound  \eqref{lem1-bound} holds,  thanks to Corollary  \ref{main-en-cor4}.
Standard arguments  using \eqref{lem1-bound} extend these assertions to all $\varphi\in\mathscr{H}_{_{AB}}$. 
  
 Finally, uniqueness holds since the difference of two solutions of \eqref{PiHvphi}, denoted $\Upsilon$, satisfies the homogeneous equation
  $\Pi_{_{AB}}H^\lambda(\tK) \Upsilon=0$. Therefore, $\Upsilon$ satisfies \eqref{HtkPsi} with $\varphi\equiv0$. 
  Applying  \eqref{en-est-cor4} yields that $\Upsilon=0$. This completes the proof of Lemma \ref{lem1-resolvent}.
  \end{proof}
  %
   %
 
 Our next step is to extend results on the invertibility of
$ \Pi_{_{AB}} H^\lambda(\tK)$ on $\mathscr{H}_{_{AB}}$ to results 
  on the invertibility of
 $ \Pi_{_{AB}}\left(H^\lambda(\bk)-\Omega\right)$ on $\mathscr{H}_{_{AB}}$, for $\bk\in\C^2$ sufficiently near $\tK\in\R^2$
  and $\Omega$ sufficiently small. 
 
For  $\varphi\in\mathscr{H}_{_{AB}}$, consider the equation 
 \begin{equation}
 \Pi_{_{AB}} H^\lambda(\tK)\psi=\varphi\ .
 \label{eqn10}
 \end{equation}
Via Lemma \ref{lem1-resolvent}, we define the mapping $A:\mathscr{H}_{_{AB}}\to\mathscr{H}_{_{AB}}^2$, 
  \begin{align}
\varphi\ \mapsto\ A\varphi\ &=\ \psi\in\mathscr{H}_{_{AB}}^2 ,
 \label{eqn11}
 \end{align}
which gives the unique solution of \eqref{eqn10}.

We then set
 \begin{align*}
 B_j\varphi\ &=\ \Pi_{_{AB}} \partial_{x_j}A\varphi\ =\ \Pi_{_{AB}} \partial_{x_j}\psi,\ \ j=1,2\  .
 \end{align*}
 Lemma \ref{lem1-resolvent} tells us that $A, B_1, B_2$ are bounded operators on $\mathscr{H}_{_{AB}}$ with norm bounds:
 \begin{equation*}
 \|A\|_{_{\mathscr{H}_{_{AB}}\to\mathscr{H}_{_{AB}}}}\lesssim1,\quad \|B_j\|_{_{\mathscr{H}_{_{AB}}\to\mathscr{H}_{_{AB}}}}\ \lesssim\ \lambda,\quad j=1,2\ .
 \end{equation*}
 
 For $\varphi\in\mathscr{H}_{_{AB}}$,  we now try to solve the equation
 \begin{equation}
 \Pi_{_{AB}} \left(H^\lambda(\bk)-\Omega\right)\psi\ =\ \varphi,
 \label{eqn14}
 \end{equation}
 for $\psi\in\mathscr{H}_{_{AB}}^2$. Here, $\bk\in\C^2$ and $\Omega\in\C$. 
 
 For $\psi\in\mathscr{H}_{_{AB}}^2$, set $\Pi_{_{AB}}H^\lambda(\tK)\psi=\tilde\varphi$. 
 Then, $\tilde\varphi\in\mathscr{H}_{_{AB}}$, $\psi=A\tilde\varphi$ and \eqref{eqn14} is equivalent to the equation
 \begin{equation*}
 \tilde\varphi - 2i\sum_{j=1,2}(k_j-\tilde{K}_j) B_j\tilde\varphi + \left(\ (\bk-\tK)\cdot(\bk-\tK)-\Omega\ \right)A\tilde\varphi
  =\ \varphi\ .
 \end{equation*}
 Therefore, the solution to \eqref{eqn14} (under conditions on $\bk$ and $\Omega$ to be spelled out below) is given by 
 \begin{align}
 &\psi\ =\ A \tilde\varphi,\ \ {\rm where}\ \tilde\varphi\in\mathscr{H}_{_{AB}}\ {\rm solves}\label{eqn17-minus}\\
& \left\{\ I\ -\ 2i\sum_{j=1,2}(k_j-\tilde{K}_j) B_j\ +\ \left(\ (\bk-\tK)\cdot(\bk-\tK)-\Omega\ \right)A\ \right\}\  \tilde\varphi\ =\ \varphi\ .\label{eqn17}
\end{align}
The operator in curly brackets in \eqref{eqn17}  can be inverted, via a Neumman series, provided the following three quantities
\[
 |\bk-\tK|\cdot\|B\|_{_{\mathscr{H}_{_{AB}}\to\mathscr{H}_{_{AB}}}}\lesssim |\bk-\tK|\ \lambda,\ \ |\bk-\tK|,\ {\rm and}\ \ 
 |\Omega| \cdot\|A\|_{_{\mathscr{H}_{_{AB}}\to\mathscr{H}_{_{AB}}}}
 \]
are all less than a small constant $\hat{c}$. Thus, \eqref{eqn17} 
has a unique solution $\tilde\varphi\in\mathscr{H}_{_{AB}}$ with 
$\|\tilde\varphi\|_{_{\mathscr{H}_{AB}}}\lesssim\|\varphi\|_{_{\mathscr{H}_{AB}}}$. By \eqref{eqn17-minus} and \eqref{eqn11} ,
 $\psi=A\tilde\varphi$ solves \eqref{eqn14}. Furthermore,  by  \eqref{lem1-bound} $\psi$ satisfies the bound:
 \begin{align*}
 &\|\psi\|_{_{\mathscr{H}_{AB}}}\ +\ \lambda^{-1} \|\nabla_\bx\psi\|_{_{L^2(\R^2/\Lambda_h)}}
  =  \|A\tilde\varphi\|_{_{\mathscr{H}_{AB}}}\ +\ \lambda^{-1} \|\nabla_\bx A\tilde\varphi\|_{_{L^2(\R^2/\Lambda_h)}}\
 \lesssim\|\varphi\|_{_{\mathscr{H}_{AB}}} .
   \label{psi-est}
 \end{align*}

Introducing the solution mapping
 $\varphi\mapsto \psi={\rm Res}^{\lambda,\tK}(\bk,\Omega)\ \varphi$, for equation \eqref{eqn14},
we have the following

\begin{proposition}[Bound on the resolvent, ${\rm Res}^{\lambda,\tK}(\bk,\Omega)$]\label{prop2-resolvent}
%
Fix $\tK\in\R^2$ with $|\tK|\le K_{max}$. Let 
\begin{equation}
U\ \equiv\ \left\{(\bk,\Omega)\in\C^2\times\C\ :\ |\bk-\tK|<\hat{c}\lambda^{-1},\ |\Omega|<\hat{c}\ \right\}\ .
\label{Udef}
\end{equation}
For a small enough constant $\hat{c}$, and $\lambda>\lambda_\star$ sufficiently large, we have the following:
\begin{enumerate}
\item Let $\varphi\in\mathscr{H}_{_{AB}}$ and $(\bk,\Omega)\in U$. Then the equation
\begin{equation}
 \Pi_{_{AB}} \left(H^\lambda(\bk)-\Omega\right)\psi\ =\ \varphi \label{eqn-star}
 \end{equation}
 has a unique solution $\psi\in\mathscr{H}_{_{AB}}^2$. 
 \item Let ${\rm Res}^{\lambda,\tK}(\bk,\Omega) \varphi$ denote the solution, $\psi$, of equation \eqref{eqn-star}.
Then, the mapping:
 \begin{align}
 &\varphi\mapsto {\rm Res}^{\lambda,\tK}(\bk,\Omega)\ \varphi,\quad {\rm Res}^{\lambda,\tK}(\bk,\Omega):
 \mathscr{H}_{_{AB}}\to \mathscr{H}_{_{AB}}^2
  \label{res-map}
 \end{align}
 depends analytically on $(\bk,\Omega)\in U$ and satisfies the estimates
 \begin{equation}
\qquad \frac{c^j}{j!}\ \left\|\ {\rm Res}_j^{\lambda,\tK}(\bk,\Omega) \right\|_{_{\mathscr{H}_{_{AB}}\to\mathscr{H}_{_{AB}}}},\ \ 
 \lambda^{-1}\ \left\|\ \nabla_\bx{\rm Res}^{\lambda,\tK}(\bk,\Omega)\right\|_{_{\mathscr{H}_{_{AB}}\to L^2}}\ \le\ C
\label{deriv-bounds}\end{equation}
for $j\ge0$ and $(\bk,\Omega)\in U$,
where
\begin{equation}
{\rm Res}_j^{\lambda,\tK}(\bk,\Omega) \equiv\  \partial_\Omega^j\ {\rm Res}^{\lambda,\tK}(\bk,\Omega)\ ,
\label{Resj-def}
\end{equation}
and we adopt the convention ${\rm Res}_0^{\lambda,\tK}={\rm Res}^{\lambda,\tK}$.
Here, the constants $\hat{c}, \lambda_\star, c, C$ may depend on $K_{max}$.
\item For $(\bk,\Omega)\in(\R^2\times\R)\cap U$, the mapping ${\rm Res}^{\lambda,\tK}(\bk,\Omega):
 \mathscr{H}_{_{AB}}\to \mathscr{H}_{_{AB}}$ is self-adjoint. 
\end{enumerate}
\end{proposition}

\begin{proof}[Proof of Proposition \ref{prop2-resolvent}]
Assertions (1) and (2) are proved in the discussion just before the proposition. 
Assertion (3) on self-adjointness is proved as follows. Let $\varphi_j\in\mathscr{H}_{_{AB}},\ j=1,2$.  Then, by construction 
\[ \Pi_{_{AB}} \left(H^\lambda(\bk)-\Omega\right)\Pi_{_{AB}} {\rm Res}^{\lambda,\tK}(\bk,\Omega)\varphi_j=\Pi_{_{AB}} \left(H^\lambda(\bk)-\Omega\right){\rm Res}^{\lambda,\tK}(\bk,\Omega)\varphi_j=\varphi_j.
\] 
By self-adjointness of $H^\lambda(\bk)-\Omega$ for  $(\bk,\Omega)\in \left(\R^2\times\R\right)\cap U$, we have 
\begin{align}
&\left\langle\varphi_1, {\rm Res}^{\lambda,\tK}(\bk,\Omega)\varphi_2\right\rangle_{_{\mathscr{H}_{_{AB}} }} \nn \\
 & \quad = \left\langle \Pi_{_{AB}} \left(H^\lambda(\bk)-\Omega\right) \Pi_{_{AB}}  {\rm Res}^{\lambda,\tK}(\bk,\Omega)\varphi_1, {\rm Res}^{\lambda,\tK}(\bk,\Omega)\varphi_2\right\rangle_{_{\mathscr{H}_{_{AB}} }}  \nn\\
 & \quad = \left\langle   {\rm Res}^{\lambda,\tK}(\bk,\Omega)\varphi_1, 
 \Pi_{_{AB}} \left(H^\lambda(\bk)-\Omega\right) \Pi_{_{AB}}{\rm Res}^{\lambda,\tK}(\bk,\Omega)\varphi_2\right\rangle_{_{\mathscr{H}_{_{AB}} }}  \nn\\
 & \quad = \left\langle   {\rm Res}^{\lambda,\tK}(\bk,\Omega)\varphi_1, \varphi_2\right\rangle_{_{\mathscr{H}_{_{AB}} }} .  \nn
\end{align}
This proves assertion (3) of the Proposition \ref{prop2-resolvent}.
\end{proof}
%
%

 We next  bootstrap the above arguments to obtain a bound on the norm of ${\rm Res}^{\lambda,\tK}(\bk,\Omega):
 \mathscr{H}_{_{AB}}\to \mathscr{H}_{_{AB}}^2$.
\begin{corollary}\label{cor1-resolvent}
 For all $(\bk,\Omega)\in U$, defined in \eqref{Udef} we have the additional bound for ${\rm Res}^{\lambda,\tK}(\bk,\Omega)$:
 \begin{equation*}
 \left\|{\rm Res}^{\lambda,\tK}(\bk,\Omega)\right\|_{_{\mathscr{H}_{_{AB}}\to\mathscr{H}^2_{_{AB}}}} \le\ C(\lambda,\tK)\ .
 \end{equation*}
\end{corollary}

\begin{proof}[Proof of Corollary \ref{cor1-resolvent}] 
Recall the mapping $A:\varphi\in\mathscr{H}_{_{AB}}\mapsto\psi=A\varphi\in\mathscr{H}^2_{_{AB}}$, which solves 
 $\Pi_{_{AB}} H^\lambda(\tK) \psi = \varphi$.  We claim that this mapping is bounded, {\it i.e.}
$\|A\|_{_{\mathscr{H}_{_{AB}}\to\mathscr{H}_{_{AB}}^2}} \ \le\ C$, with $C$ a constant depending on $\lambda$ and $\tK$. 

 Let us first prove Corollary \ref{cor1-resolvent}, assuming this claim. Assume $\varphi\in\mathscr{H}_{_{AB}}$.
   As shown in the discussion leading up 
  to the assertion of Proposition \ref{prop2-resolvent}, 
  ${\rm Res}^{\lambda,\tK}(\bk,\Omega)\varphi = A\tilde \varphi$, where $\tilde\varphi\in\mathscr{H}_{_{AB}}$ is the unique solution of  \eqref{eqn17}. Moreover, $\|\tilde\varphi\|_{_{\mathscr{H}_{_{AB}} }}\lesssim \|\varphi\|_{_{\mathscr{H}_{_{AB}} }}$. Therefore, 
  $\|{\rm Res}^{\lambda,\tK}(\bk,\Omega)\varphi\|_{_{\mathscr{H}^2_{_{AB}} }}\le C \|A\|_{_{\mathscr{H}_{_{AB}}\to\mathscr{H}_{_{AB}}^2}}\|\tilde\varphi\|_{_{\mathscr{H}_{_{AB}} }}\le C^\prime \|A\|_{_{\mathscr{H}_{_{AB}}\to\mathscr{H}_{_{AB}}^2}}
 \|\varphi\|_{_{\mathscr{H}_{_{AB}} }}$.

We now prove the above claim. By definition, $\psi$ satisfies
\[ \Pi_{_{AB}} \left[-\Delta_\bx-2 i\tK\cdot\nabla_\bx+|\tK|^2\ +\ V^\lambda(\bx)\ \right] \psi\ =\ \varphi ,\]
and therefore
\begin{equation}\label{eqn18}
 \left[-\Delta_\bx-2 i\tK\cdot\nabla_\bx+|\tK|^2\ +\ V^\lambda(\bx)\ \right] \psi\ =\ 
 \varphi\ -\ \sum_{I=A,B}\alpha_I p_{\tK,I}^\lambda ,
 \end{equation}
 for complex scalars $\alpha_A$ and $\alpha_B$. By Corollary \ref{main-en-cor4},
$ \|\psi\|_{_{\mathscr{H}_{_{AB}}}}\ \le\ C\ \|\varphi\|_{_{\mathscr{H}_{_{AB}}}}$, 
 and hence
 \[ \left\|\ \left[-\Delta_\bx-2 i\tK\cdot\nabla_\bx+|\tK|^2\ +\ V^\lambda(\bx)\ \right]\psi\ \right\|_{_{H^{-2}(\R^2/\Lambda_h)}}
 \le C\ \|\varphi\|_{_{\mathscr{H}_{_{AB}}}}\ .\]
 So, $\left\|  \varphi\ -\ \sum_{I=A,B}\alpha_I p_{\tK,I}^\lambda \right\|_{_{H^{-2}(\R^2/\Lambda_h)}}
\le C\ \|\varphi\|_{_{\mathscr{H}_{_{AB}}}}$,
 and consequently 
 \[\left\|\sum_{I=A,B}\alpha_I\ p_{\tK,I}^\lambda\right\|_{_{H^{-2}(\R^2/\Lambda_h)}}\le C\|\varphi\|_{_{\mathscr{H}_{_{AB}}}}.\]

Since any two norms on a 2-dimensional vector space are equivalent, we have
 $\sum_{I=A,B} |\alpha_I| \le C(\lambda,\tK)\ \|\varphi\|_{\mathscr{H}_{_{AB}}}$, where we make no attempt to see how 
 $C(\lambda,\tK)$ depends on $\lambda$ and $\tK$.  
 Returning now to \eqref{eqn18}, we have
 \[\|\Delta_\bx\psi\| \le 2|\tK| \|\nabla_\bx\psi\|\ +\ \left( |\tK|^2 + C\lambda^2\right)\ \|\psi\|
 \ +\ C(\lambda,\tK)\ \|\varphi\|_{_{\mathscr{H}_{_{AB}}}}\ .
 \]
By Lemma \ref{lem1-resolvent}, all terms on the right hand side are dominated by 
$C^\prime(\lambda,\tK)\ \|\varphi\|_{_{\mathscr{H}_{_{AB}}}}$. Therefore, 
$ \|\psi\|_{L^2(\R^2/\Lambda_h)},\ \ \|\Delta_\bx\psi\|_{L^2(\R^2/\Lambda_h)}\ \le\ 
C(\lambda,\tK)\ \|\varphi\|_{_{\mathscr{H}_{_{AB}}}}$. 
Consequently, $\|\psi\|_{_{\mathscr{H}_{_{AB}}^2}}\le C(\lambda,\tK)\ \|\varphi\|_{_{\mathscr{H}_{_{AB}}}}$, {\it i.e.}
$\|A\|_{_{\mathscr{H}_{_{AB}}\to\mathscr{H}_{_{AB}}^2}} \ \le\ C(\lambda,\tK)$.
 This completes the proof of the claim and therewith Corollary \ref{cor1-resolvent}.
\end{proof}

 \section{Dirac points of $H^\lambda$ in the strong binding regime}\label{dirac-points-sb}
  
 Let $\bK_\star$ be any vertex of $\B_h$. 
We study the  eigenvalue problem $(H^\lambda-\Omega) \psi = 0$,\ $\psi\in H^2_{\bK_\star}$, where
$H^\lambda=-\Delta+V^\lambda=-\Delta+\lambda^2 V-E_0^\lambda$.  
 Recall 
\begin{enumerate}
\item $H^\lambda$ and $\mathcal{R}$ map a dense subspace of $L^2_{\bK_\star}$ to itself. 
 \item The commutator vanishes; $\left[ H^\lambda,\mathcal{R}\right]=0$.
\item $L^2_{\bK_\star}= L^2_{\bK_\star,1}\oplus L^2_{\bK_\star,\tau}\oplus L^2_{\bK_\star,\overline\tau}$,
 where $L^2_{\bK_\star,\sigma}$ is the subspace of $L^2_{\bK_\star}$, that is also the eigenspace of $\mathcal{R}$ with eigenvalue $\sigma,\ \sigma=1,\tau,\overline\tau$. 
\end{enumerate} 

 Since $\left[ H^\lambda,\mathcal{R}\right]=0$, we may decouple the eigenvalue problem for $H^\lambda$ into the three eigenvalue problems, defined by the equation $H^\lambda \psi = \Omega\psi$, with  $\psi\in L^2_{\bK_\star,\sigma}$, 
 $ \sigma=1,\tau,\overline\tau$.  Thus, for $\sigma=1, \tau, \overline\tau$, we define the {\it $L^2_{\bK_\star,\sigma}-$ eigenvalue problem}
\begin{equation}
\left( H^\lambda - \Omega \right) \Psi\ =\ 0,\qquad \Psi\in H^2_{\bK_\star,\sigma}\ .
\label{L2sigma-evp}\end{equation}

To establish the existence of Dirac points at all vertices of $\brill_h$, by part (2) of Remark \ref{onedp-implies-all}, it suffices to consider $\bK_\star=\bK$.
%

\begin{theorem}\label{solve-L2tau-evp}
There exists $\lambda_\star>0$, depending on $V_0$,  such that the following holds: For all $\lambda>\lambda_\star$,
\begin{enumerate}
\item The $L^2_{\bK,\tau}$ eigenvalue problem, \eqref{L2sigma-evp} with $\sigma=\tau$, 
has a simple eigenvalue, $\Omega^\lambda$, which satisfies the bound
  $|\Omega^\lambda|\lesssim \rho_\lambda e^{-c\lambda}$, with corresponding eigenfunction 
$\Psi^\lambda=\Phi_1^\lambda$. Here, $\rho_\lambda$ is displayed in \eqref{rho-lam-def} and satisfies the upper and lower bounds \eqref{rho-lambda-bounds}.
\item $\Omega^\lambda$ is a simple $L^2_{\bK,\overline\tau}-$ eigenvalue of the eigenvalue problem, \eqref{L2sigma-evp} with $\sigma=\overline\tau$, with corresponding eigenfunction $\Phi_2^\lambda=(\mathcal{C}\circ\mathcal{I})[\Phi_1^\lambda](\bx)\equiv\overline{\Phi_1^\lambda(2\bx_c-\bx)}$.
 \item The $L^2_{\bK,1}$ eigenvalue problem, \eqref{L2sigma-evp} with $\sigma=1$, 
 has no nontrivial solution. Therefore, the eigenspace of $H^\lambda$, for the eigenvalue $\Omega^\lambda$,
  is two-dimensional and has a basis $\{\Phi_1^\lambda,\Phi_2^\lambda\}$. 
\item If in $(1)$ we consider instead the  eigenvalue problem with $\bK^\prime-$ pseudo-periodic boundary conditions ($\bK^\prime$ instead of $\bK$), then all assertions of parts $(1)-(3)$ hold with 
$(\bK,\tau)$ interchanged with $(\bK^\prime,\overline\tau)$. See Remark \ref{remark-on-PkAB}.
%
%
\end{enumerate}
\end{theorem}

\nit Theorem \ref{solve-L2tau-evp} is proved below in  Section \ref{pf-solve-evp}.
\begin{corollary}[Dirac points]\label{tb-dirac-pts}
Let $\bK_\star$ be any vertex of the hexagonal Brillouin zone, $\B_h$. There exists $\lambda_\star$, depending on $V_0$, such that for all $\lambda>\lambda_\star$,
$-\Delta +\lambda^2 V(\bx) = H^\lambda+E_0^\lambda$ has a multiplicity two $L^2_{\bK_\star}-$ eigenvalue $E^\lambda_D=E_0^\lambda+\Omega^\lambda$, where $E_0^\lambda$ denotes the ground state eigenvalue of $-\Delta+\lambda^2 V_0(\bx)$ acting in $L^2(\R^2)$ and $|\Omega^\lambda|\lesssim \rho_\lambda\ e^{-c\lambda}$. Furthermore, $(\bK_\star,E^
\lambda)$ is a Dirac point in the sense of Definition \ref{dirac-pt-defn} with Fermi velocity, $|v_F^\lambda|$, given by 
\eqref{vF-expand}.
\end{corollary}

 To prove Corollary \ref{tb-dirac-pts} it is necessary to show that $(E_D^\lambda,\bK_\star)$ is a Dirac point in the sense of Definition \ref{dirac-pt-defn}. The properties that need to be checked are a consequence of  Theorem \ref{solve-L2tau-evp} and part (2) of Remark \ref{onedp-implies-all} and the main theorem, Theorem \ref{main-theorem}.
  In particular,  the non-vanishing of the Fermi velocity: 
\begin{equation}
\textrm{$\lamsharplam\ne0$,\ \ for all $\lambda>\lambda_\star$ sufficiently large}
\label{vF-nonzero}
\end{equation}
 is a consequence of the uniform convergence
\[ \left(\ E^\lambda_\pm(\bk)-E^\lambda_D\ \right) / \rho_\lambda \to \pm\mathscr{W}_{_{TB}}(\bk)\ ,\ \ \lambda\to\infty.
\] 
stated in Theorem \ref{conv_to_wallace-47}. See the discussion of Theorem \ref{main-theorem} in the introduction.

\begin{remark}\label{vFne0-notused} We wish to emphasize the dependencies of various assertions. 
The main theorem,  Theorem \ref{main-theorem} requires Theorem \ref{solve-L2tau-evp}. 
The property of Dirac points, that  $\lamsharplam$ is non-zero, \eqref{vF-nonzero}, 
follows from Theorem \ref{main-theorem}; see the discussion around \eqref{vF-expand}. 
Corollary \ref{tb-dirac-pts} follows from Theorem  \ref{solve-L2tau-evp} and Theorem \ref{main-theorem} .
 \end{remark}
\subsection{Proof of Theorem \ref{solve-L2tau-evp}}\label{pf-solve-evp}

We first show that: 
\begin{align}
& \text{Some } \Omega^\lambda, \text{ with } |\Omega^\lambda| \lesssim \rho_\lambda e^{-c\lambda}, 
\text{is an eigenvalue of } \label{assert1} \\
& H^\lambda \equiv -\Delta+\lambda^2V(\bx)-E_0^\lambda
\text{ with corresponding eigenfunction } \Phi_1^\lambda\in L^2_{\bK,\tau} \nn
\end{align}
 implies that part (1) holds. That is, such an eigenvalue, $\Omega^\lambda$,  is necessarily a \underline{simple} $L^2_{\bK,\tau}-$ eigenvalue, and furthermore that 
  parts (2) and (3) of the theorem hold. 
The assertion \eqref{assert1} will be proved below. Verification of part (4) is straightforward. 

Assuming \eqref{assert1},  since $\mathcal{C}\circ\mathcal{I}$ is an isomorphism of $L^2_{\bK,\tau}$ and $L^2_{\bK,\overline\tau}$, and commutes with $-\Delta+\lambda^2V(\bx)$,  it follows that $\Omega^\lambda$ is a  $L^2_{\bK,\overline\tau}-$ eigenvalue of $-\Delta+\lambda^2V(\bx)-E_0^\lambda$
 with corresponding eigenfunction: $\Phi_2^\lambda\equiv(\mathcal{C}\circ\mathcal{I})[\Phi_1^\lambda]$. 
 Note that
  \[
   H^\lambda-\Omega^\lambda\ =\  -\Delta+\lambda^2 V(\bx)-E_0^\lambda-\Omega^\lambda
 = e^{i\bK\cdot\bx} \left(H^\lambda(\bK)-\Omega^\lambda\right) e^{-i\bK\cdot\bx} . \] 
 Therefore, 
  the $L^2_{\bK}-$ kernel of $H^\lambda-\Omega^\lambda$, and hence the $L^2(\R^2/\Lambda_h)-$ kernel of $H^\lambda(\bK)-\Omega^\lambda$, are at least two-dimensional. 
  Furthermore, the resolvent bounds of Section \ref{resolvent-bounds} (Proposition \ref{prop2-resolvent}) imply, for $\lambda$ sufficiently large,  that  $H^\lambda-\Omega^\lambda$
  is invertible on the $L^2_{\bK}-$ orthogonal complement of
  \[ {\rm span}\{P^\lambda_{_{\bK,A}},P^\lambda_{_{\bK,B}}\}\ =\ 
   {\rm span}\{e^{i\bK\cdot\bx}p^\lambda_{_{\bK,A}},e^{i\bK\cdot\bx}p^\lambda_{_{\bK,B}}\}.\]
It follows that the $L^2_{\bK}$ kernel of $H^\lambda-\Omega^\lambda$ is exactly two-dimensional.
 Moreover, since $\Phi_1^\lambda$ and $\Phi_2^\lambda$ lie in orthogonal subspaces of $L^2_{\bK}$,  $\Omega^\lambda$ is a simple eigenvalue in the spaces $L^2_{\bK,\tau}$ and 
 $L^2_{\bK,\overline\tau}$, respectively. Furthermore, since the kernel of $H^\lambda-\Omega^\lambda$ is two-dimensional kernel,  $\Omega^\lambda$ cannot be not a $L^2_{\bK,1}-$ eigenvalue.
This completes the proof that assertion \eqref{assert1} implies parts (2) and (3). 
 
We now turn to the proof of assertion \eqref{assert1}, from which Theorem \ref{solve-L2tau-evp} will then follow. 
Consider $P^\lambda_{\bK,A}(\bx)$ and $P^\lambda_{\bK,B}(\bx)$, defined in \eqref{p-bk-I}. 

\begin{lemma}\label{Ptau-taubar}
$P^\lambda_{\bK,A}(\bx) \equiv e^{i\bK\cdot\bx}\ p_{\bK,A}^\lambda(\bx)\in L^2_{\bK,\tau}$ and 
  $P^\lambda_{\bK,B}(\bx)\equiv e^{i\bK\cdot\bx}\ p_{\bK,B}^\lambda(\bx)\in L^2_{\bK,\overline\tau}$. \\ In particular, $\mathcal{R}[P^\lambda_{\bK,A}]=\tau\ P^\lambda_{\bK,A}$ and 
  $\mathcal{R}[P^\lambda_{\bK,B}]=\overline\tau\ P^\lambda_{\bK,B}$. If $\bK$ is replaced by $\bK^\prime=-\bK$,
  then the same relations hold with $\tau$ and $\bar\tau$ interchanged. 
\end{lemma}

We shall will see below that for large $\lambda$ these are, respectively, approximate 
$L^2_{\bK,\tau}$ and $L^2_{\bK,\overline\tau}$ eigenstates.

\begin{proof}[Proof of Lemma \ref{Ptau-taubar}]   Recall that $\bv_A=0$ and therefore $\Lambda_A=\Lambda_h$.  From \eqref{p-bk-I} we have 
\[ P^\lambda_{\bK,A}(\bx)=\sum_{\bv\in\Lambda_A}e^{i\bK\cdot\bv}p_0^\lambda(\bx-\bv).\]
Therefore, using that $\bx_c=\bv_2-\bv_B$ and $R\bx_c=-\bv_B$, we have 
 we also have 
\begin{align}
&\mathcal{R}\left[P^\lambda_{\bK,A}\right](\bx)=P^\lambda_{\bK,A}\left(\bx_c+R^*(\bx-\bx_c)\right)\nn\\
&=\ \sum_{\bv\in\Lambda_A=\Lambda_h}\ e^{i\bK\cdot\bv}p_0^\lambda\left(\bx_c+R^*(\bx-\bx_c)-\bv\right) \nn\\
&=\ \sum_{\bv\in\Lambda_h}\ e^{i\bK\cdot\bv}p_0^\lambda\left(R\bx_c+\bx-\bx_c-R\bv\right)\quad
 \textrm{($V_0$ is rotationally invariant, (PW3))} \nn\\
&=\ \sum_{\bv\in\Lambda_h}\ e^{i\bK\cdot\bv}p_0^\lambda\left(\bx-(R\bv+\bv_2)\right)\nn\\
&=\ e^{-iR\bK\cdot \bv_2}\times \sum_{\bv\in\Lambda_h}\ e^{iR\bK\cdot (R\bv+\bv_2)}p_0^\lambda\left(\bx-(R\bv+\bv_2)\right)\nn\\
&=\ e^{-i(\bK+\bk_2)\cdot\bv_2}\times \sum_{\bw\in\Lambda_h}\ e^{iR\bK\cdot \bw}p_0^\lambda\left(\bx-\bw\right)
=\ e^{-i\bK\cdot \bv_2}\ P^\lambda_{\bK,A}(\bx)\ =\ \tau\ P^\lambda_{\bK,A}(\bx)\ .\nn
\end{align}
The proof that $\mathcal{R}\left[P^\lambda_{\bK,B}\right]=\bar\tau\ P^\lambda_{\bK,B}$ is similar.
For this, one uses that $\Lambda_B=\bv_B+\Lambda_h$ and that $R\bv_B=\bv_B-\bv_1$. The proof for 
quasi-momentum $\bK^\prime=-\bK$ is similar. This completes the proof of Lemma \ref{Ptau-taubar}.
%

Continuing with the proof of part (1) of Theorem \ref{solve-L2tau-evp},  
 we consider the eigenvalue problem in $L^2_{\bK,\tau}$. Using Lemma \ref{Ptau-taubar}, the $L^2_{\bK^\prime,\bar\tau}$ eigenvalue problem is treated analogously. 
 
We seek a solution of the $L^2_{\bK,\tau}$ eigenvalue problem for the operator $H^\lambda=-\Delta+\lambda^2V-E_0^\lambda$,\ 
 $(H^\lambda-\Omega)\Psi=0$, with non-zero $\Psi\in L^2_{\bK,\tau}$  in the form
\begin{align}
\Psi\ &=\  \alpha_A\ P_{\bK,A}^\lambda(\bx)\ +\ \tilde\Psi(\bx),\label{Psi-decomp-tau}
\qquad 
\left\langle P_{\bK,A}^\lambda,\tilde\Psi \right\rangle=0,\quad 
 \tilde{\Psi}\in H^2_{\bK,\tau},
\end{align}
where $\alpha_A\in\C$ and  $\Omega$ near zero. 
Substitution of \eqref{Psi-decomp-tau} into \eqref{L2sigma-evp} with $\sigma=\tau$ 
implies that  $(\Psi,\Omega)$ solves the $L^2_{\bK,\tau}$ eigenvalue problem for the operator $H^\lambda$ if we can find $\alpha_A\in\C$ and $\tilde\Psi\in H^2_{\bK,\tau}$ with $\left\langle P_{\bK,A}^\lambda,\tilde\Psi \right\rangle=0$, such that:
\begin{equation} 
\left( H^\lambda - \Omega \right)\tilde\Psi\ =\ -\alpha_A\ \left( H^\lambda - \Omega \right)P_{\bK,A}^\lambda .
\label{tPsi-eqn}
\end{equation}

For $\tilde\Psi$, we set $\tilde\Psi=e^{i\bK\cdot\bx}\tilde\psi$ and note that the 
 condition $\mathcal{R}[\tilde\Psi]=\tau\tilde\Psi$ transforms as $\mathcal{R}_\bK \tilde\psi(\bx) = \tau \tilde \psi$, 
 where $\mathcal{R}_\bK=e^{-i\bK\cdot\bx}\circ \mathcal{R}\circ e^{i\bK\cdot\bx}$ is given by
 \begin{equation}
 \mathcal{R}_\bK[g](\bx) = e^{i\bk_2\cdot (\bx-\bx_c)}\ g\left(\bx+R^*(\bx-\bx_c)\right) .
    \nn\end{equation}
We therefore obtain the following problem for $(\tilde\psi,\Omega)$:
\begin{align}
&\left( H^\lambda(\bK)-\Omega \right)\ \tilde\psi\ =\ 
-\alpha_A\ \left( H^\lambda(\bK) - \Omega \right)p_{\bK,A}^\lambda,\label{tpsi-eqn}\\
& \mathcal{R}_\bK \tilde\psi(\bx) = \tau \tilde \psi,\quad  \tilde\psi\in H^2(\R^2/\Lambda_h) ,
\nn
\end{align}
where $H^\lambda(\bK)=e^{-i\bK\cdot\bx}H^\lambda e^{i\bK\cdot\bx}=
 -(\nabla_\bx+i\bK)^2+\lambda^2 V(\bx)-E_0^\lambda$.

Define the orthogonal projection: $\Pi_{A,\tau}:L^2(\R^2/\Lambda_h)\to\mathscr{H}_{_{A,\tau}}$ onto the Hilbert space
%
\begin{equation*}
\mathscr{H}_{A,\tau} = \left\{\tilde\psi\in L^2(\R^2/\Lambda_h) : \left\langle p_{\bK,A}^\lambda,\tilde\psi\right\rangle=0,\ 
 \tilde\psi(\bx_c+R^*(\bx-\bx_c))= \tau e^{-i\bk_2\cdot(\bx-\bx_c)} \tilde\psi(\bx) \right\} .
\end{equation*}
In the natural way, we introduce $\mathscr{H}_{_{A,\tau}}^2$ the subspace of $\mathscr{H}_{A}$ consisting of $H^2(\R^2/\Lambda_h)$ functions, and similarly $\mathscr{H}_{_{B,\bar\tau}}$ and $\mathscr{H}^2_{_{B,\bar\tau}}$, where
%
\begin{equation*}
\mathscr{H}_{B,\bar\tau} = \left\{\tilde\psi\in L^2(\R^2/\Lambda_h) : \left\langle p_{\bK,B}^\lambda,\tilde\psi\right\rangle=0,\ 
 \tilde\psi(\bx_c+R^*(\bx-\bx_c)) = \overline\tau e^{-i\bk_2\cdot(\bx-\bx_c)} \tilde\psi(\bx) \right\} .
\end{equation*}
%

 Applying   $I-\Pi_{_{A,\tau}}$ and $\Pi_{_{A,\tau}}$ to equation \eqref{tpsi-eqn} yields the equivalent system of  two equations
  for $\alpha_A\in\C$, $\tilde\psi \in\mathscr{H}_A^2$ and $\Omega\in\C$:
\begin{align}
& \alpha_A\ \times\ \left\langle p_{_{\bK,A}}^\lambda, \left[H^\lambda(\bK)-\Omega\right] p_{_{\bK,A}}^\lambda  \right\rangle\ +\ \left\langle \left[\ H^\lambda(\bK)-\overline\Omega\ \right] p_{_{\bK,A}}^\lambda,\tilde\psi \right\rangle\ =\ 0,
\label{IminusP_A0}\\
&\Pi_{_{A,\tau}}\left[H^\lambda(\bK)-\Omega\right]\tilde\psi\ =\ -\alpha_A\ \Pi_{_{A,\tau}}\  H^\lambda(\bK)\ p_{_{\bK,A}}^\lambda\ .
\label{Pi_A0}\end{align}

$\mathscr{H}_{A,\tau}$ and $\mathscr{H}_{B,\bar\tau}$ are subspaces of $\mathscr{H}_{_{A}}$ and $\mathscr{H}_{_{B}}$, respectively. Moreover, $\mathcal{R}_\bK$ (whose eigenvalues are $1, \tau$ and $\bar\tau$) commutes with $H(\bK)$ and leaves $\mathscr{H}_{A,\tau}$ and $\mathscr{H}_{B,\bar\tau}$ invariant. Furthermore, the range of $\Pi_{_{A,\tau}}$ is orthogonal to $p_{\bK,A}^\lambda$ by definition, and is also orthogonal to $p_{\bK,B}^\lambda$ since
\[ \left\langle \Pi_{_{A,\tau}}f,p^\lambda_{\bK,B} \right\rangle 
= \left\langle \mathcal{R}\Pi_{_{A,\tau}}f, \mathcal{R}p^\lambda_{\bK,B}\right\rangle
=\left\langle \tau\Pi_{_{A,\tau}}f, \overline{\tau} p^\lambda_{\bK,B}\right\rangle
=(\overline\tau)^2\left\langle \Pi_{_{A,\tau}}f, p^\lambda_{\bK,B}\right\rangle\ .\]
 Therefore, by Proposition  \ref{prop2-resolvent}, the  resolvent
 ${\rm Res}^{\lambda,\bK}(\bK,\Omega) \Pi_{{_{A,\tau}}}$ is well-defined as a mapping from $\mathscr{H}_{{_{A,\tau}}}$ to $\mathscr{H}_{_{A,\tau}}^2$, and the solution of \eqref{Pi_A0} is given by:
\begin{equation}
\tilde\psi\ =\ -\ \alpha_A\times {\rm Res}^{\lambda,\bK}(\bK,\Omega) \Pi_{A,\tau}\  H^\lambda(\bK) p_{_{\bK,A}}^\lambda\ .
\label{tpsi-A}
\end{equation}
Substitution of \eqref{tpsi-A} into \eqref{IminusP_A0} gives the scalar equation 
\begin{equation*}
 \mathcal{M}_{AA}^{\lambda,\bK}(\bK,\Omega)\ \times\ \alpha_A\ =\ 0,\ 
\end{equation*}
where $\mathcal{M}_{AA}^{\lambda,\bK}(\bK,\Omega)$ is defined by (all inner products in $L^2(\R^2/\Lambda_h)$):
\begin{align}
\mathcal{M}_{AA}^{\lambda,\bK}(\bK,\Omega)\ &
\equiv\left\langle p_{_{\bK,A}}^\lambda, \left[H^\lambda(\bK)-\Omega\right] p_{_{\bK,A}}^\lambda \right\rangle \nn\\
 &\quad\ - \ \left\langle H^\lambda(\bK) p_{_{\bK,A}}^\lambda, 
 {\rm Res}^{\lambda,\bK}(\bK,\Omega) \Pi_{_{A,\tau}} H^\lambda(\bK) p_{\bK,A}^\lambda
  \right\rangle\ .
\label{MA-def}\end{align}
%
%
Here, we have used that the range of ${\rm Res}^{\lambda,\bK}(\bK,\Omega) \Pi_{{_{A,\tau}}}=\Pi_{{_{A,\tau}}}{\rm Res}^{\lambda,\bK}(\bK,\Omega) \Pi_{{_{A,\tau}}}$ is orthogonal to ${\rm span}\{p_{\bK,A}^\lambda\}$.
To obtain a non-trivial solution we may set $\alpha_A=1$. Thus, we obtain the following result.
\begin{proposition}\label{Kstar-tau}
For $\lambda>\lambda_\star$, a non-trivial solution of the $L^2_{\bK,\tau}-$ eigenvalue problem, \eqref{L2sigma-evp} with $\sigma=\tau$, exists
if and only if  $\mathcal{M}_{AA}^{\lambda,\bK}(\bK,\Omega)=0$.
\end{proposition}
 
 The equation $\mathcal{M}_{AA}^{\lambda,\bK}(\bK,\Omega)=0$ may be rewritten as
 \begin{align}
&\left\langle p_{_{\bK,A}}^\lambda, p_{_{\bK,A}}^\lambda \right\rangle\times\Omega\ 
\nn\\ 
&= \left\langle p_{_{\bK,A}}^\lambda, H^\lambda(\bK) p_{_{\bK,A}}^\lambda \right\rangle
- \left\langle H^\lambda(\bK) p_{_{\bK,A}}^\lambda, 
 {\rm Res}^{\lambda,\bK}(\bK,\Omega) \Pi_{_{A,\tau}} H^\lambda(\bK) p_{\bK,A}^\lambda  \right\rangle\ .
\label{eqn-for-Omegalam}
 \end{align}
%
We now solve \eqref{eqn-for-Omegalam} for $\lambda\mapsto\Omega^\lambda$, for $\lambda$ sufficiently large.
We shall make use of:
\begin{enumerate}
\item[(a)]  the analyticity and bounds on the mapping  of $\Omega\mapsto {\rm Res}^{\lambda,\bK}(\bK,\Omega)$
 of Proposition \ref{prop2-resolvent} and
 \item[(b)] bounds of Proposition \ref{ME} in Section \ref{sec:M-expanded}, proved in Section \ref{ME-bounds},
 on expressions of the type appearing on the right hand side of \eqref{eqn-for-Omegalam}, for general quasi-momentum, $\bk$, in an appropriate domain in $\C^2$ and small energy, $\Omega$.
 \end{enumerate}

We may rewrite 
  \eqref{eqn-for-Omegalam} as 
  \begin{align}
 \mathfrak{F}(\Omega,\lambda)\ &=\  \Omega,  \label{eqn-for-Omegalam1}
 \end{align}
 where
\begin{align}
\mathfrak{F}(\Omega,\lambda)\ &\equiv\ \Omega^\lambda_0\ +\  \Omega\ \mathscr{R}(\Omega;\lambda)\ ,\nn\\
  \Omega^\lambda_0\ &\equiv\ 
  \left[ \left\langle p_{_{\bK,A}}^\lambda, H^\lambda(\bK) p_{_{\bK,A}}^\lambda \right\rangle
- \left\langle H^\lambda(\bK) p_{_{\bK,A}}^\lambda, 
 {\rm Res}^{\lambda,\bK}(\bK,0)  \Pi_{_{A,\tau}}  H^\lambda(\bK)\ p_{\bK,A}^\lambda  \right\rangle \right]\nn\\
&\qquad\qquad \times \left[ \left\langle p_{_{\bK,A}}^\lambda, p_{_{\bK,A}}^\lambda \right\rangle\right]^{-1}\ ,\quad {\rm and}\nn\\
\mathscr{R}(\Omega;\lambda) &\equiv 
 - \int_0^1 ds \left\langle H^\lambda(\bK) p_{\bK,A}^\lambda , {\rm Res}_1^{\lambda,\bK}(\bK,s\Omega)
  H^\lambda(\bK) p_{\bK,A}^\lambda\right\rangle  \times \left[ \left\langle p_{_{\bK,A}}^\lambda, p_{_{\bK,A}}^\lambda \right\rangle\right]^{-1} ,
  \label{scrR-def}
 \end{align}
 where the operator ${\rm Res}_1^{\lambda,\bK}$ is defined in \eqref{Resj-def}.
We proceed to solve  \eqref{eqn-for-Omegalam1} for $\Omega=\Omega^\lambda$,  for $\lambda$ sufficiently large. 
 Note that by \eqref{pkI-normalize1}
   for $I=J=A$: $\left\langle p_{_{\bK,A}}^\lambda, p_{_{\bK,A}}^\lambda \right\rangle=\Big(1+\mathcal{O}(e^{-c\lambda})\Big)$. Note also that $\mathfrak{F}(\Omega,\lambda)$ is real-valued for $\Omega$ real, by self-adjointness of 
   ${\rm Res}_1^{\lambda,\bK}(\bK,\Omega)$. 

We first claim that $|\Omega_0^\lambda|\le C_0\times \rho_\lambda\times e^{-c\lambda}$,
 for some $C_0>0$. 
   This bound for $\Omega_0^\lambda$
  is  consequence of parts (2) and (3) of Proposition \ref{ME} for the special case $\bk=\bK$.
 %
   %
   
  Let $\mathfrak{g}_0(\lambda)=C_0\times \rho_\lambda\times e^{-c\lambda}$. For all $|\Omega|\le 2\mathfrak{g}_0(\lambda)$, we have
 \begin{align*}
 \left|\ \mathfrak{F}(\Omega,\lambda)\ \right| &\le \mathfrak{g}_0(\lambda) +
  2\mathfrak{g}_0(\lambda)\times e^{-c\lambda} \le 2\mathfrak{g}_0(\lambda) ,
  \end{align*}
  provided $\lambda$ is sufficiently large. Here we use \eqref{deriv-bounds} and the exponential smallness of 
  $\|H^\lambda(\bK)p_{\bK,\bA}^\lambda\|$; see \eqref{Hlambda-pk}.

We claim furthermore, for all $\lambda$ sufficiently large, that the mapping $\Omega\mapsto \mathfrak{F}(\Omega,\lambda)$ is a strict contraction on the set: $|\Omega|\le 2\mathfrak{g}_0(\lambda)$.  Indeed, let $\Omega_1$ and $\Omega_2$ be such that 
$|\Omega_j|\le 2\mathfrak{g}_0(\lambda)$, $j=1,2$, and note that
\begin{align}
&\mathfrak{F}(\Omega_1,\lambda)-\mathfrak{F}(\Omega_2,\lambda)
 = 
(\Omega_1-\Omega_2)\mathscr{R}(\Omega_1,\lambda) + \Omega_2\ \left(\mathscr{R}(\Omega_1,\lambda)-\mathscr{R}(\Omega_2,\lambda)\right) . \label{frakFdiff}
\end{align}
The first term on the right hand side of \eqref{frakFdiff} is $\lesssim  e^{-c\lambda}\ |\Omega_1-\Omega_2|$,
 by \eqref{deriv-bounds} (for $j=1$) and the bound \eqref{Hlambda-pk},
 $\|H^\lambda(\bK)p_{\bk,A}^\lambda\|\lesssim e^{-c\lambda}$.
 To obtain a similar upper bound on the second term on the right hand side of \eqref{frakFdiff}  we proceed as follows. Note that
 \begin{align}
 \mathscr{R}(\Omega_1,\lambda)-\mathscr{R}(\Omega_2,\lambda) 
 = (\Omega_1-\Omega_2)\ \int_0^1 \partial_{_{\Omega }}\mathscr{R}\left(\widehat{\Omega}(s),\lambda\right)\ ds,
\label{Rdiff} \end{align}
where $\widehat{\Omega}(s)=s\Omega_1+(1-s)\Omega_2$.
 From \eqref{scrR-def}, we obtain 
 \begin{align}
 &\partial_{_{\Omega }}\mathscr{R}\left(\widehat{\Omega}(s),\lambda\right)\nn\\
 &= 
 -\ \int_0^1\ s\ ds\ \left\langle H^\lambda(\bK) p_{\bK,A}^\lambda , {\rm Res}_2^{\lambda,\bK}\left(\bK,\widehat{\Omega}(s)\right)
  H^\lambda(\bK) p_{\bK,A}^\lambda\right\rangle  \times \left[ \left\langle p_{_{\bK,A}}^\lambda, p_{_{\bK,A}}^\lambda \right\rangle\right]^{-1}\ ,\label{DscrR-def}
 \end{align}
where ${\rm Res}_2^{\lambda,\bK}$ is defined in \eqref{Resj-def}. Combining \eqref{Rdiff} and \eqref{DscrR-def} with 
\eqref{deriv-bounds} (for $j=2$) and the bound \eqref{Hlambda-pk}, we obtain that the second term 
on the right hand side of \eqref{frakFdiff} is $\lesssim e^{-c\lambda}\ |\Omega_1-\Omega_2|$.
 
 It follows that for $|\Omega_1|$ and $|\Omega_2|\le 2\mathfrak{g}_0(\lambda)$, we have 
 \[
 \left|\ \mathfrak{F}(\Omega_1,\lambda)-\mathfrak{F}(\Omega_2,\lambda)\ \right|\ \lesssim\ e^{-c\lambda}\ |\Omega_1-\Omega_2|.\] 
  Therefore, for $\lambda$ sufficiently large $\Omega\mapsto \mathfrak{F}(\Omega,\lambda)$ 
 is a strict contraction mapping of $[-2\mathfrak{g}_0(\lambda),2\mathfrak{g}_0(\lambda)]$ to itself, and therefore has a unique real fixed point $\Omega^\lambda$. Therefore,  \eqref{eqn-for-Omegalam1}, and thus 
 $\mathcal{M}_{AA}^{\lambda,\bK}(\bK,\Omega^\lambda)=0$,  
  has a unique real solution $\Omega^\lambda$ which satisfies 
 $ |\Omega^\lambda|\le 2C_0\times \rho_\lambda\times e^{-c\lambda}$.

We have therefore found, for $\lambda>\lambda_\star$ sufficiently  large,  $\Omega^\lambda\in\R$ near zero and $\tilde\psi^\lambda\in H^2(\R^2/\Lambda_h)$, given by:
\begin{equation}
\tilde\psi^\lambda\ =\ -\ {\rm Res}^{\lambda,\bK}(\bK,\Omega^\lambda)  \Pi_{_{A,\tau}}\ H^\lambda(\bK)\ p_{_{\bK,A}}^\lambda
\label{psi-tilde1}
\end{equation}
 ($\alpha_A=1$) such that the pair $(\Omega^\lambda,\tilde\psi^\lambda)$ solve \eqref{tpsi-eqn}. Therefore,
$\tilde\Psi^\lambda\ \equiv\ e^{i\bK\cdot\bx} \tilde\psi^\lambda \in H^2_\bK$ and satisfies \eqref{tPsi-eqn} with $\alpha_A=1$.

We claim that $\tilde\Psi^\lambda\in H^2_{\bK,\tau}$. To see this, we rewrite \eqref{psi-tilde1} as:
\begin{align*}
\tilde\Psi^\lambda &= - \left( e^{i\bK\cdot\bx}\ {\rm Res}^{\lambda,\bK}(\bK,\Omega^\lambda)\ e^{-i\bK\cdot\bx} \right)\ 
 \circ\ \left( e^{i\bK\cdot\bx}\  \Pi_{_{A,\tau}}\ e^{-i\bK\cdot\bx} \right)\
 \circ\ H^\lambda\ P_{_{\bK,A}}^\lambda .
\end{align*}
Next, recall that $P_{\bK,A}\in H^2_{\bK,\tau}$ (Lemma \ref{Ptau-taubar}); that is, $P_{\bK,A}\in H^2_{\bK}$ and $\mathcal{R}[P_{\bK,A}]=\tau P_{\bK,A}$. By $120^\circ$ rotational invariance of $H^\lambda$, we have 
$H^\lambda\ P_{_{\bK,A}}^\lambda\in L^2_{\bK,\tau}$,  $ \Pi_{_{A,\tau}}e^{-i\bK\cdot\bx}\ H^\lambda\ P_{_{\bK,A}}^\lambda\in \mathscr{H}_{_{A,\tau}}$,  and $e^{i\bK\cdot\bx}  \Pi_{_{A,\tau}}e^{-i\bK\cdot\bx} H^\lambda\ P_{_{\bK,A}}^\lambda
\in \mathscr{S}$, the subspace of $L^2_{\bK,\tau}$ defined by $\mathscr{S}=\{f\in L^2_{\bK,\tau} : f\perp P_{\bK,A}^\lambda\}$ . 

 Finally, to prove that $\tilde\Psi^\lambda\in H^2_{\bK,\tau}$,  we need to show that  $G\mapsto e^{i\bK\cdot\bx}\ {\rm Res}^{\lambda,\bK}(\bK,\Omega^\lambda)\ e^{-i\bK\cdot\bx}G$ is a mapping from $\mathscr{S}$ to $H^2_{\bK,\tau}$. Let $G\in\mathscr{S}$. Then, 
  $e^{-i\bK\cdot\bx}G\in \mathscr{H}_{_{A,\tau}}$ and, by Proposition \ref{prop2-resolvent},  ${\rm Res}^{\lambda,\bK}(\bK,\Omega^\lambda)\ e^{-i\bK\cdot\bx}G\in\mathscr{H}^2_{_{A,\tau}}$. Finally, $e^{i\bK\cdot\bx}\ {\rm Res}^{\lambda,\bK}(\bK,\Omega^\lambda)\ e^{-i\bK\cdot\bx}G\in H^2_{\bK,\tau}$ .
 %
%
 This completes the proof of \eqref{assert1} and hence the proof Theorem \ref{solve-L2tau-evp}.
  \end{proof}

\section{Low-lying dispersion surfaces via Lyapunov-Schmidt reduction}\label{LS-reduction}

Fix $K_{max}>0$ and let $\tK\in\R^2$ with $|\tK|\le K_{max}$. Let 
\begin{equation*}
U=\left\{(\bk,\Omega)\in\C^2\times\C\ :\ |\bk-\tK|<\hat{c}\lambda^{-1},\ |\Omega|<c'\ \right\},
\end{equation*}
where $\hat{c}$ is less than the constant appearing in Corollary \ref{main-en-cor3}, chosen so that $\hat{c} K_{max}$ is small enough.
Recall, from Section \ref{resolvent-bounds},  the orthogonal projection: $\Pi_{_{AB}}:L^2(\R^2/\Lambda_h)\to\mathscr{H}_{_{AB}}$ onto
\begin{equation*}
\mathscr{H}_{_{AB}}\ =\ \left\{\tilde\psi\in L^2(\R^2/\Lambda_h)\ :\ \left\langle p_{\tK,I}^\lambda,\tilde\psi\right\rangle=0,\ {\rm for}\ I=A,B \right\}
\end{equation*}
and  
 $\mathscr{H}_{_{AB}}^2=\mathscr{H}_{_{AB}}\cap H^2(\R^2/\Lambda_h)$.
%
For $(\bk,\Omega)\in U$, we look for solutions of 
\begin{equation}
\left[ H^\lambda(\bk)\ -\ \Omega \right]\psi\ =\ 0,\ \ \psi\in H^2(\R^2/\Lambda_h)\ .
\label{EVP}
\end{equation}

By part (2) of Lemma \ref{orthog},
 that any $\psi\in H^2(\R^2/\Lambda_h)$ may be written in the form
\begin{equation}
\psi\ =\ \sum_{I=A,B}\ \alpha_I\ p_{\bk,I}^\lambda\ +\ \tilde\psi,\ \ \tilde{\psi}\in \mathscr{H}_{_{AB}}^2,
\label{psi-decomp-A}\end{equation}
 $\alpha_A, \alpha_B\in\C$.
 Note that $\mathscr{H}_{_{AB}}^2$ is defined in terms  of the  modes: $p_{\tK,I}^\lambda, \ I=A, B$, and is independent of $\bk$. 
 
Substitution of \eqref{psi-decomp-A} into \eqref{EVP} yields the equation
\begin{equation}
\sum_{I=A, B}\alpha_I\ \left[\ H^\lambda(\bk)-\Omega\ \right]\ p_{\bk,I}^\lambda\ +\ 
\left[\ H^\lambda(\bk)-\Omega\ \right]\ \tilde{\psi}\ =\ 0 .
\label{substit}\end{equation}
By part (1) of Lemma \ref{orthog}, equation \eqref{substit} is equivalent to the system of equations obtained by applying 
$\Pi_{_{AB}}$ to \eqref{substit} and by taking the inner product of \eqref{substit} with $p^\lambda_{_{\bar\bk,J}}, \ J=A, B$:
\begin{align}
\Pi_{_{AB}}\left[H^\lambda(\bk)-\Omega\right]\tilde\psi\ +\ \sum_{I=A,B}\alpha_I\ \Pi_{_{AB}}\  \left[H^\lambda(\bk)-\Omega\right]\ p_{\bk,I}^\lambda\ &=\ 0\label{PiAB} , \\
\left\langle p_{_{\overline\bk,J}}^\lambda, \sum_{I=A,B} \alpha_I \left[H^\lambda(\bk)-\Omega\right] p_{_{\bk,I}}^\lambda  \right\rangle\ +\ \left\langle \left[H^\lambda(\overline\bk)-\overline\Omega\right] p_{_{\overline\bk,J}}^\lambda,\tilde\psi\right\rangle\ &=\ 0 .\label{project}
\end{align}
We next solve \eqref{PiAB} for $\tilde\psi\in\mathscr{H}_{_{AB}}^2$ in the form
\begin{equation}
\tilde\psi\ =\ -\ \sum_{I=A,B}\alpha_I\ {\rm Res}^{\lambda,\tK}(\bk,\Omega) \Pi_{_{AB}}\  \left[H^\lambda(\bk)-\Omega\right] p_{\bk,I}^\lambda,
\label{tpsi}
\end{equation}
where ${\rm Res}^{\lambda,\tK}(\bk,\Omega)$ is the resolvent defined and bounded in Proposition  \ref{prop2-resolvent}.
Substituting \eqref{tpsi} into \eqref{project} gives the equivalent system $\mathcal{M}^{\lambda,\tK} \alpha=0$:
\begin{equation*}
\sum_{I=A,B}\  \mathcal{M}^{\lambda,\tK}_{JI}(\bk,\Omega)\ \alpha_I\ =\ 0,\ \ J=A,B, 
\end{equation*}
where
\begin{align}
\mathcal{M}^{\lambda,\tK}_{JI}(\bk,\Omega)\ &=\ 
\left\langle p_{_{\overline\bk,J}}^\lambda, \left[H^\lambda(\bk)-\Omega\right] p_{_{\bk,I}}^\lambda \right\rangle
\nn\\
 &\quad\ - \ \left\langle \left[H^\lambda(\overline\bk)-\overline\Omega\right] p_{_{\overline\bk,J}}^\lambda, 
 {\rm Res}^{\lambda,\tK}(\bk,\Omega) \Pi_{_{AB}}\  \left[H^\lambda(\bk)-\Omega\right] p_{\bk,I}^\lambda
  \right\rangle\ .
\label{MJI-def}\end{align}

\begin{remark}\label{calM-analytic}
We note that
\begin{enumerate}
\item[(a)] the mapping $(\bk,\Omega)\mapsto \mathcal{M}^{\lambda,\tK}_{JI}(\bk,\Omega)$ is analytic on the domain:
\begin{equation}
U\ \equiv\ \left\{(\bk,\Omega)\in\C^2\times\C\ :\ |\bk-\tK|<\hat{c}\lambda^{-1},\ |\Omega|<\hat{c}\ \right\}\ .
\label{Udef1}
\end{equation}
\item[(b)] for real $\bk$ and $\Omega$,
the matrix $\mathcal{M}^{\lambda,\tK}(\bk,\Omega)$ is Hermitian:
\begin{equation}
\left[\ \mathcal{M}_{JI}^{\lambda,\tK}(\bk,\Omega)\ \right]^* =\ \mathcal{M}_{IJ}^{\lambda,\tK}(\bk,\Omega),\ 
\ \bk\in\R^2,\ \ \Omega\in\R.
\label{hermitian}\end{equation}
Relation \eqref{hermitian} follows from self-adjointness of $H^\lambda(\bk)-\Omega$ and 
${\rm Res}^{\lambda,\tK}(\bk,\Omega)\Pi_{AB}$ for real $(\bk,\Omega)\in U$; see Proposition \ref{prop2-resolvent}. 
\end{enumerate}
\end{remark}

From the above the discussion we have:

\begin{proposition}\label{evp-det0} 
A given $(\bk,\Omega)\in U$, defined in \eqref{Udef1}, admits a nonzero solution, $\psi\in H^2(\R^2/\Lambda_h)$, 
 of $[H^\lambda(\bk) - \Omega]\psi=0$  if and only if 
 \begin{equation*}
 {\rm det}\left(\ \mathcal{M}^{\lambda,\tK}(\bk,\Omega)\ \right)\ =\ 0,
 \end{equation*}
 where $\mathcal{M}^{\lambda,\tK}(\bk,\Omega)$ is the $2\times2$ matrix with entries displayed in \eqref{MJI-def}.
 \end{proposition}
 
 \section{Expansion of $\mathcal{M}^{\lambda,\tK}(\bk,\Omega)$}\label{sec:M-expanded}
 
  We prove
 \begin{proposition}\label{M-expanded}
 Fix $K_{max}$. There exists $\lambda_\star$ such that for all $\lambda>\lambda_\star$ the following holds:
 Let $\mathcal{M}^{\lambda,\tK}(\bk,\Omega)$ be given by \eqref{MJI-def}, which is defined and analytic on
  the domain $U$; see Remark \ref{calM-analytic}. Introduce the 
 {\it rescaled eigenvalue} parameter, $\mu$, via the relation 
 \[ \Omega=\rho_\lambda\ \mu.\] 
Let
\begin{equation}
\mathcal{U}_\lambda(\tK)\ =\ \left\{(\bk,\mu)\in\C^2\times\C\ :\ |\bk-\tK|<\hat{c}\lambda^{-1},\ |\mu|<\widehat{C} \right\}\ ,
\label{rescaled-U}
\end{equation} 
where $\widehat{C}$ is a constant chosen in Section \ref{complex-analysis} to satisfy \eqref{Chat}.
Then, the mapping $(\bk,\mu)\mapsto  \mathcal{M}^{\lambda,\tK}(\bk,\rho_\lambda \mu)$ is analytic for $(\bk,\mu)\in \mathcal{U}_\lambda$ and satisfies the expansion:
 \begin{equation}\label{MIJ-expanded}
 \mathcal{M}^{\lambda,\tK}(\bk,\rho_\lambda \mu)\ =\
  -\rho_\lambda
  \left[\ 
  \begin{pmatrix} 
  \mu & \gamma(-\bk)\\
 \gamma(\bk) & \mu
 \end{pmatrix}
 \ +\ {\rm Error}^\lambda(\bk,\mu)\ \right]
 \end{equation}
 where   (see also \eqref{gamma-def0})
 \begin{align}
 \rho_\lambda &\equiv  \lambda^2\int |V_0(\by)| p_0^\lambda(\by) p_0^\lambda(\by+\be_{A,1}) d\by\qquad 
 \textrm{(see  \eqref{rho-lam-def})}, \nn \\ 
\gamma(\bk) &\equiv \sum_{\nu=1,2,3}e^{ i\bk\cdot \be_{B,\nu}}= \sum_{1\le\nu\le3}
 e^{-i\bk\cdot\be_{A,\nu}}
=e^{ i\bk\cdot \be_{B,1}} \left( 1+e^{i\bk\cdot\bv_1}+e^{i\bk\cdot\bv_2} \right),
\label{gamma_pm-def1}
\end{align}
 and
 \begin{align}
& \sup_{(\bk,\mu)\in\mathcal{U}_\lambda}|{\rm Error}^\lambda(\bk,\rho_\lambda\mu)|\lesssim e^{-c\lambda}\ .\nn
\end{align} 
 \end{proposition}
 
 \subsection{Proof of Proposition \ref{M-expanded}}
We expand the matrix entries $\mathcal{M}^{\lambda,\tK}_{JI}(\bk,\Omega)$ for $\lambda>\lambda_\star$ and 
 $(\bk,\Omega)\in U$.

\begin{proposition}\label{ME}
  Under the hypotheses of Proposition \ref{M-expanded}, we have
  \begin{enumerate}
\item
  \begin{align*}
  \left\langle\ p_{_{\overline\bk,B}}^\lambda\ ,\ H^\lambda(\bk)p_{_{\bk,A}}^\lambda\ \right\rangle
 &=  -\ \rho_\lambda\times\ \gamma(\bk)\ +\ \mathfrak{I}_{BA}(\lambda),\\
  \left\langle\ p_{_{\overline\bk,A}}^\lambda\ ,\ H^\lambda(\bk)p_{_{\bk,B}}^\lambda\ \right\rangle
 &=  -\ \rho_\lambda\times\ \gamma(-\bk)\ +\ \mathfrak{I}_{AB}(\lambda), \nn
 \end{align*}
where
$ |\mathfrak{I}_{BA}|,\ \  |\mathfrak{I}_{AB}|\ \lesssim\ \rho_\lambda\times e^{-c\lambda}$, and 
  \item
\begin{equation*}
\Big|\ \left\langle p_{_{\overline\bk,A}}^\lambda,H^\lambda(\bk)p_{_{\bk,A}}^\lambda\right\rangle\ \Big|
\ \lesssim\ \rho_\lambda\times e^{-c\lambda}\ ,
\end{equation*}
and similarly with $B$ in place of $A$.
\item 
\begin{align}
 \left|  \left\langle \left[H^\lambda(\overline\bk)-\overline\Omega\right] p_{_{\overline\bk,J}}^\lambda, 
 {\rm Res}^{\lambda,\tK}(\bk,\Omega) \Pi_{AB}  \left[H^\lambda(\bk)-\Omega\right] p_{_{\bk,I}}^\lambda
  \right\rangle \right| \lesssim \rho_\lambda\times e^{-c\lambda} ,
 \label{ip-higherorder-bound}
\end{align}
where $J$ and $I$ vary over the set $\{A,B\}$.
\end{enumerate}
\end{proposition}
We first prove Proposition \ref{ME} (modulo several assertions to be proven in Section \ref{ME-bounds}) and then return to the proof of Proposition \ref{M-expanded}. 

 We begin with some brief review and preliminary observations.
  Recall
  \begin{enumerate}
  \item[(i)] $p_0^\lambda(\bx)$ denotes the ground state of $-\Delta+\lambda^2 V_0(\bx)$, 
  with corresponding eigenvalue $E^\lambda_0$; $[-\Delta+\lambda^2 V_0(\bx)-E^\lambda_0]p_0^\lambda(\bx)=0$. Thus, $p_\bk^\lambda(\bx) \equiv e^{-i\bk\cdot\bx}\ p_0^\lambda(\bx)$ satisfies
   $[-(\nabla_\bx+i\bk)^2+\lambda^2 V_0(\bx)-E^\lambda_0]p_\bk^\lambda(\bx)=0$.
\item[(ii)] $p_{\bk,\hat\bv}^\lambda(\bx)\ \equiv\ p_\bk^\lambda(\bx-\hat\bv)$ and 
 $V^\lambda(\bx) \equiv \lambda^2 \sum_{\bv\in\Honeycomb} V_0(\bx-\bv) - E_0^\lambda$. 
 \item[(iii)] 
For $I = A, B$, we have the $\Lambda_h-$ periodic approximate Floquet-Bloch amplitudes:
\begin{equation}
 p_{\bk,I}^\lambda(\bx)\ \equiv\ \sum_{\hat\bv\in\Lambda_I}\ p_{\bk,\hat\bv}^\lambda(\bx)
 \ =\ \sum_{\hat\bv\in\Lambda_I}\ e^{-i\bk\cdot(\bx-\hat\bv)}p_0^\lambda(\bx-\hat\bv);
\label{p-bk-I1}
\end{equation}
For $|\Im\bk|<C_1$ and $\lambda$ sufficiently large, the series \eqref{p-bk-I1} is uniformly convergent.  
 $p_{\bk,I}^\lambda(\bx)$ is $\Lambda_h-$ periodic on $\R^2$ and is regarded as a function on $\R^2/\Lambda_h$.
\end{enumerate}
 Summing  the expression for $H^\lambda(\bk) p_{\bk,\hat\bv}^\lambda(\bx)$, displayed in 
  \eqref{Hpkhatv0},  over $\hat{\bv}\in\Lambda_I$, we obtain
\begin{align*}
&\left[\ -(\nabla_\bx+i\bk)^2 + V^\lambda(\bx)\ \right] p_{\bk,I}^\lambda(\bx)\ =\ 
 \sum_{\hat{\bv}\in\Lambda_I}\ \sum_{\bv\in\Honeycomb\setminus\{\hat\bv\}}\ \lambda^2 V_0(\bx-\bv) p_{\bk}^\lambda(\bx-\hat\bv)\ .
 \end{align*}
For $\bx\in D$, the fundamental domain, we have $V_0(\bx-\bv)=0$ for all $\bv\in\Lambda$ except $\bv= \bv_A, \bv_B$. This follows since the support of $V_0$ is contained in $B({\bf 0},r_0)$ and $r_0<|\be_{A,1}|/2$; see Figure \ref{fig:fundamental-cell}.
 Therefore, the inner sum over $\bv\in\Honeycomb\setminus\{\hat\bv\}$ can only have contributions from $\bv= \bv_A, \bv_B$; hence
\begin{align*}
H^\lambda(\bk)p_{\bk,I}^\lambda(\bx) &=
\lambda^2 V_0(\bx-\bv_A)\ \sum_{\bv\in\Lambda_I\setminus\{\bv_A\}}\  p_{\bk}^\lambda(\bx-\bv)\nn\\
&\ \ +\ \lambda^2 V_0(\bx-\bv_B)\ \sum_{\bv\in\Lambda_I\setminus\{\bv_B\}}\  p_{\bk}^\lambda(\bx-\bv),\ \ \bx\in D\ .
 \end{align*}
 Therefore, for all $\bx\in D$, we have
 \begin{align}
 \label{Hpkhatv3}
& H^\lambda(\bk)p_{\bk,A}^\lambda(\bx) =
 \begin{cases}
  0 , & \textrm{if}\ \bx\in D\setminus\left(B(\bv_A,r_0)\cup B(\bv_B,r_0)\right)\\
 \lambda^2 V_0(\bx-\bv_A) \sum\limits_{\bv\in\Lambda_A\setminus\{\bv_A\}}  p_{\bk}^\lambda(\bx-\bv) ,
  & \textrm{if}\  \bx\in B(\bv_A,r_0)\\
  \lambda^2 V_0(\bx-\bv_B) \sum\limits_{\bv\in\Lambda_A}  p_{\bk}^\lambda(\bx-\bv) , 
  & \textrm{if}\  \bx\in B(\bv_B,r_0) 
  \end{cases}
  \end{align}
  and similarly for $H^\lambda(\bk)\ p_{\bk,B}^\lambda(\bx)$.

\subsection{Matrix element $\left\langle\ p_{_{\overline\bk,B}}^\lambda\ ,\ H^\lambda(\bk)p_{_{\bk,A}}^\lambda\ \right\rangle$}

Using \eqref{p-bk-I1} and \eqref{Hpkhatv3}, we have
\begin{align}
&\left\langle p_{_{\overline\bk,B}}^\lambda,H^\lambda(\bk)p_{_{\bk,A}}^\lambda\right\rangle\nn\\
&\quad = \int_{B(\bv_A,r_0)\cup B(\bv_B,r_0)} \overline{p_{_{\overline\bk,B}}^\lambda(\bx)} H^\lambda(\bk) p_{_{\bk,A}}^\lambda(\bx) d\bx\nn\\
&\quad = \sum_{\bw\in\Lambda_B}\sum_{\bv\in\Lambda_A\setminus\{\bv_A\}}
\lambda^2  e^{i\bk\cdot(\bv-\bw)} \int_{B(\bv_A,r_0)} V_0(\bx-\bv_A) p_0^\lambda(\bx-\bw) p_0^\lambda(\bx-\bv) d\bx\nn\\
&\quad\quad + \sum_{\bw\in\Lambda_B}\sum_{\bv\in\Lambda_A}
\quad  \lambda^2 e^{i\bk\cdot(\bv-\bw)} \int_{B(\bv_B,r_0)} V_0(\bx-\bv_B) p_0^\lambda(\bx-\bw) p_0^\lambda(\bx-\bv) d\bx\nn\\
&\quad \equiv I + II.
\label{BA-terms}
\end{align}
Consider the first double-sum in \eqref{BA-terms} over $\bw\in\Lambda_B$ and $\bv\in\Lambda_A\setminus\{\bv_A\}$, which is integrated over $B(\bv_A,r_0)$.  To study this integral it is convenient to express the integrand in coordinates centered at $\bv_A$. Note that 
 since $V_0(\bx-\bv_A)$ is supported in $B(\bv_A,r_0)$ and $p_0^\lambda$ is exponentially decaying,
 we expect that the dominant contribution to the summation over $\bw\in\Lambda_B$ comes from the three vertices of $\Lambda_B$ which are closest to $\bv_A$. 
  These are: 
  \begin{equation*}
  \bw\ =\ \bv_A+\be_{A,\nu},\ \nu=1,2,3.
  \end{equation*}
   The non-nearest neighbors to $\bv_A$ in $\Lambda_B$ are  
   \begin{equation*}
    \bw=\bv_A+\be_{A,1}+\bn\vec{\bv},\quad \bn\ne (0,0), (0,-1), (-1,0).
    \end{equation*}
   We therefore write:
  \begin{equation*}
  \sum_{\bw\in\Lambda_B}\sum_{\bv\in\Lambda_A\setminus\{\bv_A\}}
   =\left(\ \sum_{\substack{\bw=\bv_A+\be_{A,\nu}\\ \nu=1,2,3} }
   + \sum_{\substack{\bw=\bv_A+\be_{A,1}+\vec{\bn}\vec{\bv}\\ \bn\ne(0,0),(0,-1),(-1,0)} }\right)
   \ \sum_{\substack{\bv=\bv_A+\bfm\vec{\bv}\\ \bfm\ne0}} .
   \end{equation*}
  Therefore, the first double-sum in \eqref{BA-terms} may be expressed as follows:
  \begin{align}
 I\ =\  &\sum_{1\le\nu\le3}\ \sum_{\vec{\bfm}\ne\vec0}
  \lambda^2  e^{i\bk\cdot(\bfm\vec\bv-\be_{A,\nu})}\nn\\
  &\qquad \times\ \int_{B(\bv_A,r_0)}\ V_0(\bx-\bv_A)\ p_0^\lambda(\bx-\bv_A-\be_{A,\nu})\ 
  p_0^\lambda(\bx-\bv_A-\bfm\vec\bv)\ d\bx\nn\\
  &\nn\\
  \ + &\sum_{\bn\ne(0,0),(0,-1),(-1,0)}\ \sum_{\bfm\ne\vec0}
   \lambda^2  e^{i\bk\cdot([\bfm-\bn]\vec\bv-\be_{A,1})}\nn\\
  &\qquad \times  
 \ \int_{B(\bv_A,r_0)}\ V_0(\bx-\bv_A)\ p_0^\lambda(\bx-\bv_A-\be_{A,1}-\bn\vec\bv)\ p_0^\lambda(\bx-\bv_A-\bfm\vec\bv)\ d\bx\ .\nn
  \end{align}
Hence, because $|\Im\bk|\lesssim \lambda^{-1}$, the first double-sum in \eqref{BA-terms} is bounded by
 \begin{align}
& \mathfrak{I}^1_{BA}(\lambda)\ \equiv\  \sum_{\substack{1\le\nu\le3\\ \vec{\bfm}\ne\vec0}}\  \lambda^2\
 e^{ C\lambda^{-1} |\bfm|}\  \nn \\
 &  \qquad\qquad\qquad\ \times 
 \int\ |V_0(\bx-\bv_A)|\ p_0^\lambda(\bx-\bv_A-\be_{A,\nu})\ 
 p_0^\lambda(\bx-\bv_A-\bfm\vec\bv)\ d\bx\nn\\
  &\nn\\
  & \qquad\qquad\ +   C\ \sum_{\substack{\bn\ne\vec0,(0,-1),(-1,0)\\ \bfm\ne\vec0}} \lambda^2\ 
   e^{ C\lambda^{-1} |\bfm-\bn|}\nn\\
 &  \qquad\qquad\qquad\ \times  \int\ |V_0(\bx-\bv_A)|\ p_0^\lambda(\bx-\bv_A-\be_{A,1}-\bn\vec\bv)\ 
  p_0^\lambda(\bx-\bv_A-\bfm\vec\bv)\ d\bx\ .
\label{1st-sum-bound}
  \end{align}

We now turn to the second double-sum in \eqref{BA-terms} over $\bw\in\Lambda_B$ and $\bv\in\Lambda_A$, which is integrated over $B(\bv_B,r_0)$. Here, we express the integrand in coordinates relative to a center at $\bv_B$. The dominant contributions come from summands with $\bw=\bv_B$ and $\bv\in \Lambda_A$ varying over the three nearest neighbors to $\bv_B$. Those neighbors are given by:
\begin{equation*}
\bv = \bv_B\ +\ \be_{B,\nu},\ \ \nu=1,2,3\ .
\end{equation*}
(For real $\bk$, these contributions to the sum are equal in magnitude, by symmetry.) 
The points of $\Lambda_A$ which are not among the nearest neighbors to $\bv_B$ are
\begin{equation*}
\bv = \bv_B+\be_{B,1}+\bn\vec\bv,\ \ \bn\ne (0,0), (1,0), (0,1)\ .
\end{equation*}

We therefore write:
  \begin{equation*}
  \sum_{\bw\in\Lambda_B}\sum_{\bv\in\Lambda_A}
   = \left(\ \sum_{\bw=\bv_B}+\sum_{\substack{\bw=\bv_B+\bfm\vec\bv\\ \bfm\ne\vec0}}\ \right)\ \left(\ \sum_{\substack{\bv = \bv_B\ +\ \be_{B,\nu}\\ 1\le\nu\le3}}\ +\
   \sum_{\substack{\bv = \bv_B+\be_{B,1}+\bn\vec\bv\\ \bn\ne (0,0), (1,0), (0,1)}}    \ \right)\ .
   \end{equation*}
   
   Therefore, the second double-sum in \eqref{BA-terms} may be expressed as $II\equiv$
   {\footnotesize{
\begin{align} 
& \sum_{1\le\nu\le3}
 \lambda^2 e^{i\bk\cdot\be_{B,\nu}} \int V_0(\bx-\bv_B) p_0^\lambda(\bx-\bv_B) p_0^\lambda(\bx-\bv_B-\be_{B,\nu}) d\bx\nn\\
&+ \sum_{\bn\ne (0,0), (1,0), (0,1)}
  \lambda^2 e^{i\bk\cdot(\be_{B,1}+\bn\vec\bv)} \int V_0(\bx-\bv_B) p_0^\lambda(\bx-\bv_B) p_0^\lambda(\bx-\bv_B-\be_{B,1}-\bn\vec\bv) d\bx\nn\\
&+ \sum_{\substack{1\le\nu\le3\\ \bfm\ne(0,0)}}
  \lambda^2 e^{i\bk\cdot(\be_{B,\nu}-\bfm\vec\bv)} \int V_0(\bx-\bv_B) p_0^\lambda(\bx-\bv_B-\bfm\bv) p_0^\lambda(\bx-\bv_B-\be_{B,\nu}) d\bx   \label{2nd-ds-breakdown} \\
&+ \sum_{\substack{\bfm\ne(0,0)\\ \bn\ne (0,0), (1,0), (0,1)}}\
  \lambda^2 e^{i\bk\cdot[\be_{B,1}+(\bn-\bfm)\vec\bv]} \nn \\
  &\times \int V_0(\bx-\bv_B) p_0^\lambda(\bx-\bv_B-\bfm\vec\bv) p_0^\lambda(\bx-\bv_B-\be_{B,1}-\bn\vec\bv) d\bx . \nn
  \end{align}
   }}
 The first term in \eqref{2nd-ds-breakdown} may be simplified by symmetry. Indeed,  
 \[\int V_0(\bx-\bv_B)\ p_0^\lambda(\bx-\bv_B)\ p_0^\lambda(\bx-\bv_B-\be_{B,\nu})\ d\bx,\ \ \nu=1,2,3\]
 is independent of $\nu$. Therefore, the first term in \eqref{2nd-ds-breakdown} becomes
 \begin{align*}  
 &\left(\ \sum_{1\le\nu\le3}
 e^{i\bk\cdot\be_{B,\nu}}\ \right)\ \lambda^2\ \int V_0(\bx-\bv_B)\ p_0^\lambda(\bx-\bv_B)\ p_0^\lambda(\bx-\bv_B-\be_{B,1})\ d\bx\nn\\
 &\quad =\ \left(\ \sum_{1\le\nu\le3}
 e^{i\bk\cdot\be_{B,\nu}}\ \right)\ \lambda^2\ \int V_0(\by)\ p_0^\lambda(\by)\ p_0^\lambda(\by+\be_{A,1})\ d\by
 \equiv \gamma(\bk) \times (\ - \rho_\lambda\ ),
 \end{align*}
 where $\gamma(\bk)$ is defined in \eqref{gamma_pm-def1} and 
 \begin{align*}
 \rho_\lambda &\equiv \lambda^2\int |V_0(\by)| p_0^\lambda(\by) p_0^\lambda(\by-\be_{B,1}) d\by
=  \lambda^2\int |V_0(\by)| p_0^\lambda(\by) p_0^\lambda(\by+\be_{A,1}) d\by . 
 \end{align*}
Thus,  the second double-sum in \eqref{BA-terms} may be written as
 \begin{align*}
& -\ \gamma(\bk) \times \rho_\lambda
\end{align*}
plus an expression which is bounded by $C\cdot \mathfrak{I}^{(2)}_{BA}(\lambda)$, where
{\footnotesize{
\begin{align} 
&\mathfrak{I}^{(2)}_{BA}(\lambda)\nn\\ 
&\equiv \sum_{\bn\ne (0,0), (1,0), (0,1)}
  \lambda^2  e^{C\lambda^{-1} |\bn|} \int |V_0(\bx-\bv_B)| p_0^\lambda(\bx-\bv_B) p_0^\lambda(\bx-\bv_B-\be_{B,1}-\bn\vec\bv) d\bx\nn\\
&+ \sum_{\substack{1\le\nu\le3\ \bfm\in\Z^2\setminus\{\vec0\}}}
  \lambda^2 e^{C\lambda^{-1} |\bfm|}   \int |V_0(\bx-\bv_B)| p_0^\lambda(\bx-\bv_B-\bfm\vec\bv) p_0^\lambda(\bx-\bv_B-\be_{B,\nu}) d\bx   
  \label{2nd-ds-breakdown-est} \\
&+ \sum_{\substack{\bfm\in\Z^2\setminus\{\vec0\}\ \bn\ne (0,0), (1,0), (0,1)}}\
  \lambda^2 e^{C\lambda^{-1} |\bn-\bfm|}  \nn \\
  &\qquad\times \int |V_0(\bx-\bv_B)| p_0^\lambda(\bx-\bv_B-\bfm\vec\bv) p_0^\lambda(\bx-\bv_B-\be_{B,1}-\bn\vec\bv) d\bx .  \nn
  \end{align}
  }}

We next use the above expansions to obtain
  \begin{proposition}\label{ME-BA1}
  Let $\gamma(\bk)$ be as defined in \eqref{gamma_pm-def1}. Under the conditions of Proposition 
   \ref{M-expanded}, we have 
  \begin{equation*}
  \left\langle\ p_{_{\overline\bk,B}}^\lambda\ ,\ H^\lambda(\bk)p_{_{\bk,A}}^\lambda\ \right\rangle
 =  -\ \rho_\lambda\times \gamma(\bk)\ 
 +\ \mathfrak{I}^{(1)}_{BA}(\lambda)\ +\ \mathfrak{I}^{(2)}_{BA}(\lambda),
 \end{equation*}
 where
 $\mathfrak{I}^{(j)}_{BA}(\lambda),\ j=1,2$, are bounded by the expressions displayed in \eqref{1st-sum-bound} and \eqref{2nd-ds-breakdown-est} and satisfy the bounds
 \begin{equation}
\mathfrak{I}^{(1)}_{BA}(\lambda) +  \mathfrak{I}^{(2)}_{BA}(\lambda)\ \lesssim\ \rho_\lambda\times e^{-c\lambda}\ .
 \label{ip-BA-bound}
 \end{equation}
 Similarly, $\left\langle\ p_{_{\overline\bk,A}}^\lambda\ ,\ H^\lambda(\bk)p_{_{\bk,B}}^\lambda\ \right\rangle
 =  -\ \rho_\lambda\times \gamma(-\bk)\ +\ \mathfrak{J}_{AB}(\lambda)$, where $\mathfrak{J}_{AB}$ satisfies an estimate analogous to \eqref{ip-BA-bound}.
  \end{proposition}
 %
%
To complete the proof of Proposition \ref{ME-BA1}, we need to establish the estimate \eqref{ip-BA-bound}. This is deferred  to Section \ref{ME-bounds}.
  
\subsection{Matrix element $\left\langle\ p_{_{\overline\bk,A}}^\lambda\ ,\ H^\lambda(\bk)\ p_{_{\bk,A}}^\lambda\ \right\rangle$}

Thanks to  \eqref{Hpkhatv3} we have
\begin{align}
&\left\langle p_{_{\overline\bk,A}}^\lambda,H^\lambda(\bk)p_{_{\bk,A}}^\lambda\right\rangle\nn\\
&\quad = \int_{B(\bv_A,r_0)\cup B(\bv_B,r_0)} \overline{p_{_{\overline\bk,A}}^\lambda(\bx)} H^\lambda(\bk) p_{_{\bk,A}}^\lambda(\bx) d\bx\nn\\
&\quad = \sum_{\bw\in\Lambda_A}\sum_{\bv\in\Lambda_A\setminus\{\bv_A\}}
\lambda^2  e^{i\bk\cdot(\bv-\bw)} \int_{B(\bv_A,r_0)} V_0(\bx-\bv_A) p_0^\lambda(\bx-\bw) p_0^\lambda(\bx-\bv) d\bx\nn\\
&\quad\quad + \sum_{\bw\in\Lambda_A}\sum_{\bv\in\Lambda_A}
\quad  \lambda^2 e^{i\bk\cdot(\bv-\bw)} \int_{B(\bv_B,r_0)} V_0(\bx-\bv_B) p_0^\lambda(\bx-\bw) p_0^\lambda(\bx-\bv) d\bx . \label{AA-terms}
\end{align}

We now bound both double-sums in \eqref{AA-terms}. For the first double-sum,  
express $\bw\in\Lambda_A$ as $\bw=\bv_A+\bfm\vec\bv$ and $\bv\in\Lambda_A\setminus\{\bv_A\}$ as $\bv=\bv_A+\bn\vec\bv,\ \bn\ne\vec0$.  The first double-sum is bounded by
\begin{equation}
\mathfrak{I}_{AA}^{(1)}\equiv \lambda^2\sum_{\bfm}\sum_{\bn\ne{\bf 0}}\ e^{C|\bfm-\bn|\lambda^{-1}} \int |V_0(\bx-\bv_A)|\ p_0^\lambda(\bx-\bv_A-\bfm\vec\bv)\ p_0^\lambda(\bx-\bv_A-\bn\vec\bv)\ d\bx\ .
\label{IAA-1}\end{equation}
For the second double-sum, express $\bw\in\Lambda_A$ as $\bw=\bv_B+\be_{B,1}+\bfm\vec\bv$, and 
$\bv\in\Lambda_A$ as $\bv=\bv_B+ \be_{B,1}+\bn\vec\bv$, $\bn\in\Z^2$. 
The second double-sum is bounded $C\cdot \mathfrak{I}_{AA}^{(2)}$, where 
\begin{align}
&\mathfrak{I}_{AA}^{(2)}\ \equiv\lambda^2\sum_{\bfm}\sum_{\bn} \ e^{C|\bn-\bfm|\lambda^{-1}}\nn\\
&\qquad\qquad\times \int |V_0(\bx-\bv_B)|\ p_0^\lambda(\bx-\bv_B-\be_{B,1}-\bfm\vec\bv)\ p_0^\lambda(\bx-\bv_B-\be_{B,1}-\bn\vec\bv)\ d\bx\ .\label{IAA-2}
\end{align}

\begin{proposition}\label{ME-AA1} Under the conditions of Proposition 
   \ref{M-expanded}, we have 
\begin{equation}
\Big|\ \left\langle p_{_{\overline\bk,A}}^\lambda,H^\lambda(\bk)p_{_{\bk,A}}^\lambda\right\rangle\ \Big|\ \le\
\mathfrak{I}_{AA}^{(1)}(\lambda)\ +\ \mathfrak{I}_{AA}^{(2)}(\lambda) \ \lesssim\ \rho_\lambda\times e^{-c\lambda},
\label{ip-AA-bound}\end{equation}
where the expressions for $\mathfrak{I}_{AA}^{(j)},\ j=1,2$  are displayed in 
\eqref{IAA-1} and \eqref{IAA-2}.
\end{proposition}

\nit 
We defer the proof of Proposition \ref{ME-AA1}, along with that of Proposition \ref{ME-BA1}, to Section  
\ref{ME-bounds}.

\subsection{Bounds on the higher order matrix elements}\label{sec:higher-order}
  
  \begin{proposition}\label{ME-IJ-higher-order}
Assume $\lambda>\lambda_\star$ and $|\Omega|\le \widehat{C}\rho_\lambda$, for some constant $\widehat{C}$. Then, 
 \begin{equation}
   \left|\  \left\langle \left[H^\lambda(\overline\bk)-\overline\Omega\right] p_{_{\overline\bk,J}}^\lambda, 
 {\rm Res}^{\lambda,\tK}(\bk,\Omega) \Pi_{AB}\  \left[H^\lambda(\bk)-\Omega\right] p_{\bk,I}^\lambda
  \right\rangle\ \right|\ \lesssim\ \rho_\lambda\ e^{-c\lambda} .
  \label{higher-order-ME-est}
  \end{equation}
  \end{proposition}
  
\nit The proof of Proposition \ref{ME-IJ-higher-order} is given in Section \ref{ME-bounds}.  

Once Propositions \ref{ME-BA1}, \ref{ME-AA1} and \ref{ME-IJ-higher-order} are established, Proposition \ref{ME} follows.
Furthermore, Proposition \ref{M-expanded} follows at once from Proposition \ref{ME} and the near-orthonormality
 \eqref{pkI-normalize1} of the $p_{\bk,I}^\lambda$.

 \section{Dispersion surfaces}\label{dispersion-surfaces}
 
By Proposition \ref{M-expanded},
 \begin{equation*}
 \mathcal{M}^{\lambda,\tK}(\bk,\rho_\lambda \mu)\ =\
  -\rho_\lambda
  \left[\ 
  \begin{pmatrix} 
  \mu & \gamma(-\bk)\\
 \gamma(\bk) & \mu
 \end{pmatrix}
 \ +\ {\rm Error}^\lambda(\bk,\mu)\ \right],
 \end{equation*}
 $ \gamma(\bk) \equiv \sum_{\nu=1,2,3}e^{i\bk\cdot \be_{B,\nu}} $,
  and $|{\rm Error}^\lambda(\bk,\mu)|\lesssim e^{-c\lambda}$ on the domain $\mathcal{U}_\lambda(\tK)$, given in
   \eqref{rescaled-U}.
%
%
%
%
%
Therefore, 
 \begin{equation*}
 \rho_\lambda^{-2}\ \det(\mathcal{M}^{\lambda,\tK}\left(\bk,\rho_\lambda \mu)\right)
 = \mu^2-\gamma(\bk)\gamma(-\bk)\mu + f(\bk,\mu),
 \end{equation*}
 where $f(\bk,\mu)$ is analytic on the domain $\mathcal{U}_\lambda$, defined in \eqref{rescaled-U},  and $|f(\bk,\mu)|\lesssim e^{-c\lambda}$. {\it Although  $f(\bk,\mu)$ depends on $\lambda$,  we suppress this dependence. }\\
Recall also that  $\mathcal{M}^{\lambda,\tK}(\bk,\Omega)$ depends analytically on $(\bk,\Omega)$ and hence also on $(\bk,\mu)$. Thus, we obtain the following result.

 \begin{proposition}\label{detM-scaled-poly}
 Fix $\tK\in\R^2$ with $|\tK|\le K_{max}$.  Suppose $\lambda\ge\lambda_\star$, where $\lambda_\star$ is a large enough  constant depending only on $V$ and $K_{max}$. Let 
 \begin{equation*}
U(\tK) \equiv \left\{ (\bk,\mu)\in\C^2\times\C : |\bk-\tK|<\hat{c}\lambda^{-1},\ \ |\mu|<\widehat{C}\ \right\},
 \end{equation*}
 where $\hat{c}$ is a  small enough constant, dependent on $K_{max}$, and $\widehat{C}$ is a large enough constant, chosen below to 
  satisfy \eqref{Chat}.
  
 Then, there exists an analytic function $f(\bk,\mu)$ defined on $U(\tK)$, with the following properties:
 \begin{itemize}
 \item[(i)] For real $(\bk,\mu)\in U(\tK)$, the quantity $\Omega=\rho_\lambda\times\mu$ is an eigenvalue of $H^\lambda(\bk)$, {\it i.e.} there exists $(\Omega,\psi)$ with $\psi\ne0$ such that 
 \[ \left(\ -(\nabla_\bx+i\bk)^2\ +\ V^\lambda(\bx)\ -\ \Omega\ \right)\psi\ =\ 0,\ \ \psi\in H^2(\R^2/\Lambda_h),\]
  if and only if $\mu$ is a root of the equation:
  \begin{equation*}
  \rho_\lambda^{-2}\ \det(\mathcal{M}^{\lambda,\tK}\left(\bk,\rho_\lambda \mu)\right)\ \equiv\ \mu^2\ -\ \gamma(\bk)\cdot\gamma(-\bk)+f(\bk,\mu)=0\ .
  \end{equation*}
  \item[(ii)] $|f(\bk,\mu)| \lesssim e^{-c\lambda}$ for all $(\bk,\mu)\in U(\tK)$.
  \end{itemize}
 \end{proposition}
 
 \begin{definition}\label{rescaled-eig}
 If $\rho_\lambda\times \mu$ is an eigenvalue of $H^\lambda(\bk)$, then we say that $(\bk,\mu)$ belongs to the 
 {\it rescaled dispersion surface} and we say that $\mu$ is a rescaled eigenvalue of $H^\lambda(\bk)$.
\end{definition}
 
 \subsection{Rescaled dispersion surfaces}
 
 We show that the locus of quasi-momentum / energy pairs which comprise a rescaling of the two lowest dispersion surfaces
  of $H=-\Delta+\lambda^2V(\bx)$
is uniformly approximated, for  large $\lambda$ and on any prescribed compact quasi-momentum set, by the zero-set of an analytic perturbation of Wallace{'}s dispersion relation 
   of the 2-band tight-binding model.
 
\subsection{Complex analysis}\label{complex-analysis}
 We pick $\widehat{C}$ large enough, depending on $K_{max}$, to guarantee that for $(\bk,\mu)\in U(\tK)$ 
\begin{equation}
|\mu|\ge\frac{\widehat{C}}2\ \ \implies\ \ \left|\ \mu^2\ -\ \gamma(\bk)\ \gamma(-\bk)\ \right|\ >\ 2 .
\label{Chat}
\end{equation}
 Denote by $F(\mu,\bk)$ the function
 \begin{equation*}
 F(\mu,\bk)\ =\ \mu^2\ -\ \gamma(\bk)\ \gamma(-\bk)\ + f(\mu,\bk)\ .
 \end{equation*}
We suppress the dependence of  $F(\mu,\bk)$ on $\lambda$, inherited from $f(\mu,\bk)$. 

The function $\mu\mapsto \mu^2\ -\ \gamma(\bk)\ \gamma(-\bk)$ has two zeros (multiplicity counted), for fixed $\bk\in\C^2$ such that $|\bk-\tK|<\hat{c}\lambda^{-1}$. These zeros lie in disc $\{\mu\in\C: |\mu|<\widehat{C}/2\}$, and moreover 
$ \left|\mu^2\ -\ \gamma(\bk)\ \gamma(-\bk)\right|>2$ and $|f(\mu,\bk)|\le e^{-c\lambda}$ on the boundary of that disc. Hence, by Rouch\'e's theorem, the function $\mu\mapsto F(\mu,\bk)$ has two zeros (multiplicities counted)
in the disc $\{\mu\in\C: |\mu|<\widehat{C}/2\}$. 
\begin{equation}
\textrm{We call the two zeros of  $F(\mu,\bk)$:\ $\mu_+(\bk)$ and $\mu_-(\bk)$.}
\nn\end{equation}

If $\bk\in\R^2$, then $\mu_+(\bk)$ and $\mu_-(\bk)$ are the two real rescaled eigenvalues of $H^\lambda(\bk)$ in the interval $\left[-\widehat{C}/2,\widehat{C}/2\right]$. The dependence of $\mu_+(\bk)$ and $\mu_-(\bk)$ on $\lambda$ has been suppressed. 

Standard residue calculations give
\begin{equation}
(\mu_+(\bk))^l\ +\ (\mu_-(\bk))^l\ =\ \frac{1}{2\pi i}\oint_{|\mu|=\frac12\widehat{C}}\ 
\mu^l\ \frac{\D_\mu F(\mu,\bk)}{F(\mu,\bk)}\ d\mu,\ \ l=1,2.
\label{residue-l}
\end{equation}
Since $|f(\mu,\bk)|\le e^{-c\lambda}$ for $|\mu|<\widehat{C}$ and $|\bk-\tK|\le\hat{c}\lambda^{-1}$,
we have for $|\mu|=\frac12\widehat{C}$ and $|\bk-\tK|\le\hat{c}\lambda^{-1}$ the estimates:
\begin{align}
&\left|\ \partial_\mu F(\mu,\bk)\ -\ 2\mu\ \right| \lesssim\ e^{-c\lambda}\nn\\
& \left|\ F(\mu,\bk)\ -\ \left[\mu^2-\gamma(\bk)\gamma(-\bk)\right]\ \right|\ \lesssim\ 
e^{-c\lambda} \label{bounds_on_F} \\
& \left|\ \mu^2-\gamma(\bk)\gamma(-\bk)\ \right|\ >\ 2.\nn
\end{align}
From \eqref{residue-l} and  \eqref{bounds_on_F} 
we obtain that $\mu_+(\bk)+\mu_-(\bk)$ and  $(\mu_-(\bk))^2\ +\ (\mu_-(\bk))^2$ are analytic functions 
 of $\bk$.
Furthermore, 
\begin{equation*}
 \left|\ 
\frac{\partial_\mu F(\mu,\bk)}{F(\mu,\bk)}\ -\ 
\frac{\partial_\mu[  \mu^2-\gamma(\bk)\gamma(-\bk)\ ]}{\mu^2-\gamma(\bk)\gamma(-\bk)}\ \right| 
\lesssim e^{-c\lambda}
\end{equation*}
for $|\mu|=\frac12\widehat{C}$,\ \ $|\bk-\tK|<\hat{c}\lambda^{-1}$. Therefore, for $l=1,2$:
 \begin{align}
 (\mu_+(\bk))^l\ +\ (\mu_-(\bk))^l\ &=\ 
 \frac{1}{2\pi i}\oint_{|\mu|=\frac12\widehat{C} }\ 
\frac{2\mu^{l+1}}{\mu^2-\gamma(\bk)\gamma(-\bk)}\ d\mu\ +\ \mathcal{O}(e^{-c\lambda})\nn\\
\ &=\ \begin{cases}
\mathcal{O}(e^{-c\lambda}), & l=1\\
2\gamma(\bk)\gamma(-\bk)\ +\ \mathcal{O}(e^{-c\lambda}), & l=2
\end{cases}
\label{mu-mu_squared} 
\end{align}
%
%
for $|\bk-\tK|<\hat{c}\lambda^{-1}$.
Note that 
\[\mu_+(\bk)\mu_-(\bk)=\frac12(\mu_+(\bk)+\mu_-(\bk))^2-\frac12(\mu_+(\bk)^2+\mu_-(\bk)^2).\]
It follows that $\mu_+(\bk)\mu_-(\bk)$ is analytic for $|\bk-\tK|<\hat{c}\lambda^{-1}$ and satisfies the bound
\begin{equation}
\left| \mu_+(\bk)\mu_-(\bk) + \gamma(\bk)\gamma(-\bk) \right|\lesssim e^{-c\lambda},\ \  |\bk-\tK|<\hat{c}\lambda^{-1}\ .
\label{mu_pm2gamma_pm}
\end{equation}

From \eqref{mu-mu_squared}  and \eqref{mu_pm2gamma_pm}, we obtain the following results regarding the quadratic equation $(\mu-\mu_+(\bk))\times (\mu-\mu_-(\bk))=0$.

 \begin{lemma}\label{F-poly-mu}
 For $\bk\in\C^2$, $|\bk-\tK|<\hat{c}\lambda^{-1}$, the roots of $F(\mu,\bk)=0$ in the disc $|\mu|<\widehat{C}/2$ are the roots of a quadratic equation:
 \begin{equation*}
 \mu^2\ +\ b_1(\bk)\mu\ +\ b_0(\bk)\ =\ 0,
\end{equation*}
 where $b_1(\bk), b_0(\bk)$ are analytic functions on 
 \begin{equation}
 D(\tK)\ =\ \{\bk\in\C^2 : |\bk-\tK|<\hat{c}\lambda^{-1}\}
 \label{DtK}
 \end{equation}
 and $|b_1(\bk)|$,\ $|b_0(\bk)+\gamma(\bk)\gamma(-\bk)|\lesssim e^{-c\lambda}$ for all  such $\bk$.
 \end{lemma}
 
 Therefore, after possibly interchanging $\mu_+(\bk)$ and $\mu_-(\bk)$ for each $\bk$, we find:
 \begin{lemma}\label{solve-mu_pm}
 For $\bk\in D(\tK)$, defined in \eqref{DtK}, the roots of $F(\mu,\bk)=0$, are given by
 \begin{equation}
 \mu_\pm(\bk) = h_\tK(\bk)\ \pm\ \sqrt{\gamma(\bk)\gamma(-\bk)\ +\ g_\tK(\bk)},
 \label{mupmofk}
 \end{equation}
 where $h_\tK$ and $g_\tK$ are analytic functions on $D(\tK)$ and 
 \[ |h_\tK(\bk)|,\ |g_\tK(\bk)|\ \lesssim\ e^{-c\lambda},\ \ \bk\in D(\tK).\]
 When $\bk\in\R^2\cap D(\tK)$, $\mu_\pm(\bk)$ are real and therefore
  $h_\tK$ and $g_\tK$ are real for $\bk\in\R^2\cap D(\tK)$.
\end{lemma}

\begin{lemma}\label{analytic-extension}
Let $\tK_1, \tK_2\in\R^2$ with $|\tK_1|, |\tK_2|\le K_{max}$. Suppose $D(\tK_1)\cap D(\tK_2)$ is non-empty.
Let $h_{\tK_1}, g_{\tK_1}$ and $h_{\tK_2}, g_{\tK_2}$ be analytic functions arising from $\tK_1$ and $\tK_2$, respectively, in Lemma \ref{solve-mu_pm}.  Then, $h_{\tK_1}=h_{\tK_2}$ and $g_{\tK_1}=g_{\tK_2}$ on
 $D(\tK_1)\cap D(\tK_2)$.
  \end{lemma}
 \begin{proof}[Proof of Lemma \ref{analytic-extension}] For $\bk\in\R^2\cap\left[D(\tK_1)\cap D(\tK_2)\right]$ we know that $\mu_\pm(\bk)$ are the rescaled eigenvalues of $H^\lambda(\bk)$  in the interval $[-\frac12\widehat{C},\frac12\widehat{C}]$. Since these $\mu_\pm(\bk)$ are given by the formulae \eqref{mupmofk} in 
   $D(\tK_1)\cap\R^2$ and in $D(\tK_2)\cap\R^2$, it follows that $h_{\tK_1}=h_{\tK_2}$ and $g_{\tK_1}=g_{\tK_2}$ on
 $D(\tK_1)\cap D(\tK_2)\cap\R^2$. The Lemma now follows by analytic continuation.
 \end{proof}
 
 Thanks to Lemma \ref{analytic-extension}, we can put together all $h_\tK$ into a single analytic function defined on 
  \begin{equation*}
  \mathcal{D}_{K_{max}}=\{\bk\in\C^2: |\Re\bk|< K_{max}, \ |\Im\bk|<\hat{c}\lambda^{-1}\},
  \end{equation*} 
  and similarly for $g_\tK$. Thus, we obtain the following result.
 
  \begin{lemma}\label{global-mupm}
 There exist analytic functions $g$ and $h$, defined on $\mathcal{D}_{K_{max}}$, with the following properties:
  \begin{enumerate}
  \item $|g(\bk)|$, $|h(\bk)|\lesssim e^{-c\lambda}$ for all $\bk\in \mathcal{D}_{K_{max}}$.
  \item $g(\bk)$ and $h(\bk)$ are real for real $\bk\in \mathcal{D}_{K_{max}}$, and for each $\bk\in \mathcal{D}_{K_{max}}\cap\R^2$, the rescaled eigenvalues of $H^\lambda(\bk)$ in the interval  $[-\frac12\widehat{C},\frac12\widehat{C}]$ are given by
  \begin{equation*}
  \mu_\pm(\bk)\ =\ h(\bk)\pm\sqrt{\gamma(\bk)\gamma(-\bk)+g(\bk)}\ .
  \end{equation*}
  \end{enumerate}
  \end{lemma}
  
  The following result is a consequence of Corollary \ref{atmost2evalues}.
  
  \begin{corollary}\label{low-lying}
The maps $\bk\mapsto E^\lambda_\pm(\bk)$, where
 \begin{equation}
 E^\lambda_\pm(\bk)\ \equiv\ E_0^\lambda+\rho_\lambda\ \mu_\pm(\bk) = E_0^\lambda+\rho_\lambda\ \left(\ 
  \ h(\bk)\pm\sqrt{\gamma(\bk)\gamma(-\bk)+g(\bk)}\ \right)
  \label{Elam_pm}
  \end{equation}
  define the two lowest dispersion surfaces of $-\Delta +\lambda^2 V(\bx)$.  
%
  \end{corollary} 

Recall, by Theorem \ref{solve-L2tau-evp} and Corollary \ref{tb-dirac-pts}, the operator $-\Delta+\lambda^2 V(\bx)$ has a Dirac point at $(\bK_\star,E_D^\lambda)$,  where $\bk=\bK_\star$, any vertex of $\B_h$.
 Therefore, $E_+(\bK_\star)=E_-(\bK_\star)=E_D^\lambda$. 
 By Lemma \ref{gamma0}, $\gamma(\bK_\star)\gamma(-\bK_\star)=0$ and therefore 
 it follows that 
$ E_D^\lambda\ =\ E_0^\lambda + \rho_\lambda\ \left(\ 
  \ h(\bK_\star)\pm\sqrt{ g(\bK_\star)}\ \right),$
and since $E_D^\lambda$ is a double-eigenvalue, we have $g(\bK_\star)=0$, {\it i.e.} $g(\bk)$ vanishes at the vertices of $\B_h$. Thus,
\begin{equation*}
 E_D^\lambda\ =\ E^\lambda_+(\bK_\star)=E^\lambda_-(\bK_\star)= E_0^\lambda\ +\ \rho_\lambda\ h(\bK_\star) .
 \end{equation*}
Note $\mu_\pm(\bK_\star)=h(\bK_\star)$.  By Corollary \ref{low-lying}
 \begin{align*}
 E^\lambda_\pm(\bk)\ -\ E_D^\lambda &\equiv\ \rho_\lambda\ \left(\ \mu_\pm(\bk)\ -\ \mu_\pm(\bK_\star)\ \right)\\   %
&=\ \rho_\lambda\ \left(h(\bk)-h(\bK_\star)\right) \pm\ \rho_\lambda\ \sqrt{\gamma(\bk)\gamma(-\bk)+g(\bk)}\ 
  \end{align*}
  Dividing by $\rho_\lambda$ gives
   \begin{align}
 \left(\ E^\lambda_\pm(\bk)\ -\ E_D^\lambda\ \right) / \rho_\lambda\ &=\ \left(h(\bk)-h(\bK_\star)\right) \pm\ \sqrt{\gamma(\bk)\gamma(-\bk)+g(\bk)}.  \label{Elam_pm1}
  \end{align}
  
  The expression \eqref{Elam_pm1} is an expression for the rescaled low-lying dispersion surfaces, which we now study 
  for $\lambda$ large.

\section{Expansion of $\mu_\pm(\bk)$ and rescaled dispersion surfaces for $\lambda$ large}\label{mu_pm-rescaled}

Introduce the rescaled low-lying dispersion maps:
\begin{equation}
\mu_\pm(\bk)\ \equiv\ \left(\ E^\lambda_\pm(\bk)\ -\ E_D^\lambda\ \right) / \rho_\lambda
\label{mu_pm-def}
\end{equation}
Also, recall (Lemma \ref{gamma0}) that for $\bk\in\R^2$ :
\begin{equation}
\mathscr{W}_{_{\rm TB}}(\bk)\ \equiv\ |\gamma(\bk)|\ =\ \sqrt{\gamma(\bk)\gamma(-\bk)}.
\end{equation}

\subsection{Rescaled dispersion surfaces away from Dirac points}\label{mu-away-from-dpts}

Assume $\bk\in\R^2$, $|\bk|<K_{max}$ and  $|\gamma(\bk)\gamma(-\bk)|\ \ge\  \lambda^{-1/2}$.
(Note that for $\bk\in\R^2$,  we have $\gamma(\bk)\gamma(-\bk)=\mathscr{W}_{_{\rm TB}}(\bk)$.) 
 Then, using \eqref{Elam_pm1} we write
 \begin{align*}
\mu_\pm(\bk)\ &=\  \left(\ E^\lambda_\pm(\bk)\ -\ E_D^\lambda\ \right) / \rho_\lambda\\
 &=\  \pm\sqrt{\gamma(\bk)\gamma(-\bk)}\ \left[1+\frac{g(\bk)}{\gamma(\bk)\gamma(-\bk)}\right]^{\frac{1}{2}} +
   h(\bk)-h(\bK_\star)\\
   &=\ \pm\sqrt{\gamma(\bk)\gamma(-\bk)}\ \left[\ 1\ +\ \tilde{f}_{1,\pm}(\bk)\ \right]\ .
   \end{align*}
%
Thanks to our estimates for $|g(\bk)|$ and $|h(\bk)|$, and the assumed  lower bound for $\gamma(\bk)\gamma(-\bk)$,
 we have
%
%
\begin{proposition}[Rescaled dispersion surfaces away from Dirac points]\label{Proposition-RDSADP}
On the set $\{\bk\in\R^2:|\bk|< K_{max},\ |\gamma(\bk)\gamma(-\bk)|>\lambda^{-{\frac{1}{2}}}\}$, the rescaled eigenvalues of $H^\lambda(\bk)$ in the interval $[-\frac12\widehat{C},\frac12\widehat{C}]$ are given by
\begin{equation}
\mu_\pm(\bk) = \pm\sqrt{\gamma(\bk)\gamma(-\bk)}\ \left[\ 1\ +\ \tilde{f}_{2,\pm}(\bk)\ \right].
\label{mu-away1}
\end{equation}
Equivalently, for the rescaled eigenvalues of $-(\nabla+i\bk)^2+\lambda^2 V(\bx)$, we have:
\begin{align}
& \left(\ E^\lambda_\pm(\bk)\ -\ E_D^\lambda\ \right) / \rho_\lambda\ =\ \pm\sqrt{\gamma(\bk)\gamma(-\bk)}\ \left[\ 1\ +\ \tilde{f}_{1,\pm}(\bk)\ \right]\ .\label{Epm-away}
   \end{align}
Here, the corrections $\tilde{f}_{j,\pm}(\bk),\ j=1,2$ in \eqref{mu-away1} and \eqref{Epm-away}  are real-valued and satisfy 
$|\partial_\bk^\beta\tilde{f}_{\pm,j}(\bk)|\le C(\beta_{max})\ e^{-c\lambda}$ for $|\beta|\le \beta_{max}$.
\end{proposition}

\subsection{Rescaled dispersion surfaces in a neighborhood of Dirac points}\label{mu-near-dpts}
  We shall use Lemma \ref{global-mupm} to study the  rescaled dispersion surfaces in a neighborhood of Dirac points, $(\bk,E)$,  with $\bk\in\mathcal{D}_{K_{max}}$ and rescaled energy, within $[-\frac12\widehat{C},\frac12\widehat{C}]$. 

We begin by noting that rotational symmetry of the Hamiltonian $-\Delta+\lambda^2 V$ implies rotational symmetry of the maps $\bk\mapsto\mu_\pm(\bk)$ with respect to Dirac points. 
\begin{lemma}\label{symmetry-lemma}
Let $(\bK_\star,E_D^\lambda)$, where $\bK_\star$ is a vertex of $\brill_h$, denote a Dirac point of $-\Delta+\lambda^2 V$
 (see Definition \ref{dirac-pt-defn}) 
guaranteed, for $\lambda$ sufficiently large, by Theorem \ref{solve-L2tau-evp} and Corollary \ref{tb-dirac-pts}.
 Thus ,  $E_D^\lambda\ =\ E^\lambda_+(\bK_\star)=E^\lambda_-(\bK_\star)$. 
Then, for all $\bkappa\in\R^2$ with $0<|\bkappa|<\bkappa_0$ sufficiently small, we have
\[ \{\mu_+(\bK_\star+R\bkappa)\ ,\ \mu_-(\bK_\star+R\bkappa)\}\ =\ 
\{\mu_+(\bK_\star+\bkappa)\ ,\ \mu_-(\bK_\star+\bkappa)\}\ .\]
Here, $R$ denotes the $120^{\circ}$ clockwise rotation matrix. The analogous assertion holds with $\mu_\pm(\cdot)$ replaced by $E_\pm^\lambda(\cdot)$.  Here, $\bkappa_0$ is independent of $\lambda$. 
\end{lemma}

\begin{proof}[Proof of Lemma \ref{symmetry-lemma}] 
Consider $0<|\bkappa|<\kappa_0$ sufficiently small. 
Let $E_\bkappa\in[-\frac12\widehat{C},\frac12\widehat{C}]$ be an eigenvalue of $-\Delta+\lambda^2 V$ acting in the space $L^2_{\bK_\star+\bkappa}$. Thus, $E_\bkappa=E^\lambda_-(\bK_\star+\bkappa)$ or $E_\bkappa=E^\lambda_+(\bK_\star+\bkappa)$. Denote 
the corresponding eigenfunction, by $\psi(\bx)$; $(-\Delta+\lambda^2 V)\psi=E_\bkappa\psi$.  Now consider $\widetilde\psi\equiv (\mathcal{R}\psi)(\bx)=\psi(\bx_c+R^*(\bx-\bx_c))$. Recall that $\mathcal{R}$ commutes with $-\Delta+\lambda^2 V$  (Proposition \ref{VisHLP}) and therefore, 
 $(-\Delta+\lambda^2 V)\widetilde\psi=E_\bkappa\widetilde\psi$. 
 By \eqref{RfK}, we have $\widetilde\psi(\bx+\bv) =e^{i(\bK_\star+R\bkappa)\cdot\bv}\widetilde\psi(\bx)$ for $\bv\in\Lambda_h$.
Therefore, $\psi(\bx)$ and $\widetilde\psi(\bx)$ are respectively $\bK_\star+\bkappa$ and $\bK_\star+R\bkappa$ pseudo-periodic eigenstate of $-\Delta+\lambda^2 V$ with the same eigenvalue, $E_\bkappa$. Thus, 
\[ \{E^\lambda_+(\bK_\star+\bkappa),E^\lambda_-(\bK_\star+\bkappa)\}\subset
 \{E^\lambda_+(\bK_\star+R\bkappa),E^\lambda_-(\bK_\star+R\bkappa)\}\ .\]
To prove the reverse inclusion, we start with an $L^2_{\bK_\star+R\bkappa}-$ eigenvalue, $E_{_{R\bkappa}}$, with corresponding eigenstate $\tilde{\phi}\in L^2_{\bK_\star+R\bkappa}$. Then,  $(\mathcal{R}^2\tilde\phi)(\bx)\in L^2_{\bK_\star+\bkappa}$ is an eigenfunction with eigenvalue $E_{_{R\bkappa}}$ and therefore
\[ \{E^\lambda_+(\bK_\star+R\bkappa),E^\lambda_-(\bK_\star+R\bkappa)\}
\subset \{E^\lambda_+(\bK_\star+\bkappa),E^\lambda_-(\bK_\star+\bkappa)\}\ .\]
By \eqref{Elam_pm}, this result transfers to $\mu_\pm(\bk)$, completing the proof of Lemma \ref{symmetry-lemma}. 
\end{proof}

By a careful study of the Taylor expansions of the analytic functions $h(\bk)$ and $g(\bk)$ in a neighborhood of $\bK_\star$ we will prove the following characterization of the local behavior of the low-lying dispersion surfaces near the vertices of $\B_h$.
\begin{proposition}[Rescaled dispersion surfaces near Dirac points]\label{Proposition-nearDPs}
Let $c_{**}$ denote a small constant and  $\widehat{C}$ denote a sufficiently large constant, determined by $V_0$.
For 
 $\lambda\ge\lambda_\star(V_0,\beta_{max})$ sufficiently large,
\begin{enumerate}
\item  the rescaled eigenvalues of $H^\lambda(\bk)$ in the interval $[-\frac12\widehat{C},\frac12\widehat{C}]$ are given by $\mu_\pm(\bk)$, for $|\bk-\bK_\star|<c_{**}$, where 
\begin{align*}
\left|\ \partial_\bk^\beta
\left\{\ \mu_\pm(\bk) - \left[\ h(\bK_\star)\pm\sqrt{\gamma(\bk)\gamma(-\bk)}\ \right]\ \right\}
\ \right|
&\lesssim e^{-c\lambda}\ |\bk-\bK_\star|^{1-|\beta|},
\end{align*}
for $0\le|\beta|\le\beta_{max}$. 
Here, $\mu_\pm(\bK_\star)=h(\bK_\star)$ (see Lemma \ref{global-mupm})  satisfies $|h(\bK_\star)|\lesssim e^{-c\lambda}$. 
\item Equivalently, $E^\lambda_\pm(\bk)$, the low-lying eigenvalues of $-(\nabla+i\bk)^2+\lambda^2 V(\bx)$,   when rescaled, satisfy
\begin{align*}
& \left|\ \partial_\bk^\beta
\left\{\ \left(\ E^\lambda_\pm(\bk)\ -\ E_D^\lambda\ \right) / \rho_\lambda\ -\ 
\left[\ \pm\sqrt{\gamma(\bk)\gamma(-\bk)} \ \right]\ \right\}\ \right| \lesssim\ 
\ e^{-c\lambda}\ |\bk-\bK_\star|^{1-|\beta|}.
   \end{align*}

\end{enumerate}
\end{proposition}
 We now embark on the proof of Proposition \ref{Proposition-nearDPs}, which will occupy the remainder of Section 
 \ref{mu_pm-rescaled}. 
 
  We Taylor expand $h(\bk)$ and $g(\bk)$ using Symmetry Lemma \ref{symmetry-lemma}
 and the results of Section \ref{complex-analysis}. For $\bk\in\R^2$ and $|\bk-\bK_\star|<K_{max}$ we write
  $\bk=\bK_\star+\bkappa$, and we have
  \begin{align*}
  g(\bK_\star+\bkappa)\ &=\ g_0 + \vec{g}_1\cdot\bkappa + \frac12 \bkappa^T g_2 \bkappa 
  + \sum_{|\bn|=3}g_\bn(\bkappa)\ \bkappa^\bn \ , \\
  h(\bK_\star+\bkappa)\ &=\ h_0 + \vec{h}_1\cdot\bkappa + \frac12 \bkappa^T h_2 \bkappa 
  + \sum_{|\bn|=3}h_\bn(\bkappa)\ \bkappa^\bn \ ,
  \end{align*}
  where $g_0=g(\bK_\star), h_0=h(\bK_\star)$ are numbers, $\vec{g}_1, \vec{h}_1\in\R^2$ are vectors, $g_2, h_2$ are symmetric $2\times2$ matrices,
 and  $g_\bn, h_\bn$ ($|\bn|=3$) are scalar and single-valued functions which  depend on multi-indices $\bn$. Moreover, because of the analyticity and bounds given for $g$ and $h$ in Lemma \ref{global-mupm}, we have the norm bounds:
  \begin{align}
  |g_0|, |h_0|\lesssim e^{-c\lambda},\quad  |\vec{g}_1|, |\vec{h}_1|\lesssim e^{-c\lambda},\quad
    |g_2|, |h_2|\lesssim e^{-c\lambda}\ ,
  \label{bounds-gh012}  \end{align}
    and 
     \begin{align}
      \left| \partial^\beta_\bkappa g_\bn(\bkappa)\right|,  \left| \partial^\beta_\bkappa h_\bn(\bkappa)\right|\lesssim e^{-c\lambda},
   \label{bounds-DgDhn}     \end{align}
   for $\bkappa\in\R^2$, $|\bkappa|<\frac{1}{2}K_{max}$, $|\beta|\le\beta_{max}$, $|\bn|=3$.
Note: The derivative bounds follow from
 Cauchy estimates for derivatives of analytic functions. The small thickness $|\Im\bk|<\hat{c} \lambda^{-1}$ of $D_{K_{max}}$ in the imaginary directions in Lemma \ref{global-mupm}
 is overwhelmed by the tiny upper bound $e^{-c\lambda}$ in that lemma.  

As observed above, by part (1) of Lemma \ref{symmetry-lemma} and the vanishing of 
$\gamma(\bK_\star)\gamma(-\bK_\star)$, we have $g_0=0$. By part (2) of Lemma \ref{symmetry-lemma}, we have
$R\vec{g}_1=\vec{g}_1$, $R\vec{h}_1=\vec{h}_1$. Since $1$ is not an eigenvalue of $R$, $\vec{g}_1=0$
and $\vec{h}_1=0$.  Furthermore, $R^T g_2 R=g_2$ and $R^T h_2 R=h_2$. Therefore, by symmetry of $g_2$ and $h_2$, $g_2=g_2^0 I$ and $h_2=h_2^0 I$, for scalars $g_2^0, h_2^0$.
%
%
 Our estimates for $|g_2|, |h_2|$ yield $|g_2^0|, |h_2^0|\lesssim e^{-c\lambda}$. Thus, we obtain the following result.

 \begin{lemma}\label{US1}
 For $\bk=\bK_\star+\bkappa$, $\bkappa\in\R^2$, $|\bkappa|<\frac{1}{2}K_{max}$, the rescaled eigenvalues of $H^\lambda(\bk)$ in the interval $[-\frac12 \widehat{C},\frac12 \widehat{C}]$ are given by
 \begin{equation*}
 \mu_\pm(\bk)\ =\ h(\bk)\ \pm\ \sqrt{\gamma(\bk)\gamma(-\bk)+g(\bk)},
 \end{equation*}
 where 
 \begin{align*}
 h(\bK_\star+\bkappa)\ &=\ h_0 + \frac12 h_2^0 |\bkappa|^2 + \sum_{|\bn|=3}h_\bn(\bkappa)\ \bkappa^\bn
 \ \ \textrm{and} \\ 
 g(\bK_\star+\bkappa)\ &=\  \frac12 g_2^0 |\bkappa|^2 + \sum_{|\bn|=3}g_\bn(\bkappa)\ \bkappa^\bn\ .
 \end{align*}
Here, 
\begin{align*} 
&|h_0|,\ |h_2^0|,\ |g_2^0|\ \lesssim\ e^{-c\lambda} \quad \text{and} \quad
\left| \partial^\beta_\bkappa g_\bn(\bkappa)\right|,  \left| \partial^\beta_\bkappa h_\bn(\bkappa)\right|\lesssim e^{-c\lambda} , \\ 
     & \textrm{for } \bkappa\in\R^2,\ |\bkappa|<\frac{1}{2}K_{max},\ |\beta|\le\beta_{max},\ |\bn|=3\ .\nn
 \end{align*}
 \end{lemma}
 
 In the next sections we use Lemma \ref{US1} to write $\mu_\pm(\bk)=\pm\sqrt{\gamma(\bk)\gamma(-\bk)}$
  plus an error term, which we estimate explicitly. 
  
  \subsection{Bookkeeping}\label{bookkeeping}
  
  Suppose $F(\bkappa)$ and $G(\bkappa)$ are functions, defined on an open subset of $\R^2$,
   satisfying at some point $\bkappa_0$ the estimates:
   \begin{align*}
   \left|\partial_\bkappa^\beta F(\bkappa_0)\right|\ &\le\ A\cdot \delta^{-|\beta|}, \ \ |\beta|\le\beta_{max};\\
   \left|\partial_\bkappa^\beta G(\bkappa_0)\right|\ &\le\ B\cdot {\tilde\delta}^{-|\beta|}, \ \ |\beta|\le\beta_{max}.
   \end{align*}
   Because $\partial_\bkappa^\beta(FG)(\bkappa)$ is a sum of terms $\partial_\bkappa^{\beta_1} F(\bkappa_0)\partial_\bkappa^{\beta_2}G(\bkappa_0)$ with $|\beta_1|+|\beta_2|=|\beta|$, we have
   \[  \left|\partial_\bkappa^\beta(FG)(\bkappa_0)\right|\ \lesssim\ C\ A B\cdot \left[\min(\delta,\tilde\delta)\right]^{-|\beta|},\ \ 
   |\beta|\le\beta_{max}\ ,\]
   where $C$ is a constant depending only on $\beta_{max}$. 
   
   Next, suppose $H(\bkappa)$ is a function defined in a neighborhood of $\bkappa_0$ in $\R^2$, and suppose $H$ satisfies 
   $\left|\partial_\kappa^\beta H(\kappa_0)\right|\ \le\ A \delta^{-|\beta|}$ for $ |\beta|\le\beta_{max}$, 
   and $10^{-1} \le H(\bkappa_0) \le 10$. Because $\partial_\kappa^\beta\left(1/H\right)(\bkappa_0)$ is a sum of terms
    $\left(H(\bkappa_0)\right)^{-\nu-1} \partial_\kappa^{\beta_1}H(\bkappa_0)\cdots\partial_\bkappa^{\beta_\nu}H(\bkappa_0)$
    with $|\beta_1|+\cdots+|\beta_\nu|=|\beta|$, we find that 
    \[ \left|\ \partial_\bkappa^\beta\left(\frac{1}{H}\right)(\bkappa_0)\ \right|\ \le C\delta^{-|\beta|},\ \ |\beta|\le\beta_{max} , \]
    with $C$ determined by $A$ and $\beta_{max}$. 
    
    \begin{lemma}[Bookkeeping Lemma]\label{BL}
   Fix a positive integer, $\beta_{max}$.  Let $F(\bkappa)$ and $F_0(\bkappa)$ be complex-valued functions defined in an open subset of $\C^2$. Assume that there exist positive constants $A$ and $\eta$, such that  the following estimates hold:
    \begin{align}
    \left|F(\bkappa)\right|,\  \left|F_0(\bkappa)\right|\ &\le\ \frac{1}{10};\label{FF0}\\
     \left|\partial^\beta_\bkappa F(\bkappa)\right|,\  \left|\partial^\beta_\bkappa F_0(\bkappa)\right|\ 
     &\le\ A \ \delta^{-|\beta|},\ \ {\rm for}\ |\beta|\le\beta_{max};\label{DFDF0}\\
     \left|\partial^\beta_\bkappa \left( F(\bkappa)-F_0(\bkappa) \right)\right|\ &\le\ \eta\ \delta^{-|\beta|},\ \
     {\rm for}\ |\beta|\le\beta_{max}\ .
     \label{DFminusF0}
     \end{align} 
     Then,
         \begin{equation*}
     \left|\partial^\beta_\bkappa \left( \left(1+ F(\bkappa)\right)^{\frac{1}{2}}-\left(1+F_0(\bkappa)\right)^{\frac{1}{2}} \right)\right|\ \le C\ \eta\ \delta^{-|\beta|},\   {\rm for}\ |\beta|\le\beta_{max},\ 
     \end{equation*}
     where $C$ is determined by $A$ and $\beta_{max}$. 
          \end{lemma}
     
     \begin{proof}[Proof of Bookkeeping Lemma \ref{BL}] Set 
     \begin{equation*}
     \mathscr{A}(X,Y)\ \equiv\ \frac{(1+X)^{\frac{1}{2}}-(1+Y)^{\frac{1}{2}}}{X-Y} .
     \end{equation*}
     Then, $\mathscr{A}(X,Y)$ is analytic as a function of $(X, Y)$ in $D\times D$, where $D$ is the disc of radius $1/10$ about $0$ in $\C$. 
      Thus, $(1+F(\bkappa))^{\frac{1}{2}}-(1+F_0(\bkappa))^{\frac{1}{2}}=\mathscr{A}\left(F(\bkappa) , F_0(\bkappa)\right)\cdot \left[ F(\bkappa)-F_0(\bkappa) \right]$. 
      
      Now $\partial^\beta_\bkappa\left[ \mathscr{A}\left(F(\bkappa),F(\bkappa_0)\right) \right]$ is a sum of terms
      \begin{equation}
      \left[\partial_{X,Y}^\gamma \mathscr{A}(X,Y) \Big|_{\substack{X=F(\bkappa)\\ Y= F_0(\bkappa)}} \right]
      \cdot \prod_{\nu=1}^{\nu_{max}}\partial_\bkappa^{\beta_{\nu}} F(\bkappa)\cdot 
       \prod_{\tilde\nu=1}^{\tilde\nu_{max}}\partial_\bkappa^{\sigma_{\tilde\nu}} F_0(\bkappa)\ ,
       \label{DXYA}
      \end{equation}
      where $|\beta_1|+\cdots+|\beta_{\nu_{max}}|+|\sigma_1|+\dots+|\sigma_{\tilde\nu_{max}}|=|\beta|$. Using the hypothesized bounds \eqref{FF0} and \eqref{DFDF0}, we see that each term \eqref{DXYA} has absolute value at most $C\cdot\delta^{-|\beta|}$, where the constant $C$ is determined by $A$ and $\beta_{max}$. Thus,
      \begin{equation*}
      \left|\D^\beta_\bkappa \mathscr{A}(F(\bkappa),F_0(\bkappa))\right|\ \le\ C\delta^{-|\beta|},\ \ |\beta|\le\beta_{max} ,
      \end{equation*}
      where $C$ is determined by $A$ and $\beta_{max}$. 
      
      We next use the assumed bound \eqref{DFminusF0}. Because
      \[ \left(1+F(\bkappa)\right)^{\frac{1}{2}}\ -\ \left(1+F_0(\bkappa)\right)^{\frac{1}{2}}\ =\ \mathscr{A}\left(F(\bkappa),F_0(\bkappa)\right)\cdot \left[F(\bkappa)-F_0(\bkappa)\right]\]
      it follows that 
      \[\Big|\ \partial_\bkappa^\beta\left[ \left(1+F(\bkappa)\right)^{\frac{1}{2}}\ -\ \left(1+F_0(\bkappa)\right)^{\frac{1}{2}}\right]\ \Big|
      \le\ C\ \eta\ \delta^{-|\beta|},\ \ |\beta|\le\beta_{max} , \]
      with $C$  determined by $A$ and $\beta_{max}$. This completes the proof of the lemma. 
        \end{proof}
        
     Recall that Lemma \ref{US1} gives $\mu_\pm(\bk)=h(\bk)\pm\sqrt{\gamma(\bk)\gamma(-\bk)+g(\bk)}$, 
     and specifies the form of the Taylor expansions of $g(\bK_\star+\bkappa)$ and $h(\bK_\star+\bkappa)$ for  $|\bkappa|=|\bk-\bK_\star|$ small and $\bkappa$ real. We next compare the functions
     \begin{equation}
     \mathbb{F}(\bkappa)\ \equiv\ \sqrt{\gamma(\bk)\gamma(-\bk)+g(\bk)},\ \ \bk=\bK_\star+\bkappa
     \label{bbF}
     \end{equation}
     and 
     \begin{equation}
     \mathbb{F}_0(\bkappa)\ \equiv\ \sqrt{\gamma(\bk)\gamma(-\bk)},\ \ \bk=\bK_\star+\bkappa\ ,
     \label{bbF0}
     \end{equation}
     for $|\bkappa|<c_{\star\star}$, a small constant to be chosen below. 
     Note that since here $\bk$ is real, $\gamma(\bk)\gamma(-\bk)$ and $\gamma(\bk)\gamma(-\bk)+g(\bk)$ are non-negative; therefore we may use the non-negative square root. Consequently, $\mathbb{F}$ and $\mathbb{F}_0$ are well-defined.

     By Lemma \ref{gamma0} and Taylor expansion
      ($\bk=\bK_\star+\bkappa$ real), we have
     \begin{equation}
  \gamma(\bk)\gamma(-\bk)\ =\ a_{00}^2\left|\bkappa\right|^2
      \ +\ \sum_{|\bfm|=3} \bkappa^\bfm F_{0,\bfm}(\bkappa),\ \ {\rm for}\ \ |\bkappa|<c_{\star\star},
      \label{gplus-gminus-taylor}
      \end{equation}
      with $a_{00}=\sqrt{3/4}$, $\left|\ \partial^\beta_\bkappa F_{0,\bfm}(\bkappa)\ \right|\lesssim C$, for $|\bfm|=3$, 
      $|\beta|\le\beta_{max}$ and $|\bkappa|<c_{\star\star}$. 
      
      Therefore, Lemma \ref{US1} yields $(\bk=\bK_\star+\bkappa)$
      \begin{equation}
    \gamma(\bk)\gamma(-\bk) + g(\bk) = 
      \left( a_{00}^2 + \frac12 g_2^0 \right) |\bkappa|^2 + \sum_{|\bfm|=3}\bkappa^\bfm\left(F_{0,\bfm}(\bkappa)+g_\bfm(\bkappa)\right) .
      \label{gam2plusg}\end{equation}
      
      We rewrite \eqref{bbF} and \eqref{bbF0} using   \eqref{gplus-gminus-taylor} and \eqref{gam2plusg} in the form
      \begin{align*}
      \mathbb{F}(\bkappa)\ &=\  \left(\ a_{00}^2\ +\ \frac12 g_2^0\ \right)^{\frac{1}{2}} |\bkappa|
      \times\left[1+\sum_{|\bfm|=3} \frac{\bkappa^\bfm}{|\bkappa|^2}\cdot 
      \frac{F_{0,\bfm}(\bkappa)+g_\bfm(\bkappa)}{a_{00}^2\ +\ \frac12 g_2^0}\right]^{\frac{1}{2}}\ ,\nn\\
       \mathbb{F}_0(\bkappa)\ &=\  a_{00}\ |\bkappa|
      \times\left[1+\sum_{|\bfm|=3} \frac{\bkappa^\bfm}{|\bkappa|^2}\cdot 
      \frac{F_{0,\bfm}(\bkappa)}{a_{00}^2}\right]^{\frac{1}{2}}\ .\nn
      \end{align*}
      
      Therefore, 
      {\footnotesize{
      \begin{align}
    &  \mathbb{F}(\bkappa) -  \mathbb{F}_0(\bkappa)\nn\\
      &\quad = \left[ \left( a_{00}^2 + \frac12 g_2^0 \right)^{\frac{1}{2}}-a_{00} \right]\cdot |\bkappa|\cdot
      \left(1+\sum_{|\bfm|=3} \frac{\bkappa^\bfm}{|\bkappa|^2}\cdot 
      \frac{F_{0,\bfm}(\bkappa)+g_\bfm(\bkappa)}{a_{00}^2 + \frac12 g_2^0}\right)^{\frac{1}{2}} \nn \\
      &\quad  + a_{00} |\bkappa| 
      \left[ \left(1+\sum_{|\bfm|=3} \frac{\bkappa^\bfm}{|\bkappa|^2}\cdot 
      \frac{F_{0,\bfm}(\bkappa)+g_\bfm(\bkappa)}{a_{00}^2 + \frac12 g_2^0}\right)^{\frac{1}{2}} - 
      \left(1+\sum_{|\bfm|=3} \frac{\bkappa^\bfm}{|\bkappa|^2}\cdot 
      \frac{F_{0,\bfm}(\bkappa)}{a_{00}^2}\right)^{\frac{1}{2}} \right] \nn \\
      &\quad \equiv \textrm{Term $1$} +  \textrm{Term $2$} . \label{term12}
      \end{align}
      }}

\noindent{\it Estimation of Term $1$:} 
We apply $\partial_\bkappa^\beta$ to Term $1$ and estimate. Consider the first factor in Term $1$: $\left(\ a_{00}^2\ +\ \frac12 g_2^0\ \right)^{\frac{1}{2}}-a_{00}$, which is independent of $\bkappa$. 
From Lemma \ref{gamma0} we have  $a_{00}=\sqrt{3/4}$.  Also, $|g_2^0|\lesssim e^{-c\lambda}$. Therefore, the first factor is $\lesssim e^{-c\lambda}$. 
 Concerning the second factor, $|\bkappa|$, we have $\partial_\bkappa^\beta |\bkappa|\lesssim |\bkappa|^{1-|\beta|}$ for $|\bkappa|>0$, $|\beta|\le\beta_{max}$.

We turn to the third factor. Here we apply our Bookkeeping Lemma with the choices 
\[
F(\bkappa)= \sum_{|\bfm|=3}\ \frac{\bkappa^\bfm}{|\bkappa|^2}\cdot
  \left(\frac{F_{0,\bfm}(\bkappa)+g_\bfm(\bkappa)}{a_{00}^2\ +\ \frac12 g_2^0}\right)\ \ \textrm{and}\ \  F_0(\bkappa)=0.
\]
 Recall that $\left|\partial_\bkappa^\beta F_{0,\bfm}(\bkappa)\right|\le C$ (Lemma \ref{gamma0}) , 
 $\left|\partial_\bkappa^\beta g_{\bfm}(\bkappa)\right|\lesssim e^{-c\lambda}$ for $|\bkappa|<c_{\star\star}$,
  $|\beta|\le\beta_{max}$, $|\bfm|=3$ (by \eqref{bounds-DgDhn}), and that $a_{00}^2+ g_2^0/2\ge3/8$ (Lemma \ref{gamma0} and \eqref{bounds-gh012}). Therefore, for the same range
  of $\bkappa$, $\beta$ and $\bfm$,
  \begin{equation}\label{ratio-bound}
  \left|\ \partial_\bkappa^\beta
  \left[ \frac{F_{0,\bfm}(\bkappa)+g_\bfm(\bkappa)}{a_{00}^2\ +\ \frac12 g_2^0}\right]\ \right|
  \le C \ \left(\lesssim\ |\bkappa|^{-|\beta|}\right)\ .
  \end{equation}
  Also, for $|\bfm|=3$. $|\beta|\le\beta_{max}$ and any $\bkappa$, we have
  \begin{equation}
  \left|\ \partial_\bkappa^\beta \left[\frac{\bkappa^\bfm}{|\bkappa|^2}\right]\ \right|\ \lesssim  |\bkappa|^{1-|\beta|}
  \label{Dabsbkappa}
  \end{equation}
  because $\frac{\bkappa^\bfm}{|\bkappa|^2}$ is homogeneous of degree $1$ and smooth away from zero. Therefore, 
  \begin{equation}\label{pre-book-term1}
   \left|\ \partial_\bkappa^\beta F(\bkappa)\ \right|\ \equiv\ 
  \left|\ \partial_\bkappa^\beta
  \left[
  \sum_{|\bfm|=3}\ \frac{\bkappa^\bfm}{|\bkappa|^2}\cdot
  \left(\frac{F_{0,\bfm}(\bkappa)+g_\bfm(\bkappa)}{a_{00}^2\ +\ \frac12 g_2^0}\right)
  \right]\ \right| \lesssim |\bkappa|^{1-|\beta|} ,
        \end{equation}
        and of course $\left|\ \partial_\bkappa^\beta F_0(\bkappa)\ \right|\lesssim |\bkappa|^{1-|\beta|}$. 
         In particular, the case $\beta=0$ implies that $|F(\bkappa)|, |F_0(\bkappa)| \le1/10$ on the set where $|\bkappa|<c_{\star\star}$,
        provided we take $c_{\star\star}$ small enough.  
        
        Applying Bookkeeping Lemma \ref{bookkeeping} with $F$ given by the expression within square brackets in \eqref{pre-book-term1} and $F_0\equiv0$ we conclude that 
        \begin{equation}
        \left|\ \partial_\bkappa^\beta
  \left[\ \left(\ 1+\sum_{|\bfm|=3}\ \frac{\bkappa^\bfm}{|\bkappa|^2}\cdot
  \left(\frac{F_{0,\bfm}(\bkappa)+g_\bfm(\bkappa)}{a_{00}^2\ +\ \frac12 g_2^0}\right)\ \right)^{\frac12}-(1+0)^{\frac{1}{2}}
\   \right]\
   \right|\ \lesssim\ |\bkappa|^{1-|\beta|}\ ,
   \label{book-term1}     \end{equation}
        for all $|\bkappa|<c_{\star\star}$ and $|\beta|\le\beta_{max}$. 
Finally, using \eqref{Dabsbkappa} and \eqref{book-term1} we obtain 
\[|\partial^\beta_\bkappa({\rm Term}\ 1)|\lesssim e^{-c\lambda}\ |\bkappa|^{1-|\beta|}.\] 

\noindent{\it Estimation of Term $2$:} Our strategy is again to apply the Bookkeeping Lemma, this time with the choices 
\begin{equation*}
F_0(\bkappa)\equiv \sum_{|\bfm|=3} \frac{\bkappa^\bfm}{|\bkappa|^2}\cdot 
    \frac{F_{0,\bfm}(\bkappa)}{a_{00}^2},\qquad
F(\bkappa)\equiv \sum_{|\bfm|=3} \frac{\bkappa^\bfm}{|\bkappa|^2}\cdot 
     \frac{F_{0,\bfm}(\bkappa)+g_\bfm(\bkappa)}{a_{00}^2\ +\ \frac12 g_2^0} .
      \end{equation*}

 From \eqref{ratio-bound} we have, for $|\beta|\le\beta_{max}, |\bkappa|<c_{\star\star}$ and $|\bfm|=3$
\begin{equation}\label{ratio-bound1}
  \left|\ \partial_\bkappa^\beta
  \left[ \frac{F_{0,\bfm}(\bkappa)+g_\bfm(\bkappa)}{a_{00}^2\ +\ \frac12 g_2^0}\right]\ \right|
  \le C ,\ \ \textrm{ and similarly,}\ \  \left|\ \partial_\bkappa^\beta
  \left[ \frac{F_{0,\bfm}(\bkappa)}{a_{00}^2}\right]\ \right|\ 
  \le C \ .
  \end{equation}
  Moreover, 
  \begin{align*}
&   \frac{F_{0,\bfm}(\bkappa)+g_\bfm(\bkappa)}{a_{00}^2\ +\ \frac12 g_2^0}\ -\ 
   \frac{F_{0,\bfm}(\bkappa)}{a_{00}^2} 
   \nn\\
   &\quad 
   =\ 
   \left[ (a_{00}^2\ +\ \frac12 g_2^0)^{-1}-a_{00}^{-2} \right]\cdot \left(F_{0,\bfm}(\bkappa)+g_\bfm(\bkappa)\right)
   \ +\ a_{00}^{-2} g_\bfm(\bkappa)
   \end{align*}
   Recalling that $a_{00}=\sqrt{3}/2$, $|g_2^0|\lesssim e^{-c\lambda}$, we have 
\[ |(a_{00}^2\ +\ \frac12 g_2^0)^{-1}-a_{00}^{-2}|\lesssim e^{-c\lambda},\ \ {\rm and}\ \ 
 \left| \partial_\bkappa^\beta\left(\ F_{0,\bfm}(\bkappa)+g_\bfm(\bkappa)\ \right)\ \right|\le C
 \]
 for $|\beta|\le\beta_{max}$, $|\bfm|=3$ and $|\bkappa|<c_{\star\star}$, $\bkappa\in\R^2$. Also, 
 \[
 \left| \partial_\bkappa^\beta g_\bfm(\bkappa)\right|\lesssim e^{-c\lambda}
 \]
 for $|\beta|\le \frac12\beta_{\max}$, $|\bfm|=3$ and $|\bkappa|<c_{\star\star}$.
Consequently, 
\[
\left|\ \partial_\bkappa^\beta\left[\frac{F_{0,\bfm}(\bkappa)+g_\bfm(\bkappa)}{a_{00}^2\ +\ \frac12 g_2^0}\ -\ 
   \frac{F_{0,\bfm}(\bkappa)}{a_{00}^2}\right]\right|\ \lesssim\ e^{-c\lambda}\lesssim e^{-c\lambda}|\kappa|^{-|\beta|}
   \]
   for $|\beta|\le\beta_{max}$, $0<|\bkappa|<c_{\star\star}$, $|\bfm|=3$. Also, 
   \[ 
  \left| \partial_\bkappa \left( \frac{\bkappa^\bfm}{|\bkappa|^2} \right) \right| \lesssim |\bkappa|^{1-|\beta|}\ 
  \]
   for the same range of $\beta$, $\bkappa$ and $\bfm$. Therefore, 
  \begin{align}
&\left|\  \partial_\bkappa^\beta\left( F(\bkappa)\ -\ F_0(\bkappa)\right) \right|\nn\\
&=\ \left|\  \partial_\bkappa^\beta\left\{\ \left[\sum_{|\bfm|=3}\frac{\bkappa^\bfm}{|\bkappa|^2}\frac{F_{0,\bfm}(\bkappa)+g_\bfm(\bkappa)}{a_{00}^2\ +\ \frac12 g_2^0}\right] -\ 
   \left[\sum_{|\bfm|=3}\frac{\bkappa^\bfm}{|\bkappa|^2}\frac{F_{0,\bfm}(\bkappa)}{a_{00}^2}\right]\
   \right\}\right|
   \nn\\
   &  
   \lesssim\ e^{-c\lambda}\ |\bkappa|\ \cdot |\bkappa|^{-|\beta|}\nn
   \end{align}
   for $\bkappa\in\R^2$,\ $0<|\bkappa|<c_{\star\star}$ and $|\beta|\le\beta_{max}$. 
   Moreover, if $c_{\star\star}$ is sufficiently small, then for $0<|\bkappa|<c_{\star\star}$ we have $|F(\bkappa)|, |F_0(\bkappa)|\le 1/10$ because
    $\left|\frac{\bkappa^\bfm}{|\bkappa|^2}\right|\lesssim|\bkappa|$. From \eqref{ratio-bound1} we have
    \begin{align*}
  &  \left|\  \partial_\bkappa^\beta F(\bkappa)\ \right|\ =\ 
\left|\    \partial_\bkappa^\beta \left[\sum_{|\bfm|=3}\frac{\bkappa^\bfm}{|\bkappa|^2}\frac{F_{0,\bfm}(\bkappa)+g_\bfm(\bkappa)}{a_{00}^2\ +\ \frac12 g_2^0}\right] 
   \right|\lesssim |\bkappa|\cdot |\bkappa|^{-|\beta|}
   \end{align*}    
   and 
   \begin{align*}
  &  \left|\  \partial_\bkappa^\beta F_0(\bkappa)\ \right|\ =\ 
\left|\    \partial_\bkappa^\beta \left[\sum_{|\bfm|=3}\frac{\bkappa^\bfm}{|\bkappa|^2}\frac{F_{0,\bfm}(\bkappa)}{a_{00}^2}\right] 
   \right|\lesssim |\bkappa|\cdot |\bkappa|^{-|\beta|}\ .
   \end{align*} 
   Thus we have verified all hypothesis of the Bookkeeping Lemma. That lemma implies that the expression in square brackets in Term $2$  satisfies the bound: 
  \begin{align*}
 &\left|\ \partial_\bkappa^\beta\left[\ \left(1+ \sum_{|\bfm|=3}\frac{\bkappa^\bfm}{|\bkappa|^2}\frac{F_{0,\bfm}(\bkappa)+g_\bfm(\bkappa)}{a_{00}^2\ +\ \frac12 g_2^0}\right)^{\frac{1}{2}}-\left(1+\sum_{|\bfm|=3}\frac{\bkappa^\bfm}{|\bkappa|^2}\frac{F_{0,\bfm}(\bkappa)}{a_{00}^2}\right)^{\frac{1}{2}}\ \right]\ \right|\\
  &\lesssim e^{-c\lambda}|\bkappa|^{1-|\beta|},
  \end{align*}
  for $0<|\bkappa|<c_{\star\star}$ and $|\beta|\le\beta_{max}$. Because also $|\partial_\bkappa^\beta(a_{00}|\bkappa|)|\lesssim |\bkappa|^{1-|\beta|}$ for $0<|\bkappa|<c_{\star\star}$ and $|\beta|\le\beta_{max}$ it now follows that 
  {\footnotesize{
  \begin{align*}
  &\left| \partial_\bkappa^\beta\ \left[\ \textrm{Term\ $2$}\ \right]\ \right|\nn\\
  &\equiv\ \partial_\bkappa^\beta\left[\ a_{00}|\bkappa|\ \left\{\left(1+ \sum_{|\bfm|=3}\frac{\bkappa^\bfm}{|\bkappa|^2}\frac{F_{0,\bfm}(\bkappa)+g_\bfm(\bkappa)}{a_{00}^2\ +\ \frac12 g_2^0}\right)^{\frac{1}{2}}-\left(1+\sum_{|\bfm|=3}\frac{\bkappa^\bfm}{|\bkappa|^2}\frac{F_{0,\bfm}(\bkappa)}{a_{00}^2}\right)^{\frac{1}{2}}\ \right\}\ \right]\nn\\
  &\lesssim e^{-c\lambda}|\bkappa|^{2-|\beta|}\ . \nn 
  \end{align*}
  }}
  
    Recall that  $ \mathbb{F}(\bk)- \mathbb{F}_0(\bk)= \textrm{Term $1$}\ +\ \textrm{Term $2$}$, where
     $\mathbb{F}(\bk)$ and $\mathbb{F}_0(\bk)$ are given by expressions, which are displayed  in 
   \eqref{bbF} and \eqref{bbF0}. Combining our estimates for the derivatives of  Term $1$ and Term $2$, we find that 
   \begin{align}
\left| \partial_\bkappa^\beta\left(   \mathbb{F}(\bk)- \mathbb{F}_0(\bk) \right) \right|
&\equiv \left| \partial_\bkappa^\beta\left( \sqrt{\gamma(\bk)\gamma(-\bk)+g(\bk)} - 
\sqrt{\gamma(\bk)\gamma(-\bk)} \right) \right|
 \lesssim e^{-c\lambda} |\bkappa|^{1-|\beta|}
\nn   \end{align}
for $\bkappa\in\R^2$, $0<|\bkappa|<c_{\star\star}$ and $|\beta|\le\beta_{max}$.

Recall from Lemma \ref{US1} that the rescaled eigenvalues of $H^\lambda(\bk)$, with $\bk=\bK_\star+\bkappa$, in the interval $[-\frac12\widehat{C},\frac12\widehat{C}]$ are
\begin{align}
\mu_\pm(\bK_\star+\bkappa)\ &=\ h(\bkappa)\ \pm\ \mathbb{F}(\bkappa) \nn\\
&=\ h(\bkappa)\ \pm\ \mathbb{F}_0(\bkappa)\ +\ \left(\mathbb{F}(\bkappa)-\mathbb{F}_0(\bkappa)\right),
\nn\end{align}
where 
\begin{align}
&h(\bkappa)=h_0+\frac12 h_2^0|\bkappa|^2+\sum_{|\bfm|=3}\kappa^\bfm h_\bfm(\bkappa) , \nn
\end{align}
with $|h_0|,\ |h_2^0|\ \lesssim e^{-c\lambda}$ and $ |\partial_\bkappa^\beta h_\bfm(\bkappa)|\lesssim e^{-c\lambda}$
for $|\beta|\le\beta_{max}$ and $|\bkappa|<c_{\star\star}$.  

Combining the above with our estimate for the derivatives of $\mathbb{F}-\mathbb{F}_0$, we obtain
\begin{align}
&\left|\ \partial_\bkappa^\beta\left\{\ \mu_\pm(\bK_\star+\bkappa)
-\left[h_0\pm\sqrt{\gamma(\bK_\star+\bkappa)\gamma(-\bK_\star-\bkappa)}\right]\ \right\}\ \right|
\ \lesssim\ e^{-c\lambda} |\bkappa|^{1-|\beta|}
\end{align}
for $|\beta|\le\beta_{max}$ and $0<|\bkappa|<c_{\star\star}$. Here, $|h_0|\lesssim e^{-c\lambda}$. 
 
 This completes the proof Proposition \ref{Proposition-nearDPs}.

\section{Controlling the perturbation theory; completion of the proof of Theorem \ref{main-theorem}}\label{ME-bounds}

In this section we complete the proof of Theorem \ref{main-theorem}.
To do this, we must complete the proofs of Propositions \ref{ME-BA1}, \ref{ME-AA1} and \ref{ME-IJ-higher-order}
 by establishing the following estimates on corrections to the leading order behavior of the entries of $\mathcal{M}^{\lambda,\tK}(\bk,\rho_\lambda \mu)$ (see \eqref{MIJ-expanded}):
\begin{align*}
&\mathfrak{I}^{(1)}_{BA}(\lambda) +  \mathfrak{I}^{(2)}_{BA}(\lambda)\ \lesssim\ \rho_\lambda\times e^{-c\lambda}\ , \\
&\mathfrak{I}_{AA}^{(1)}(\lambda) + \mathfrak{I}_{AA}^{(2)}(\lambda) \ \lesssim\ \rho_\lambda\times e^{-c\lambda}\ ,
\end{align*}
and 
\begin{align*}
 &  \left|\  \left\langle \left[H^\lambda(\overline\bk)-\overline\Omega\right] p_{_{\overline\bk,J}}^\lambda, 
 {\rm Res}^{\lambda,\tK}(\bk,\Omega)\ \Pi_{AB}  \left[H^\lambda(\bk)-\Omega\right] p_{\bk,I}^\lambda
  \right\rangle\ \right|\ \lesssim\ \rho_\lambda\ e^{-c\lambda}\ . 
\end{align*}
for some $c>0$; see \eqref{ip-BA-bound}, \eqref{ip-AA-bound} and \eqref{higher-order-ME-est}. Here, 
 \begin{equation*}
  \rho_\lambda\ \equiv\ \lambda^2 \int |V_0(\by)|\ p_0^\lambda(\by-\be_{A,1})\ p_0^\lambda(\by)\ d\by\ , 
  \end{equation*}
  which, by Proposition \ref{prop:rho-lam-bounds}, satisfies the upper and lower bounds 
  \[ e^{-c_1\lambda}\ \lesssim\ \rho_\lambda\ \lesssim\ e^{-c_2\lambda}.\] 
 These bounds are proved Section \ref{proof-rho-lam-bounds}. In previous sections, we required that ${\rm supp}\ V_0\subset B({\bf 0},r_0)$, where $0<r_0<\frac{1}{2}|\be_{A,1}|$. To prove the above bounds, we impose a stricter constraint on $r_0$, 
  namely $r_0<r_{critical}$, where $r_{critical}$ arises in the following geometric lemma. 
  %
  %
We note that the assertions of this lemma are easily seen to hold for $r_0$ positive and
sufficiently small. A non-trivial lower bound for $r_{critical}$ is of interest in applications.
%

 \begin{lemma}[Geometric Lemma]\label{euclid1}
  Recall  the vectors $\be_{A,\nu},\ \be_{B,\nu}$, $\nu=1,2,3$; see \eqref{vA-nneighbors}, \eqref{vB-nneighbors} and  Figure \ref{fig:fundamental-cell}. 
Let $R_{_{60}}$ denote the matrix which rotates a vector in the plane by $60^\circ$ clockwise.
  There exists a positive universal constant, $r_{critical}$, satisfying
\begin{equation}
 0.33\ |\be_{A,1}|\ \le\ r_{critical}\ <\ 0.5\ |\be_{A,1}| \ ,
 \label{r-critical-bound}
 \end{equation}
 and small positive constants, $c^\prime$, $c^{\prime\prime}$, $c^{\prime\prime\prime}$, $c^{\prime\prime\prime\prime}$,
for which the following assertions hold for $I=A,B$ and $\nu=1,2,3$, and all $r_0<r_{critical}$ and all $\bz, \by \in B({\bf 0},r_0)$:
 \begin{enumerate}
\item There exists  $l_0\in\{0,1,\ldots,5\}$, such that:
 \begin{equation}
 |\bz+\be_{I,\nu}-\by|\ -\  |\bz-R^{l_0}_{_{60}}\by|\ \ge\  c^\prime\ |\be_{I,\nu}|\ .
 \label{euclid-est1}
 \end{equation}
 %
 \item For any $\bfm\in\Z^2$ 
there exists  $l_0\in\{0,1,\ldots,5\}$ such that:
 \begin{equation}
  |\bz-\bfm\vec\bv-\by|\ -\ |\bz-R^{l_0}_{_{60}}\by|\ \ge\ c^{\prime\prime}\ |\bfm|\ .
 \label{euclid-est2}
 \end{equation}
\item Let $N_{bad}(\be_{I,\nu})$ denote the set of all $\bfm\in\Z^2$ such that $|\be_{I,\nu}+\bfm\vec\bv|=|\be_{I,\mu}|$, for $\mu=1,2,3$. 
For any $\bfm\in\Z^2\setminus N_{bad}(\be_{I,\nu})$ there exists $l_0\in\{0,1,\ldots,5\}$ such that:  
 \begin{equation}
 |\bz+(\be_{I,\nu}+\bfm\vec\bv)-\by|\ -\  |\bz+\be_{I,\nu}-R^{l_0}_{_{60}}\by|\ \ge\ c^{\prime\prime\prime}\ |\bfm|\ .
 \label{euclid-est3}
 \end{equation}
 \item For any $\bn\in\Z^2\setminus\{(0,0)\}$, there exists $l_0\in\{0,1,\ldots,5\}$ such that:  
 \begin{equation}
 |\bz+\bn\vec\bv-\by|\ -\  |\bz+\be_{I,\nu}-R^{l_0}_{_{60}}\by|\ \ge\ c^{\prime\prime\prime\prime}\ |\bn|\ .
 \label{euclid-est4}
 \end{equation}
 \end{enumerate}
 \end{lemma}

\begin{proof}[Proof of Lemma \ref{euclid1}]
Fix $\delta=\frac{1}{2} \times 10^{-2}$ and let $\Gamma_\delta$ denote the grid of points with rational coordinates of the form $\bx^\star=(x^\star_1,x^\star_2)=(p_1,p_2)\delta$ with $(p_1,p_2)\in\Z^2$, satisfying 
\begin{equation}
 \label{set_lambda}
 |\bx^\star|=|(x^\star_1,x^\star_2)| = \sqrt{p_1^2+p_2^2}\ \delta < r_0+\delta/\sqrt{2}.
\end{equation}
Our strategy to prove assertions (1) and (4) of Lemma \ref{euclid1} is to reduce the continuum assertions \eqref{euclid-est1} and \eqref{euclid-est4} to a finite computation on the grid $\Gamma_\delta$. This finite computation is then verified on a computer which performs arithmetic with sufficiently high precision. We prove assertions (2) and (3) of the lemma without resorting to computer simulations.

Our reduction of the continuum assertions to finite computations relies on the following observation:

\begin{remark}\label{continuum_discrete}
Let $r_0>0$ and $\bx$ be any point in the disc $B({\bf 0},r_0)$: $|\bx|<r_0$. Then, there exists $\bx^\star\in\Gamma_\delta$ such that $|\bx-\bx^\star| \leq \delta/\sqrt{2}$.
\end{remark}

\begin{proof}[Proof of Remark \ref{continuum_discrete}]
 Every point of $\R^2$ lies within a distance $1/\sqrt{2}$ of some lattice point $(p_1,p_2)\in\Z^2$, the worst case being points of the form $(p_1+1/2,p_2+1/2)$. 
 Therefore, $\bx \in B({\bf 0},r_0)$ lies within a distance $\delta/\sqrt{2}$ of a point $\bx^\star =(p_1,p_2)\delta$, which belongs to $\Gamma_\delta$ because $|\bx^\star| \leq |\bx|+|\bx-\bx^\star| < r_0 + \delta/\sqrt{2}$; see \eqref{set_lambda}.
\end{proof}

Continuing the proof of Lemma \ref{euclid1}, note that by symmetry ($\be_{B,\nu}=-\be_{A,\nu}$ and $\be_{A,\nu}=R_{_{120}}^{\nu-1}\be_{A,1}$, $\nu=1,2,3$; see \eqref{eAB-def}), it suffices to prove the lemma for the case where $I=A$ and $\nu=1$: $\be_{I,\nu} = \be_{A,1}$. 

In what follows, we fix $r_0=0.33|\be_{A,1}|$, $\eps=10^{-8}$ and $\delta=\frac{1}{2} \times 10^{-2}$ as above. 

\nit {\it Assertion (1).} 
For each $\by^\star, \bz^\star\in\Gamma_\delta$, we confirm by computer that the following inequality is true: 
There exists $l_0\in\{0,1,\ldots,5\}$ such that, 
for some (tiny) positive constant $c^\prime$, we have
\begin{equation*}
 |\bz^\star+\be_{A,1}-\by^\star| - |\bz^\star-R^{l_0}_{_{60}} \by^\star| > \frac{4\delta}{\sqrt{2}} + c^\prime |\be_{A_,1}|.
\end{equation*}
Now let $|\bz|, |\by| < r_0 $. Let $\bz^\star, \by^\star \in \Gamma_\delta$ be as in Remark \ref{continuum_discrete} above, {\it i.e.}
$|\bz-\bz^\star|, |\by-\by^\star| \leq \delta/\sqrt{2}$. Then $|\bz^\star+\be_{A,1}-\by^\star|$ differs from $|\bz+\be_{A,1}-\by|$ by at most $2\delta/\sqrt{2}$, and 
$|\bz^\star-R^{l_0}_{_{60}}\by^\star|$ differs from $|\bz-R^{l_0}_{_{60}}\by|$ by at most $2\delta/\sqrt{2}$. 
Therefore, $|\bz^\star+\be_{A,1}-\by^\star| - |\bz^\star-R^{l_0}_{_{60}}\by^\star|$ differs from 
$|\bz+\be_{A,1}-\by| - |\bz-R^{l_0}_{_{60}}\by|$ by at most $4\delta/\sqrt{2}$.
Therefore, $|\bz+\be_{A,1}-\by| - |\bz-R^{l_0}_{_{60}}\by| > c^\prime |\be_{A,1}|$, proving inequality \eqref{euclid-est1} and assertion (1) of the lemma.

\nit {\it Assertion (2).} 
We prove assertion (2) without a computer. 
Firstly, observe that if $\bfm=(0,0)$, then we may satisfy inequality \eqref{euclid-est2} by choosing $l_0=0$.
Next, let $0 < |\bfm| \leq 10^6$. Recall that $\bz,\by \in B({\bf 0},r_0)$: $|\bz|, |\by| < r_0$.
We can pick $l_0$ so that $\bz$ and $R^{l_0}_{_{60}}\by$ lie in the same $60^{\circ}-$ sector in the disc of radius $r_0$ about the origin.
Therefore, $|\bz-R^{l_0}_{_{60}}\by| \leq r_0 $; so 
\begin{align*}
|\bz-\bfm\vec\bv-\by| - |\bz-R^{l_0}_{_{60}}\by| &\geq |\bfm\vec\bv| - |\bz| - |\by| - |\bz-R^{l_0}_{_{60}}\by| \\
&\geq |\bfm\vec\bv| - 3 r_0  \\
&\geq \sqrt{3}|\be_{A,1}| - 3 r_0  \qquad (\text{because } |\bfm\vec\bv| \geq 1 = \sqrt{3}|\be_{A,1}|)  \\
&> \frac{1}{2} |\be_{A,1}|  \qquad (\text{because } 3 r_0 < |\be_{A,1}| \text{ and } \sqrt{3} - 1 \geq \frac{1}{2})  \\
& \geq c^{\prime\prime} |\bfm| ,
\end{align*}
where we may take $c^{\prime\prime} \equiv \frac{1}{2} |\be_{A,1}| 10^{-6}$ because $|\bfm| \leq 10^6$.
Finally, if $|\bfm| > 10^6$, it is obvious that inequality \eqref{euclid-est2} holds. 
So assertion (2) of the lemma holds.

\nit {\it Assertion (3).} 
We prove assertion (3) without a computer.
We claim that
\begin{equation}
 \label{claim_assert_3}
\text{for } \bfm\in\Z^2, \text{ if } \bfm \notin N_{bad}(\be_{A,1}) \text{ then } |\be_{A,1} + \bfm\vec\bv| > |\be_{A,1}| + 3 r_0  + \eps. 
\end{equation}
We verify \eqref{claim_assert_3} below, but first show how we use it to prove assertion (3). By picking $l_0$ so that $|\bz-R^{l_0}_{_{60}}\by| \leq r_0$, we have
\begin{align*}
 &|\bz+(\be_{A,1}+\bfm\vec\bv)-\by| - |\bz + \be_{A,1} - R^{l_0}_{_{60}} \by | \\
 &\quad \geq
 \left( |\be_{A,1}+\bfm\bv| -|\bz| - |\by| \right) - \left( |\be_{A,1}| + |\bz - R^{l_0}_{_{60}}\by| \right) \\
 &\quad \geq |\be_{A,1}+\bfm\bv| - |\be_{A,1}| - 3 r_0  > \eps \geq c^{\prime\prime\prime} |\bfm|,
\end{align*}
provided $|\bfm|$ is bounded, by say $|\bfm| \leq 10^6$, 
in which case we may take $c^{\prime\prime\prime} \equiv \eps\ 10^{-6}$.
Finally, if $|\bfm| > 10^6$, then inequality \eqref{euclid-est3} clearly holds. 
Thus, assertion (3) holds provided \eqref{claim_assert_3} is true.

We proceed to verify \eqref{claim_assert_3}. Because $r_0=0.33|\be_{A,1}|$ and $\eps=10^{-8}$, the right hand side of \eqref{claim_assert_3} satisfies 
$|\be_{A,1}| + 3 r_0  + \eps = 1.99 |\be_{A,1}| + \eps < 2|\be_{A,1}|$.
Therefore, if we can show that the set
$\{ \bfm\in\Z^2 : \bfm \notin N_{bad}(\be_{A,1})$ and $|\be_{A,1} + \bfm\vec\bv| < 2 |\be_{A,1}| \}$ is empty, we will have verified \eqref{claim_assert_3}.
Refer now to Figure \ref{fig:fundamental-cell}. For simplicity, consider centering the coordinates in the figure about the blue lattice point $\bv_A$. It follows that, for any $\bfm\in\Z^2$, $\be_{A,1} + \bfm\vec\bv$ is a red lattice point. The only red lattice points that satisfy $|\be_{A,1} + \bfm\vec\bv| < 2 |\be_{A,1}|$ are the three red lattice points closest to the origin (at $\bv_A$), which satisfy $|\be_{A,1}+\bfm\vec\bv|=|\be_{A,\mu}|$, for $\mu=1,2,3$. But this is exactly the condition defining the set $N_{bad}(\be_{A,1})$, and therefore, because 
$\bfm \notin N_{bad}(\be_{A,1})$, these three red lattice points are excluded. This completes the proof of the claim \eqref{claim_assert_3} and assertion (3).

\nit {\it Assertion (4).} 
To prove assertion (4), we confirm by computer that the following is true: 
For each $\by^\star, \bz^\star\in\Gamma_\delta$ and each nonzero $\bn\in\Z^2$ such that $|\bn\vec\bv| \leq |\be_{A,1}|+3r_0+\eps$,
we check that there exists $l_0\in\{0,1,\ldots,5\}$ such that
\begin{equation}
\label{comp_sim_4}
 |\bz^\star+\bn\vec\bv-\by^\star| - |\bz^\star+\be_{A,1}-R^{l_0}_{_{60}} \by^\star| > \frac{4\delta}{\sqrt{2}} + \eps .
\end{equation}
It is the verification of \eqref{comp_sim_4} by computer that imposes the strongest constraint on $r_0$, and forces us to choose $r_0$ equal to or very slightly larger than $0.33|\be_{A,1}|$.

For each $\bz,\by\in\R^2$ with $|\bz|,|\by|<r_0$, and for each nonzero $\bn\in\Z^2$ satisfying $|\bn\vec\bv| \leq |\be_{A,1}|+3r_0+\eps$ (as in the computer run), we pick $\bz^\star,\by^\star\in\Gamma_\delta$ such that $|\bz^\star-\bz|, |\by^\star-\by| \leq \delta/\sqrt{2}$. 
Because $|\bz^\star+\bn\vec\bv-\by^\star| - |\bz^\star+\be_{A,1}-R^{l_0}_{_{60}} \by^\star|$ differs from
$|\bz+\bn\vec\bv-\by| - |\bz+\be_{A,1}-R^{l_0}_{_{60}} \by|$
by at most $4\delta/\sqrt{2}$, we have (for some $l_0$):
\begin{equation*}
 |\bz+\bn\vec\bv-\by| - |\bz+\be_{A,1}-R^{l_0}_{_{60}} \by| \geq \eps \geq c^{\prime\prime\prime\prime}|\bn|;
\end{equation*}
the last inequality holds because the family of $\bn\in\Z^2$ arising in the computer run is bounded.

On the other hand, let $|\bz|, |\by| < r_0 $, and suppose $|\bn\vec\bv| > |\be_{A,1}|+3r_0+\eps$ but $|\bn| \leq 10^6$.
Then we pick $l_0$ such that $|\bz-R^{l_0}_{_{60}}\by| \leq r_0 $, and we have 
\begin{align*}
 |\bz+\bn\vec\bv-\by| - |\bz+\be_{A,1}-R^{l_0}_{_{60}} \by| &\geq 
 \left( |\bn\vec\bv|-|\bz|-|\by| \right) - \left( |\be_{A,1}|+|\bz-R^{l_0}_{_{60}}\by| \right) \\
 &\geq |\bn\vec\bv| - |\be_{A,1}| -3 r_0  \geq \eps \geq c^{\prime\prime\prime\prime}|\bn| ,
\end{align*}
the last inequality holding because we assumed that $|\bn| \leq 10^{6}$. Here we may take $c^{\prime\prime\prime\prime}=\eps\ 10^{-6}$.

Finally, if $|\bn| > 10^6$, then obviously (for any $l_0\in\{0,1,\ldots,5\}$),
\begin{align*}
 |\bz+\bn\vec\bv-\by| - |\bz+\be_{A,1}-R^{l_0}_{_{60}} \by| &\geq 
 \left( |\bn\vec\bv|-|\bz|-|\by| \right) - \left( |\be_{A,1}|+|\bz| + |\by| \right) \\
 &\geq |\bn\vec\bv| - |\be_{A,1}| - 4r_0 \geq c^{\prime\prime\prime\prime} |\bn|;
\end{align*}
the last inequality holds because $|\be_{A,1}| + 4 r_0 \leq \frac{1}{2} |\bn\vec\bv|$ for $|\bn|>10^6$. This completes the proof of assertion (4) and therewith Lemma \ref{euclid1}.
\end{proof}

Consider the eigenvalue problem satisfied by the ground state eigenfunction, $p_0^\lambda(\bx)$, with corresponding simple eigenvalue, $E_0^\lambda$:
 \begin{equation*}
\left(\ -\Delta_\bx+\lambda^2V_0(\bx) - E_0^\lambda\ \right) p_0^\lambda(\bx)=0,\ \ p_0^\lambda\in L^2(\R^2).
\end{equation*}
This may be rewritten as
$
\left(\ -\Delta_\bx + |E_0^\lambda|\ \right) p_0^\lambda(\bx)\ =\ \lambda^2 |V_0(\bx)|\ p^\lambda_0(\bx),\ \ \bx\in\R^2,
$
and therefore, 
\begin{equation}
 p_0^\lambda(\bx)\ =\  \int\ \mathcal{K}_\lambda\left(\bx-\by\right)\ \lambda^2 |V_0(\by)|\ p^\lambda_0(\by)\ d\by,\ \  \bx\in\R^2,
\label{gs2}\end{equation}
where $\mathcal{K}_\lambda(\bx)=\ \mathcal{K}\left(\sqrt{|E_0^\lambda|}\bx\right)$.
Here,  $\mathcal{K}(\bx)$ is the fundamental solution for $-\Delta_\bx+1$ satisfying
$\left(\ -\Delta_\bx+1\ \right)\ \mathcal{K}(\bx)=\delta(\bx),\ \bx\in\R^2$,
$\delta(\bx)$ is the Dirac delta function, and 
 $\mathcal{K}=K_0(|\bx|)$ is the modified Bessel function of order zero, which decays to zero exponentially  as $|\bx|\to\infty$ \cites{WW:02}.

 An alternative representation to \eqref{gs2} for $p_0^\lambda(\bx)$, which we find useful, is obtained as follows. 
Note from \eqref{gs2} that $p_0^\lambda$ is a convolution with $\lambda^2|V_0|p_0^\lambda$, a function which  is non-negative, supported in a disc of radius $r_0$ about ${\bf 0}$, and is invariant under a $60^\circ$ rotation about ${\bf 0}$. (The latter is a consequence of the $120^\circ$ rotational symmetry and inversion symmetry of $V_0$.)
Thus, in addition to \eqref{gs2} we have
\begin{equation}
 p_0^\lambda(\bz)\ =\    \int\ \left[\frac16\ \sum_{l=0}^5\mathcal{K}_\lambda\left(\bz-R^l_{_{60}}\by\right)\right]\ \lambda^2 |V_0(\by)|\ p^\lambda_0(\by)\ d\by,\ \  \bz\in\R^2,
\label{gs2-symm}\end{equation}
where $R_{_{60}}$ is the rotation by $60^\circ$. 

\begin{lemma}[Properties of $\mathcal{K}(\bx)$]\label{K0-properties} For $\bx\in\R^2$, 
%
%
\begin{enumerate}
\item $\mathcal{K}(\bx)=\mathcal{K}(|\bx|)$ is positive and strictly decreasing for $|\bx|\ge0$.
\item There exist entire functions $f$ and $g$ and constants $C_1, C_2$,  such that
\[ \mathcal{K}(\bx)\ =\ f(|\bx|) \log |\bx|\ +\ g(|\bx|)\ ,\]
where $f(0)=-1/2\pi$ and $|f(s)|, |g(s)|\le C_1 e^{-C_2s}$, for all $s\in[0,\infty)$.
\item $\mathcal{K}(\bx)\ \lesssim |\bx|^{-\frac12} e^{-|\bx|}$ for $|\bx|$ large. 
\item For $\bx' , \bx'' \in\R^2 $ such that $|\bx'|>|\bx''|$,  we have
\begin{equation}
\mathcal{K}(\bx')\ \lesssim\ e^{-[\ |\bx'|-|\bx''|\ ]}\ \mathcal{K}(\bx''),\label{Kxpxpp}
\end{equation} 
  \end{enumerate}
  \end{lemma}
  
  \begin{remark}[$\mathcal{K}$ and $\mathcal{K}_\lambda$]\label{KtoKlam}
 Since for large $\lambda$: $|E_0^\lambda|\approx\lambda^2$,
  by \eqref{Kxpxpp} we have:
 \begin{equation*}
\mathcal{K}_\lambda(\bx')\ \lesssim\ e^{-c\lambda[\ |\bx'|-|\bx''|\ ]}\ \mathcal{K}_\lambda(\bx''),\ \ {\rm for}\ \ 
|\bx'|>|\bx''|.
\end{equation*}
\end{remark} 

\begin{proof}[Proof of Lemma \ref{K0-properties}] Recall that $(4\pi |\bz|)^{-1}\exp(-|\bz|)$ is the fundamental solution for $-\Delta_\bz+1$ on $\R^3$, {\it i.e.} $(-\Delta_\bz+1)(4\pi |\bz|)^{-1}\exp(-|\bz|)=\delta(\bz)$.  Let $\bz=(\bx,t)=(x_1,x_2,t)$. Integrating against $dt$ over $\R$, we obtain that the fundamental solution of $-\Delta_\bx+1$ on $\R^2$ is given by:
 \[ \mathcal{K}(\bx)= \frac1{4\pi}\ \int_{-\infty}^\infty \frac{e^{-(|\bx|^2+t^2)^{\frac12}}}{(|\bx|^2+t^2)^{\frac12}}\ dt,\  \ \bx\in\R^2.  \]
Next, introduce the change of variables $(|\bx|+\zeta)^2=|\bx|^2+t^2$. Then,
 \begin{equation}
  \mathcal{K}(\bx)= \frac1{2\pi}\ \int_{0}^\infty \frac{e^{-\zeta}}{\zeta^{\frac{1}{2}}(2|\bx|+\zeta)^{\frac12}}\ d\zeta\times e^{-|\bx|}\ . 
 \label{Krep} \end{equation}
 Part (1) of the lemma follows since  the expression in \eqref{Krep} is clearly positive and decreasing as a function of $|\bx|$.
 Part (2) of the lemma was proved
 in \cites{Simon:76} (Lemma 3.1). Parts (3) and (4) are immediate consequences of \eqref{Krep}.
 \end{proof}
 
 A consequence of Lemma \ref{K0-properties} is the following exponential decay bound on $p_0^\lambda(\bx)$.
\begin{corollary}\label{cor:expo-decay}
Assume $\supp(V_0)\subset B(0,r_0)$ with $r_0>0$, and let $c_0>0$ denote any positive constant.
There exist positive constants $C_1, C_2$ and $c_1$, which depend on $V_0$, $r_0$ and $c_0$, such that 
$p_0^\lambda(\bx)$ satisfies the bound:
 \begin{align*}
 p_0^\lambda(\bx) &\le
  \begin{cases} 
C_1\ e^{-c_1\lambda|\bx|}, & |\bx|\ge r_0+c_0\\
C_2\  \lambda,& |\bx|< r_0+c_0 .
  \end{cases}
  \end{align*}
  \end{corollary}
  \begin{proof}[Proof of Corollary \ref{cor:expo-decay}]  Take $\by\in\supp(V_0)$ and $|\bx|\ge r_0+c_0$. Then, 
   $|\bx-\by|\ge c_0$ and hence $\lambda|\bx-\by|\ge \lambda c_0$. By part (3) of Lemma \ref{K0-properties}, 
   there exists $\lambda_\star>0$, which depends on $c_0$, such that for all $\lambda\ge\lambda_\star$ we have
    $\mathcal{K}_\lambda(\bx-\by) = \mathcal{K}\left(|E_0^\lambda|(\bx-\by)\right)\le (c\lambda|\bx-\by|)^{-\frac12} e^{-c\lambda|\bx-\by|}\le
    ( c c_0\lambda)^{-\frac12} e^{-c\lambda|\bx-\by|}$. Estimating $|p_0^\lambda(\bx)|$ using the  integral equation \eqref{gs2}, the above bound on $\mathcal{K}_\lambda(\bx-\by)$,  the Cauchy-Schwarz inequality and 
    $\|p_0^\lambda\|_{_{L^2}}=1$,  we obtain:
    \[ |p_0^\lambda(\bx)|\ \le\ \lambda^2\ \|V_0\|_{_{L^\infty}}\ (cc_0\lambda )^{-\frac12}\  \left[\ \int_{|\by|\le r_0}  e^{-2c\lambda|\bx-\by|}\ d\by\ \right]^{\frac12}\ .\]
  Note that  
  \begin{align}
  |\bx-\by|\ge|\bx|-|\by|\ge |\bx|-r_0\ge |\bx|-r_0\ \frac{|\bx|}{r_0+c_0}= \left(\frac{c_0}{r_0+c_0}\right)|\bx|.
 \nn \end{align}
    Therefore, for all $\bx$ such that $|\bx|\ge r_0+c_0$:
     \[ |p_0^\lambda(\bx)|\ \le\ \lambda^2\ \|V_0\|_{_{L^\infty}}\ (cc_0\lambda )^{-\frac12}\  e^{-c\frac{c_0}{r_0+c_0}\lambda|\bx|}\left[\ \pi r_0^2 \right]^{\frac12}.
     \]
      For $0<|\bx|\le r_0+c_0$, we use \eqref{gs2} and the Cauchy-Schwarz inequality to get
     $p_0^\lambda(\bx)\le \|V_0\|_{_{L^\infty}}\ \lambda^2\ \|\mathcal{K}_\lambda\|_{_{L^2}}=\lambda\ \|V_0\|_{_{L^\infty}}\ \|\mathcal{K}_1\|_{_{L^2}}$.  This completes the proof of Corollary \ref{cor:expo-decay}.
    \end{proof}

 The following bounds on $p_0^\lambda$ are used in completing the proofs of Propositions \ref{ME-BA1}, \ref{ME-AA1}
  and \ref{ME-IJ-higher-order}.

  \begin{lemma}\label{lem:p0-lam-bounds}
  Recall that $N_{bad}(\be_{I,\nu})$ denotes the set of $\bn\in\Z^2$ such that $|\be_{I,\nu}+\bn\vec\bv|=|\be_{I,\mu}|$, for $\mu=1,2,3$.
 There exists a constant $c$ such that for $\by\in {\rm supp}(V_0) \subset B({\bf 0},r_0)$, {\it i.e.} $|\by|\le r_0$, we have
  for $I=A,B$ and $\nu=1,2,3$:
  \begin{align}
 p_0^\lambda(\by-\bfm\vec\bv)\  &\lesssim\ e^{-c|\bfm|\lambda}\ p_0^\lambda(\by)  ,   \label{p0-bound1} \\
p^\lambda_0\left(\by+\be_{I,\nu}+\bn\vec\bv \right)\ &\lesssim\ 
  e^{-c|\bn|\lambda}\ p^\lambda_0\left(\by+\be_{I,\nu}\right),\ \ \bn \notin N_{bad}(\be_{I,\nu})  ,
  \label{p0-bound2}\\
 p_0^\lambda(\by+\be_{I,\nu})\ &\lesssim\ e^{-c\lambda}\ p_0^\lambda(\by) ,  \quad {\rm and}
 \label{p0-bound3}\\
 p_0^\lambda(\by-\bn\vec\bv)\ &\lesssim\ e^{-c\lambda|\bn|}\ p_0^\lambda(\by+\be_{I,\nu}) ,\ \bn\ne(0,0).
 \label{p0-bound4}
 \end{align}
  \end{lemma}
  
  \begin{proof}[Proof of Lemma \ref{lem:p0-lam-bounds}]
   We first prove the bound \eqref{p0-bound1}. 
  Applying part (2) of geometric Lemma \ref{euclid1} we obtain, for all $|\bz|,\ |\by|\le r_0$ and all $\bfm\in\Z^2$,  an $l_0$ such that 
 \begin{equation*}
 |\bz+\bfm\vec\bv-\by|\ -\ |\bz-R^{l_0}_{_{60}}\by|\ \geq \ c^{\prime\prime}\ |\bfm|\ .
 \end{equation*}
 for some $l_0$.
 %
%
Therefore, by Remark \ref{KtoKlam}, just after Lemma \ref{K0-properties},
 we have for all $|\bz|,\ |\by|\le r_0$
 \begin{align*}
& \mathcal{K}_\lambda(|\bz-\bfm\vec\bv-\by|)\lesssim e^{-cc''|\bfm|\lambda} 
  \mathcal{K}_\lambda(|\bz-R^{l_0}_{_{60}}\by|)\ \lesssim e^{-cc''|\bfm|\lambda} \frac16\sum_{l=0}^5 \mathcal{K}_\lambda(|\bz-R^{l}_{_{60}}\by|) .
 \end{align*} 
 
 Therefore, by \eqref{gs2} and \eqref{gs2-symm},
 \begin{align*}
 p_0^\lambda(\bz-\bfm\vec\bv)\ &=\ \int\ \mathcal{K}_\lambda\left(\bz-\bfm\vec\bv-\by\right)\ \lambda^2 |V_0(\by)|\ p^\lambda_0(\by)\ d\by\nn\\
 &\lesssim\ 
e^{-c|\bfm|\lambda}\ \int\ \frac16\sum_{l=0}^5 \mathcal{K}_\lambda(|\bz-R^{l}_{_{60}}\by|)\ \lambda^2 |V_0(\by)|\ p^\lambda_0(\by)\ d\by\nn\\
&=\ e^{-c|\bfm|\lambda}\ p_0^\lambda(\bz),\ \ |\bz|\le r_0.
\end{align*}

Next we turn to the proof of the bound \eqref{p0-bound2}. In a manner similar to the proof of part \eqref{p0-bound1}, we have by part (3) of Geometric Lemma \ref{euclid1} that for all $|\bz|, |\by|\le r_0$ and all $\bn\in\Z^2\setminus N_{bad}(\be_{I,\nu})$, there exists $l_0$ such that 
\begin{equation}
\lambda |\bz +(\be_{I,\nu}+\bn\vec\bv)-\by|\ \ -\  \lambda|\bz+\be_{I,\nu}-R^{l_0}_{_{60}}\by|\ \geq \ \lambda c^{\prime\prime\prime}\ |\bn|\ .
 \label{euclid-est3a}
 \end{equation}
  By Remark \ref{KtoKlam} we have for all $|\bz|,\ |\by|\le r_0$,
 \[ \mathcal{K}_\lambda(\bz+(\be_{I,\nu}+\bn\vec\bv)-\by)\ \lesssim\ e^{-cc^{\prime\prime\prime} |\bn| \lambda }\
 \frac{1}{6}\sum_{l=0}^5 \mathcal{K}_\lambda(\bz+\be_{I,\nu}-R^{l}_{_{60}}\by)\ .
  \]
  By \eqref{gs2} and \eqref{gs2-symm},
 \begin{align*}
p_0^\lambda\left(\bz+(\be_{I,\nu}+\bn\vec\bv)\right)
& \lesssim\ 
e^{-c |\bn|\lambda}\ \int\ \frac16\sum_{l=0}^5 \mathcal{K}_\lambda(\bz+\be_{I,\nu}-R^{l}_{_{60}}\by)\ \lambda^2 |V_0(\by)|\ p^\lambda_0(\by)\ d\by\nn\\
& =\ e^{-c |\bn|\lambda}\ p_0^\lambda(\bz+\be_{I,\nu}),\ \ |\bz|\le r_0.
\end{align*}
Thus, bound \eqref{p0-bound2} holds. 

To prove bound \eqref{p0-bound3} we have by part (1) of Lemma \ref{euclid1} that, given $|\bz|, |\by|\le r_0$,
 there exists $l_0$ such that:
 \begin{equation}
\lambda |\bz+\be_{I,\nu}-\by|\ \ -\  \lambda|\bz-R^{l_0}_{_{60}}\by|\ \geq \ \lambda c^{\prime}\ |\be_{I,\nu}| \ .
 \label{euclid-est1a}
 \end{equation}
 By Remark \ref{KtoKlam}, 
 \[ \mathcal{K}_\lambda(\bz+\be_{I,\nu}-\by)\ \lesssim\ e^{-cc^{\prime} |\be_{I,\nu}| \lambda }\
 \frac{1}{6}\sum_{l=0}^5 \mathcal{K}_\lambda(\bz-R^{l}_{_{60}}\by),\ \ {\rm for}\ \ |\bz|,\ |\by|\le r_0.
  \]
  By  \eqref{gs2},
 \begin{align*}
p_0^\lambda(\bz+\be_{I,\nu})
&\lesssim\ 
e^{-c |\be_{I,\nu}|\lambda}\ \int\ \frac16\sum_{l=0}^5 \mathcal{K}_\lambda(\bz-R^{l}_{_{60}}\by)\ \lambda^2 |V_0(\by)|\ p^\lambda_0(\by)\ d\by\nn\\
& =\ e^{-c |\be_{I,\nu}|\lambda}\ p_0^\lambda(\bz),\ \ |\bz|\le r_0. 
\end{align*}
Thus, bound \eqref{p0-bound3} holds.

Finally, to prove \eqref{p0-bound4} we have by part (4) of Lemma \ref{euclid1} that, $|\bz|, |\by|\le r_0$,
 there exists $l_0$ such that:
 \begin{equation*}
\lambda |\bz-\bn\vec\bv-\by|\ -\ \lambda |\bz+\be_{I,\nu}-R^{l_0}_{_{60}}\by|\ \geq \  \lambda c^{\prime\prime\prime\prime}\ |\bn|\ ,\ \ \bn\ne(0,0).
 \end{equation*}
By Remark \ref{KtoKlam}, 
 \[ \mathcal{K}_\lambda(\bz-\bn\vec\bv-\by)\ \lesssim\ e^{-cc^{\prime\prime\prime\prime} |\bn| \lambda }\
 \frac{1}{6}\sum_{l=0}^5 \mathcal{K}_\lambda(\bz+\be_{I,\nu}-R^{l_0}_{_{60}}\by),\ \ {\rm for}\ \ |\bz|,\ |\by|\le r_0 ,  \]
and hence, by \eqref{gs2},
\begin{align*}
p_0^\lambda(\bz-\bn\vec\bv)\ & \lesssim\ 
e^{-c |\bn|\lambda}\ \int\ \frac16\sum_{l=0}^5 \mathcal{K}_\lambda(\bz+\be_{I,\nu}-R^{l}_{_{60}}\by)\ \lambda^2 |V_0(\by)|\ p^\lambda_0(\by)\ d\by\nn\\
&\qquad =\ e^{-c |\bn|\lambda}\ p_0^\lambda(\bz+\be_{I,\nu}),\ \ |\bz|\le r_0. \nn\\
\end{align*}
The proof of Lemma \ref{lem:p0-lam-bounds} is complete.
\end{proof}

 \subsection{ Completion of the proofs Proposition \ref{ME-BA1} and  Proposition \ref{ME-AA1} }
 %
 
 %
%
%
%
To prove estimate  \eqref{ip-BA-bound} of Proposition \ref{ME-BA1}  we bound the five sums in \eqref{1st-sum-bound}
 and \eqref{2nd-ds-breakdown-est}, and to prove estimate \eqref{ip-AA-bound}  of Proposition \ref{ME-AA1} we bound the two sums appearing in \eqref{IAA-1} and \eqref{IAA-2}. The full list of seven expressions to be bounded is as follows (recall $\bv_A={\bf 0}$, and note that all sums are over subsets of $\Z^2$):
\begin{align*}
 \mathscr{I}_{1} &= \sum_{\substack{\bfm\ne(0,0)\\ 1\le\nu\le3}}\ e^{C\lambda^{-1}|\bfm|}\ \int |V_0(\by)|\ p_0^\lambda(\by-\be_{A,\nu})\ p_0^\lambda(\by-\bfm\vec\bv)\ d\by, \\ 
  \mathscr{I}_2 &= \sum_{\substack{\bfm\ne (0,0) \\ \bn\ne(0,0), (0,-1), (-1,0)} }
\ e^{C\lambda^{-1}|\bfm-\bn|}\  \int |V_0(\by)|\ p_0^\lambda(\by-\be_{A,1}-\bn\vec\bv)\ p_0^\lambda(\by-\bfm\vec\bv)\ d\by , \\ 
 \mathscr{I}_3  &= \sum_{\bn\ne (0,0), (1,0), (0,1)}\ e^{C\lambda^{-1}|\bn|}\ \int |V_0(\by)|\ p_0^\lambda(\by)\ 
 p_0^\lambda(\by-\be_{B,1}-\bn\vec\bv)\ d\by ,  \\ 
  \mathscr{I}_4 &= \sum_{\substack{\bfm\ne(0,0)\\ 1\le\nu\le3}}\ e^{C\lambda^{-1}|\bfm|}\ \int |V_0(\by)|\ p_0^\lambda(\by-\bfm\vec\bv)\ p_0^\lambda(\by-\be_{B,\nu})\ d\by,  \\ 
  \mathscr{I}_5 &= \sum_{\substack{\bfm\ne(0,0) \\ \bn\ne (0,0), (1,0), (0,1) } }\ e^{C\lambda^{-1}|\bn-\bfm|}\ \int |V_0(\by)|\ p_0^\lambda(\by-\bfm\vec\bv)\ p_0^\lambda(\by-\be_{B,1}-\bn\vec\bv)\ d\by , \\ 
  \mathscr{I}_6 &= \sum_{\substack{\bn\ne(0,0)} } \sum_{\substack{\bfm } }\ e^{C\lambda^{-1}|\bn-\bfm|}\ \int |V_0(\by)|\ p_0^\lambda(\by-\bfm\vec\bv)\ p_0^\lambda(\by-\bn\vec\bv)\ d\by , \\ 
  \mathscr{I}_7 &= \sum_{\substack{\bn,\bfm } }\ e^{C\lambda^{-1}|\bn-\bfm|}\ \int |V_0(\by)|\ 
  p_0^\lambda(\by-\be_{B,1}-\bfm\vec\bv)\ p_0^\lambda(\by-\be_{B,1}-\bn\vec\bv)\ d\by . \\ 
 \end{align*}
\noindent  In particular, we prove
  \begin{proposition} \label{BA-AA-control} There exists a positive constants $\tilde{c}$ and $\lambda_\star$,  such that for all $\lambda>\lambda_\star$,
\begin{equation}
 | \mathscr{I}_j |\ \lesssim\ \rho_\lambda\times e^{-\tilde{c} \lambda}, \ \ 1\le j\le7. 
 \label{I123-bound}
 \end{equation}
 (In $ \mathscr{I}_j,\ 1\le j\le7$, we have dropped the factor of $\lambda^2$ multiplying $|V_0|$, since this 
 can be absorbed by adjusting the constant $\tilde{c}$ in the exponential factor in \eqref{I123-bound}.)
 \end{proposition}

 \begin{proof}[Proof of Proposition \ref{BA-AA-control}] We proceed to estimate $ \mathscr{I}_j,\ 1\le j\le7$, using 
  the bounds of Lemma \ref{lem:p0-lam-bounds}, and that $V_0(\by)$ and $p_0^\lambda(\by)$ are even functions.
 
 \noindent{\it Estimation of  $| \mathscr{I}_1 |$}:  
By estimate \eqref{p0-bound1} of Lemma \ref{lem:p0-lam-bounds}, we have that $p_0^\lambda(\by-\bfm\vec\bv) \lesssim \ e^{-c|\bfm|\lambda}\ p_0^\lambda(\by)$ for $\by\in{\rm supp}\ V_0$. Therefore, for $\nu=1,2,3$:
\begin{align*}
&\int |V_0(\by)|\ p_0^\lambda(\by-\be_{A,\nu})\ p_0^\lambda(\by-\bfm\vec\bv)\ d\by\nn\\
&\quad \lesssim\ e^{-c|\bfm|\lambda}\ \int |V_0(\by)|\ p_0^\lambda(\by-\be_{A,\nu})\ p_0^\lambda(\by)\ d\by\ \le\ 
e^{-c|\bfm|\lambda}\times\rho_\lambda .
\end{align*}
Next, multiplying by $e^{C\lambda^{-1}|\bfm|}$ and summing over $\bfm\in\Z^2\setminus\{(0,0)\}$ gives the bound $\left|\mathscr{I}_1\right|\ \lesssim\ \rho_\lambda\times e^{-\tilde{c} \lambda}$.

\noindent{\it Estimation of  $| \mathscr{I}_2 |$}:  
 The strategy is similar to that used to bound $| \mathscr{I}_1 |$. 
 By \eqref{p0-bound1}  and \eqref{p0-bound2} we have
  \begin{align*}
 &  \int |V_0(\by)|\ p_0^\lambda(\by-\be_{A,1}-\bn\vec\bv)\ p_0^\lambda(\by-\bfm\vec\bv)\ d\by\nn\\
 &\quad\lesssim\ e^{-c|\bn|\lambda}\times\ e^{-c|\bfm|\lambda}\times 
  \int |V_0(\by)|\ p_0^\lambda(\by-\be_{A,1})\ p_0^\lambda(\by)\ d\by\nn\\
 &\quad \lesssim\ e^{-c|\bn|\lambda}\times\ e^{-c|\bfm|\lambda}\times\ \rho_\lambda ,
 \nn 
\end{align*}
 for $ \bfm\ne(0,0)$ and $\bn\notin N_{bad}(\be_{A,1})=\{(0,0),(-1,0),(0,-1)\}$.
 Multiplying by $e^{C\lambda^{-1}|\bfm-\bn|}$ and summing over 
 $ \bfm\ne(0,0)$ and $\bn\notin N_{bad}(\be_{A,1})$, we obtain the bound $|\mathscr{I}_2|\lesssim  e^{-\tilde{c}\lambda}\ \rho_\lambda$. 
 
 \noindent{\it Estimation of  $| \mathscr{I}_3 |$}: 
 To bound $|\mathscr{I}_3|$, we apply the bound \eqref{p0-bound2}
 to obtain, for $\bn\ne (0,0),(1,0),(0,1)$:
 \begin{align*}
& \int |V_0(\by)|\ p_0^\lambda(\by)\ p_0^\lambda(\by-\be_{B,1}-\bn\vec\bv)\ d\by\nn\\
 &\ \lesssim\ e^{-c|\bn|\lambda}\times\ \int  |V_0(\by)|\ p_0^\lambda(\by)\ p_0^\lambda(\by-\be_{B,1})\ d\by
\ \lesssim\  e^{-c|\bn|\lambda}\times\ \rho_\lambda.
\end{align*}
Multiplying by $e^{C\lambda^{-1}|\bn|}$ and summing over $\bn\notin\{(0,0),(1,0),(0,1)\}$ we obtain that 
 $| \mathscr{I}_3 |\lesssim e^{-\tilde{c}\lambda}\ \rho_\lambda$. 

 \noindent{\it Estimation of  $| \mathscr{I}_4 |$}: 
Because $\be_{B,\nu}=-\be_{A,\nu}$, $\nu=1,2,3$, and $p_0^\lambda(-\by)=p_0^\lambda(\by)$, we have that $\mathscr{I}_4= c\ \mathscr{I}_1$ for some constant $c$. Therefore the bound for 
$|\mathscr{I}_4|$ follows from $| \mathscr{I}_1|$: $|\mathscr{I}_4|\lesssim  e^{-\tilde{c}\lambda}\ \rho_\lambda$.

 \noindent{\it Estimation of  $| \mathscr{I}_5 |$}: 
To bound $| \mathscr{I}_5 |$ we
 apply  \eqref{p0-bound1} and \eqref{p0-bound2} to obtain, for $\bfm\ne(0,0)$ and $\bn\ne (0,0), (1,0), (0,1) $:
 \begin{align*}
 &\int |V_0(\by)|\ p_0^\lambda(\by-\bfm\vec\bv)\ p_0^\lambda(\by-\be_{B,1}-\bn\vec\bv)\ d\by\nn\\ 
&\ \lesssim \ e^{-c\lambda|\bfm|}\times e^{-c\lambda|\bn|}\times \int |V_0(\by)|\ p_0^\lambda(\by-\be_{B,1})\ 
p_0^\lambda(\by)\ d\by
\lesssim \ e^{-c\lambda|\bfm|}\times e^{-c\lambda|\bn|}\times\ \rho_\lambda .
 \end{align*}
 Multiplying by $e^{C\lambda^{-1}|\bn-\bfm|}$ and summing over all $\bfm\ne(0,0)$ and $\bn\ne (0,0), (1,0), (0,1) $, we obtain that
 $| \mathscr{I}_5 | \lesssim e^{-\tilde{c}\lambda}\times \rho_\lambda$.  

  \noindent{\it Estimation of  $| \mathscr{I}_6 |$}:  To bound $| \mathscr{I}_6 |$ we apply  \eqref{p0-bound1} and \eqref{p0-bound4} to obtain, for $\bn\ne(0,0)$ and $\bfm\in\Z^2$:
  \begin{align*}
 & \int |V_0(\by)|\ p_0^\lambda(\by-\bfm\vec\bv)\ p_0^\lambda(\by-\bn\vec\bv)\ d\by \\
  &\quad \lesssim \ e^{-c\lambda|\bfm|}\times e^{-c\lambda|\bn|} \int |V_0(\by)|\ p_0^\lambda(\by)\ p_0^\lambda(\by-\be_{B,1})\ d\by\nn  .
 \end{align*}
 Multiplying by $e^{C\lambda^{-1}|\bn-\bfm|}$ and summing over all $\bn\ne(0,0)$ and $\bfm\in\Z^2$, we obtain that
 $| \mathscr{I}_6 | \lesssim e^{-\tilde{c}\lambda}\times \rho_\lambda$.  

  \noindent{\it Estimation of  $| \mathscr{I}_7 |$}: 
 Because \eqref{p0-bound2} holds only for $\bfm\notin N_{bad}(\be_{I,\nu})$, we need to expand $\mathscr{I}_7$ out and bound terms separately.
 Let 
 \[I^{\bfm,\bn}_7 \equiv e ^{C \lambda^{-1} |\bfm-\bn|} \int |V_0(\by)| p_0^\lambda(\by-\be_{B,1}-\bfm\vec\bv)\ p_0^\lambda(\by-\be_{B,1}-\bn\vec\bv)\ d\by.\]
 Then
 \begin{align*}
  \mathscr{I}_7 = \sum_{\bfm,\bn\in\Z^2} I^{\bfm,\bn}_7 &=
   \left(\sum_{\substack{\bfm,\bn \notin N_{bad}(\be_{B,1})}} + \sum_{\substack{\bfm \notin N_{bad}(\be_{B,1}) \\ \bn\in N_{bad}(\be_{B,1})}} + 
   \sum_{\substack{\bn \notin N_{bad}(\be_{B,1}) \\ \bfm \in N_{bad}(\be_{B,1})}}    + \sum_{\substack{\bfm,\bn \in N_{bad}(\be_{B,1})}} \right) I^{\bfm,\bn}_7 \\
   & \equiv \mathscr{I}_{7,A} + \mathscr{I}_{7,B} + \mathscr{I}_{7,C} + \mathscr{I}_{7,D} .
 \end{align*}
 We estimate $\mathscr{I}_{7,A}, \mathscr{I}_{7,B}, \mathscr{I}_{7,C}, \mathscr{I}_{7,D}$ separately. 
 
 \noindent{\it Estimation of  $| \mathscr{I}_{7,A} |$}: 
 For $\bfm,\bn \notin N_{bad}(\be_{B,1})$, we apply \eqref{p0-bound2} and \eqref{p0-bound3} to obtain
  \begin{align*}
 & \int |V_0(\by)| p_0^\lambda(\by-\be_{B,1}-\bfm\vec\bv)\ p_0^\lambda(\by-\be_{B,1}-\bn\vec\bv)\ d\by \nn\\ 
 &\ \lesssim \ e^{-c\lambda|\bfm|}\times e^{-c\lambda|\bn|}\times \int |V_0(\by)|\ p_0^\lambda(\by)\ p_0^\lambda(\by-\be_{B,1}) d\by
 \lesssim \ e^{-c\lambda|\bfm|}\times e^{-c\lambda|\bn|}\times\ \rho_\lambda .
 \end{align*}
 Multiplying by $e^{C\lambda^{-1}|\bn-\bfm|}$ and summing over all $\bfm,\bn \notin N_{bad}(\be_{B,1})$, we find that
 $| \mathscr{I}_{7,A} | \lesssim e^{-\tilde{c}\lambda}\times \rho_\lambda$.  
 
 \noindent{\it Estimation of  $| \mathscr{I}_{7,B} |$}: 
 Suppose $\bfm \notin N_{bad}(\be_{B,1})$, $\bn \in N_{bad}(\be_{B,1})$, and $|\by| \leq r_0$.
 Then, by the definition of $N_{bad}(\be_{B,1})$, $\be_{B,1} + \bn\vec\bv = \be_{B,\nu}$ for some $\nu\in\{1,2,3\}$.
 By \eqref{p0-bound2} and \eqref{p0-bound3}, we have
   \begin{align*}
 & \int |V_0(\by)| p_0^\lambda(\by-\be_{B,1}-\bfm\vec\bv)\ p_0^\lambda(\by-\be_{B,1}-\bn\vec\bv)\ d\by \nn\\ 
 &\ \lesssim \ e^{-c\lambda|\bfm|} \times \int |V_0(\by)|\ p_0^\lambda(\by) \ p_0^\lambda(\by-\be_{B,\nu}) d\by \\
 &\ = \ e^{-c\lambda|\bfm|} \times \lambda^{-2} \times\ \rho_\lambda 
 \lesssim \ e^{-c'\lambda|\bfm|} \times \rho_\lambda ,
 \end{align*}
 where in the final line we have used that $\rho_\lambda$ is independent of $\nu$; see Remark \ref{rho-indep}.
 Multiplying by $e^{C\lambda^{-1}|\bn-\bfm|}$ and summing over all $\bfm \notin N_{bad}(\be_{B,1})$ and $\bn \in N_{bad}(\be_{B,1})$, we find that
 $| \mathscr{I}_{7,B} | \lesssim e^{-\tilde{c}\lambda}\times \rho_\lambda$.  
 
 \noindent{\it Estimation of  $| \mathscr{I}_{7,C} |$}: 
 Note that $\mathscr{I}_{7,C}$ and $\mathscr{I}_{7,B}$ are equal (just interchange dummy indices $\bfm$ and $\bn$). Therefore, 
 $| \mathscr{I}_{7,C} | \lesssim e^{-\tilde{c}\lambda}\times \rho_\lambda$.  
 
 \noindent{\it Estimation of  $| \mathscr{I}_{7,D} |$}: 
 Suppose $\bfm, \bn \in N_{bad}(\be_{B,1})$, and $|\by| \leq r_0$.
 Then, by the definition of $N_{bad}(\be_{B,1})$, $\be_{B,1} + \bn\vec\bv = \be_{B,\nu}$ and $\be_{B,1} + \bfm\vec\bv = \be_{B,\nu'}$ for some $\nu,\nu'\in\{1,2,3\}$.
 By \eqref{p0-bound3}, we have
   \begin{align*}
 & \int |V_0(\by)| p_0^\lambda(\by-\be_{B,1}-\bfm\vec\bv)\ p_0^\lambda(\by-\be_{B,1}-\bn\vec\bv)\ d\by \nn\\ 
 & = \int |V_0(\by)| p_0^\lambda(\by-\be_{B,\nu'})\ p_0^\lambda(\by-\be_{B,\nu})\ d\by \nn\\ 
 &\ \lesssim \ e^{-c\lambda} \times \int |V_0(\by)|\ p_0^\lambda(\by) \ p_0^\lambda(\by-\be_{B,\nu}) d\by
 \ = \ e^{-c\lambda} \times \lambda^{-2} \times\ \rho_\lambda \leq \ e^{-c'\lambda} \ \rho_\lambda .
 \end{align*}
 Multiplying by $e^{C\lambda^{-1}|\bn-\bfm|}$ and summing over all $\bfm, \bn \in N_{bad}(\be_{B,1})$, we find that
 $| \mathscr{I}_{7,D} | \lesssim e^{-\tilde{c}\lambda}\times \rho_\lambda$.  
 
 Combining our estimates for $\mathscr{I}_{7,A}, \mathscr{I}_{7,B}, \mathscr{I}_{7,C}, \mathscr{I}_{7,D}$, we see that 
 $| \mathscr{I}_{7} | \lesssim e^{-\tilde{c}\lambda}\times \rho_\lambda$, completing the estimation of $\mathscr{I}_{7}$. 
 Together, the above bounds on $\mathscr{I}_1,\dots,\mathscr{I}_7$ imply Proposition \ref{BA-AA-control}, from which  Propositions \ref{ME-BA1} and \ref{ME-AA1} follow.
 \end{proof}

\subsection{Completion of the proof of Proposition \ref{ME-IJ-higher-order}}
{\ }
%
%
By the Cauchy-Schwarz inequality and the resolvent bound of Lemma \ref{lem1-resolvent}, the bound 
 \eqref{higher-order-ME-est} will follow if we can prove
\[ \|\left(H^\lambda(\overline\bk)-\overline\Omega\right)p_{_{\overline\bk,J}}\| \times 
\|\left(H^\lambda(\bk)-\Omega\right)p_{_{\bk,I}}\|\ \lesssim\ e^{-c\lambda}\ \rho_\lambda.
\]
for $I, J = A, B$. We prove that each factor on the left hand side is bounded by a $C\times e^{-c^\prime\lambda}\ \sqrt{\rho_\lambda}$. Both factors are bounded in the same manner; we focus on the second factor and prove
\begin{equation*}
 \|\left(H^\lambda(\bk)-\Omega\right)p_{_{\bk,I}}\|^2\ \lesssim\ e^{-c\lambda}\ \rho_\lambda\ . 
 \end{equation*}
By hypothesis we have $|\Omega|\le \widehat{C}\rho_\lambda$. Also,  by Proposition \ref{prop:rho-lam-bounds} we have
 $\rho_\lambda\lesssim e^{-c\lambda}$. Therefore, 
\begin{align}
\|\left(H^\lambda(\bk)-\Omega\right)p_{_{\bk,I}}\|^2\ &\lesssim\  \| H^\lambda(\bk)p_{_{\bk,I}}\|^2\ 
+\ |\Omega|^2\ \|p_{_{\bk,I}}\|^2\nn\\
 &\lesssim\ \| H^\lambda(\bk)p_{_{\bk,I}}\|^2\ +\ \rho_\lambda^2
 \lesssim 
 \| H^\lambda(\bk)p_{_{\bk,I}}\|^2\ +\ e^{-c\lambda}\ \rho_\lambda\ .\nn
\end{align}
Hence, it suffices to prove
\begin{equation}
 \| H^\lambda(\bk)p_{_{\bk,I}}\|\ \lesssim\ e^{-c\lambda}\ \sqrt{\rho_\lambda}\ ,\ I=A, B.
 \label{sqrt-bound}\end{equation}
 We consider the case $I=A$; the case $I=B$ is treated similarly. 

By \eqref{Hpkhatv3} we have
\begin{align*}
\| H^\lambda(\bk)p_{_{\bk,A}}\|\ &\le\ \mathscr{J}_1(\lambda) \ +\ \mathscr{J}_2(\lambda) , 
\end{align*}
where
\begin{align*}
\mathscr{J}_1(\lambda)\ &\equiv\ \sum_{\bw\in\Lambda_A\setminus\{\bv_A\}}\ \lambda^2\ \|V_0(\bx-\bv_A)\ p_\bk^\lambda(\bx-\bw)\|_{_{L^2(|\bx-\bv_A|<r_0)}} , \\
\mathscr{J}_2(\lambda)\ &\equiv\ \sum_{\bw\in\Lambda_A}\ \lambda^2\ \|V_0(\bx-\bv_B)\ p_\bk^\lambda(\bx-\bw)\|_{_{L^2(|\bx-\bv_B|<r_0)}} .
\end{align*}

We claim that $\mathscr{J}_1(\lambda)\lesssim e^{-c\lambda} \sqrt{\rho_\lambda}$ and $ \mathscr{J}_2(\lambda) \lesssim e^{-c\lambda} \sqrt{\rho_\lambda}$ .
We present the details of the bound on $\mathscr{J}_2(\lambda)$ and then remark on the bound for $\mathscr{J}_1(\lambda)$.

Partition $\Lambda_A$ into those points in $\Lambda_A$,   which are the nearest neighbors to $\bv_B$:
 $\bv_B+\be_{B,\nu}$, $\nu=1,2,3$ and those points in $\Lambda_A$ which are not nearest neighbors of $\bv_B$:
 $\bw=\bv_B+\be_{B,1}+\bn\vec\bv$, where $\bn\ne (0,0), (1,0), (0,1)$. 
 Therefore,
\begin{align*}
\mathscr{J}_2(\lambda)\ &=\ \sum_{\nu=1,2,3}\ 
\lambda^2\ \|V_0(\bx-\bv_B)\ p_\bk^\lambda(\bx-\bv_B-\be_{B,\nu})\|_{_{L^2(|\bx-\bv_B|<r_0)}}
\nn\\
&+\ \sum_{\bn\ne (0,0), (1,0), (0,1)}\ \lambda^2\ \|V_0(\bx-\bv_B)\ p_\bk^\lambda(\bx-\bv_B-\be_{B,1}-\bn\vec\bv)\|_{_{L^2(|\bx-\bv_B|<r_0)}}\nn\\
&\equiv\ \mathscr{J}_{2a}(\lambda)\ +\ \mathscr{J}_{2b}(\lambda).
\end{align*}

\nit {\it Bound on $\mathscr{J}_{2a}(\lambda)$:} By symmetry, all three terms in the sum are equal.
Changing variables and using the bound \eqref{p0-bound3}  we find
 \begin{align*}
\mathscr{J}_{2a}(\lambda) &=  3\lambda^2 \left( \int_{|\by|<r_0} |V_0(\by)|^2  \left( p_\bk^\lambda(\by+\be_{A,1}) \right)^2 d\by \right)^{\frac{1}{2}}\nn\\
&\lesssim \|V_0\|_{_\infty}^{\frac{1}{2}}   \lambda^2 \left( \int_{|\by|<r_0} |V_0(\by)|  \left( p_0^\lambda(\by+\be_{A,1}) \right)^2 d\by \right)^{\frac{1}{2}}\nn\\
&\lesssim \|V_0\|_{_\infty}^{\frac{1}{2}}  e^{-c\lambda} \lambda^2 \left( \int_{|\by|<r_0} |V_0(\by)|  p_0^\lambda(\by+\be_{A,1})  p_0^\lambda(\by) d\by \right)^{\frac{1}{2}}
\lesssim e^{-c\lambda} \sqrt{\rho_\lambda} .
 \end{align*}
 
 \nit {\it Bound on $\mathscr{J}_{2b}(\lambda)$:} Consider the general term in this infinite sum over points in $\Z^2$
 except $(0,0), (1,0)$ and  $(0,1)$. Changing variables and estimating, using Lemma \ref{lem:p0-lam-bounds},
 we obtain:
 \begin{align*}
&\lambda^2 \left( \int_{|\by|<r_0} |V_0(\by)|^2 |p_\bk^\lambda(\by-\be_{B,1}-\bn\vec\bv)|^2 d\by
 \right)^{\frac{1}{2}}\nn\\
 &\lesssim  e^{c|\bn|\lambda^{-1}}  \|V_0\|_{_\infty}^{\frac{1}{2}} \lambda^2
 \left( \int_{|\by|<r_0} |V_0(\by)| \left( p_0^\lambda(\by-\be_{B,1}-\bn\vec\bv) \right)^2 d\by
 \right)^{\frac{1}{2}}\nn\\
&  \lesssim   e^{-c^\prime\lambda |\bn|} e^{c|\bn|\lambda^{-1}} \|V_0\|_{_\infty}^{\frac{1}{2}} \lambda^2
 \left( \int_{|\by|<r_0} |V_0(\by)| p_0^\lambda(\by-\be_{B,1}) p_0^\lambda(\by)  d\by
\right)^{\frac{1}{2}}
 \lesssim e^{-c |\bn| \lambda} \sqrt{\rho_\lambda} .
 \end{align*} 
 Summing over admissible $\bn$ yields the bound $\mathscr{J}_{2b}(\lambda)\lesssim e^{-c \lambda} \sqrt{\rho_\lambda}$.  The bound $|\mathscr{J}_1(\lambda)|\lesssim e^{-c \lambda}\ \sqrt{\rho_\lambda}$ is proved in a manner similar to the bound on $\mathscr{J}_{2b}(\lambda)$,  making use of 
 \eqref{p0-bound1} and \eqref{p0-bound4}.  This completes the proof of Proposition \ref{ME-IJ-higher-order}.

\subsection{Proof of Proposition \ref{prop:rho-lam-bounds}}\label{proof-rho-lam-bounds} 

We first prove the upper bound in \eqref{rho-lambda-bounds}.
 From \eqref{p0-bound3} of Lemma \ref{lem:p0-lam-bounds} we have 
 $p_0^\lambda(\by+\be_{A,1})\lesssim e^{-c\lambda} p_0^\lambda(\by)$ for $\by\in {\rm supp}\ V_0$. Thus, 
  \[ \rho_\lambda \lesssim \|V_0\|_{_{\infty}}\ \lambda^2\ e^{-c\lambda}\ \int \left(  p_0^\lambda(\by) \right)^2 d\by\ 
   = \|V_0\|_{_{\infty}}\ \lambda^2\ e^{-c\lambda}\ \lesssim\ C_2\ e^{-c_2\lambda} .\]

To prove the lower bound in \eqref{rho-lambda-bounds}, we first use \eqref{gs2} to obtain 
\begin{equation*}
p_0^\lambda(\by+\be_{A,1})\ =\ 
\int\ \mathcal{K}_\lambda\left(\by+\be_{A,1}-\bz\right)\ \lambda^2 |V_0(\bz)|\ p^\lambda_0(\bz)\ d\bz\ .
\end{equation*}
Substitution into \eqref{rho-lam-def} yields
\begin{equation}
\rho_\lambda \ \equiv\ \int\ d\by\ \int\ d\bz\  \lambda^4\ |V_0(\by)|\ |V_0(\bz)|\ p_0^\lambda(\by)\ p_0^\lambda(\bz)\ \mathcal{K}_\lambda(\by+\be_{A,1}-\bz)\ .
\label{rho-lam-double}
\end{equation} 

Recall that the support of $V_0$ is contained in the $B({\bf 0},r_0)$, the disc of radius $r_0$ about the origin. 
Note that $|\by+\be_{A,1}-\bz|\le C_1$ for all 
 $|\bz|\le r_0$ and $|\by|\le r_0$, and therefore from \eqref{Krep} we have $ \mathcal{K}_\lambda(\by+\be_{A,1}-\bz)\ge e^{-C_1\lambda}$
 %
%
 and therefore, by \eqref{rho-lam-double}, we have the lower bound
 \begin{align}
 \rho_\lambda &\ge\ 
C\ e^{-c\lambda}\  \int_{|\by|<r_0}\ d\by\ \int_{|\bz|<r_0}\ d\bz\ \ \lambda^4\ |V_0(\by)|\ |V_0(\bz)|\ p_0^\lambda(\by)\ p_0^\lambda(\bz)\nn\\
&\ge\   C \lambda^2 \ e^{-c\lambda}\ \left(\ \lambda \int_{|\by| < r_0}\ |V_0(\by)|\ p_0^\lambda(\by)\ d\by\ \right)^2\ .
\label{lower1}\end{align}

Note that  by \eqref{gs2}
\begin{equation*}
p_0^\lambda(\by)\ =\  \lambda^2\ \left(\mathcal{K}_\lambda\star |V_0|\ p_0^\lambda \right)(\by)\ .
\end{equation*}
By Proposition \ref{K0-properties}, $\mathcal{K}\in L^2$ and therefore 
$\|\mathcal{K}_\lambda\|_{_{L^2}}= \|\mathcal{K}\|_{_{L^2}} \times |E_0^\lambda|^{-1/2}$.
Taking the $L^2(\R^2)$ norm and estimating using Young{'}s inequality gives  
\begin{align*}
1 =\ \|p_0^\lambda\|_{_{L^2}}\ &\le\ \lambda^2\ \|\mathcal{K}_\lambda\|_{_{L^2}}\ \int |V_0(\by)|\ p_0^\lambda(\by)\ d\by \\
\ &=\ \left(\ \lambda^2\ /\ |E_0^\lambda|^{1/2}\ \right)\ \|\mathcal{K}\|_{_{L^2}} \ \int |V_0(\by)|\ p_0^\lambda(\by)\ d\by\ .
\end{align*}
Recalling the lower bound on the ground state \eqref{GS},  $E_0^\lambda\le -c_0\lambda^2$,
 for some positive constant $c_0$, we have 
\begin{equation}
 \frac{c_0^{1/2} }{\|\mathcal{K}\|_{_{L^2}}}\ \lesssim \lambda \int |V_0(\by)|\ p_0^\lambda(\by)\ d\by. \label{lower2}
 \end{equation}
Substituting the lower bound \eqref{lower2} into \eqref{lower1}, we find 
\[
\rho_\lambda\ge \frac{C\ c_0}{ \|\mathcal{K}\|_{_{L^2}}^2} \lambda^{2} e^{-c\lambda}\ge C^\prime e^{-c^\prime\lambda}.\]

  This completes the proof of  Proposition \ref{prop:rho-lam-bounds}
   and therewith the last details of 
the proof of the main theorem, Theorem \ref{main-theorem}.

\section{Scaled convergence of the resolvent}\label{resolvent}

In this section we indicate how our analysis of the scaled convergence of dispersion surfaces of $H^\lambda=-\Delta+\lambda^2 V(\bx)-E_D^\lambda$ can be used  to obtain results on the scaled convergence of the resolvent. Introduce the scaled operator:
\begin{equation}
\widetilde{H}^\lambda\ \equiv (\rho^\lambda)^{-1}\ H^\lambda\ ,
\label{Hscaled}
\end{equation}
whose two lowest energy dispersion surfaces are uniformly close those of the tight-binding model. 
Introduce the restriction of $\widetilde{H}_\bk^\lambda$ to $\bk-$ pseudo-periodic functions. Since $H^\lambda$ commutes with lattice ($\Lambda_h$) translations,  $\widetilde{H}_\bk^\lambda\ =\  \widetilde{H}^\lambda\Big|_{_{L^2_\bk}}$ maps $H^2_\bk(\R^2/\Lambda_h)$ into $L^2_\bk(\R^2/\Lambda_h)$. 


For each $\bk\in\brill$,  let
 $\mathscr{Q}^\lambda_{AB,\bk}:L^2_\bk\to L^2_\bk$ denote the orthogonal projection onto
   the span of the two states: $P_{\bk,A}^\lambda(\bx)$ and $ P_{\bk,B}^\lambda(\bx)$, defined in \eqref{P-bk-I}, and 
   let $\mathscr{P}^\lambda_{AB,\bk}\equiv I-\mathscr{Q}^\lambda_{AB,\bk}:L^2_\bk\to L^2_\bk$ denote the projection onto its orthogonal complement.
 
 \nit By Lemma \ref{lem1-resolvent}, for fixed $z\in\C\setminus\R$ and any $\bk\in\brill_h$: 
  \begin{equation}
 \|\ \mathscr{P}^\lambda_{AB,\bk}\ \left(H^\lambda_\bk\ -\ \rho^\lambda zI\right)^{-1}\ \mathscr{P}^\lambda_{AB,\bk}\ \|_{_{L^2_\bk\to L^2_\bk}}\ 
 \lesssim\ 1\ ,
 \label{orthog-bound1}
 \end{equation}
where  $e^{-c_1\lambda}\ \lesssim\ \rho_\lambda\ \lesssim\ \ e^{-c_2\lambda}$; see \eqref{rho-lambda-bounds}. 

Represent $f\in L^2_\bk$ by $f(\bx)=\alpha_A\ P_{\bk,A}^\lambda\ +\ \alpha_B\ P_{\bk,B}^\lambda\ +\ f_\perp$,
 where $\alpha_A, \alpha_B\in\C$ and $f_\perp\in {\rm Range}(\mathscr{P}_{AB,\bk}^\lambda)$.
 Define the map $J_\bk: L^2_\bk\to \C\oplus\C\oplus {\rm Range}(\mathscr{P}_{AB,\bk}^\lambda)$ by:
 \begin{equation}
  J_\bk:\ f\mapsto \begin{pmatrix} \alpha_A[f] \\ \alpha_B[f] \\ f_\perp \end{pmatrix}\ =\ 
   \begin{pmatrix} \left\langle P_{\bk,A}^\lambda , f\right\rangle +\mathcal{O}(e^{-c\lambda}\|f\|)\\ \left\langle P_{\bk,B}^\lambda , f\right\rangle+\mathcal{O}(e^{-c\lambda}\|f\|) \\ f_\perp \end{pmatrix} .
 \label{Jk-def}
 \end{equation}
 The equality in \eqref{Jk-def} holds  since $\left\langle P_{\bk,I}^\lambda, P_{\bk,J}^\lambda\right\rangle = \delta_{_{IJ}}+\mathcal{O}(e^{-c\lambda})$ for $I,J=\{A,B\}$. 

  We use the notation  $\mathcal{O}_{_{X\to Y}}(a)$ to denote an operator
 from $X\to Y$ with norm $\lesssim a$, 
  $\mathcal{O}_{_{X}}(a)$
to  denote a function whose $X-$ norm is $\lesssim\ a$.

\begin{proposition}\label{prop:res-expand}
 For fixed $z\in\C\setminus\R$, $\bk\in\brill_h$ and $\lambda>\lambda_\star$ sufficiently large,   
 \begin{align}
& J^\lambda_\bk \left(\widetilde{H}^\lambda_\bk\ -\ zI\right)^{-1}(J_\bk^\lambda)^*\ \begin{bmatrix} \alpha_A\\ \alpha_B \\ f_\perp \end{bmatrix}\nn\\
 &= 
 \left[
    \begin{array}{c;{2pt/2pt}r}
    \mbox{$(H_{_{\rm TB}}(\bk)-zI_{_{2\times2}})^{-1}$} & \begin{matrix} 0 \\ 0 \end{matrix} \\ \hdashline[2pt/2pt]
    \begin{matrix} 0 & 0 \end{matrix} & 0
    \end{array}
    \right]\ \begin{bmatrix} \alpha_A\\ \alpha_B \\ f_\perp \end{bmatrix}
+\ 
 \mathcal{O}_{_{\C\oplus \C\oplus L^2_\bk}}(e^{-c\lambda})\ \begin{bmatrix} \alpha_A\\ \alpha_B \\ f_\perp \end{bmatrix}, 
\label{res-expand}
 \end{align}
where, $H_{_{\rm TB}}(\bk)$ is displayed in \eqref{HTBa}-\eqref{HTBb}. The error term in \eqref{res-expand} is uniform in $\bk\in\brill_h$. 

Consequently,  for $\lambda>\lambda_\star$
\begin{align}
& \left(\widetilde{H}^\lambda_\bk - zI\right)^{-1}
-
(J^\lambda_\bk)^*  
\left[
    \begin{array}{c;{2pt/2pt}r}
    \mbox{$(H_{_{\rm TB}}(\bk)-zI_{_{2\times2}})^{-1}$} & \begin{matrix} 0 \\ 0 \end{matrix} \\ \hdashline[2pt/2pt]
    \begin{matrix} 0 & 0 \end{matrix} & 0
    \end{array}
    \right]
  J_\bk^\lambda = 
 \mathcal{O}_{_{L^2_\bk\to L^2_\bk}}(e^{-c\lambda}) ,
\label{res-expand1}
 \end{align}
 uniformly in $\bk\in\brill_h$.
 \end{proposition}
The assertions of Theorem \ref{res-conv} follow from this proposition.
 
 To complete the proof of Theorem \ref{res-conv} we now prove Proposition \ref{prop:res-expand}. 
 
For $f=\alpha_A P_{\bk,A}^\lambda+\alpha_B P_{\bk,B}^\lambda + f_\perp\in L^2_\bk$, let 
\begin{align}
\begin{bmatrix} \xi_A \\  \xi_B \\ \xi_\perp \end{bmatrix}\ =\ J_\bk^\lambda\ \left(\widetilde{H}^\lambda_\bk\ -\ zI\right)^{-1}\
(J_\bk^\lambda)^*\ \begin{bmatrix} \alpha_A \\  \alpha_B \\ f_\perp \end{bmatrix} .
\end{align}
We now calculate $\xi_A$, $\xi_B$ and  $\xi_\perp$ in terms of $\alpha_A$, $\alpha_B$ and $f_\perp$.  By definition 
\[ \left(\widetilde{H}^\lambda_\bk\ -\ zI\right)^{-1}\left(\ \sum_{J=A,B}\alpha_J\ P_{\bk,J}^\lambda+f_\perp\right),
= \sum_{J=A,B} \xi_J\ P_{\bk,J}^\lambda\ +\ \xi_\perp\]
or equivalently
\begin{equation}
 \sum_{J=A,B}\ \alpha_J\ P_{\bk,J}^\lambda\ +\ f_\perp
= \sum_{J=A,B}\ \xi_J\ \left(\widetilde{H}^\lambda_\bk\ -\ zI\right)P_{\bk,J}^\lambda\ +\ \left(\widetilde{H}^\lambda_\bk\ -\ zI\right)\xi_\perp.
\label{hh}
\end{equation}
Next, apply the orthogonal projection $\rho^\lambda\times \mathscr{P}_{_{AB}}^\lambda$ to obtain
\[ \rho^\lambda f_\perp
=\ \sum_{J=A,B}\ \xi_J\ \mathscr{P}_{_{AB}}^\lambda\ H^\lambda_\bk\ P_{\bk,J}^\lambda\ +\ \mathscr{P}_{_{AB}}^\lambda\left(H^\lambda_\bk\ -\ \rho^\lambda zI\right)\mathscr{P}_{_{AB}}^\lambda\ \xi_\perp,
\]
where we have used $\rho^\lambda \widetilde{H}^\lambda_\bk= H^\lambda_\bk$ and $\mathscr{P}_{_{AB}}^\lambda 
P_{\bk,A}^\lambda=\mathscr{P}_{_{AB}}^\lambda 
P_{\bk,B}^\lambda=0$. 

Applying $\mathscr{P}_{_{AB}}^\lambda\left(H^\lambda_\bk\ -\ zI\right)^{-1}\mathscr{P}_{_{AB}}^\lambda$ and rearranging
\begin{align}
\xi_\perp & =  - \sum_{J=A,B} \xi_J \mathscr{P}_{_{AB}}^\lambda\left(H^\lambda_\bk - \rho^\lambda zI\right)^{-1}\mathscr{P}_{_{AB}}^\lambda H^\lambda_\bk P_{\bk,J}^\lambda  
 + \rho^\lambda \mathscr{P}_{_{AB}}^\lambda\left(H^\lambda_\bk - \rho^\lambda zI\right)^{-1}\mathscr{P}_{_{AB}}^\lambda f_\perp .
\nn  \end{align}

By the bound \eqref{orthog-bound1}, we obtain 
\begin{align}
\xi_\perp\ =\  \sum_{J=A,B}\ \xi_J\times  \mathcal{O}_{_{L^2_\bk\to L^2_\bk}}(1) H^\lambda_\bk\ P_{\bk,J}^\lambda\ +\
  \rho^\lambda\times  \mathcal{O}_{_{L^2_\bk\to L^2_\bk}}(1) f_\perp.
\label{xiperp}\end{align}
Furthermore, by the bound \eqref{sqrt-bound} $\|H^\lambda_\bk\ P_{\bk,B}^\lambda\|\lesssim \sqrt{\rho^{\lambda}}\ e^{-c\lambda}$.
 Therefore,
\begin{equation}
\xi_\perp\ =\  \sum_{J=A,B}\ \xi_J\times  \mathcal{O}_{_{L^2_\bk}}(\sqrt{\rho^{\lambda}}\ e^{-c\lambda})\ +\
  \rho^\lambda\times  \mathcal{O}_{_{L^2_\bk\to L^2_\bk}}(1) f_\perp.
\label{xi-perp-bound}
\end{equation}

Next, take the inner product of \eqref{hh} with $P_{\bk,M}^\lambda$, $M=A, B$ and obtain, using that 
$\left\langle P_{\bk,M}^\lambda, f_\perp\right\rangle=0$, $\widetilde{H}^\lambda_\bk=H^\lambda_\bk/\rho^\lambda$ and self-adjointness of $H^\lambda_\bk$:
\begin{align}
& \sum_{J=A,B}\ \alpha_J \left\langle P_{\bk,M}^\lambda, P_{\bk,J}^\lambda\right\rangle \nn \\
&\qquad =\ \sum_{J=A,B}\ \xi_J\ 
\left\langle P_{\bk,M}^\lambda, \left(\widetilde{H}^\lambda_\bk\ -\ zI\right)P_{\bk,J}^\lambda\right\rangle\ 
 +\ (\rho^\lambda)^{-1}\ \left\langle H^\lambda_\bk\ P_{\bk,M}^\lambda, \xi_\perp\right\rangle\ .
\label{hh2}\end{align}

Bounding the latter term in \eqref{hh2}, we find:
\begin{align*}
& \left| (\rho^\lambda)^{-1}\ \left\langle H^\lambda_\bk\ P_{\bk,M}^\lambda, \xi_\perp\right\rangle\right|
\le\ (\rho^\lambda)^{-1}\ \|H^\lambda_\bk P_{\bk,M}^\lambda\|\ \|\xi_\perp\|\nn\\
&\lesssim  (\rho^\lambda)^{-1}\times \sqrt{\rho^{\lambda}}\ e^{-c\lambda}\ \left(\sqrt{\rho^{\lambda}}\ e^{-c\lambda}\ (\ |\xi_A| + |\xi_B| \ )\ +\ \rho^\lambda\ \|f_\perp\|\right)\nn\\
&\lesssim\ e^{-2c\lambda}\ (\ |\xi_A| + |\xi_B| \ )\ +\ e^{-c\lambda}\ \sqrt{\rho^\lambda}\ \|f_\perp\|\ .
\end{align*}

Thus, using the expansion of matrix elements in Proposition \ref{ME} we have
\begin{align}
&\left[ I_{_{2\times2}} + \mathcal{O}(e^{-c\lambda}) \right] \begin{pmatrix}\alpha_A\ \alpha_B\end{pmatrix}
\nn\\
&\qquad  =  \left[ \left\langle P_{\bk,M}^\lambda, \left(\widetilde{H}^\lambda_\bk - zI\right)P_{\bk,J}^\lambda\right\rangle + \mathcal{O}(e^{-c\lambda})\
\right]_{_{M,J=A,B}} \begin{pmatrix}\xi_A\ \xi_B\end{pmatrix} + \mathcal{O}_{_{L^2_\bk\to\C^2}}\left( e^{-c\lambda} \sqrt{\rho^\lambda}\right) f_\perp\nn\\
&\qquad = \left[  H_{_{\rm TB}}(\bk)-zI_{_{2\times2}} + \mathcal{O}(e^{-c\lambda})\
\right]_{_{M,J=A,B}} \begin{pmatrix}\xi_A\ \xi_B\end{pmatrix} + \mathcal{O}_{_{L^2_\bk\to\C^2}}\left( e^{-c\lambda} \sqrt{\rho^\lambda}\right) f_\perp ,\nn\\
&\label{hh3}
\end{align}
where $H_{_{\rm TB}}(\bk)$ is displayed in \eqref{HTBa}-\eqref{HTBb}.
Proposition \ref{res-expand} now follows from \eqref{xi-perp-bound} and \eqref{hh3}.

\section{Remarks on dependencies of constants}\label{sec:constants}

By hypothesis {\bf (GS)}, \eqref{GS},   the one-atom ground state energy, $E_0^\lambda$, satisfies the bounds
\begin{equation} 
-\|V_0\|_{_{L^\infty}}\ \lambda^2\ \le\ E_0^\lambda \ \le\ -C \lambda^2, 
\label{GS1}\end{equation}
for some constant $C=C(V_0)>0$.

By {\bf (EG)},  the assumed {\bf Energy Gap Property}, \eqref{EG}, we have that for $\psi\in H^2(\R^2)$ orthogonal to the ground state of $-\Delta+\lambda^2V_0$ in $L^2(\R^2)$,
\begin{equation}
\left\langle\ \left(-\Delta + \lambda^2 V_0\right)\psi,\psi\ \right\rangle_{_{L^2(\R^2)}}\ \ge\ (E_0^\lambda+c_{gap})\ \|\psi\|_{_{L^2(\R^2)}}^2,\label{EG1}
\end{equation}
for a positive constant $c_{gap}$. 

Moreover, we assumed that our atomic potential, $V_0(\bx)$, is supported in a disc of radius, $r_0$, where $0<r_0<r_{critical}$, where $r_{critical}$ is a universal constant satisfying the bounds \eqref{r-critical-bound} of 
 Geometric Lemma \ref{euclid1}. 

Our main result, Theorem \ref{main-theorem}, and the proof of Corollaries \ref{spectral-gaps} and  \ref{rational-edge-states},   concern the behavior of the rescaled dispersion functions 
$\rho_\lambda^{-1}(E_\pm^\lambda(\bk)-E_D^\lambda)$, or equivalently, the behavior of $\mu_\pm^\lambda(\bk)$,  and their derivatives up to order $\beta_{max}$, with $\beta_{max}$ as large as we please. 
 
By going carefully through the arguments in this paper, one can check that all constants that appear, including those in Theorem \ref{main-theorem}, {\it e.g.} $\lambda_\star$ and $c_{\star\star}$, depend only on $\beta_{max}$, on $C$ in \eqref{GS1}, on $c_{gap}$ in \eqref{EG1}, and on a lower bound for $r_{critical}-r_0$. The sole exception is in Theorem \ref{res-conv}, where the constants depend on $z\in\C\setminus\R$ as well. 
This allows us to treat atomic potentials $V_1(\bx)$ not explicitly given in the form $\lambda^2V_0(\bx)$; 
we simply define $\lambda=\|V_1\|^{\frac{1}{2}}_{_{\infty}}$ and set $V_0=\lambda^{-2}V_1$.

\bibliographystyle{amsplain}
\bibliography{strong-binding}

\end{document}